\documentclass[12pt,a4paper,twoside,openright]{report}

\usepackage{graphicx}
\usepackage{subfigure}
\usepackage{amsmath}
\usepackage{amssymb}
\usepackage{algorithmic}
\usepackage{algorithm}
\usepackage{amsthm}
\usepackage{fancyhdr}
\usepackage{url}
\usepackage{longtable}

\newtheorem{theorem}{Theorem}[section]

\newtheorem{lemma}[theorem]{Lemma}
\newtheorem{proposition}[theorem]{Proposition}
\newtheorem{remark}[theorem]{Remark}

\renewcommand{\chaptermark}[1]{\markboth{\MakeUppercase{\chaptername\ \thechapter.\ #1}}{}}
\renewcommand{\sectionmark}[1]{\markright{\thesection.\ #1}}
\begin{document}

\pagenumbering{roman}
\pagestyle{empty}

\renewcommand{\chaptermark}[1] {\markboth{{\chaptername\ \thechapter.\ #1}}{}}
\renewcommand{\sectionmark}[1] {\markright{{\thesection.\ #1}}}
\lhead[\small \thepage]{\small \rightmark}
\rhead[\small \leftmark]{\small \thepage}
\lhead[\small \thepage]{\small \rightmark}
\rhead[\small \leftmark]{\small \thepage}

\renewcommand{\baselinestretch}{1}\Large\normalsize
\thispagestyle{empty}

\newcommand{\vs}{\vspace{1ex}}

\newcommand{\hos}{\hspace{5 in}}

\vskip 0.3in

\begin{center}
{\bf \LARGE 
On Wireless Link Scheduling and \\
\vspace{1ex}
Flow Control
} \\ 
\vspace{10ex}
\large
A thesis submitted in partial fulfillment of \\
the requirements for the degree of \\
\vspace{5ex}
Doctor of Philosophy \\
\vspace{5ex}
by \\
\vspace{5ex}
{\bf Ashutosh Deepak Gore} \\ 
\vspace{1ex}
(Roll number: 02407007) \\
\vspace{10ex} 
Advisor: Prof. Abhay Karandikar \\ 
\vspace{10ex}
\includegraphics{iitblogo.epsi} \\
\vspace{5ex}
Department of Electrical Engineering,\\ 
Indian Institute of Technology Bombay,\\
Powai, Mumbai, 400076.\\ 
December 2008

\end{center}

\renewcommand{\baselinestretch}{1.5}\Large\normalsize

\chapter*{ }
\begin{center}
{\Large Indian Institute of Technology Bombay\\\vs Certificate of Course
Work}\vs
\end{center}
This is to certify that {\bf Ashutosh Deepak Gore} was admitted to the
candidacy of the Ph.D. degree in January 2003 after successfully
completing all the courses required for the Ph.D. degree
programme. The details of the course work done are given below:\\\\\\
\bigskip
\begin{center}
\begin{table}[h]
\centering
\begin{tabular}{|l|l|l|l|} \hline
Sr. No. & Course code & Course name & Credits \\ \hline
1 & EE708 & Information Theory and Coding & 6.00 \\ \hline
2 & MA402 & Algebra & 8.00 \\ \hline
3 & EES801 & Seminar & 4.00 \\ \hline
4 & EE659 & A First Course in Optimization & 6.00 \\ \hline
5 & EE621 & Markov Chains and Queueing Systems & 6.00 \\ \hline
6 & MA403 & Real Analysis & 8.00 \\ \hline
7 & HS699 & Communication Skills & 4.00 \\ \hline
\end{tabular}
\end{table}
\end{center}
\bigskip
IIT Bombay\\\\Date:........  \hspace{3 in} Deputy Registrar (Academic)

\chapter*{Acknowledgments}

I joined the Ph.D. programme at my alma mater with the intent of
honing my knowledge in networking and wireless communications.  In
retrospect, I feel that I have gained knowledge in many other domains
as well. This is primarily due to close interaction with intellectuals
(both faculty and students) at IIT Bombay.

A doctoral thesis can never be produced by the thoughts and actions of
a single person. Repeated technical discussions, mathematical workouts
and simulations are the major factors that contribute to the
``evolution'' of a thesis. In this space, I wish to explicitly thank
various individuals who have helped me during my doctoral adventure.

I would like to thank my exuberant advisor Prof.  Abhay Karandikar,
who has taught me engineering in the true sense of the word. His keen
insight into the nitti-gritty of every problem and his perfectionism
in technical documentation have significantly moulded my grey matter.
I will always remember his words ``A Ph.D.  thesis is a piece of
scholarly work. It is not a sequence of papers stapled together!'' I
have also sharpened my knowledge and pedagogy as a teaching assistant
in various courses taught by Prof. Karandikar.

I would like to express my gratitude to my research progress committee
members, namely, Prof. H. Narayanan, Prof. Harish Pillai and Prof.
Varsha Apte. They have provided valuable tips and guidance throughout
my research career. I would especially like to thank Prof. Narayanan
for encouraging me to pursue a Ph.D. at IIT Bombay. I also wish to
thank my Ph.D. thesis reviewers for their insightful comments which
helped to improve the quality of the final thesis.

I have closely interacted with many bright people at Information
Networks Laboratory, which has been my second home for the past six
years. In particular, I would like to thank my peers, Nitin Salodkar,
Hemant Rath and Punit Rathod, and my juniors, Mukul Agarwal and S.
Sundhar Ram, for many a discussion, both technical and non-technical.
I also wish to thank Srikanth Jagabathula and N. Praneeth Kumar, who
have been my collaborators in some of my work.

I wish to sincerely thank my wife Chaitali for her constant love and
support. Our wonderful baby girl, born on $9^{th}$ December 2008, has
infused a lot of energy in me over the past few weeks!  My brother
Hrishikesh and cousin sister Namrata have enthused me at various
stages of my doctoral journey.

My father, Deepak Keshav Gore, and my mother, Jayshree Deepak Gore,
had recognized my proclivity for mathematics right from my childhood.
They did not flinch a bit when I decided to tread the off-beaten track
towards a Ph.D. Their unconditional love, inspiration and ethics have
been the pillars of my motivation all along. This thesis is dedicated
to them.
\\\\\\
Ashutosh Deepak Gore\\
$26^{th}$ December 2008

\renewcommand{\baselinestretch}{1.4}\Large\normalsize
\chapter*{Abstract}

This thesis focuses on link scheduling in wireless mesh networks by
taking into account physical layer characteristics.  The assumption
made throughout is that a packet is received successfully only if the
Signal to Interference and Noise Ratio (SINR) at the receiver exceeds
a certain threshold, termed as communication threshold.  The thesis
also discusses the complementary problem of flow control.

First, we consider various problems on centralized link scheduling in
Spatial Time Division Multiple Access (STDMA) wireless mesh networks.
We motivate the use of spatial reuse as performance metric and provide
an explicit characterization of spatial reuse. We propose link
scheduling algorithms based on certain graph models (communication
graph, SINR graph) of the network.  Our algorithms achieve higher
spatial reuse than that of existing algorithms, with only a slight
increase in computational complexity.

Next, we investigate a related scenario involving link scheduling,
namely random access algorithms in wireless networks. We assume that
the receiver is capable of power-based capture and propose a splitting
algorithm that varies transmission powers of users on the basis of
quaternary channel feedback.  We model the algorithm dynamics by a
Discrete Time Markov Chain and consequently show that its maximum
stable throughput is 0.5518.  Our algorithm achieves higher maximum
stable throughput and significantly lower delay than the First Come
First Serve (FCFS) splitting algorithm with uniform transmission
power.

Finally, we consider the complementary problem of flow control in
packet networks from an information-theoretic perspective. We derive
the maximum entropy of a flow which conforms to traffic constraints
imposed by a generalized token bucket regulator, by taking into
account the covert information present in the randomness of packet
lengths. Our results demonstrate that the optimal generalized token
bucket regulator has a near uniform bucket depth sequence and a
decreasing token increment sequence.

\renewcommand{\baselinestretch}{1.5}\Large\normalsize

\pagestyle{plain}
\tableofcontents\newpage

\chapter*{List of Acronyms}
\addcontentsline{toc}{chapter}{List of Acronyms}

\begin{longtable}{ll}
3GPP LTE & $3^{rd}$ Generation Partnership Project Long Term Evolution \\
3GPP2 & $3^{rd}$ Generation Partnership Project 2 \\
ACK & Acknowledgment \\
ALS & ArboricalLinkSchedule \\
AWGN & Additive White Gaussian Noise \\
BS & Base Station \\
BS & BroadcastSchedule \\
BTA & Basic Tree Algorithm \\
CAA & Channel Access Algorithm \\
CDMA & Code Division Multiple Access \\
CFLS & ConflictFreeLinkSchedule \\
CRP & Collision Resolution Period \\
CSI & Channel State Information \\
CSMA/CA & Carrier Sense Multiple Access with Collision Avoidance \\
CSMA/CD & Carrier Sense Multiple Access with Collision Detection \\
CTS & Clear To Send \\
DTMC & Discrete Time Markov Chain \\
FCFC & FirstConflictFreeColor \\
FCFS & First Come First Served \\
FDMA & Frequency Division Multiple Access \\
FEC & Forward Error Correction \\
GP & GreedyPhysical \\
GTBR & Generalized Token Bucket Regulator \\
IETF & Internet Engineering Task Force \\
i.i.d. & independent and identically distributed \\
ISP & Internet Service Provider \\
LAN & Local Area Network \\
LMMSE & Linear Minimum Mean Square Error \\
OFDM & Orthogonal Frequency Division Multiplexing \\
PCFCFS & Power Controlled First Come First Served \\
MAC & Medium Access Control \\
MANET & Mobile Ad Hoc Network \\
MASC & MaxAverageSINRColor \\
MASS & MaxAverageSINRSchedule \\
MIMO & Multiple Input Multiple Output \\
MPR & MultiPacket Reception \\
MTA & Modified Tree Algorithm \\
NDMA & Network-Assisted Diversity Multiple Access \\
NP & Non-deterministic Polynomial time \\
pdf & probability density function \\
pmf & probability mass function \\
QoS & Quality of Service \\
RTS & Request To Send \\
SGLS & SINRGraphLinkSchedule \\
SINR & Signal to Interference and Noise Ratio \\
SLA & Service Level Agreement \\
SS & Subscriber Station \\
STBR & Standard Token Bucket Regulator \\
STDMA & Spatial Time Division Multiple Access \\
TBR & Token Bucket Regulator \\
TCP & Transmission Control Protocol \\
TDMA & Time Division Multiple Access \\
TGSA & Truncated Graph Based Scheduling Algorithm \\
VBR & Variable Bit Rate \\
WiMAX & Worldwide Interoperability for Microwave Access \\
WLAN & Wireless Local Area Network \\
WMAN & Wireless Metropolitan Area Network \\
WMN & Wireless Mesh Network
\end{longtable}

\chapter*{List of Symbols}
\addcontentsline{toc}{chapter}{List of Symbols}

\begin{longtable}{ll}
$N$ & number of nodes in STDMA wireless network \\
$(X_j,Y_j)$ & Cartesian coordinates of $j^{th}$ node \\
$(R_j,\Theta_j)$ & polar coordinates of $j^{th}$ node \\
$P$ & power with which a node transmits its packet \\
$N_0$ & thermal noise power spectral density \\
$\beta$ & path loss exponent \\
$D(j,k)$ & Euclidean distance between nodes $j$ and $k$ \\
$C$ & number of slots (colors) in STDMA link schedule \\
$\gamma_c$ & communication threshold \\
$\gamma_i$ & interference threshold \\
$R_c$ & communication range \\
$R_i$ & interference range \\
$\Phi(\cdot)$ & STDMA wireless network \\
$\mathcal V$ & set of vertices \\
$\mathcal E$ & set of directed edges \\
${\mathcal E}_c$ & set of communication edges \\
${\mathcal E}_i$ & set of interference edges \\
$\mathcal G_c(\mathcal V,\mathcal E_c)$ & 
  communication graph representation of STDMA network \\
$\mathcal G(\mathcal V,\mathcal E_c \cup \mathcal E_i)$ &
  two-tier graph representation of STDMA network \\
$\Psi(\cdot)$ & point to point link schedule for STDMA network \\
$\sigma$ & spatial reuse of point to point link schedule \\
$v$ & number of vertices in communication graph \\
$e$ & number of edges in communication graph \\
$\theta$ & thickness of communication graph \\
$\rho$ & maximum degree of any vertex \\
$t_{i,j}$ & index of $j^{th}$ transmitter in $i^{th}$ slot \\
$r_{i,j}$ & index of $j^{th}$ receiver in $i^{th}$ slot \\
$\mathcal S_i$ & set of transmissions in $i^{th}$
  slot of point to point link schedule \\
$M_i$ &  number of concurrent transmitters in $i^{th}$ slot \\
$\mbox{SINR}_{r_{i,j}}$ & SINR at receiver $r_{i,j}$ \\
$\mbox{SNR}_{r_{i,j}}$ & SNR at receiver $r_{i,j}$ \\
$G_c(\cdot)$ & undirected equivalent of communication graph \\
$I(\cdot)$ & indicator function \\
$v_i$ & $i^{th}$ vertex in communication or two-tier graph \\
$T_i$ & $i^{th}$ oriented graph \\
$C(x)$ & colour assigned to edge $x$ \\
$L(u)$ & label assigned to vertex $u$ \\
$\omega$ & maximum number of neighbors with lower labels \\
$\tau[k_1,k_2]$ & number of successful links from slot $k_1$ to slot $k_2$ \\
$\eta[k_1,k_2]$ & 
  number of successful links per time slot from slot $k_1$ to slot $k_2$ \\
$\mathcal G_r(\cdot)$ & residual subgraph of communication graph \\
$\mathcal C$ & set of existing colors \\
$\mathcal C_c$ & set of conflicting colors \\
$\mathcal C_1$ & set of colors with primary edge conflict \\
$\mathcal C_2$ & set of colors with secondary edge conflict \\
${\mathcal C}_{cf}$ & set of conflict-free colors \\
${\mathcal C}_{nc}$ & set of non-conflicting colors \\
$R$ & radius of circular deployment region \\
$V(\cdot)$ & fading channel gain \\
$W(\cdot)$ & shadowing channel gain measured in bels \\
$f_X(x)$ & probability density function of random variable $X$ \\
$t_j$ & $j^{th}$ transmitter in a given time slot \\
$r_j$ & $j^{th}$ receiver in a given time slot \\
$M$ & number of concurrent transmissions in a given time slot \\
$\mathcal V'$ & set of vertices of SINR graph \\
$\mathcal E'$ & set of directed edges of SINR graph \\
$\mathcal G'(\mathcal V',\mathcal E')$ 
  & SINR graph representation of STDMA network \\
$w_{ij}$ & interference weight function for edges $i,j \in \mathcal E_c$ \\
$w_{ij}'$ 
  & co-schedulability weight function for edges $i,j \in \mathcal E_c$ \\
$\mathcal N(v')$ & normalized noise power for vertex $v' \in \mathcal V'$ \\
$\mathcal V_{uc}'$ & set of uncolored vertices of $\mathcal V'$ \\
$\mathcal V_{c_p}'$ & set of vertices of $\mathcal V'$ colored with color $p$\\
$\mathcal C(v')$ & color assigned to vertex $v' \in \mathcal V'$ \\
$\mathcal E_t'$ & set of directed edges of truncated SINR graph \\
$\mathcal G_t'(\mathcal V',\mathcal E_t')$ & truncated SINR graph \\
$\mathcal V_{cc}'$ & set of co-colored vertices of $\mathcal V'$ \\
$\Omega(\cdot)$ & point to multipoint link schedule for STDMA network \\
$r_{i,j,k}$ & index of $k^{th}$ receiver of $j^{th}$ transmission in
  time slot $i$ \\ 
$\varsigma$ & spatial reuse of point to multipoint link schedule \\
$\mathcal B_i$ & set of transmissions in $i^{th}$
  slot of point to multipoint link schedule \\
$\mbox{SINR}_{r_{i,j,k}}$ & SINR at receiver $r_{i,j,k}$ \\
$\eta(j)$ & number of neighbors of node $j$ \\
$\mathcal C_p$ & set of colors with primary vertex conflict \\
$\mathcal C_s$ & set of colors with secondary vertex conflict \\
$\lambda$ & Poisson packet arrival rate \\
$\mathcal D$ & average packet delay \\
$\mathcal T$ & throughput \\
$T(k)$ & left endpoint of allocation interval for slot $k$ \\
$\phi(k)$ & length of allocation interval for slot $k$ of PCFCFS algorithm \\
$\phi_0$ & maximum size of allocation interval of PCFCFS algorithm \\
$a_i$ & arrival time of $i^{th}$ packet \\
$d_i$ & departure time of $i^{th}$ packet \\
$P_i(k)$ & transmission power of $i^{th}$ packet in slot $k$ \\
$P_1$ & nominal transmission power \\
$P_2$ & higher transmission power \\
$\mathcal L$ & left tag \\
$\mathcal R$ & right tag \\
$\sigma(k)$ & tag of allocation interval in slot $k$ \\
$L$ & left allocation interval \\
$R$ & right allocation interval \\
$\alpha_0$ & maximum size of allocation interval of FCFS algorithm \\
$G_i$ & expected number of packets in an interval split $i$ times \\
$P_{A_i,B_j}$ & transition probability from state $(A,i)$ to $(B,j)$ \\
$x_Z$ & number of packets in allocation interval $Z$ \\
$Q_{X_i}$ & probability of hitting state $(X,i)$ in a CRP \\
$K$ & random variable denoting number of slots in a CRP \\
$F$ & random variable denoting fraction of original allocation interval
returned \\
& to waiting interval \\
$U_{X_i}$ & probability that state $(X,i)$ has a collision or a capture \\
$D$ & expected change in time backlog \\
$\tau$ & number of slots for which algorithm operates \\
$n_{suc}$ & number of successful packets in $[0,\tau)$ \\
$r$ & token increment rate of STBR \\
$B$ & bucket depth (maximum burst size) of STBR \\
$S$ & number of slots of operation of TBR \\
$r_k$ & token increment of GTBR in slot $k$ \\
$B_k$ & bucket depth of GTBR in slot $k$ \\
$\ell_k$ & length of packet transmitted by GTBR in slot $k$ \\
$u_k$ & number of residual tokens of GTBR at start of slot $k$ \\
$\mathbf r$ & token increment sequence of GTBR \\
$\mathbf B$ & bucket depth sequence of GTBR \\
$\mathcal R_s(\cdot)$ & standard token bucket regulator \\
$\mathcal R_g(\cdot)$ & generalized token bucket regulator \\
$p_{\ell_k}(u_k)$ & probability of transmitting packet of length
 $\ell_k$ bits with $u_k$ residual tokens \\
$H_k(u_k)$ & flow entropy of GTBR in slot $k$ with $u_k$ residual tokens \\
$H_k^*(u_k)$ 
  & optimal flow entropy of GTBR in slot $k$ with $u_k$ residual tokens \\
$\mu_i$ & maximum number of tokens possible in slot $i$ \\
\end{longtable}

\listoftables
\addcontentsline{toc}{chapter}{List of Tables}
\listoffigures
\setcounter{page}{26} 
\addcontentsline{toc}{chapter}{List of Figures}

\clearpage

\pagenumbering{arabic}
\pagestyle{fancy}

\chapter{Introduction}
\label{ch:introduction}

\section{Link Scheduling in Wireless Networks}

Wireless and mobile communications have revolutionized the way we
communicate over the past decade. This impact has been felt both in
voice communications and wireless Internet access.  The
ever-increasing need for applications like video and images have
driven the need for technologies like $3^{rd}$ Generation Partnership
Project Long Term Evolution (3GPP LTE), $3^{rd}$ Generation
Partnership Project 2 (3GPP2), IEEE 802.16 Worldwide Interoperability
for Microwave Access (WiMAX) networks and IEEE 802.11 Wireless Local
Area Networks (WLANs) which promise broadband data rates to wireless
users.  This revolution in wireless communications has had a great
impact in India, where the number of cellular subscribers is 250
million (as of November 2008) and is growing at a rate of approximately
$3\%$ per month \cite{cellular_operators_india}.

Wireless networks can be broadly classified into cellular networks and
ad hoc networks. A wireless ad hoc network is a collection of wireless
nodes that can dynamically self-organize into an arbitrary topology to
form a network without necessarily using any pre-existing
infrastructure. Based on their application, ad hoc networks can be
further classified into Mobile Ad Hoc Networks (MANETs), wireless mesh
networks and wireless sensor networks.  A wireless mesh network can be
considered to be an infrastructure-based ad hoc network with a mesh
backbone carrying most of the traffic.

\begin{figure}[thbp]
  \centering
  \includegraphics[width=6in]{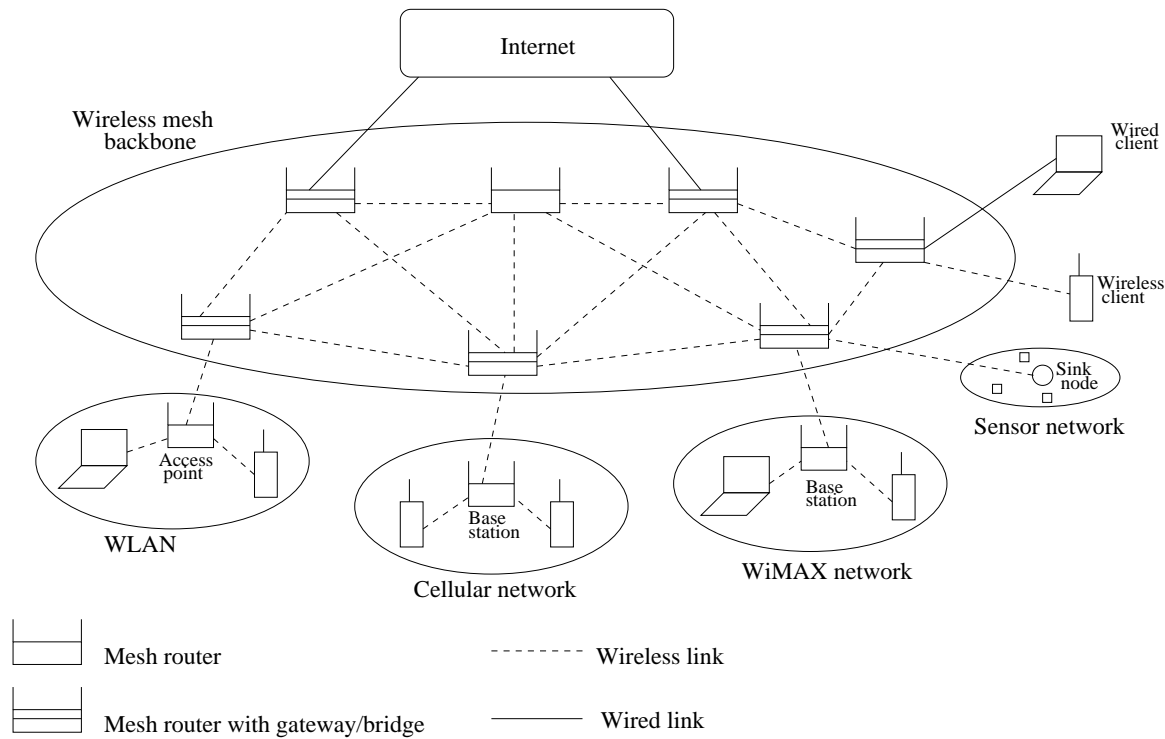}
  \caption{Wireless mesh network, adapted from
    \cite{akyildiz_wang__survey_wireless}.}
  \label{fig:wmn_mesh}
\end{figure}

Wireless Mesh Networks (WMNs) have been recently advocated to provide
connectivity and coverage, especially in sparsely populated and rural
areas.  For example, several Wireless Community Networks (WCNs) are
operational in Europe, Australia and USA
\cite{efstathiou_frangoudis__stimulating_participation}.  Peer to peer
wireless technology is also being developed by companies such as
\cite{terranet_p2p}.  WMNs are dynamically self-organized and
self-configured, with nodes in the network automatically establishing
an ad hoc network and maintaining mesh connectivity
\cite{akyildiz_wang__survey_wireless}.  An example of a WMN is shown
in Figure \ref{fig:wmn_mesh}. Typically, a WMN comprises of two types
of nodes: mesh routers and mesh clients. A mesh router consists of
gateway/bridge functions and the capability to support mesh
networking.  Mesh routers have little or no mobility and form a
wireless backbone for mesh clients. The gateway/bridge functionalities
in mesh routers aid in the integration of WMNs with heterogeneous
networks such as Ethernet \cite{metcalfe_boggs__ethernet_distributed},
cellular networks, WLANs \cite{ieee__wlan}, WiMAX networks
\cite{ieee__wimax} and sensor networks. WMNs are witnessing
commercialization in various applications like broadband home
networks, enterprise networks, community networks and metropolitan
area networks. Moreover, WMNs diversify the functionalities of ad hoc
networks, instead of just being another type of ad hoc network.  These
additional functionalities necessitate novel design principles and
efficient algorithms for the realization of WMNs.

Significant research efforts are required to realize the full
potential of WMNs. Among the many challenging issues in the design of
WMNs, the design of the physical as well as the Medium Access Control
(MAC) layers is important, especially from a perspective of achieving
high network throughput. At the physical layer, techniques like
adaptive modulation and coding, Orthogonal Frequency Division
Multiplexing (OFDM) \cite{cimini__analysis_simulation},
\cite{bahai_saltzberg_ergen__multicarrier_digital} and Multiple Input
Multiple Output (MIMO) techniques
\cite{tse_viswanath__fundamentals_wireless} can be used to increase
the capacity of a wireless channel and achieve high data transmission
rates. At the MAC layer, various solutions like directional antenna
based MAC \cite{choudhury_yang__designing_mac}, MAC with power control
\cite{jung_vaidya__power_control} and multi-channel MAC
\cite{so_vaidya__multichannel_mac} have been proposed in the
literature.

In this thesis, we primarily focus on the design of the MAC layer for
wireless mesh networks. We abstract out essential features of the MAC
and physical layers of a WMN and propose techniques that deliver high
network throughput. We take into account wireless channel effects such
as propagation path loss, fading and shadowing
\cite{sklar__rayleigh_fading}.  Towards the end of the thesis, we
provide an information-theoretic perspective on flow control.  The
main body of this thesis, however, focuses on MAC layer design for two
types of networks: Spatial Time Division Multiple Access (STDMA)
networks and random access networks.  We next describe these two types
of networks along with their potential applications in WMNs.

An STDMA network can be thought of as a mesh network in which multiple
transmitter receiver pairs can communicate at the same time. More
specifically, consider a WMN comprising of store-and-forward nodes
connected by ``point to point'' wireless communication channels
(links). A link is an ordered pair $(t,r)$, where $t$ is a transmitter
and $r$ is a receiver. Time is divided into fixed-length intervals
called slots. In STDMA, we allow concurrent communications between
collections of nodes that are ``reasonably far'' from each other,
i.e., we exploit spatial reuse. An STDMA link schedule describes the
transmission rights for each time slot in such a way that
communicating entities assigned to the same slot do not ``collide''.
In this thesis, we design centralized STDMA link scheduling algorithms
that take into account physical layer characteristics such as Signal
to Interference and Noise Ratio (SINR) at a receiver.

\begin{figure}[thbp]
  \centering
  \includegraphics[width=6in]{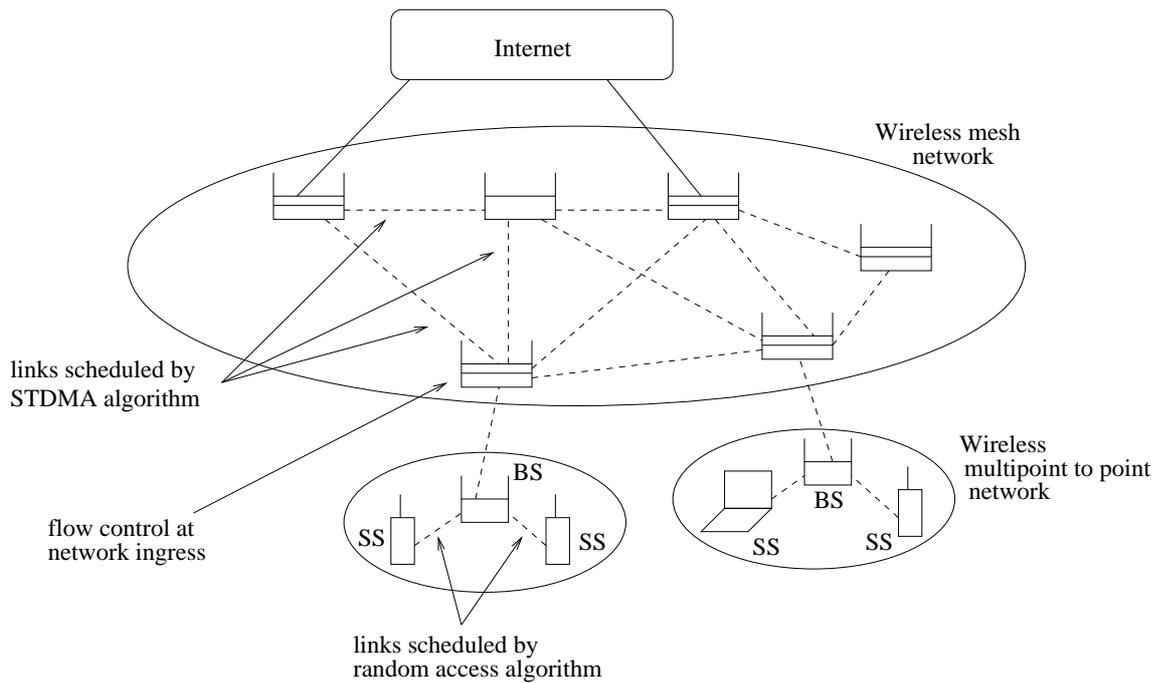}
  \caption{Potential applications of link scheduling and flow control
    in wireless networks.}
  \label{fig:potential_apps}
\end{figure}

STDMA link scheduling algorithms can be implemented at the MAC layer
of wireless mesh networks, as shown in Figure
\ref{fig:potential_apps}. A mesh network can be constructed with mesh
routers and mesh clients functioning as relay nodes in addition to
their sender and receiver roles. The link schedule can be computed by
a designated mesh router and then disseminated to all other nodes.
The mesh routers form the mesh backbone to provide connectivity to
(possibly mobile) mesh clients.

In a related problem involving link scheduling, we consider a
multipoint to point wireless network with random access. When random
access algorithms are directly translated from a wired network to a
wireless network, they yield equal or lower throughput.  This is
because they do not consider the time variation of the wireless
channel and interference conditions at the receiver.  In this thesis,
we design a distributed random access algorithm that takes into
account wireless channel attributes such as propagation path loss and
physical layer characteristics such as SINR at the receiver.

Random access algorithms can be applied to the MAC layer of wireless
networks, as shown in Figure \ref{fig:potential_apps}. The BS and SSs
are organized into a cell-like structure. Both uplink (from SS to BS)
and downlink (from BS to SS) channels are shared among the SSs. This
mode requires all SSs to be within the communication range and line of
sight of the BS.  A random access algorithm can be implemented in the
SSs to resolve contentions on the uplink channel.

In a complementary problem, we consider a packet level flow from a
source to a destination over a data network. The packets transmitted
by the source are regulated at the ingress of the network, as shown in
Figure \ref{fig:potential_apps}. In this thesis, we investigate the
maximum amount of information that can be transmitted from the source
to the destination by utilizing the idea of covert information
channels.

To summarize, this thesis deals with the design of MAC layer
algorithms (equivalently, link scheduling algorithms) for mesh
networks.  The proposed link scheduling algorithms take into account
physical layer characteristics such as SINR at a receiver.  Finally,
we also consider the problem of flow control.

Various solutions to the link scheduling problem have been proposed in
literature depending on the modeling of the wireless network and
interference conditions. In the next section, we motivate our work by
briefly outlining the essential differences between our approach and
the methodology of existing approaches.

\section{Motivation for the Thesis}

Consider the problem of determining a link schedule for an STDMA
wireless network. STDMA link schedules can be classified into point to
point and point to multipoint link schedules. In a point to point link
schedule, the transmission right in each slot is assigned to certain
links, while in a point to multipoint link schedule, the transmission
right in each slot is assigned to certain nodes. An STDMA scheduling
algorithm is a set of rules that is used to determine a link schedule
so as to satisfy certain objectives. An STDMA link schedule should be
so designed that, in every time slot, all packets transmitted by the
scheduled transmitters are received successfully at the corresponding
(intended) receivers.

Two models have been proposed in literature for specifying the
criteria for successful packet reception. According to the protocol
interference model \cite{gupta_kumar__capacity_wireless}, a packet is
received successfully at a receiver only if its intended transmitter
is within the communication range and other unintended transmitters
are outside the interference range of the receiver. In essence, the
protocol interference model mandates a ``silence zone'' around every
scheduled receiver in a time slot.  On the other hand, according to
the physical interference model \cite{gupta_kumar__capacity_wireless},
a packet is received successfully at a receiver only if the SINR at
the receiver is no less than a certain threshold, called communication
threshold.

Throughout this thesis, we assume that a packet is received
successfully if the SINR at the receiver is greater than or equal to
the communication threshold, i.e., we employ the physical interference
model. Moreover, we assume that, as long as the SINR threshold
condition is satisfied at the receiver of a link, a constant rate of
data transfer occurs along that link. In other words, the existence of
a channel coding technique that guarantees a fixed data rate is
assumed, when the SINR threshold condition is satisfied.

To maximize the aggregate traffic transported by an STDMA wireless
network, most link scheduling algorithms employ the protocol
interference model and seek to minimize the schedule length. These
algorithms model the network by a communication graph and employ novel
techniques to color all the edges of the graph using minimum number of
colors \cite{ramanathan_lloyd__scheduling_algorithms}. Such approaches
have three lacunae. First, they transform the link scheduling problem
to an edge coloring problem in a graph, which is a simplification of
the true system model. Second, they do not incorporate wireless
channel effects like propagation path loss, fading and shadowing.
Finally, they do not consider SINR threshold conditions at a receiver.

In this thesis, we seek to address these issues by designing
polynomial time link scheduling algorithms that employ the physical
interference model, provide a reasonably accurate representation of
the wireless network and aim to maximize the number of successful
packet transmissions per time slot. These algorithms take into account
wireless channel effects like propagation path loss, fading and
shadowing, as well as SINR conditions at a receiver.  We design and
evaluate algorithms for both point to point and point to multipoint
link scheduling. Our work falls under the realm of joint PHY-MAC
design of wireless networks.

In a related scenario involving link scheduling, consider the problem
of designing a random access algorithm for a multipoint to point
wireless network.  When traditional random access algorithms like
ALOHA \cite{abramson__aloha_system} and tree-like algorithms
\cite{capetanakis__tree_algorithms} are employed in a wireless
network, they yield equal or lower throughput compared to the wired
case. This is because such algorithms are incognizant of wireless
channel effects and physical layer characteristics. Thus, it is
important to design a random access algorithm that incorporates
wireless channel effects and exploits flexibilities provided by the
physical layer. Towards this step, we assume a receiver that is
capable of power-based capture
\cite{nguyen_wieselthier__capture_wireless}.  Also, we assume that
users can vary their transmission powers to increase the chances of
successful packet reception under the physical interference model.
Consequently, we design and analyze a variable-power tree-like
algorithm for a random access wireless network.

In the final scenario, we formulate the problem of analyzing flow
control in packet networks from an information-theoretic perspective.
We focus on the problem of analyzing regulated flows in a point to
point network.  It is well-known that information (in the Shannon
sense) can be transmitted from a source to a destination only by
encoding it in the contents, lengths and timings of data packets from
the source to the destination \cite{gallager__basic_limits},
\cite{anantharam_verdu__bits_through}.  We investigate the maximum
amount of information that can be transmitted by a source whose flow
is linearly bounded.  Specifically, we assume that covert information
is conveyed by randomness in packet lengths and investigate properties
of the regulating mechanism that leads to maximum information
transfer.

\section{Overview and Contributions of the Thesis}

In the first part of the thesis (Chapters \ref{ch:framework_link} to
\ref{ch:broadcastschedule}), we consider various problems on
centralized link scheduling in STDMA wireless networks; each problem
represents a different nuance of the overall link scheduling problem.
In the second part of the thesis (Chapters \ref{ch:review_random} and
\ref{ch:powercontrolled}), we consider a related link scheduling
problem, namely, distributed medium access control in a random access
wireless network.  In the third and final part of the thesis (Chapter
\ref{ch:flow_control}), we consider flow control in networks from an
information-theoretic perspective.

Chapter \ref{ch:framework_link} presents a generic framework and
system model for link scheduling in STDMA wireless networks.  We
describe the system parameters of an STDMA wireless network and
explain two prevalent models used to specify the criteria for
successful packet reception, namely protocol interference model and
physical interference model \cite{gupta_kumar__capacity_wireless}.  We
argue that STDMA link scheduling algorithms can be classified into
three classes: algorithms based on modeling the network by a two-tier
or communication graph, ``hybrid'' algorithms based on modeling the
network by a communication graph and verifying SINR conditions and
algorithms based on modeling the network by an SINR graph. We review
representative research papers from each of these classes. We explain
the relative merits and demerits of each class of algorithms in terms
of computational complexity, performance and accuracy of the network
model. We discuss limitations of link scheduling algorithms based only
on the communication graph model by providing illustrative examples.
Finally, to compare the performance of various link scheduling
algorithms, we motivate and introduce spatial reuse as a performance
metric.  Various ``spinoffs'' of the ``parent'' link scheduling
problem constitute the subproblems considered in Chapters
\ref{ch:comm_graph}, \ref{ch:sinr_graph} and
\ref{ch:broadcastschedule}.

In Chapter \ref{ch:comm_graph}, we consider STDMA point to point link
scheduling algorithms which utilize a communication graph
representation of the network. Initially, we examine the
ArboricalLinkSchedule (ALS) algorithm
\cite{ramanathan_lloyd__scheduling_algorithms}, which represents the
network by a communication graph, partitions the graph into minimum
number of planar subgraphs and colors each subgraph in a greedy
manner. We suggest a modification to the ALS algorithm based on
reusing colors from previously colored subgraphs to color the current
subgraph. We compare the performance of the modified algorithm with
the ALS algorithm and derive its running time complexity.
Subsequently, we propose the ConflictFreeLinkSchedule algorithm, which
is a hybrid algorithm based on the communication graph and verifying
SINR conditions. Under various wireless channel conditions, we
demonstrate that ConflictFreeLinkSchedule achieves higher spatial
reuse than existing link scheduling algorithms based on the
communication graph.  However, this improvement in performance is
achieved at a cost of slightly higher computational complexity.

In Chapter \ref{ch:sinr_graph}, we consider the point to point link
scheduling problem under the physical interference model.  The STDMA
network is represented by an SINR graph, in which weights of edges
correspond to interferences between pairs of nodes and weights of
vertices correspond to normalized noise powers at receiving nodes. We
propose a link scheduling algorithm based on the SINR graph
representation of the network.  We prove the correctness of the
algorithm and show that it has polynomial running time complexity.
Finally, we demonstrate that the proposed algorithm achieves higher
spatial reuse than ConflictFreeLinkSchedule.

In Chapter \ref{ch:broadcastschedule}, we consider point to multipoint
link scheduling (broadcast scheduling) under the physical interference
model.  The problem addressed herein can be considered as the ``dual''
of the problem considered in Chapters \ref{ch:comm_graph} and
\ref{ch:sinr_graph}.  We generalize the definition of spatial reuse to
the point to multipoint link scheduling problem.  We propose a greedy
scheduling algorithm which has demonstrably higher spatial reuse than
existing algorithms, without any increase in computational complexity.

In Chapter \ref{ch:review_random}, we consider another flavor of the
link scheduling problem, namely random access algorithms for wireless
networks.  While random access algorithms for satellite networks,
packet radio networks, multidrop telephone lines and multitap bus
(``traditional random access algorithms'') is a well-researched and
mature subject, the study of random access algorithms for wireless
networks that take into account physical layer characteristics such as
SINR and channel variations has yet to gain momentum.  This chapter
reviews representative research work which investigate such random
access algorithms, most of them being generalizations of the ALOHA
protocol (by adapting the retransmission probability) or the tree
algorithm (by adapting the set of contending users). We motivate the
use of variable transmission power to increase the throughput in
random access wireless networks.

We consider random access for wireless networks under the physical
interference model in Chapter \ref{ch:powercontrolled}. We design an
algorithm that adapts the set of contending users and their
corresponding transmission powers based on quaternary (2 bit) channel
feedback.  We model the algorithm dynamics by a Discrete Time Markov
Chain and subsequently derive its maximum stable throughput. Finally,
we demonstrate that the proposed algorithm achieves higher throughput
and substantially lower delay than the well-known First Come First
Serve splitting algorithm \cite{bertsekas_gallager__data_networks}.

In Chapter \ref{ch:flow_control}, we formulate the problem of
analyzing flow control in packet networks from a perspective of
maximizing mutual information between a source and a destination.  We
focus on the simpler, yet insightful, problem of analyzing regulated
flows in a point to point network. More specifically, we consider a
source whose flow is bounded by a ``generalized'' Token Bucket
Regulator (TBR) and analyze the maximum amount of information (in the
Shannon sense) that the source can convey to its destination by
encoding information in the randomness of packet lengths. This chapter
reveals two interesting results. First, under certain ``bandwidth''
constraints on cumulative tokens and cumulative bucket depth, we
demonstrate that a generalized TBR can achieve higher flow entropy
than that of a standard TBR. Second, we provide information-theoretic
arguments for the observations that the optimal generalized TBR has a
decreasing token increment sequence and a near-uniform bucket depth
sequence.

In Chapter \ref{ch:conclusions}, we summarize the thesis and provide
possible directions for future work. Specifically, we suggest
generalizations of the two-level power control algorithm proposed in
Chapter \ref{ch:powercontrolled}. We also provide pointers for
deriving the approximation factors of the algorithms proposed in
Chapters \ref{ch:comm_graph} and \ref{ch:sinr_graph}.

\clearpage{\pagestyle{empty}\cleardoublepage}

\chapter{A Framework for Link Scheduling Algorithms for STDMA Wireless Networks}
\label{ch:framework_link}

An STDMA wireless network consists of a finite set of nodes wherein
multiple pairs of nodes can communicate concurrently, as discussed in
Chapter \ref{ch:introduction}.  In this chapter, we outline a
framework for modeling STDMA link scheduling algorithms.  We consider
a general representation of an STDMA wireless network, i.e., this
model is not specific to any technology or protocol.  This abstraction
lends simplicity to the network model and helps us focus on the design
of scheduling algorithms for the network.  Since the problem of
determining an optimal link schedule is NP-hard
\cite{ramanathan_lloyd__scheduling_algorithms}, researchers have
proposed various heuristics to obtain close-to-optimal solutions.  In
our view, such heuristics can be broadly classified into three
categories: algorithms based on modeling the network by a two-tier or
communication graph, ``hybrid'' algorithms based on modeling the
network by a communication graph and verifying SINR conditions and
algorithms based on modeling the network by an SINR graph. We review
representative research papers from each of these classes. The
relative merits and demerits of each class of algorithms are also
elucidated in the chapter. Our observations motivate us to propose a
performance metric that is proportional to aggregate network
throughput.

The rest of this chapter is structured as follows.  In Section
\ref{sec:system_model_stdma}, we describe the system model of an STDMA
wireless network and explain the protocol and physical interference
models.  In Section \ref{sec:protocol_model}, we elucidate the
equivalence between a point to point link schedule for an STDMA
network and the colors of edges of the communication graph model of
the network. This is followed by a review of research work on point to
point link scheduling algorithms based on the protocol interference
model.  In Section \ref{sec:limitations_protocol}, we describe the
limitations of algorithms based on the protocol interference model
from a perspective of maximizing network throughput in wireless
networks.  We review research work on link scheduling algorithms
based on the physical interference model in Sections
\ref{sec:comm_graph_and_sinr} and \ref{sec:sinr_graph}.  Specifically,
Section \ref{sec:comm_graph_and_sinr} reviews algorithms based on
communication graph model of the network and SINR conditions, while
Section \ref{sec:sinr_graph} reviews algorithms based on an SINR graph
model of the network.  Finally, in Section
\ref{sec:performance_metric}, we propose spatial reuse as a
performance metric and argue that it corresponds to network throughput
from a physical layer viewpoint.

\section{System Model}
\label{sec:system_model_stdma}

We consider a general model of an STDMA wireless network with $N$
static store-and-forward nodes in a two-dimensional plane, where $N$
is a positive integer. Nodes are indexed as $1,2,\ldots,N$.  In a
wireless network, a link is an ordered pair of nodes $(t,r)$, where
$t$ is a transmitter and $r$ is a receiver.  We assume equal length
packets.  Time is divided into slots of equal duration.  During a time
slot, a node can either transmit, receive or remain idle.  The slot
duration equals the amount of time it takes to transmit one packet
over the wireless channel. We make the following additional
assumptions:
\begin{itemize}
\item Synchronized nodes: All nodes are synchronized to slot
  boundaries.

\item Homogeneous nodes: Every node has identical receiver
  sensitivity, transmission power and thermal noise characteristics.

\item Backlogged nodes: We assume a node to be continuously backlogged,
i.e., a node always has a packet to transmit and
  cannot transmit more than one packet in a time slot.
\end{itemize}
Let:
\begin{eqnarray*}
(x_j,y_j) &=& \mbox{Cartesian coordinates of node $j$} \;=:\;{\mathbf r}_j, \\
P &=& \mbox{power with which a node transmits its packet}, \\
N_0 &=& \mbox{thermal noise power spectral density}, \\
D(j,k) &=& \mbox{Euclidean distance between nodes $j$ and $k$}.
\end{eqnarray*}
The received signal power at a distance $D$ from the transmitter is
given by $\frac{P}{D^{\beta}}$, where $\beta$ is the path loss
exponent\footnote{We do not consider fading and shadowing effects.}.
An STDMA link schedule is a mapping from the set of links to time
slots. We only consider static link schedules, i.e., link schedules
that repeat periodically throughout the operation of the network. Let
$C$ denote the number of time slots in a link schedule, i.e., the {\em
  schedule length.}  For a given time slot $i$, $j^{th}$ communicating
transmitter-receiver pair is denoted by $t_{i,j} \rightarrow r_{i,j}$,
where $t_{i,j}$ denotes the index of the node which transmits a packet
and $r_{i,j}$ denotes the index of the node which receives the packet.
Let $M_i$ denote the number of concurrent transmitter-receiver pairs
in time slot $i$. A point to point link schedule for the STDMA network
is denoted by $\Psi(\mathcal S_1,\cdots,\mathcal S_C)$, where
\begin{eqnarray*}
\mathcal S_i 
 &:=& \{t_{i,1}\rightarrow r_{i,1},\cdots,t_{i,M_i}\rightarrow r_{i,M_i}\} \\
 &=& \mbox{set of transmitter-receiver pairs which can communicate 
     concurrently} \\ && \mbox{in time slot $i$}.
\end{eqnarray*}
Note that a link schedule repeats periodically throughout the
operation of the network. More specifically, transmitter-receiver
pairs that communicate concurrently in time slot $i$ also communicate
concurrently in time slots $i+C$, $i+2C$ and so on. Thus, $\mathcal
S_i = \mathcal S_{i\pmod{C}}$.  Finally, note that all transmitters
and receivers are stationary.

Every point to point link schedule must satisfy the following:
\begin{itemize}
\item Operational constraint: During a time slot, a node can transmit
  to exactly one node, receive from exactly one node or remain idle,
  i.e.,
  \begin{eqnarray}
    \{t_{i,j},r_{i,j}\} \cap \{t_{i,k},r_{i,k}\} = \phi \;\;\forall\;\;
    i=1,\ldots,C \;\;\forall\;\; 1 \leqslant j < k \leqslant M_i.
    \label{eq:operational_constraint}
  \end{eqnarray}
\end{itemize}

\begin{figure}
\centering
\subfigure[An STDMA wireless network with six nodes.]
{
\label{fig:deployment_stdma_link}
\includegraphics[width=5.5in]{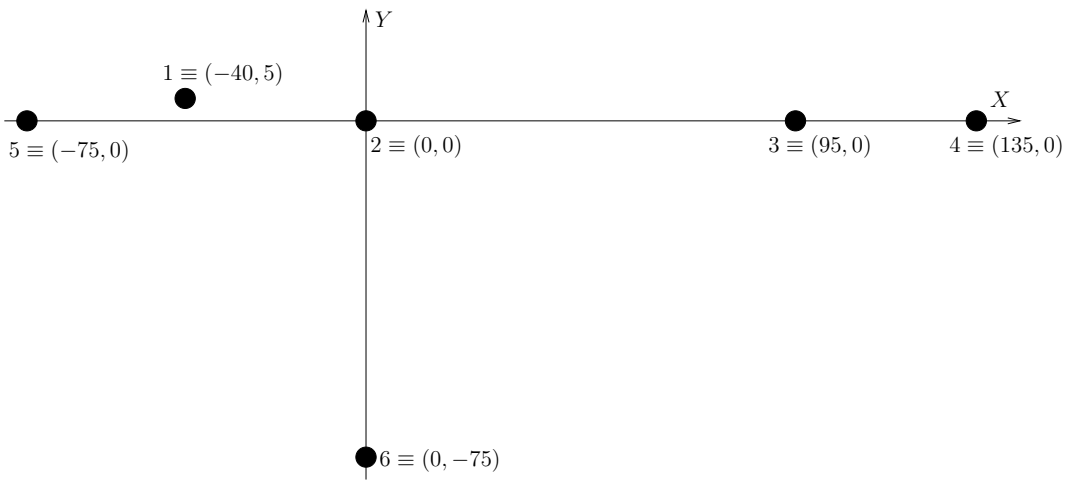}
}\\
\subfigure[A point to point link schedule for the network shown in Figure \ref{fig:deployment_stdma_link}.]
{
\label{fig:point_link_schedule}
\includegraphics[width=5.5in]{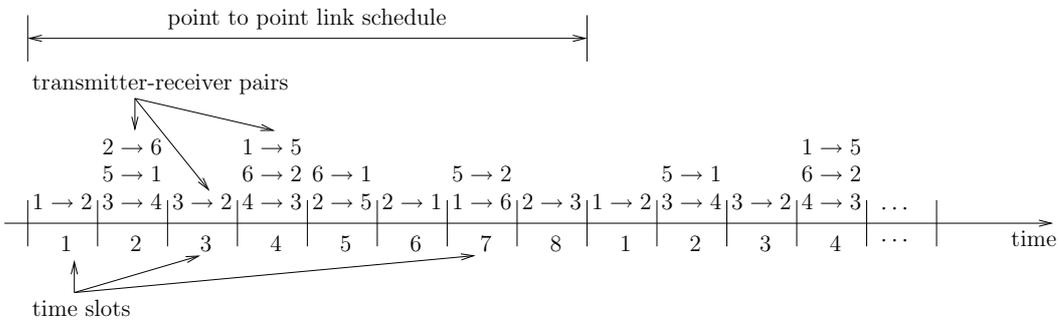}
}
\caption{Example of STDMA network and point to point link schedule.}
\label{fig:deployment_link_schedule}
\end{figure}

As an illustration, consider the STDMA wireless network shown in
Figure \ref{fig:deployment_stdma_link}. It consists of six nodes whose
coordinates (in meters) are $1 \equiv (-40,5)$, $2 \equiv (0,0)$, $3
\equiv (95,0)$, $4 \equiv (135,0)$, $5 \equiv (-75,0)$ and $6 \equiv
(0,-75)$. An example point to point link schedule for this STDMA
network is  shown in Figure \ref{fig:point_link_schedule}.  Note
that this schedule is only one of the several possible schedules and
is given here only for illustrative purposes.  The schedule length is
$C=8$ time slots and the schedule is defined by $\Psi(\mathcal
S_1,\mathcal S_2,\mathcal S_3,\mathcal S_4,\mathcal S_5, \mathcal
S_6,\mathcal S_7,\mathcal S_8)$, where
\begin{eqnarray*}
\mathcal S_1 
&=& \{t_{1,1} \rightarrow r_{1,1}\} \\
&=& \{1 \rightarrow 2\}, \\
\mathcal S_2
&=& \{t_{2,1} \rightarrow r_{2,1}, \; t_{2,2} \rightarrow r_{2,2}, 
      \; t_{2,3} \rightarrow r_{2,3}\} \\
&=& \{3 \rightarrow 4, \; 5 \rightarrow 1, \; 2 \rightarrow 6\}, \\
\mathcal S_3
&=& \{t_{3,1} \rightarrow r_{3,1}\} \\
&=& \{3 \rightarrow 2\}, \\
\mathcal S_4
&=& \{t_{4,1} \rightarrow r_{4,1}, \; t_{4,2} \rightarrow r_{4,2}, \; t_{4,3} \rightarrow r_{4,3}\} \\
&=& \{4 \rightarrow 3, \; 6 \rightarrow 2, \; 1 \rightarrow 5\}, \\
\mathcal S_5
&=& \{t_{5,1} \rightarrow r_{5,1}, \; t_{5,2} \rightarrow r_{5,2}\} \\
&=& \{2 \rightarrow 5, \; 6 \rightarrow 1\}, \\
\mathcal S_6
&=& \{t_{6,1} \rightarrow r_{6,1}\} \\
&=& \{2 \rightarrow 1\}, \\
\mathcal S_7
&=& \{t_{7,1} \rightarrow r_{7,1}, \; t_{7,2} \rightarrow r_{7,2}\} \\
&=& \{1 \rightarrow 6, \; 5 \rightarrow 2\}, \\
\mathcal S_8
&=& \{t_{8,1} \rightarrow r_{8,1}\} \\
&=& \{2 \rightarrow 3\}.
\end{eqnarray*}
After 8 time slots, the schedule repeats periodically, as shown in
Figure \ref{fig:point_link_schedule}.

A scheduling algorithm is a set of rules that is used to determine a
link schedule $\Psi(\cdot)$. Usually, a scheduling algorithm needs to
satisfy certain objectives.

Consider $j^{th}$ receiver in time slot $i$, i.e., receiver $r_{i,j}$.
The power received at $r_{i,j}$ from its intended transmitter
$t_{i,j}$ (signal power) is $\frac{P}{D^{\beta}(t_{i,j},r_{i,j})}$.
Similarly, the power received at $r_{i,j}$ from its unintended
transmitters (interference power) is $\sum_{\stackrel{k=1}{k\neq
    j}}^{M_i}\frac{P}{D^{\beta}(t_{i,k},r_{i,j})}$.  Thus, the Signal
to Interference and Noise Ratio (SINR) at receiver $r_{i,j}$ is given
by
\begin{eqnarray}
{\mbox{SINR}}_{r_{i,j}} &=& 
 \frac{\frac{P}{D^{\beta}(t_{i,j},r_{i,j})}}{N_0+\sum_{\stackrel{k=1}{k\neq j}}^{M_i}\frac{P}{D^{\beta}(t_{i,k},r_{i,j})}}.
\label{eq:definition_sinr}
\end{eqnarray}
Without considering the interference power,
the Signal to Noise Ratio (SNR) at receiver $r_{i,j}$ is given by
\begin{eqnarray}
{\mbox{SNR}}_{r_{i,j}} &=& \frac{P}{N_0D^{\beta}(t_{i,j},r_{i,j})}.
\label{eq:definition_snr}
\end{eqnarray}

According to the {\em protocol interference model}
\cite{gupta_kumar__capacity_wireless}, 
transmission $t_{i,j}\rightarrow r_{i,j}$ is successful if:
\begin{enumerate}

\item the SNR at receiver $r_{i,j}$ is no less than a certain
  threshold $\gamma_c$, termed as the {\em communication threshold}.
  From (\ref{eq:definition_snr}), this translates to
  \begin{eqnarray}
    D(t_{i,j},r_{i,j}) 
    &\leqslant& \left(\frac{P}{N_0 \gamma_c}\right)^\frac{1}{\beta}
    \;\;=:\;\; R_c,
    \label{eq:communication_range}
  \end{eqnarray}
  where $R_c$ is termed as communication range, and

\item the signal from any unintended transmitter $t_{i,k}$ is received
  at $r_{i,j}$ with an SNR less than a certain threshold $\gamma_i$,
  termed as the {\em interference threshold}. From
  (\ref{eq:definition_snr}), this translates to
  \begin{eqnarray}
    D(t_{i,k},r_{i,j}) 
    &>& \left(\frac{P}{N_0 \gamma_i}\right)^\frac{1}{\beta}
    =: R_i \;\;\forall\;\; k=1,\ldots,M_i, \;k \neq j, 
    \label{eq:interference_range}
  \end{eqnarray}
  where $R_i$ is termed as interference range.
\end{enumerate}
In essence, the transmission on a link is successful if the distance
between the nodes is less than or equal to the {\em communication
  range} and no other node is transmitting within the {\em
  interference range} from the receiver.

The STDMA network is denoted by $\Phi(N,({\mathbf r}_1,\ldots,{\mathbf
  r}_N),P, \gamma_c,\gamma_i,\beta,N_0)$. Note that $0 < \gamma_i <
\gamma_c$, thus $R_i > R_c$.  The relation $R_i=2R_c$ is widely
assumed in literature \cite{gupta_walrand__sufficient_rate},
\cite{lim_lim_hou__coordinate_based},
\cite{alicherry_bhatia_li__joint_channel},
\cite{jain_padhye__impact_interference}.

According to the {\em physical interference model}
\cite{gupta_kumar__capacity_wireless}, the transmission on a link is
successful if the SINR at the receiver is greater than or equal to the
communication threshold $\gamma_c$. More specifically, the physical
interference model states that transmission $t_{i,j}\rightarrow
r_{i,j}$ is successful if:
\begin{eqnarray}
\frac{\frac{P}{D^\beta(t_{i,j},r_{i,j})}}{N_0+\sum_{\stackrel{k=1}{k\neq j}}^{M_i} \frac{P}{D^\beta(t_{i,k},r_{i,j})}} \geqslant \gamma_c.
\label{eq:sinr_ge_gammac}
\end{eqnarray}
Note that the physical interference model is less restrictive but more
complex. Usually, this representation has been employed to model mesh
networks with TDMA like access mechanisms
\cite{brar_blough_santi__computationally_efficient}.  We will discuss
this aspect later in the thesis.

A point to point link schedule $\Psi(\cdot)$ is {\it conflict-free} if
the SINR at every intended receiver does not drop below the
communication threshold, i.e.,
\begin{eqnarray}
\mbox{SINR}_{r_{i,j}} \geqslant \gamma_c \;\;\forall\;\; i=1,\ldots,C, 
\;\;\forall\;\; j=1,\ldots,M_i.
\label{eq:conflict_free}
\end{eqnarray}

\section{Link Scheduling based on Protocol Interference Model}
\label{sec:protocol_model}

\subsection{Equivalence of Link Scheduling and Graph Edge Coloring}
\label{sec:equivalence_coloring}

In this section, we describe the communication and two-tier graph
representations of an STDMA wireless network. We explain the
equivalence between a point to point link schedule for the STDMA
network and the colors of edges of the communication graph
representation of the network, and illustrate this equivalence with an
example.

The STDMA network $\Phi(\cdot)$ can be modeled by a directed graph
$\mathcal G(\mathcal V,\mathcal E)$, where $\mathcal V$ is the set of
vertices and $\mathcal E$ is the set of edges. Let $\mathcal V =
\{v_1,v_2,\ldots,v_N\}$, where vertex $v_j$ represents node $j$ in
$\Phi(\cdot)$.  In the graph representation, if node $k$ is within
node $j$'s communication range, then there is an edge from $v_j$ to
$v_k$, denoted by $v_j \stackrel{c}{\rightarrow} v_k$ and termed as
communication edge. Similarly, if node $k$ is outside node $j$'s
communication range but within its interference range, then there is
an edge from $v_j$ to $v_k$, denoted by $v_j \stackrel{i}{\rightarrow}
v_k$ and termed as interference edge.  Thus, $\mathcal E = \mathcal
E_c \,\cup\, \mathcal E_i$, where $\mathcal E_c$ and $\mathcal E_i$
denote the set of communication and interference edges respectively.
The {\em two-tier graph} representation of the STDMA network
$\Phi(\cdot)$ is defined as the graph $\mathcal G(\mathcal V,\mathcal
E_c \cup \mathcal E_i)$ comprising of all vertices and both
communication and interference edges. The {\em communication graph}
representation of the STDMA network $\Phi(\cdot)$ is defined as the
graph $\mathcal G_c(\mathcal V,\mathcal E_c)$ comprising of all
vertices and communication edges only.  We will illustrate these
representations with an example.

\begin{table}[tbhp]
\centering
\begin{tabular}{|l|l|l|} \hline
  Parameter & Symbol & Value \\ \hline
  transmission power & $P$ & 10 mW \\ \hline
  path loss exponent & $\beta$ & 4 \\ \hline
  noise power spectral density & $N_0$ & -90 dBm \\ \hline
  communication threshold & $\gamma_c$ & 20 dB \\ \hline
  interference threshold & $\gamma_i$ & 10 dB \\ \hline
\end{tabular}
\caption{System parameters for STDMA networks shown in Figures 
    \ref{fig:deployment_stdma_link}, 
    \ref{fig:deploy_high_interference} and \ref{fig:deploy_high_colors}.}
\label{tab:system_parameters_stdma}
\end{table}

\begin{figure}[thbp]
\centering
\includegraphics[width=5.5in]{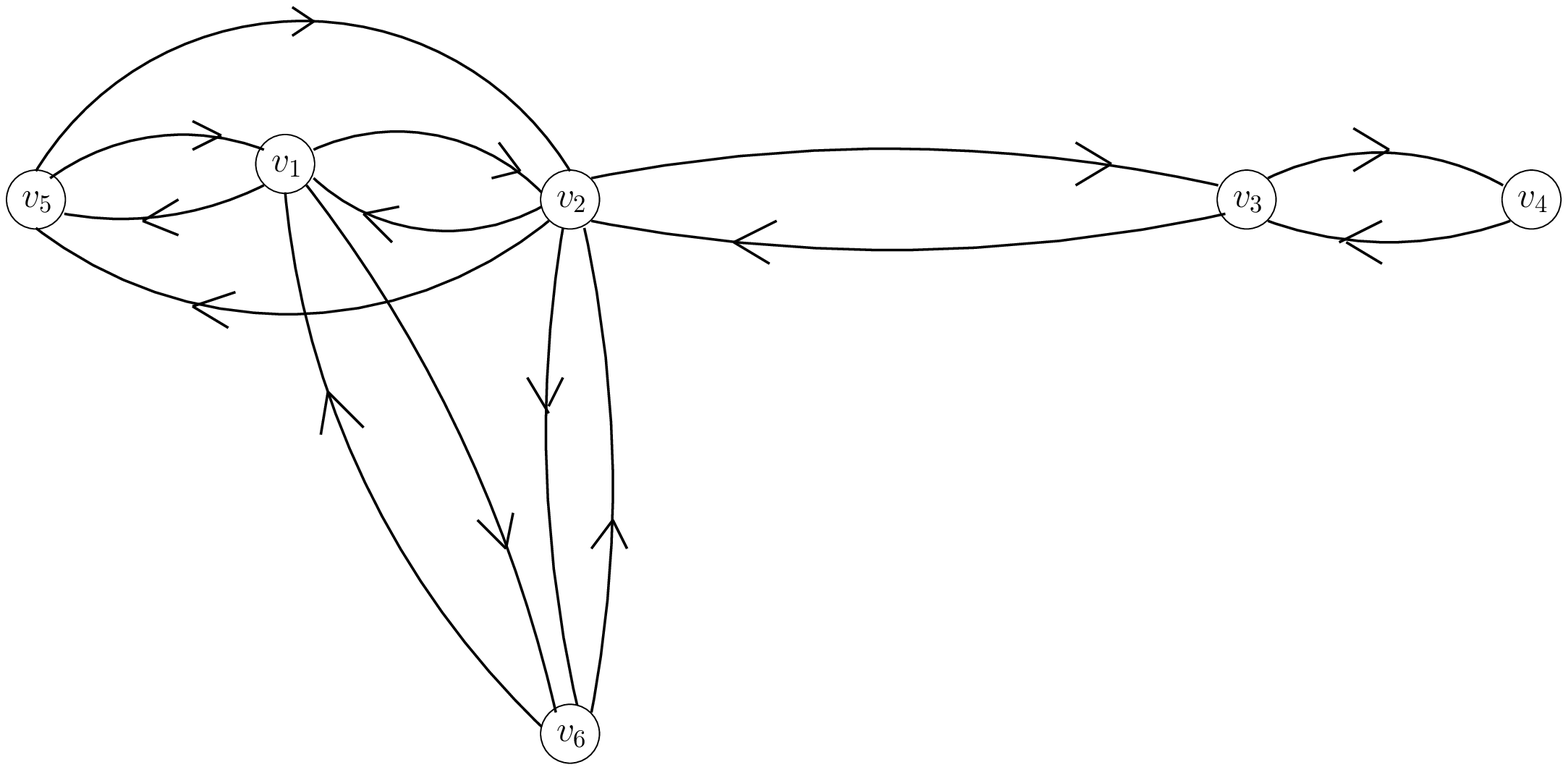}
\caption{Communication graph model of STDMA network described by
  Figure \ref{fig:deployment_stdma_link} and Table
  \ref{tab:system_parameters_stdma}.}
\label{fig:communication_graph_critique}
\end{figure}

\begin{figure}[thbp]
\centering
\includegraphics[width=5.5in]{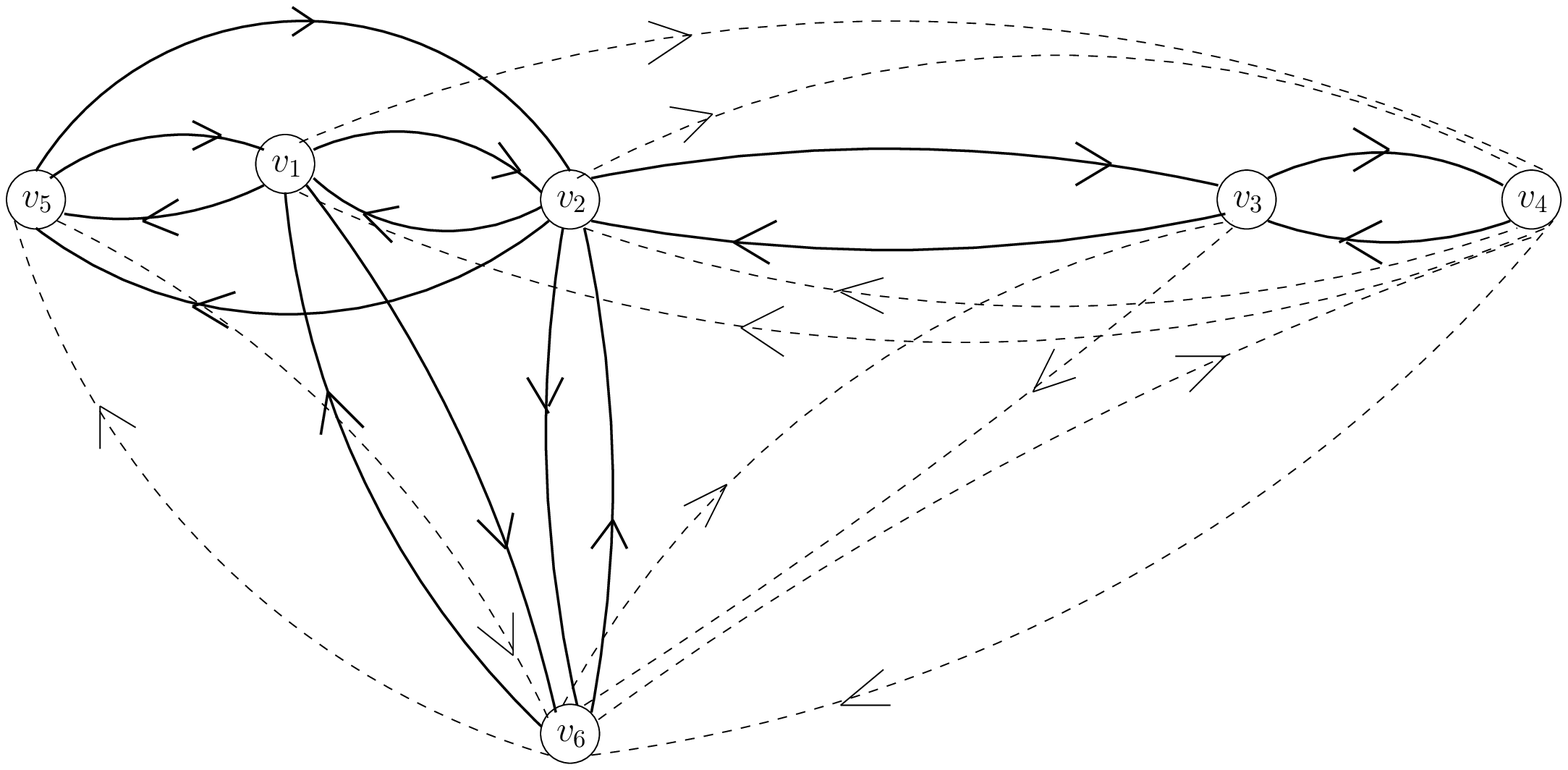}
\caption{Two-tier graph model of STDMA network described by Figure
  \ref{fig:deployment_stdma_link} and Table
  \ref{tab:system_parameters_stdma}.}
\label{fig:two_tier_graph_critique}
\end{figure}

Consider the STDMA wireless network $\Phi(\cdot)$ whose deployment is
shown in Figure \ref{fig:deployment_stdma_link}.  The system
parameters for this network are given in Table
\ref{tab:system_parameters_stdma}.  From
(\ref{eq:communication_range}) and (\ref{eq:interference_range}), it
can be easily shown that $R_c=100$ m and $R_i=177.8$ m. The
corresponding communication graph representation $\mathcal
G_c(\mathcal V,\mathcal E_c)$ is shown in Figure
\ref{fig:communication_graph_critique}. The communication graph
comprises of 6 vertices and 14 directed communication edges.  The
vertex and communication edge sets are given by
\begin{eqnarray}
\mathcal V
&=& \{v_1,v_2,v_3,v_4,v_5,v_6\} \label{eq:vertex_set}, \\
\mathcal E_c
&=& \{
v_1 \stackrel{c}{\rightarrow} v_2,
v_2 \stackrel{c}{\rightarrow} v_1,
v_1 \stackrel{c}{\rightarrow} v_5,
v_5 \stackrel{c}{\rightarrow} v_1,
v_1 \stackrel{c}{\rightarrow} v_6,
v_6 \stackrel{c}{\rightarrow} v_1,
v_2 \stackrel{c}{\rightarrow} v_5, \nonumber \\
&&
v_5 \stackrel{c}{\rightarrow} v_2,
v_2 \stackrel{c}{\rightarrow} v_6, 
v_6 \stackrel{c}{\rightarrow} v_2,
v_2 \stackrel{c}{\rightarrow} v_3,
v_3 \stackrel{c}{\rightarrow} v_2,
v_3 \stackrel{c}{\rightarrow} v_4,
v_4 \stackrel{c}{\rightarrow} v_3
\}. \label{eq:comm_edge_set}
\end{eqnarray}
The two-tier graph model $\mathcal G(\mathcal V,\mathcal E_c \cup
\mathcal E_i)$ of the STDMA network $\Phi(\cdot)$ is shown in Figure
\ref{fig:two_tier_graph_critique}. The two-tier graph comprises of 6
vertices, 14 directed communication edges and 10 directed interference
edges. The vertex and communication edge sets are given by
(\ref{eq:vertex_set}) and (\ref{eq:comm_edge_set}) respectively, while
the interference edge set is given by
\begin{eqnarray}
\mathcal E_i
&=& \{
v_1 \stackrel{i}{\rightarrow} v_4,
v_4 \stackrel{i}{\rightarrow} v_1,
v_2 \stackrel{i}{\rightarrow} v_4,
v_4 \stackrel{i}{\rightarrow} v_2,
v_3 \stackrel{i}{\rightarrow} v_6,
v_6 \stackrel{i}{\rightarrow} v_3, \nonumber \\
&&
v_4 \stackrel{i}{\rightarrow} v_6,
v_6 \stackrel{i}{\rightarrow} v_4,
v_5 \stackrel{i}{\rightarrow} v_6,
v_6 \stackrel{i}{\rightarrow} v_5,
\}.
\end{eqnarray}

\begin{figure}[thbp]
\centering
\includegraphics[width=5.5in]{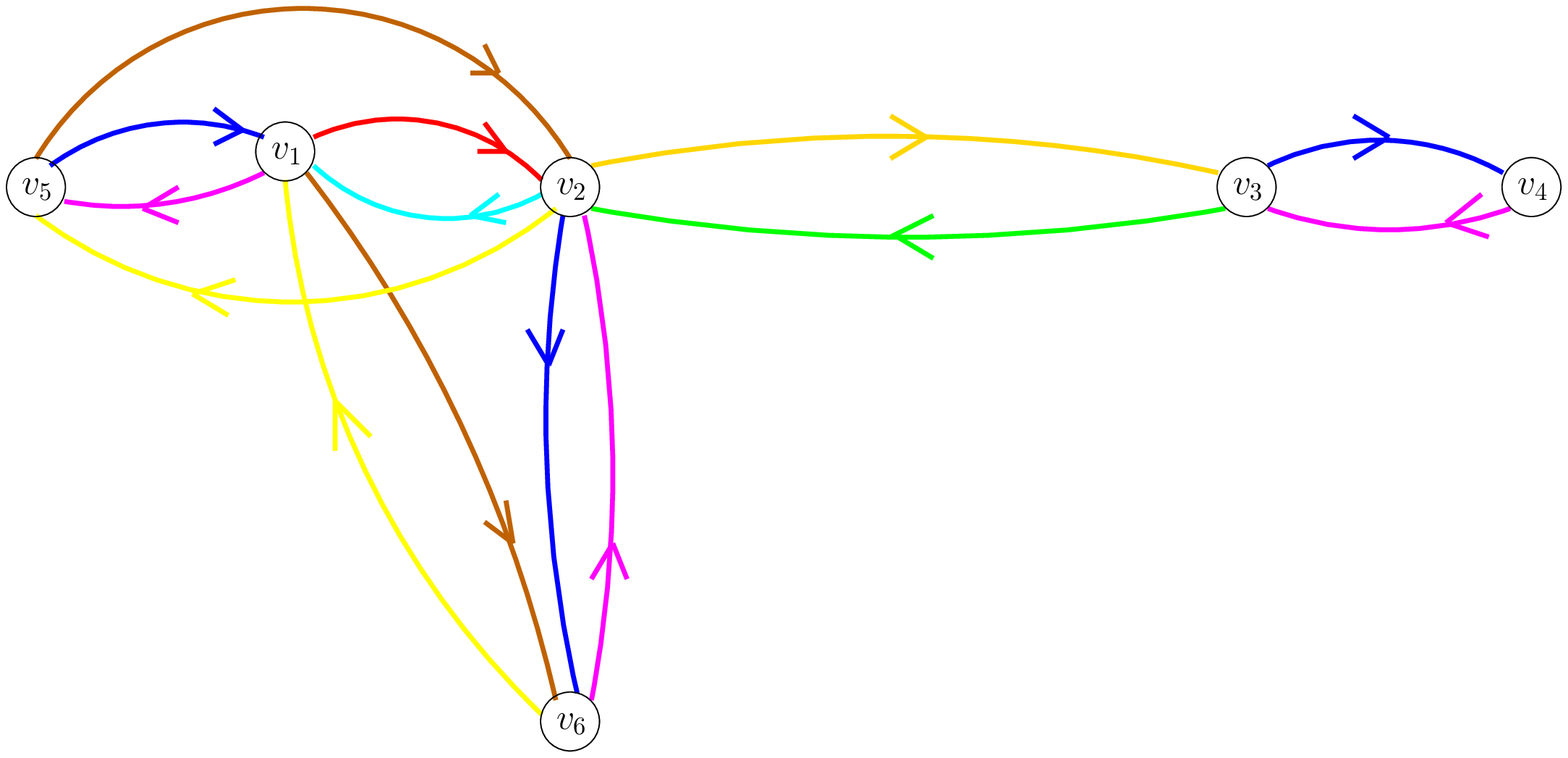}
\caption{Edge coloring of communication graph shown in Figure
  \ref{fig:communication_graph_critique} corresponding to the link
  schedule shown in Figure \ref{fig:point_link_schedule}.}
\label{fig:comm_graph_critique_colored}
\end{figure}

Given the above representations, a point to point link schedule
$\Psi(\cdot)$ for an STDMA wireless network $\Phi(\cdot)$ can be
considered as equivalent to assigning a unique color to every edge in
the communication graph, such that transmitter-receiver pairs with the
same color transmit simultaneously in a particular time slot. For the
example network considered, the link schedule shown in Figure
\ref{fig:point_link_schedule} corresponds to the coloring of the edges
of the communication graph shown in Figure
\ref{fig:comm_graph_critique_colored}.  Time slots 1, 2, 3, 4, 5, 6, 7
and 8 in $\Psi(\cdot)$ correspond to colors red, blue, green, magenta,
yellow, cyan, brown and gold in $\mathcal E_c$ respectively.  Note
that a coloring algorithm that uses the least number of colors also
minimizes the schedule length. This aspect is further addressed in
subsequent sections.

\subsection{Review of Algorithms}

In this section, we provide an overview of past research in the field
of STDMA point to point link scheduling algorithms based on the
protocol interference model.  The protocol interference model is
widely studied in literature because of its simplicity. It has been
usually employed to model networks such as Carrier Sense Multiple
Access with Collision Avoidance (CSMA/CA) based WLANs\footnote{
  Consider an IEEE 802.11 based WLAN wherein CSMA with RTS/CTS/ACK is
  used to protect unicast transmissions. Due to carrier sensing, a
  transmission between nodes $j$ and $k$ may block all transmissions
  that are within a distance of $R_i$ from either $j$ (due to sensing
  RTS and DATA) or $k$ (due to sensing CTS and ACK).}
\cite{brar_blough_santi__computationally_efficient},
\cite{alicherry_bhatia_li__joint_channel}.  Centralized algorithms
\cite{ramanathan_lloyd__scheduling_algorithms},
\cite{nelson_kleinrock__spatial_tdma},
\cite{hajek_sasaki__link_scheduling},
\cite{kodialam_nandagopal__characterizing_achievable},
\cite{alicherry_bhatia_li__joint_channel} as well as distributed
algorithms \cite{salonidis_tassiulas__distributed_dynamic} have been
proposed for generating link schedules based on the protocol
interference model.

A link scheduling algorithm based on the protocol interference model
utilizes a communication or two-tier graph model of the STDMA network
to determine a point to point link schedule
\cite{behzad_rubin__performance_graph},
\cite{gronkvist_hansson__comparison_between}.  Algorithms based on the
protocol interference model for assigning links to time slots
(equivalently, colors) require that two communication edges $v_i
\stackrel{c}{\rightarrow} v_j$ and $v_k \stackrel{c}{\rightarrow} v_l$
can be colored the same if and only if:
\begin{enumerate}
\renewcommand{\theenumi}{\roman{enumi}}
\item
vertices $v_i$, $v_j$, $v_k$, $v_l$ are all mutually distinct, i.e.,
there is no {\it primary edge conflict,} and
\label{cond:primary_conflict}

\item
$v_i \rightarrow v_l \not\in \mathcal G(\cdot)$ and $v_k \rightarrow v_j
\not\in \mathcal G(\cdot)$, i.e, there is no {\it secondary edge conflict}.
\label{cond:secondary_conflict}
\end{enumerate}
The first criterion is based on the operational constraint
(\ref{eq:operational_constraint}). The second criterion states that a
node cannot receive a packet if it lies within the interference range
of any other transmitting node.  A scheduling algorithm utilizes
various graph coloring methodologies to obtain a non-conflicting link
schedule, i.e., a link schedule devoid of primary and secondary edge
conflicts.

To maximize the throughput of an STDMA network, algorithms based on
the protocol interference model\footnote{Link scheduling algorithms
  based on the protocol interference model are sometimes referred to
  as ``graph based algorithms'' in literature
  \cite{behzad_rubin__performance_graph},
  \cite{gronkvist_hansson__comparison_between}. This term is slightly
  confusing since scheduling algorithms based on the physical
  interference model also construct graphs prior to determining a link
  schedule.}  seek to minimize the total number of colors used to
color all the communication edges of $\mathcal G(\cdot)$. This will in
turn minimize the schedule length.  It is well known that for an
arbitrary communication graph, the problem of determining a minimum
length schedule (optimal schedule) is NP-hard
\cite{ramanathan_lloyd__scheduling_algorithms},
\cite{hajek_sasaki__link_scheduling}. Hence, the approach followed in
the literature is to devise algorithms that produce close to optimal
(sub-optimal) solutions. The efficiency of a sub-optimal algorithm is
typically measured in terms of its computational (run time) complexity
and performance guarantee (approximation factor).

The concept of STDMA for wireless networks was formalized in
\cite{nelson_kleinrock__spatial_tdma}. The authors assume a multihop
packet radio network with fixed node locations and consider the
problem of assigning an integral number of slots to every link in an
STDMA cycle (frame). To solve this problem, they model the network by
a communication graph, determine a set of maximal cliques and then
assign a certain number of slots to all the links in each maximal
clique.  Finally, the authors develop a fluid approximation for the
mean system delay and validate it using simulations.

In \cite{hajek_sasaki__link_scheduling}, the authors consider
pre-specified link demands in a spread spectrum packet radio network.
They formulate the problem as a linear optimization problem and use
the ellipsoid algorithm
\cite{grotschel_lovasz_schrijver__ellipsoid_method} to solve the
problem.  They assume that the desired link data rates are rational
numbers and develop a strongly polynomial algorithm\footnote{An
  algorithm is strongly polynomial if (a) the number of arithmetic
  operations (addition, multiplication, division or comparison) is
  polynomially bounded by the dimension of the input, and (b) the
  precision of numbers appearing in the algorithm is bounded by a
  polynomial in the dimension and precision of the input.}  that
computes a minimum length schedule. Finally, they consider the problem
of link scheduling to satisfy pre-specified end-to-end demands in the
network.  They formulate this problem as a multicommodity flow problem
and describe a polynomial time algorithm that computes a minimum
length schedule. As pointed out by the authors, their algorithm is not
practical due to its high computational complexity.

A significant work in link scheduling under protocol interference
model is reported in \cite{ramanathan_lloyd__scheduling_algorithms},
in which the authors show that tree networks can be scheduled
optimally, oriented graphs\footnote{An in-oriented graph is a directed
  graph in which every vertex has at most one outgoing edge. An
  out-oriented graph is a directed graph in which every vertex has at
  most one incoming edge.}  can be scheduled near-optimally and
arbitrary networks can be scheduled such that the schedule is bounded
by a length proportional to the graph thickness\footnote{The thickness
  of a graph $\mathcal G(\cdot)$ is the minimum number of planar
  graphs into which $\mathcal G(\cdot)$ can be partitioned.}  times
the optimum number of colors.

In \cite{ramanathan_lloyd__scheduling_algorithms}, the authors seem to
have missed a subtle point that colors from previously colored
oriented graphs can be used to color the current oriented graph.
Instead, they use a {\em fresh} set of colors to color each successive
oriented graph. Consequently, their algorithm leads to a higher
numbers of colors, especially if the number of oriented graphs is
large.  The authors employ such a heuristic primarily to upper bound
the number of colors used by the algorithm
(\cite{ramanathan_lloyd__scheduling_algorithms}, Lemma 3.4) and
consequently obtain bounds on the running time complexity and
performance guarantee of the algorithm
(\cite{ramanathan_lloyd__scheduling_algorithms}, Theorem 3.3).  Though
the ArboricalLinkSchedule algorithm has nice theoretical properties
such as low computational complexity, it can be shown that it may
yield a higher number of colors {\em in practice}. This leads to lower
network throughput.

We should point out here that, if we modify the ArboricalLinkSchedule
algorithm to {\em reuse} colors from previously colored oriented
graphs to color the current oriented graph, then the schedule length
will always be {\em lower} than the schedule length obtained by the
ArboricalLinkSchedule algorithm. This can lead to higher network
throughput.  We develop this idea further in Chapter
\ref{ch:comm_graph}. Furthermore, we show that this can be achieved
with only a slight increase in computational complexity.

In \cite{jain_padhye__impact_interference}, the authors investigate
throughput bounds for a given wireless network and traffic workload
under the protocol interference model. They use a conflict
graph\footnote{Under the protocol interference model, the conflict
  graph $F(V_F,E_F)$ is constructed from the communication graph
  $\mathcal G_c(\mathcal V,\mathcal E_c)$ as follows.  Let $l_{ij}$
  denote the communication edge $v_i \stackrel{c}{\rightarrow} v_j$.
  Vertices of $F(\cdot)$ correspond to directed edges $l_{ij}$ in
  $\mathcal E_c$. In $F(\cdot)$, there exists an edge from vertex
  $l_{ij}$ to vertex $l_{pq}$ if any of the following is true: (a)
  $D(i,q) \leqslant R_i$ or (b) $D(p,j) \leqslant R_i$.}  to represent
interference constraints.  The problem of finding maximum throughput
for a given source-destination pair under the flexibility of multipath
routing is formulated as a linear program with flow constraints and
conflict graph constraints.  They show that this problem is NP-hard
and describe techniques to compute lower and upper bounds on
throughput. Finally, the authors numerically evaluate throughput
bounds and computation time of their heuristics for simple network
scenarios and IEEE 802.11 MAC (bidirectional MAC). Though the authors
provide a general framework for joint routing and scheduling, they
neither derive the computational complexity of their heuristics nor
describe their link scheduling algorithm explicitly.

Recently, in \cite{alicherry_bhatia_li__joint_channel}, the authors
investigate joint link scheduling and routing under the protocol
interference model for a wireless mesh network consisting of static
mesh routers and mobile client devices.  Assuming that $l(u)$ denotes
the aggregate traffic demand on node $u$, they consider the problem of
maximizing $\lambda$, such that at least $\lambda l(u)$ amount of
traffic can be routed from each node $u$ to a fixed gateway node.
Since this problem is NP-hard, the authors propose heuristics based on
linear programming and re-routing flows on the communication graph.
They derive the worst case bound of their algorithm and evaluate its
performance via simulations. Though the authors make a reasonable
attempt to solve the joint routing and scheduling problem, their
algorithm is extremely complex\footnote{The algorithm in
  \cite{alicherry_bhatia_li__joint_channel} consists of five steps:
  solve linear program, channel assignment, post processing, flow
  scaling and interference free link scheduling.  Moreover, the
  channel assignment step consists of three algorithms.}  and brute
force in nature.  Furthermore, the authors have not provided intuitive
arguments for their algorithm.

Another recent work which jointly investigates link scheduling and
routing under protocol interference model is reported in
\cite{kodialam_nandagopal__characterizing_achievable}.  The authors
consider wireless mesh networks with half duplex and full duplex
orthogonal channels, wherein each node can transmit to at most one
node and/or receive from at most $k$ nodes ($k \geqslant 1$) during
any time slot. They investigate the joint problem of routing and
scheduling to analyze the achievability of a given rate vector between
multiple source-destination pairs. The scheduling algorithm is
equivalent to an edge-coloring on a multi-graph
representation\footnote{A multi-graph is a directed graph in which
  multiple edges can emanate from a vertex $v_i$ and terminate at
  another vertex $v_j$ $(v_j \neq v_i)$.}  and the corresponding
necessary conditions lead the routing problem to be formulated as a
linear optimization problem.  The authors describe a polynomial time
approximation algorithm to obtain an $\epsilon$-optimal solution of
the routing problem using the primal dual approach. Finally, they
evaluate the performance of their algorithms via simulations.

It has been observed that high data rates are achievable in a wireless
mesh network by allowing a node to transmit to only one neighboring
node at fixed peak power in any time slot
\cite{kodialam_nandagopal__characterizing_achievable}.  We point out
here that a similar assumption of uniform transmission power has been
made in our system model in subsequent chapters of the thesis.

Algorithms based on the protocol interference model represent the
network by a communication or two-tier graph and employ a plethora of
techniques from graph theory \cite{west__graph_theory} and
approximation algorithms \cite{vazirani__approximation_algorithms},
\cite{hochbaum__approximation_algorithms} to devise heuristics which
yield a minimum length schedule.  Consequently, such algorithms have
the advantage of low computational complexity (in general). However,
recent research suggests that these algorithms result in low network
throughput. This aspect is further illustrated in the following
section.

\section{Limitations of Algorithms based on Protocol Interference Model}
\label{sec:limitations_protocol}

Due to its inherent simplicity, the protocol interference model has
been traditionally employed to represent a wide variety of wireless
networks. However, it leads to low network throughput in wireless mesh
networks.  To emphasize this point, we provide examples to demonstrate
that algorithms based on the protocol interference model can result in
schedules that yield low network throughput.

Intuitively, the protocol interference model divides the deployment
region of the STDMA wireless network into ``communication zones'' and
``interference zones''. This transforms the scheduling problem to an
edge coloring problem for the communication graph representation of
the network. However, this simplification can result in schedules that
do not satisfy the SINR threshold condition (\ref{eq:conflict_free}).

Specifically, algorithms based on the protocol interference model do
not necessarily maximize the throughput of an STDMA wireless network
because:
\begin{enumerate}

\item \label{exam:high_interference} They can lead to high cumulative
  interference at a receiver, due to hard-thresholding based on
  communication and interference radii
  \cite{behzad_rubin__performance_graph},
  \cite{gronkvist_hansson__comparison_between}.  This is because the
  SINR at receiver $r_{i,j}$ decreases with an increase in the number
  of concurrent transmissions $M_i$, while the communication radius
  $R_c$ and the interference radius $R_i$ have been defined for a
  single transmission only.

  \begin{figure}[thbp]
    \centering
    \includegraphics[width=5.5in]{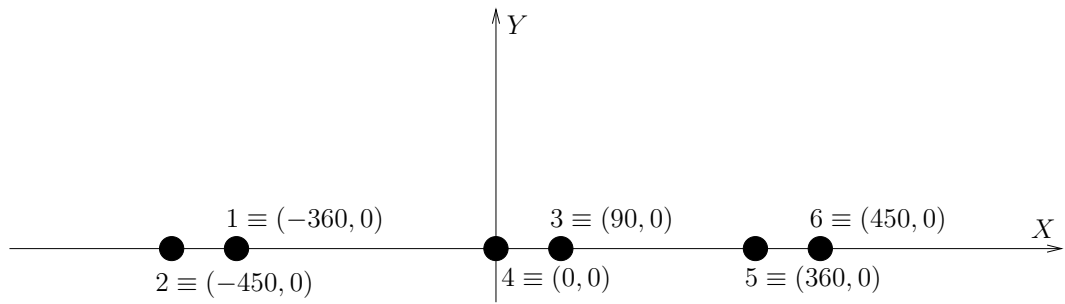}
    \caption{An STDMA wireless network with six nodes.}
    \label{fig:deploy_high_interference}
  \end{figure}

  \begin{figure}[thbp]
    \centering
    \includegraphics[width=4.5in]{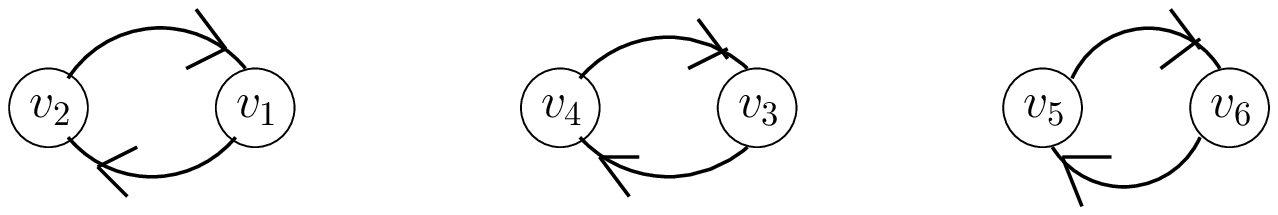}
    \caption{Two-tier graph model of the STDMA wireless network
      described by Figure \ref{fig:deploy_high_interference} and Table
      \ref{tab:system_parameters_stdma}.}
    \label{fig:comm_graph_deploy_high_intf}
  \end{figure}

  \begin{figure}[thbp]
    \centering
    \includegraphics[width=4.5in]{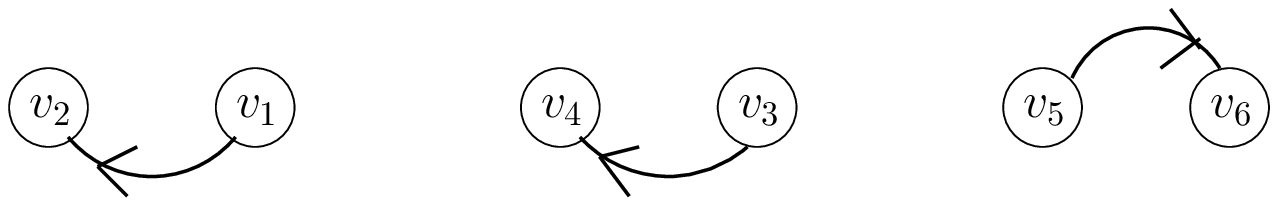}
    \caption{Subgraph of two-tier graph shown in Figure
      \ref{fig:comm_graph_deploy_high_intf}.}
    \label{fig:subgraph_deploy_high_intf}
  \end{figure}

  \begin{figure}[thbp]
    \centering
    \includegraphics[width=4.5in]{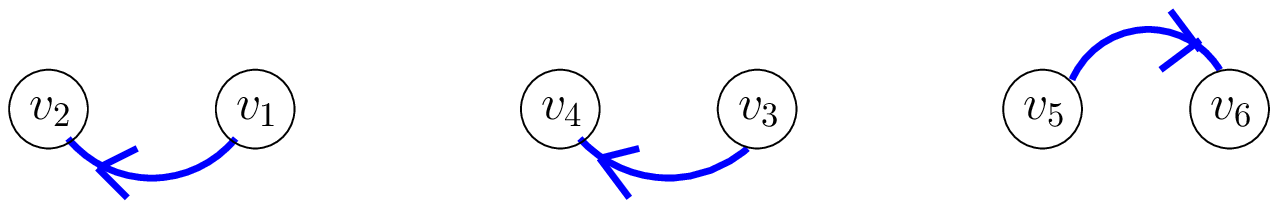}
    \caption{Coloring of subgraph shown in Figure
      \ref{fig:subgraph_deploy_high_intf}.}
    \label{fig:colored_subgraph_deploy_high_intf}
  \end{figure}

  \begin{figure}[thbp]
    \centering
    \includegraphics[width=5.5in]{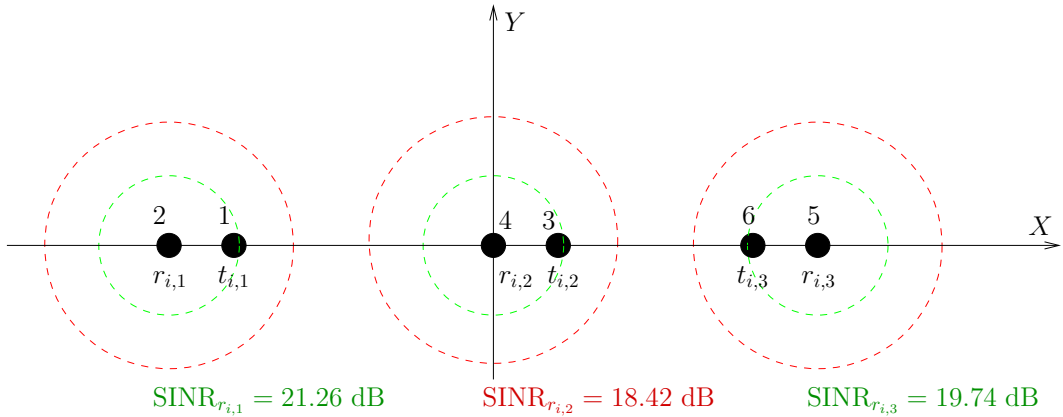}
    \caption{Point to point link scheduling algorithms based on
      protocol interference model can lead to high interference.}
    \label{fig:high_interference}
  \end{figure}

  For example, consider the STDMA wireless network whose deployment is
  shown in Figure \ref{fig:deploy_high_interference}. The network
  consists of six labeled nodes whose coordinates (in meters) are $1
  \equiv (-360,0)$, $2 \equiv (-450,0)$, $3 \equiv (90,0)$, $4 \equiv
  (0,0)$, $5 \equiv (360,0)$ and $6 \equiv (450,0)$. The system
  parameters are shown in Table \ref{tab:system_parameters_stdma},
  which yield $R_c=100$ m and $R_i=177.8$ m.  The two-tier graph model
  of the STDMA network is shown in Figure
  \ref{fig:comm_graph_deploy_high_intf}; note that interference edges
  are absent.  Consider the transmission requests $1 \rightarrow 2$,
  $3 \rightarrow 4$ and $5 \rightarrow 6$, which correspond to
  communication edges of the subgraph shown in Figure
  \ref{fig:subgraph_deploy_high_intf}.  The communication edges $v_1
  \stackrel{c}{\rightarrow} v_2$, $v_3 \stackrel{c}{\rightarrow} v_4$
  and $v_5 \stackrel{c}{\rightarrow} v_6$ shown in Figure
  \ref{fig:subgraph_deploy_high_intf} do not have primary or secondary
  edge conflicts. To minimize the number of colors, such an algorithm
  will color these edges with the same color, as shown in Figure
  \ref{fig:colored_subgraph_deploy_high_intf}. Equivalently,
  transmissions $1 \rightarrow 2$, $3 \rightarrow 4$ and $5
  \rightarrow 6$ will be scheduled in the same time slot, say time
  slot $i$.  However, our computations show that the SINRs at
  receivers $r_{i,1}$, $r_{i,2}$ and $r_{i,3}$ are $21.26$ dB, $18.42$
  dB and $19.74$ dB respectively. Figure \ref{fig:high_interference}
  shows the nodes of the network along with the labeled
  transmitter-receiver pairs, receiver-centric communication and
  interference zones and the SINRs at the receivers. From the SINR
  threshold condition (\ref{eq:sinr_ge_gammac}), transmission $t_{i,1}
  \rightarrow r_{i,1}$ is successful, while transmissions $t_{i,2}
  \rightarrow r_{i,2}$ and $t_{i,3} \rightarrow r_{i,3}$ are
  unsuccessful. This leads to low network throughput.

\item Moreover, these algorithms can be extremely conservative and
  result in higher number of colors.

  \begin{figure}[thbp]
    \centering
    \includegraphics[width=5.5in]{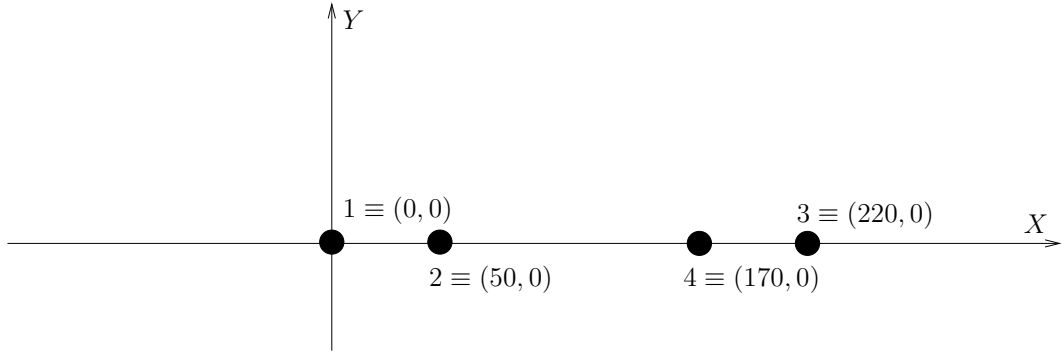}
    \caption{An STDMA wireless network with four nodes.}
    \label{fig:deploy_high_colors}
  \end{figure}

  \begin{figure}[thbp]
    \centering
    \includegraphics[width=4.5in]{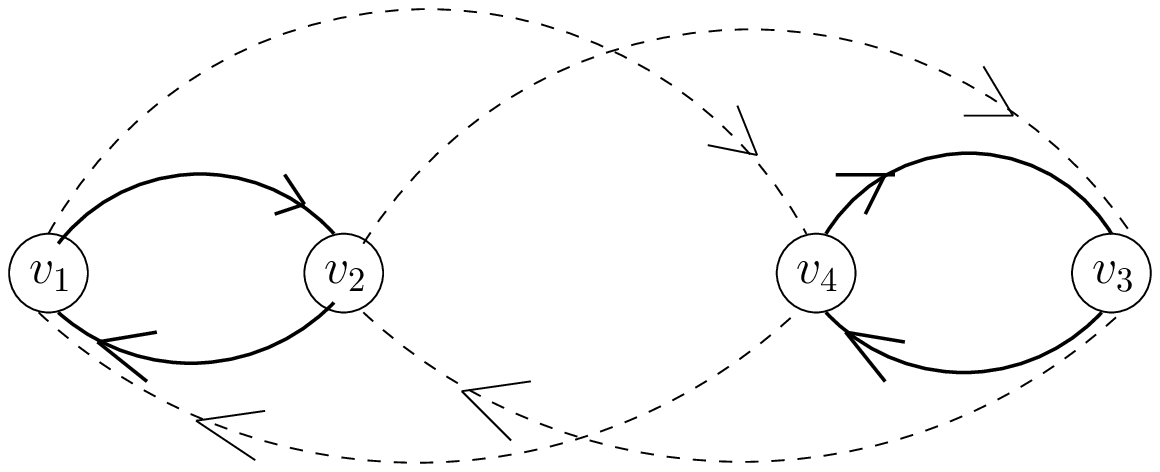}
    \caption{Two-tier graph model of STDMA wireless network described
      by Figure \ref{fig:deploy_high_colors} and Table
      \ref{tab:system_parameters_stdma}.}
    \label{fig:tiered_graph_deploy_high_colors}
  \end{figure}

  \begin{figure}[thbp]
    \centering
    \includegraphics[width=4.5in]{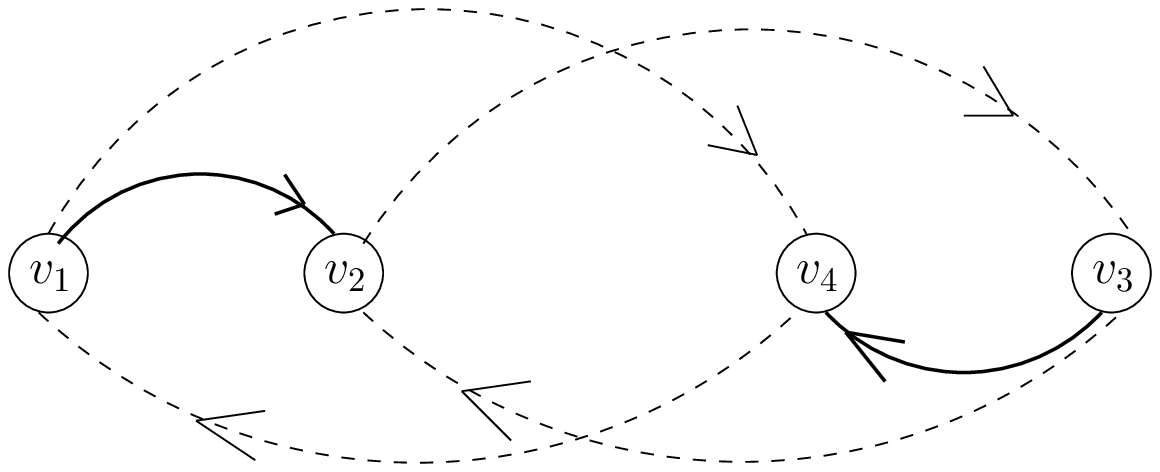}
    \caption{Subgraph of two-tier graph shown in Figure
      \ref{fig:tiered_graph_deploy_high_colors}.}
    \label{fig:subgraph_deploy_high_colors}
  \end{figure}

  \begin{figure}[thbp]
    \centering
    \includegraphics[width=4.5in]{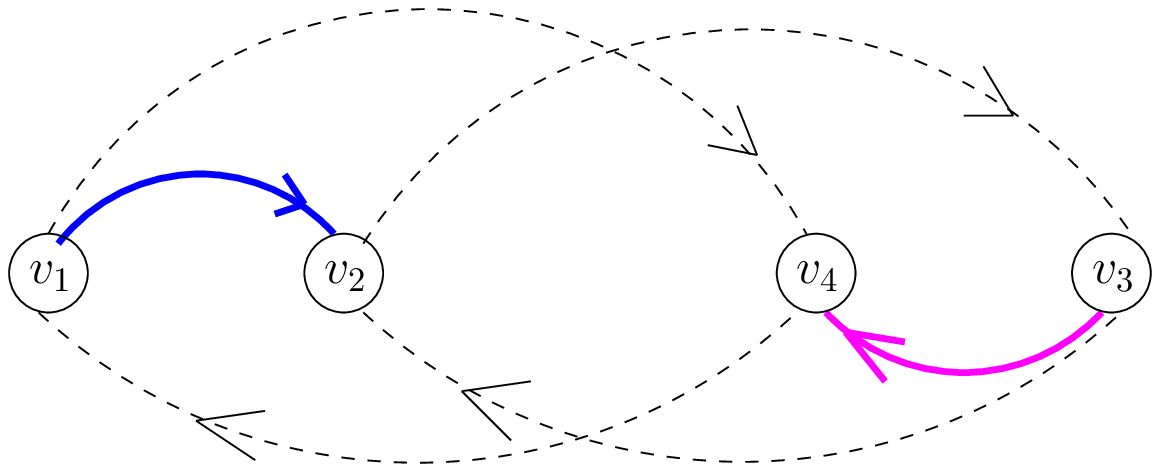}
    \caption{Coloring of subgraph shown in Figure
      \ref{fig:subgraph_deploy_high_colors}.}
    \label{fig:colored_subgraph_deploy_high_colors}
  \end{figure}

  \begin{figure}[thbp]
    \centering
    \includegraphics[width=5.5in]{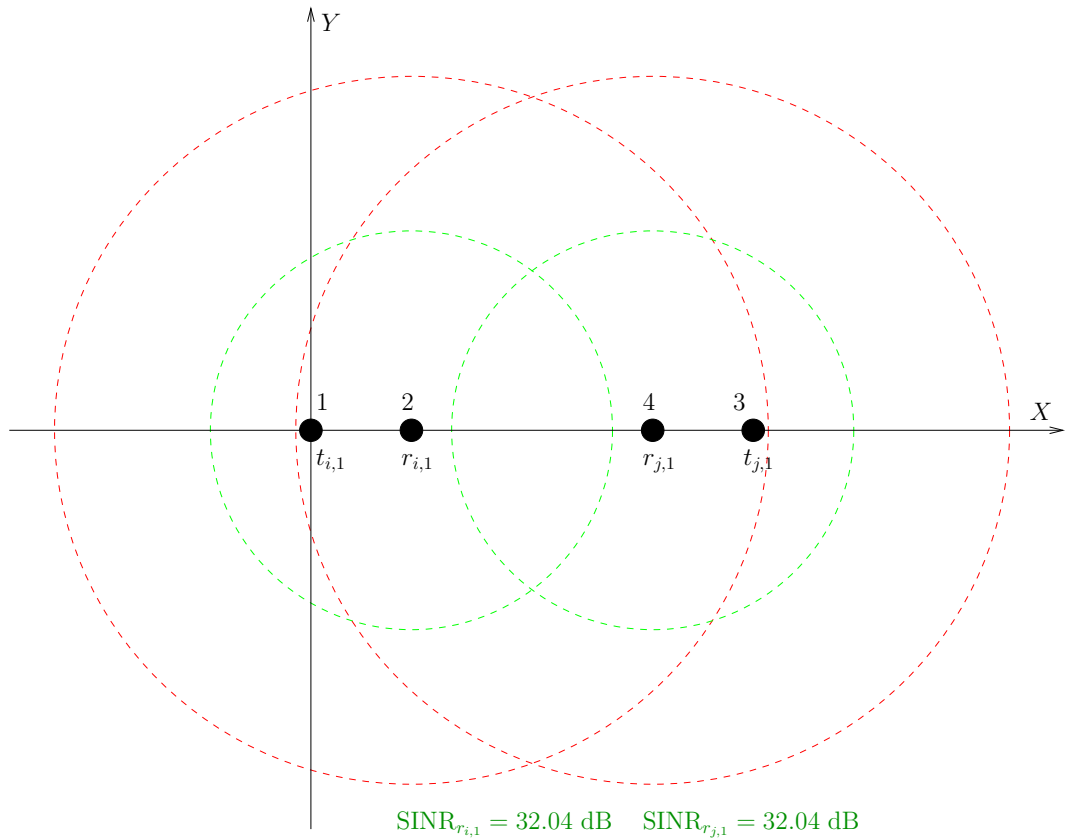}
    \caption{Point to point link scheduling algorithms based on
      protocol interference model can lead to higher number of
      colors.}
    \label{fig:high_colors}
  \end{figure}

  \begin{figure}[thbp]
    \centering
    \includegraphics[width=4.5in]{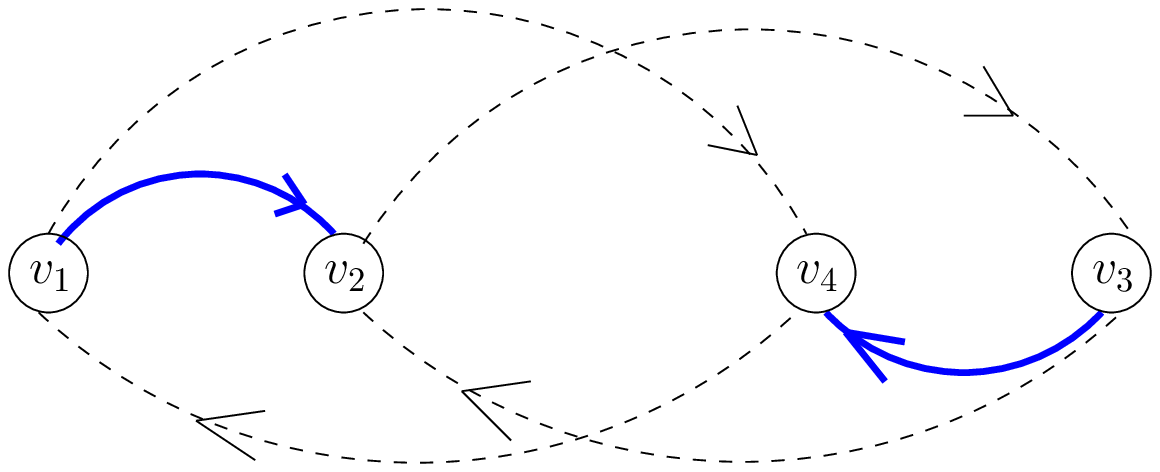}
    \caption{Alternative coloring of subgraph shown in Figure
      \ref{fig:subgraph_deploy_high_colors}.}
    \label{fig:alternative_coloring_subgraph}
  \end{figure}

  \begin{figure}[thbp]
    \centering
    \includegraphics[width=5.5in]{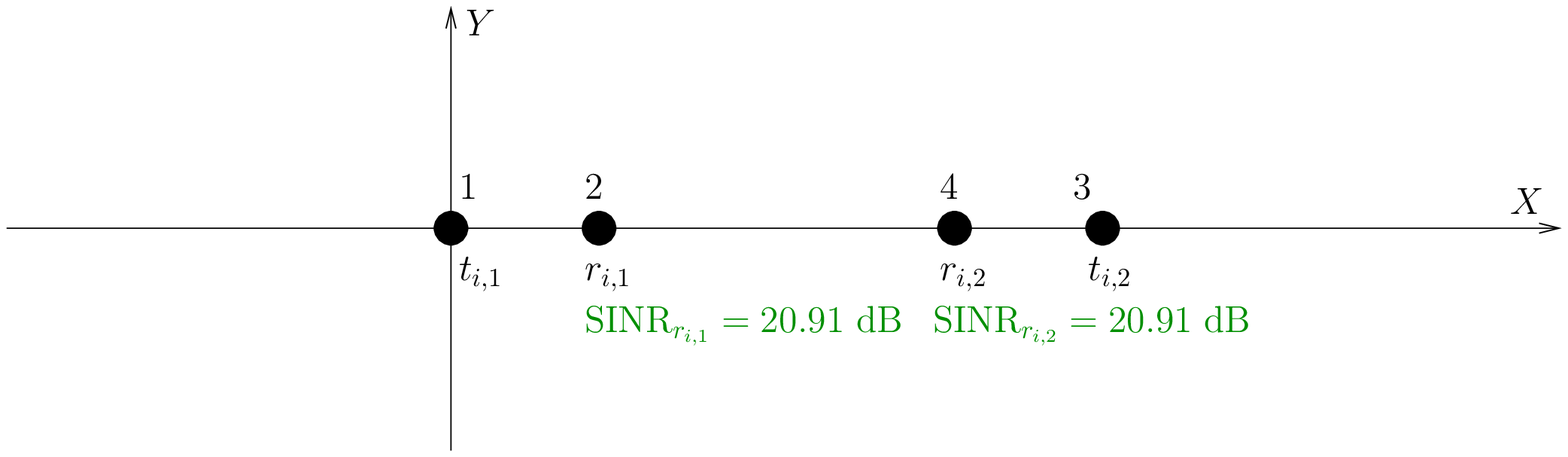}
    \caption{A point to point link schedule corresponding to Figure
      \ref{fig:alternative_coloring_subgraph} that yields lower number
      of colors.}
    \label{fig:alternative_colors}
  \end{figure}

  For example, consider the STDMA wireless network whose deployment is
  shown in Figure \ref{fig:deploy_high_colors}. The network consists
  of four labeled nodes whose coordinates (in meters) are $1 \equiv
  (0,0)$, $2 \equiv (50,0)$, $3 \equiv (220,0)$ and $4 \equiv
  (170,0)$. The system parameters are shown in Table
  \ref{tab:system_parameters_stdma}, which lead to $R_c=100$ m and
  $R_i=177.8$ m. The two-tier graph model of the STDMA network is
  shown in Figure \ref{fig:tiered_graph_deploy_high_colors}.  Consider
  the transmission requests $1 \rightarrow 2$ and $3 \rightarrow 4$,
  which correspond to communication edges of the subgraph shown in
  Figure \ref{fig:subgraph_deploy_high_colors}.  The communication
  edges $v_1 \stackrel{c}{\rightarrow} v_2$ and $v_3
  \stackrel{c}{\rightarrow} v_4$ shown in Figure
  \ref{fig:subgraph_deploy_high_colors} have secondary edge conflicts.
  Hence, such an algorithm will typically color these edges with
  different colors, as shown in Figure
  \ref{fig:colored_subgraph_deploy_high_colors}.  Equivalently, a link
  scheduling algorithm based on the protocol interference model will
  schedule transmissions $1 \rightarrow 2$ and $3 \rightarrow 4$ in
  different time slots, say time slots $i$ and $j$ respectively, where
  $i \neq j$.  Our computations show that the resulting SINRs at
  receivers $r_{i,1}$ and $r_{j,1}$ are both equal to $32.04$ dB.
  Figure \ref{fig:high_colors} shows the nodes of the network along
  with the labeled transmitter-receiver pairs, receiver-centric
  communication and interference zones and SINRs at the receivers.
  Observe that, with an algorithm based on the protocol interference
  model, the SINRs at both receivers are well above the communication
  threshold of $20$ dB. Alternatively, consider an algorithm (perhaps
  based on the physical interference model) that schedules
  transmissions $1 \rightarrow 2$ and $3 \rightarrow 4$ in the same
  time slot, say time slot $i$.  The corresponding edge coloring is
  shown in Figure \ref{fig:alternative_coloring_subgraph}.  Our
  computations show that the resulting SINRs at receivers $r_{i,1}$
  and $r_{j,1}$ are both equal to $20.91$ dB, which are also above the
  communication threshold.  Figure \ref{fig:alternative_colors} shows
  the nodes of the network along with the labeled transmitter-receiver
  pairs and SINRs at the receivers.  In essence, with the alternate
  algorithm, both transmissions $t_{i,1} \rightarrow r_{i,1}$ and
  $t_{i,2} \rightarrow r_{i,2}$ are successful, since signals levels
  are so high at the receivers that strong interferences can be
  tolerated. In summary, a point to point link scheduling algorithm
  based on the protocol interference model will typically schedule the
  above transmissions in different slots and yield lower network
  throughput compared to the alternate algorithm.

\item Lastly, these algorithms are not aware of the topology of the
  network, i.e., they determine a link schedule without being
  cognizant of the exact positions of the transmitters and receivers.

\end{enumerate}

The above examples demonstrate that scheduling algorithms based on the
protocol interference model can result in low network throughput.
Observe that algorithms that construct an approximate model of the
STDMA network (two tier graph or communication graph) and focus on
minimizing the schedule length do not necessarily maximize network
throughput. This observation is developed into a proposal for an
appropriate performance metric in Section
\ref{sec:performance_metric}.

Since link scheduling algorithms based on the protocol interference
model yield low throughput, researchers have propounded algorithms
based on the physical interference model to improve the throughput of
STDMA wireless networks. To achieve higher throughput, one possible
technique is to model the STDMA network by a communication graph and
check SINR threshold conditions during assignment of links to time
slots; this is the approach most commonly employed, for example in
\cite{brar_blough_santi__computationally_efficient},
\cite{behzad_rubin__performance_graph},
\cite{gore_jagabathula_karandikar__high_spatial}. The other technique
is to incorporate SINR threshold conditions into a special graph model
of the network; this approach is more challenging and (to the best of
our knowledge) is considered only in research work such as
\cite{moscibroda_wattenhofer__complexity_connectivity},
\cite{moscibroda_wattenhofer_zollinger__topology_control},
\cite{kumar_gore_karandikar__link_scheduling}.  Research papers which
employ the former approach are reviewed in Section
\ref{sec:comm_graph_and_sinr}, while research papers which employ the
latter approach are reviewed in Section \ref{sec:sinr_graph}.

\section{Link Scheduling based on Communication Graph Model and SINR
  Conditions}
\label{sec:comm_graph_and_sinr}

In this section, we examine recent research in link scheduling based
on modeling the STDMA network by a communication graph and verifying
SINR conditions at the receivers. Though algorithms based on this
model \cite{lim_lim_hou__coordinate_based},
\cite{kim_lim_hou__improving_spatial}, yield higher throughput, they
usually result in higher computational complexity than algorithms
based on the protocol interference model.

In \cite{brar_blough_santi__computationally_efficient}, the authors
investigate throughput improvement in an IEEE 802.11 like wireless
mesh network with CSMA/CA channel access scheme replaced by STDMA. For
a successful packet transmission, they mandate that two-way
communication be successful, i.e., a packet transmission is defined to
be successful if and only if both data and acknowledgement packets are
received successfully.  Under this ``extended physical interference
model'', they present a greedy algorithm which computes a point to
point link transmission schedule in a centralized manner. Assuming
uniform random node distribution and using results from occupancy
theory \cite{kolchin_sevastyanov__random_allocations}, they derive an
approximation factor for the length of this schedule relative to the
shortest schedule. Though the analysis presented in
\cite{brar_blough_santi__computationally_efficient} is novel, their
model is restrictive because it is only applicable to wireless
networks using link-layer reliability protocols.

The throughput performance of link scheduling algorithms based on
two-tier graph model $\mathcal G(\mathcal V,\mathcal E_c \, \cup \,
\mathcal E_i)$ has been analyzed under physical interference
conditions in \cite{behzad_rubin__performance_graph}.  The authors
determine the optimal number of simultaneous transmissions by
maximizing a lower bound on the throughput and subsequently propose
Truncated Graph-Based Scheduling Algorithm (TGSA), an algorithm that
provides probabilistic guarantees for network throughput. Though the
analysis presented in \cite{behzad_rubin__performance_graph} is
mathematically elegant and based on the Edmundson-Madansky bound
\cite{madansky__inequalities_stochastic},
\cite{dokov_morton__higher_order}, their algorithm does not yield high
network throughput.  This is because the partitioning of a maximal
independent set of communication edges into multiple subsets (time
slots) is arbitrary and not based on network topology, which can lead
to significant interference in certain regions of the network. This is
further elucidated by the simulation results in Chapter
\ref{ch:comm_graph}.

The performance of algorithms based on the protocol interference model
versus those based on communication graph model and SINR conditions is
evaluated and compared in
\cite{gronkvist_hansson__comparison_between}.  To generate a
non-conflicting link schedule based on the protocol interference
model, the authors use a two-tier graph model with certain SINR
threshold values chosen based on heuristics and examples.  To generate
a conflict-free point to point link schedule based on the physical
interference model, the authors employ a method suggested in
\cite{somarriba__multihop_packet} which describes heuristics based on
two path loss models, namely terrain-data based ground wave
propagation model and Vogler's five knife-edge model. Their
simulations results indicate that, under a Poisson arrival process,
algorithms based on the protocol interference model result in higher
average packet delay than algorithms based on communication graph
model and SINR conditions.

In \cite{kim_lim_hou__improving_spatial}, the authors investigate the
tradeoff between the average number of concurrent transmissions
(spatial reuse) and sustained data rate per node for an IEEE 802.11
wireless network.  They show that spatial reuse depends only on the
ratio of transmit power to carrier sense threshold \cite{ieee__wlan}.
Keeping the carrier sense threshold fixed, they propose a distributed
power and rate control algorithm based on interference measurement and
evaluate its performance via simulations.

In \cite{lim_lim_hou__coordinate_based}, the authors investigate
mitigation of inter-flow interference in an IEEE 802.11e wireless mesh
network from a temporal-spatial diversity perspective.  Measurements
of received signal strengths are used to construct a virtual
coordinate system to identify concurrent transmissions with minimum
inter-flow interference. Based on this new coordinate system, one of
the nodes, designated as gateway node, determines the scheduling order
for downlink frames of different connections. Through extensive
simulations with real-life measurement traces, the authors demonstrate
throughput improvement with their algorithm.

Algorithms based on representing the network by a communication graph
and verifying SINR threshold conditions yield higher network
throughput than algorithms based on the protocol interference model.
However, this is achieved at the cost of higher computational
complexity. Furthermore, the gains in throughput may not be
significant enough to justify the increase in computational
complexity. This has prompted few researchers to solve the link
scheduling problem in a more fundamental manner. These researchers
have proposed an altogether different model of the network, termed as
SINR graph model, and developed heuristics. Such algorithms are
reviewed in the following section.

\section{Link Scheduling based on SINR Graph Model}
\label{sec:sinr_graph}

In literature, many authors refer to algorithms based on communication
graph model and checking SINR conditions as ``algorithms based on
physical interference model''.  In this thesis, only algorithms that
embed SINR threshold conditions into an appropriate graph model of the
network are referred to as ``algorithms based on the physical
interference model''.  Though the physical interference model is more
realistic, algorithms based on this model
\cite{moscibroda_wattenhofer__complexity_connectivity},
\cite{moscibroda_wattenhofer_zollinger__topology_control},
\cite{kumar_gore_karandikar__link_scheduling} have, in general, higher
computational complexity than algorithms based on the protocol
interference model.

Point to point link scheduling for power-controlled STDMA networks
under the physical interference model is analyzed in
\cite{moscibroda_wattenhofer__complexity_connectivity}. The authors
define scheduling complexity as the minimum number of time slots
required for strong connectivity of the graph\footnote{A directed
  graph $\mathcal G(\cdot)$ is strongly connected if there exists a
  directed path from every vertex to every other vertex.}  constructed
from the point to point link schedule. They develop an algorithm
employing non-linear power assignment\footnote{In uniform power
  assignment, all nodes transmit with the same transmission power. In
  linear power assignment
  \cite{moscibroda_wattenhofer__complexity_connectivity}, a node
  transmits with minimum power required to satisfy the SINR threshold
  condition at the receiver, i.e., transmission power equals $N_0
  \gamma_c D^\beta$.  Non-linear power assignment refers to a power
  assignment scheme that is neither uniform nor linear.} and show that
its scheduling complexity is polylogarithmic in the number of nodes.
In a related work
\cite{moscibroda_wattenhofer_zollinger__topology_control}, the authors
investigate the time complexity of scheduling a set of communication
requests in an arbitrary network.  They consider a ``generalized
physical model'' wherein the actual received power of a signal can
deviate from the theoretically received power by a multiplicative
factor. Their algorithm successfully schedules all links in time
proportional to the squared logarithm of the number of nodes times the
static interference measure
\cite{rickenbach_schmid__robust_interference}.  Though the authors of
\cite{moscibroda_wattenhofer__complexity_connectivity},
\cite{moscibroda_wattenhofer_zollinger__topology_control} allow
non-uniform transmission power at all nodes and develop novel
algorithms, their algorithms are impractical. This is because wireless
devices have constraints on maximum transmission power, while the
algorithms in \cite{moscibroda_wattenhofer__complexity_connectivity},
\cite{moscibroda_wattenhofer_zollinger__topology_control} can result
in arbitrarily high transmission power at some nodes.

In \cite{jain_padhye__impact_interference}, the authors provide a
general framework for computation of throughput bounds for a given
wireless network and traffic workload. Though their work primarily
focuses on the protocol interference model, they briefly allude to
the physical interference model too. Specifically, they describe a
technique to construct a weighted conflict graph to represent
interference constraints.  They briefly describe methods to compute
lower and upper bounds on throughput and the issues involved therein.
However, the authors do not describe simulation results under the
physical interference model, perhaps due to the tremendous complexity
incurred in solving linear programs for representative network
scenarios.

\begin{remark}
  Under physical interference model, the weighted conflict graph
  $F(V_F,E_F)$ \cite{jain_padhye__impact_interference} is constructed
  from the network as follows.  Let $S_{ij} := \frac{P}{D^\beta(i,j)}$
  denote the received signal power at node $j$ due to the transmission
  from node $i$.  In $F(\cdot)$, a vertex corresponds to a directed
  link $l_{ij}$ (equivalently, node pair $(i,j)$) provided
  $\frac{S_{ij}}{N_0} \geqslant \gamma_c$.  $F(\cdot)$ is a perfect
  graph wherein the weight $w_{ij}^{pq}$ of the directed edge from
  vertex $l_{pq}$ to vertex $l_{ij}$ is given by $w_{ij}^{pq} =
  \frac{S_{pj}}{\frac{S_{ij}}{\gamma_c}-N_0}$.
\end{remark}

We should point out here that, analogous to a conflict graph, an SINR
graph representation of an STDMA wireless network has been proposed by
us in Chapter \ref{ch:sinr_graph}. Furthermore, the authors of
\cite{jain_padhye__impact_interference} do not propose any specific
link scheduling algorithm and use the weighted conflict graph only to
compute bounds on network throughput.  On the other hand, we use an
SINR graph representation of the network under the physical
interference model and develop a link scheduling algorithm with lower
time complexity and demonstrably superior performance.

More specifically, in Chapter \ref{ch:sinr_graph}, we investigate link
scheduling for STDMA wireless networks under the physical interference
model.  Unlike \cite{moscibroda_wattenhofer__complexity_connectivity},
\cite{moscibroda_wattenhofer_zollinger__topology_control}, we assume
that a node transmits at fixed power, i.e., we assume uniform power
assignment. Moreover, unlike
\cite{moscibroda_wattenhofer__complexity_connectivity},
\cite{moscibroda_wattenhofer_zollinger__topology_control}, we do not
assume a minimum distance of unity between any two nodes.
Consequently, our system model is more practical than those of
\cite{moscibroda_wattenhofer__complexity_connectivity},
\cite{moscibroda_wattenhofer_zollinger__topology_control}. Under these
realistic assumptions, we propose a link scheduling algorithm based on
an SINR graph representation of the network. In the SINR
graph\footnote{The SINR graph is analogous to a line graph
  \cite{west__graph_theory} constructed from the communication graph
  representation of the network.}, weights of the edges correspond to
interferences between pairs of nodes. We prove the correctness of the
algorithm and derive its computational complexity. We demonstrate that
the proposed algorithm achieves higher throughput than existing
algorithms, without any increase in computational complexity.

So far, we have provided a brief glimpse into three classes of link
scheduling algorithms, each with its relative merits and demerits.
For example, algorithms based on the protocol interference model have
low computational complexity and are simple to implement, but yield
low network throughput. On the other hand, algorithms based on SINR
graph representation have higher computational complexity and are more
cumbersome to implement, but achieve higher network throughput. Also,
there exist algorithms based on communication graph and SINR
conditions whose performance characteristics lie between these two
classes. Hence, in general, these three classes of algorithms exhibit
a tradeoff between complexity and performance. Finally, algorithms
based on the protocol interference model are better suited to model
WLANs, while the latter two classes of algorithms are better suited to
model wireless mesh networks. For these reasons, we investigate and
develop algorithms from each of these classes in this thesis.

Prior to proposing efficient algorithms in each of these classes, we
seek to address the following question: Is schedule length an
appropriate performance metric for an algorithm that considers the
SINR threshold condition (\ref{eq:sinr_ge_gammac}) as the criterion
for successful packet reception? In other words, should algorithms
based on communication graph and SINR conditions and algorithms based
on SINR graph representation focus on minimizing the schedule length?
We answer this important question in detail in the following section.

\section{Spatial Reuse as Performance Metric}
\label{sec:performance_metric}

In literature, link scheduling algorithms have only focused on
minimizing the schedule length. However, algorithms that minimize the
schedule length do not necessarily maximize network throughput, as
explained in Section \ref{sec:limitations_protocol}.  Thus, from a
perspective of maximizing network throughput observed by the physical
layer, it is imperative to consider a performance metric that takes
into account SINR threshold condition (\ref{eq:sinr_ge_gammac}) as the
criterion for successful packet reception, i.e., a metric also
suitable for the physical interference model.  We propose such a
performance metric, spatial reuse, in this section.  We show that
maximizing spatial reuse directly translates to maximizing network
throughput.

Consider an STDMA wireless network that operates over $(k_2-k_1+1)$
time slots $k_1,k_1+1,\ldots,k_2-1,k_2$. The total number of
successfully scheduled links from slot $k_1$ to slot $k_2$ is
\begin{eqnarray}
\tau[k_1,k_2]
&=& \sum_{i=k_1}^{k_2} \sum_{j=1}^{M_i} 
     I(\mbox{SINR}_{r_{ij}} \geqslant \gamma_c).
\label{eq:total_successful_packets}
\end{eqnarray}
So, the number of successfully scheduled links per time slot from slot
$k_1$ to slot $k_2$ is
\begin{eqnarray}
\eta[k_1,k_2]
&=& \frac{\sum_{i=k_1}^{k_2} \sum_{j=1}^{M_i} 
  I(\mbox{SINR}_{r_{ij}} \geqslant \gamma_c)}{k_2-k_1+1}.
\label{eq:network_throughput}
\end{eqnarray}

We define {\em spatial reuse} $\sigma$ as the limiting value of
$\eta[k_1,k_2]$ (assuming that the limit exists). In other words,
spatial reuse is the limiting value of $\eta[k_1,k_2]$ as the duration
of the time interval becomes very large. Mathematically,
\begin{eqnarray}
\mbox{Spatial Reuse}
&:=& \lim_{|k_2-k_1| \rightarrow \infty} \eta[k_1,k_2], \nonumber \\
\therefore \sigma &=& \lim_{|k_2-k_1| \rightarrow \infty}
\frac{\sum_{i=k_1}^{k_2} \sum_{j=1}^{M_i} I(\mbox{SINR}_{r_{ij}} \geqslant \gamma_c)}{k_2-k_1+1}.
\label{eq:defn_spatial_reuse}
\end{eqnarray}

Assuming a constant data rate of $R$ bits per second on each
successful link and a slot duration of $\tau_s$ seconds, the
(aggregate) network throughput is given by $\sigma R \tau_s$ bits per
second.  Thus, spatial reuse is directly proportional to network
throughput.  Note that spatial reuse is cognizant of the physical
interference model, thereby making it an appropriate performance
metric for the comparison of various link scheduling algorithms.
 
The fact that the interference at a receiver is an increasing function
of the number of concurrent transmissions in a time slot limits the
value of spatial reuse (for a given STDMA network).  More
specifically, if too many transmissions are scheduled in a single time
slot, the interference at some receivers will be high enough to drive
the SINRs below the communication threshold, leading to lower spatial
reuse.  Therefore, for a given STDMA network, there are certain
fundamental limits (upper bounds) on the spatial reuse.

In our system model, we only consider static link schedules, i.e., the
same fixed pattern of slots repeats cyclically. Hence, for our system
model, the equation for spatial reuse, (\ref{eq:defn_spatial_reuse}),
can be simplified to
\begin{eqnarray}
\mbox{Spatial Reuse} 
&=& \sigma 
\;=\; \frac{\sum_{i=1}^C \sum_{j=1}^{M_i} 
   I(\mbox{SINR}_{r_{ij}} \geqslant \gamma_c)}{C}.
\label{eq:spatial_reuse}
\end{eqnarray}

The essence of STDMA is to have a reasonably large number of
concurrent and successful transmissions.  For a network which is
operational for a long period of time, say $L$ time slots, the total
number of successfully received packets is $L\sigma$.  Thus, a high
value of spatial reuse directly translates to higher network
throughput and the number of colors $C$ is relatively unimportant.
Hence, spatial reuse\footnote{Note that spatial reuse in our network
  model is analogous to spectral efficiency in digital communication
  systems. Both performance metrics correspond to the ``rate of data
  transfer'' and are upper bounded by their respective system
  parameters.}  turns out to be a crucial metric for the comparison of
different STDMA algorithms in Chapters \ref{ch:comm_graph},
\ref{ch:sinr_graph} and \ref{ch:broadcastschedule}.

\clearpage{\pagestyle{empty}\cleardoublepage}

\chapter{Point to Point Link Scheduling based on Communication Graph Model}
\label{ch:comm_graph}

We begin our investigation in link scheduling by critically examining
the ArboricalLinkSchedule algorithm proposed in
\cite{ramanathan_lloyd__scheduling_algorithms}.  The algorithm is
based only on the communication graph (protocol interference model)
and seeks to minimize the schedule length.  Though
ArboricalLinkSchedule has good properties such as low computational
complexity, it can yield higher schedule length in practice. Towards
this end, we propose a novel modification to ArboricalLinkSchedule
that results in lower schedule length.  We compare the performance of
the modified algorithm with the ArboricalLinkSchedule algorithm and
derive its run time (computational) complexity in Section
\ref{sec:improvement_als}.  We then propose the
ConflictFreeLinkSchedule point to point link scheduling algorithm,
which is based on communication graph model and SINR conditions, in
Section \ref{sec:high_spatial}. The performance of the proposed
algorithm is compared with existing link scheduling algorithms under
various wireless channel conditions.  We show that the proposed
algorithm has polynomial run time complexity.  Finally, we summarize
the implications of our work.

\section{ArboricalLinkSchedule Algorithm Revisited}
\label{sec:improvement_als}

In this section, we propose a modification to the
ArboricalLinkSchedule point to point link scheduling algorithm.  Since
both the original algorithm and the proposed modification are based on
the protocol interference model, we compare their performance in terms
of average schedule length.  Finally, we also derive the run time
complexity of the modified algorithm.

Our system model and notation are same as described in Section
\ref{sec:protocol_model}. We seek an algorithm that determines a
minimum length point to point link schedule for an STDMA wireless
network under the protocol interference model.  For consistency with
the graph model described in
\cite{ramanathan_lloyd__scheduling_algorithms}, we assume that the
STDMA wireless network $\Phi(\cdot)$ is modeled by the communication
graph $\mathcal G_c(\mathcal V, \mathcal E_c)$ only, i.e.,
interference edges are absent $(\mathcal E=\mathcal E_c)$.

It is well known that, under the protocol interference model, the
problem of determining an optimal schedule, i.e., a minimum length
schedule, is NP-hard \cite{arikan__complexity_results}.  As pointed
out in Section \ref{sec:equivalence_coloring}, this is closely related
to the problem of coloring all edges of the communication graph with
minimum number of colors, which is also known to be NP-hard
\cite{ramanathan_lloyd__scheduling_algorithms}.  Consequently, the
only recourse is to devise approximation algorithms (heuristics) and
show their efficiency theoretically and experimentally.

One such algorithm, ArboricalLinkSchedule, has been described in
\cite{ramanathan_lloyd__scheduling_algorithms}. First, the algorithm
uses the labeler function to label all the vertices of the
communication graph. Next, it partitions the communication graph into
edge-disjoint subgraphs, which are termed as ``oriented graphs''.
Finally, the oriented graphs are colored in sequence. Specifically,
the vertices in each oriented graph are scanned in increasing order of
label and the unique edge associated with each vertex is colored using
the NonConflictingEdge function
\cite{ramanathan_lloyd__scheduling_algorithms}. The labeler function
and the partitioning technique are described later in the section.

In \cite{ramanathan_lloyd__scheduling_algorithms}, the authors appear
to have missed a delicate point that colors from previously colored
oriented graphs can be used to color the present oriented graph.
Specifically, they use a {\em fresh} set of colors to color each
successive oriented graph. In our opinion, the authors employ this
method to upper bound the number of colors used by the algorithm
(\cite{ramanathan_lloyd__scheduling_algorithms}, Lemma 3.4) and thus
derive the running time complexity of the algorithm
(\cite{ramanathan_lloyd__scheduling_algorithms}, Theorem 3.3).
However, such a heuristic can potentially lead to a higher number of
colors (and higher schedule length) in practice.

Therefore, we propound a modification to the ArboricalLinkSchedule
algorithm that {\em reuses} colors from previously colored oriented
graphs to colors the current oriented graph.  The resulting schedule
length will always be {\em lower} than that of ArboricalLinkSchedule,
leading to potentially higher throughput.  Our proposed link
scheduling algorithm is ALSReuseColors, which considers the
communication graph $\mathcal G_c({\mathcal V},\mathcal E)$ and is
described in Algorithm \ref{algo:als_reuse_colors}.

In Phase 1, we label all the vertices using the labeler function
\cite{ramanathan_lloyd__scheduling_algorithms}. The labeler function
is reproduced in Algorithm \ref{func:labeler2} for convenience.  It is
a recursive function that assigns a unique label (from $1$ to $N$) to
every vertex of the communication graph.  Let $L(w)$ denote the label
assigned to vertex $w$.  The notation $\mathcal G_r \setminus \{u\}$
denotes the graph that results when vertex $u$ and all its incident
edges are removed from graph $\mathcal G_r(\cdot)$.  At every step in
the recursion, it chooses the minimum degree vertex $u$ in the
residual graph $\mathcal G_r(\cdot)$ and assigns it the highest label
that has not been assigned so far.  Note that vertices with lower
degree tend to be assigned higher labels.  The labeler function
ensures that, for any given node, the number of neighbors with lower
labels is much lower than the number of vertices in $\mathcal
G_c(\cdot)$.

In Phase 2, the communication graph $\mathcal G_c(\cdot)$ is
decomposed into what are called as out-oriented and in-oriented graphs
$T_1,T_2,\ldots,T_k$, similar to the technique employed in
\cite{ramanathan_lloyd__scheduling_algorithms}.  Recall that an
in-oriented graph is a directed graph in which every vertex has at
most one outgoing edge, while an out-oriented graph is a directed
graph in which every vertex has at most one incoming edge.  Each $T_i$
is a forest\footnote{A graph that is a collection of trees is termed
  as a forest.}  and every edge of $\mathcal G_c(\cdot)$ is in exactly
one of the $T_i$'s.  This decomposition is achieved by partitioning
graph $G_c(\cdot)$, the undirected equivalent of $\mathcal
G_c(\cdot)$, into undirected forests. The number of forests can be
minimized by using techniques from Matroid theory
(\cite{gabow_westermann__forests_frames}, $k$-forest problem).
However, this optimal decomposition requires extensive computation.
Hence, we adopt a faster albeit non-optimal approach of using
successive breadth first searches
\cite{introduction_algorithms__cormen_leiserson_rivest} to decompose
$G_c(\cdot)$ into undirected forests. Each undirected forest is
further mapped to two directed forests. In one forest, the edges in
every connected graph point away from the root and every vertex has at
most one incoming edge, thus producing an out-oriented graph.  In the
other forest, the edges in every connected graph point toward the root
and every vertex has at most one outgoing edge, thus producing an
in-oriented graph.

In Phase 3, the oriented graphs are considered sequentially.  For each
oriented graph, the vertices are considered in increasing order of
label and the unique edge associated with each vertex is colored using
the NCEReuseColors function.  The NCEReuseColors function is explained
in Algorithm \ref{hywsl}.  For the edge $x$ under consideration, it
discards any color from any oriented graph that has an edge with a
primary or secondary conflict with $x$. It returns the least color
among the residual set of non-conflicting colors from all oriented
graphs colored so far. If no non-conflicting color from any oriented
graph is found, it returns a new color.

\begin{algorithm}
\caption{ALSReuseColors}
\label{algo:als_reuse_colors}
\begin{algorithmic}[1]
\STATE {\bf input:} Directed communication graph $\mathcal G_c(\mathcal V, \mathcal E)$
\STATE {\bf output:} A coloring $C: \mathcal E \rightarrow \{1,2,\ldots\}$
\STATE $n \leftarrow \mbox{labeler}(\mathcal G_c)$ \COMMENT{Phase 1}
\STATE use successive breadth first searches to partition $\mathcal G_c(\cdot)$ into oriented graphs $T_i$, $1 \leqslant i \leqslant k$ \COMMENT{Phase 2}
\FOR[Phase 3 begins]{$i \leftarrow 1 \mbox{ to } k$}
\FOR{$j \leftarrow 1 \mbox{ to } n$}
\IF{$T_i$ is out-oriented}
\STATE let $x=(s,d)$ be such that $L(d)=j$
\ELSE
\STATE let $x=(s,d)$ be such that $L(s)=j$
\ENDIF
\STATE $C(x) \leftarrow \mbox{NCEReuseColors}(x)$
\ENDFOR
\ENDFOR \COMMENT{Phase 3 ends}
\end{algorithmic}
\end{algorithm}

\begin{algorithm}
\caption{integer labeler($\mathcal G_r$)}
\label{func:labeler2}
\begin{algorithmic}[1]
\IF{$\mathcal G_r(\cdot)$ is not empty}
\STATE let $u$ be a vertex of $\mathcal G_r(\cdot)$ of minimum degree
\STATE $L(u) \leftarrow 1 + \mbox{labeler}(\mathcal G_r \setminus \{u\})$
\ELSE
\STATE return $0$
\ENDIF
\end{algorithmic}
\end{algorithm}

\begin{algorithm}
\caption{integer NCEReuseColors($x$)}
\label{hywsl}
\begin{algorithmic}[1]
\STATE {\bf input:} Directed communication graph $\mathcal G_c(\mathcal V,\mathcal E)$
\STATE {\bf output:} A non-conflicting color
\STATE ${\mathcal C} \leftarrow \mbox{set of existing colors}$
\STATE $\mathcal C_1 \leftarrow \{C(h):$ $h$ is colored and $x$ and $h$ have a 
primary edge conflict$\}$
\STATE $\mathcal C_2 \leftarrow \{C(h):$ $h$ is colored and $x$ and $h$ have a 
secondary edge conflict$\}$
\STATE ${\mathcal C}_{nc} = {\mathcal C} \setminus \{\mathcal C_1 \cup \mathcal C_2\}$
\IF{${\mathcal C}_{nc} \neq \phi$}
\STATE return the least color $\in$ $\mathcal C_{nc}$
\ELSE
\STATE return $|\mathcal C|+1$
\ENDIF
\end{algorithmic}
\end{algorithm}

\subsection{Performance Results}
\label{subsec:performance_als_reuse}

In the simulation experiment, every node location is generated
randomly, using a uniform distribution for its $X$ and $Y$ coordinates
in the deployment area. We assume that the deployment region is a
square of length $L$. Thus, if $(X_j,Y_j)$ are the Cartesian
coordinates of $j^{th}$ node, then $X_j \sim U[0,L]$ and $Y_j \sim
U[0,L]$.  The values chosen for system parameters $P$, $\gamma_c$,
$\beta$ and $N_0$, are prototypical values of system parameters in
wireless networks \cite{kim_lim_hou__improving_spatial}.  After
generating random positions for $N$ nodes, we have complete
information of $\Phi(\cdot)$. Using (\ref{eq:communication_range}), we
compute the communication range, and then map the STDMA network
$\Phi(\cdot)$ to the communication graph $\mathcal G_c(\cdot)$.  Once
the schedule $\Psi(\cdot)$ is computed by every algorithm, we know its
schedule length $|\mathcal C|$.  For a given set of parameters
$(N,L,R_c)$, we calculate the average schedule length by averaging
$|\mathcal C|$ over 1000 randomly generated networks.  Keeping all
other parameters fixed, we observe the effect of increasing the number
of nodes on the average schedule length.  In our experiments, we
compare the performance of the following algorithms:
\begin{itemize}
\item ArboricalLinkSchedule
  \cite{ramanathan_lloyd__scheduling_algorithms},

\item Proposed ALSReuseColors.
\end{itemize}

\begin{figure}[thbp]
  \centering
  \includegraphics[width=5in]{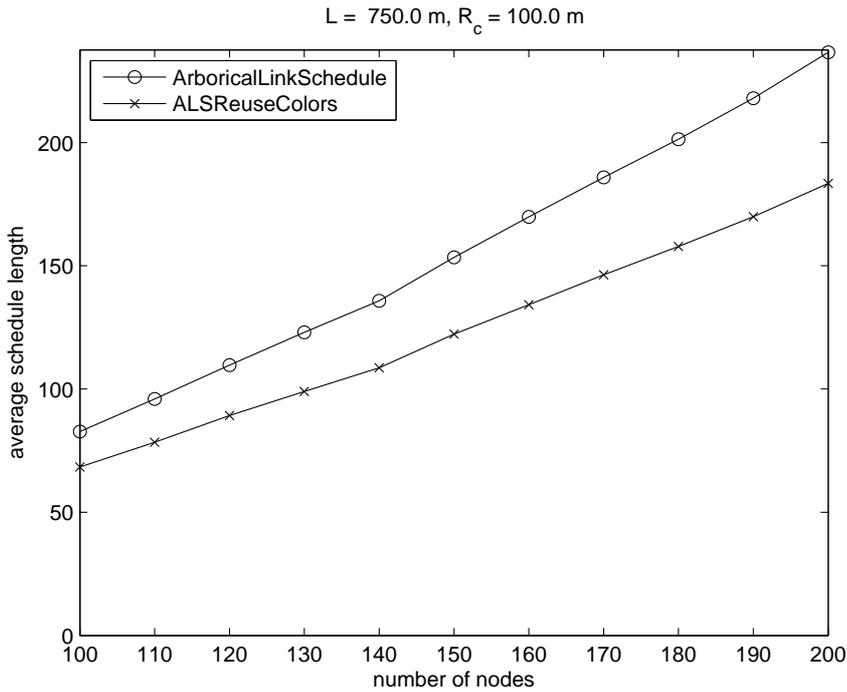}
  \caption{Schedule length vs. number of nodes.}
  \label{fig:schedule_length_als_reuse}
\end{figure}

We assume that $P=10$ mW, $\beta = 4$, $N_0=-90$ dBm and $\gamma_c =
20$ dB. From (\ref{eq:communication_range}), we obtain $R_c=100$ m.
We assume that $L=750$ m, and vary the number of nodes from 100 to 200
in steps of 10. Figure \ref{fig:schedule_length_als_reuse} plots the
average schedule length vs. number of nodes for both the algorithms.

For both the algorithms, we observe that average schedule length
increases almost linearly with the number of nodes. The average
schedule length of ALSReuseColors is about $23\%$ lower than that of
ArboricalLinkSchedule.

Note that an increase in the number of nodes in a given geographical
area leads to an increase in the number of edges incident on a vertex
and a subsequent increase in the number of oriented graphs.
ArboricalLinkSchedule, which is based on using a fresh set of colors
for each oriented graph, requires increasingly higher number of colors
to color the communication graph compared to ALSReuseColors.
Consequently, the gap between the average schedule lengths increases
with number of nodes in Figure \ref{fig:schedule_length_als_reuse}.

\subsection{Analytical Results}
\label{subsec:simulation_als_reuse}

We now derive upper bounds on the running time (computational)
complexity of the ALSReuseColors algorithm.  With respect to the
communication graph ${\mathcal G}_c({\mathcal V},\mathcal E)$, let:
\begin{eqnarray*}
e    &=& \mbox{number of edges},\\
v    &=& \mbox{number of vertices},\\
\rho &=& \mbox{maximum degree of any vertex},\\
\theta 
&=&  \mbox{thickness of the graph}\\
&:=& \mbox{minimum number of planar graphs into which the undirected 
       equivalent of ${\mathcal G}_c(\cdot)$}\\
&&   \mbox{can be partitioned},\\
\omega 
&=& \mbox{maximum number of neighbors with lower labels (for any vertex)}.
\end{eqnarray*}

Recall that the modified algorithm partitions the communication graph
$\mathcal G_c(\cdot)$ into oriented graphs $T_1,T_2,\ldots,T_k$, and
colors the oriented graphs in that order. $T_1$ is termed as the {\em
  first oriented graph}, while any oriented graph $T_j$, where $2
\leqslant j \leqslant k$, is termed as a {\em subsequent oriented
  graph}.

\begin{lemma}
  Suppose that each vertex of the first oriented graph $T_1$ has at
  most $\omega$ neighbors with lower labels. Then, $T_1$ may be
  colored using no more than $O(\omega\rho)$ colors.
\end{lemma}

\begin{proof}
  This is similar to the proof of Lemma 3.1 in
  \cite{ramanathan_lloyd__scheduling_algorithms}.
\end{proof}

\begin{lemma}
  \label{lem:subsequent_colors}
  Any subsequent oriented graph $T_j$, where $2 \leqslant j \leqslant
  k$, can be colored using no more than $O(\rho^2)$ colors.
\end{lemma}

\begin{proof}

\begin{figure}[thbp]
  \centering
  \includegraphics[width=5in]{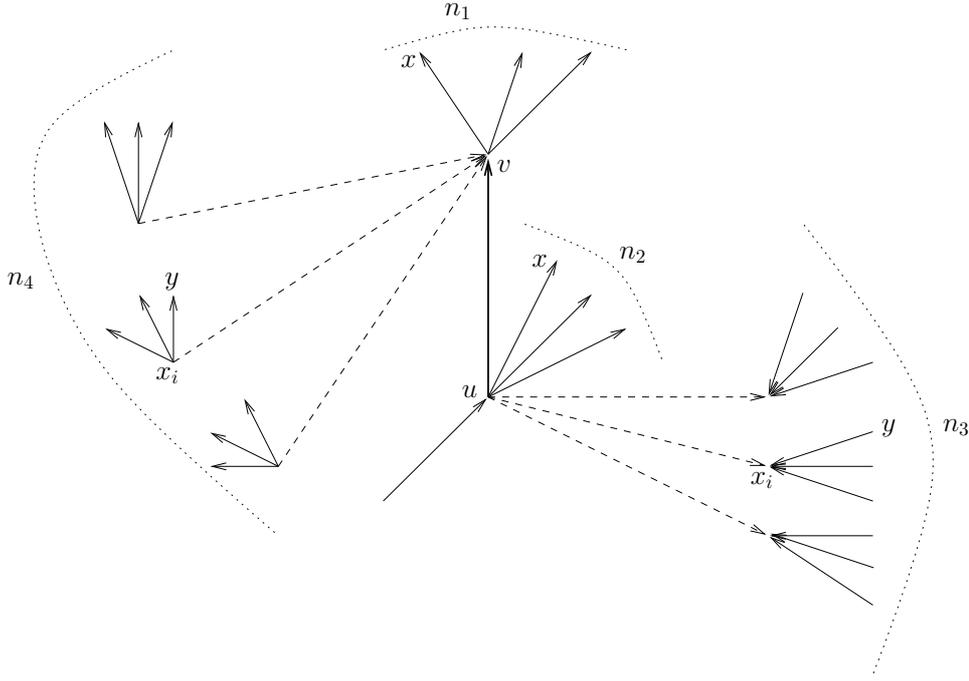}
  \caption{Potential conflicting edges when coloring edge $(u,v)$.}
  \label{fig:conflicting_edges}
\end{figure}

We prove the lemma for an out-oriented graph. A similar proof holds
for an in-oriented graph.  Let $\mathcal G_c$ be partitioned into
edge-disjoint oriented graphs $T_1,\ldots,T_k$.  Consider the coloring
of edge $(u,v)$ in $j^{th}$ oriented graph $T_j$, where $2 \leqslant j
\leqslant k$, as shown in Figure \ref{fig:conflicting_edges}.  Now,
edges of previously colored oriented graphs $T_1,\ldots,T_{j-1}$ must
also be considered for potential edge conflicts with edge $(u,v)$ of
$T_j$.  Define
\begin{eqnarray*}
S_1 &:=& 
 \Big\{ (v,x): (v,x) \in \bigcup_{i=1}^{j} T_i 
 \mbox{ and $(v,x)$ is colored} \Big\}, \\
S_2 &:=& 
 \Big\{(u,x): (u,x) \in \bigcup_{i=1}^{j} T_i 
 \mbox{ and $(u,x)$ is colored} \Big\}, \\
S_3 &:=& 
 \Big\{(y,x_i): (y,x_i) \in \bigcup_{i=1}^{j} T_i 
 \mbox{ and $(y,x_i)$ is colored and } (u,x_i) \in \mathcal G_c \Big\}, \\
S_4 &:=& 
 \Big\{(x_i,y): (x_i,y) \in \bigcup_{i=1}^{j} T_i 
 \mbox{ and $(x_i,y)$ is colored} 
 \mbox{ and } (x_i,v) \in \mathcal G_c \Big\}.
\end{eqnarray*}
Any edge which can cause a primary edge conflict with $(u,v)$ must
belong to $S_1$ or $S_2$. Also, any edge which can cause a secondary
edge conflict with $(u,v)$ must belong to $S_3$ or $S_4$. Let $n_i =
|S_i|$ for $i=1,2,3,4$. The lemma reduces to proving that
$n_1+n_2+n_3+n_4$ is $O(\rho^2)$.

By definition of maximum vertex degree, $n_1 \leqslant \rho-1$ and
$n_2 \leqslant \rho-1$. Thus, $n_1+n_2$ is $O(\rho)$. For the
computation of $n_3$, we must also consider secondary edge conflicts
with edges of previously colored oriented graphs, as shown in Figure
\ref{fig:conflicting_edges}.  The worst-case value of $n_3$ is
$(\rho-1)(\rho-1)$. Thus, $n_3$ is $O(\rho^2)$. Similarly, by
considering secondary edge conflicts with edges of previously colored
oriented graphs, it follows that $n_4$ is $O(\rho^2)$. Finally,
$n_1+n_2+n_3+n_4$ is $O(\rho^2)$.
\end{proof}

\begin{lemma}
  For the first oriented graph $T_1$, the running time of Phase 3 of
  ALSReuseColors is $O(v\omega\rho)$.
\end{lemma}

\begin{proof}
  This is similar to the proof of Lemma 3.2 in
  \cite{ramanathan_lloyd__scheduling_algorithms}.
\end{proof}

\begin{lemma}
  \label{lem:subsequent_complexity}
  For any subsequent oriented graph $T_j$, where $2 \leqslant j
  \leqslant k$, the running time of Phase 3 of ALSReuseColors is
  $O(v\rho^2)$.
\end{lemma}

\begin{proof}
  From Lemma \ref{lem:subsequent_colors}, for any subsequent oriented
  graph $T_j$, the size of the set of conflicting colors $(\mathcal
  C_1 \cup \mathcal C_2)$ of function NCEReuseColors is $O(\rho^2)$.
  Thus, determining a new color for an edge in Phase 3 of
  ALSReuseColors takes $O(\rho^2)$ steps. Since this is done for every
  label and hence for every vertex, it follows that the overall
  running time of Phase 3 of ALSReuseColors is $O(v\rho^2)$.
\end{proof}

\begin{theorem}
\label{theo:arbitrary_graph}
For an arbitrary graph of thickness $\theta$ and maximum degree
$\rho$, ALSReuseColors has a running time of $O(ev\log
v+v\theta\rho^2)$.
\end{theorem}

\begin{proof}
  The running time of the labeler function is $O(e+v\log v)$ using a
  Fibonacci Heap \cite{fredman_tarjan__fibonacci_heaps}. The
  partitioning method of \cite{gabow_westermann__forests_frames}
  results in a decomposition of a graph of thickness $\theta$ into at
  most $6\theta$ oriented graphs in time $O(ev\log v)$. Thus, $k
  \leqslant 6\theta$.  From Lemma 3.2 in
  \cite{ramanathan_lloyd__scheduling_algorithms}, the first oriented
  graph $T_1$ can be colored in time $O(v\omega\rho)$.  However,
  consider the coloring of $j^{th}$ oriented graph $T_j$, where $2
  \leqslant j \leqslant k$.  From Lemma
  \ref{lem:subsequent_complexity}, $T_j$ can be colored in time
  $O(v\rho^2)$. Hence, the for loop of ALSReuseColors runs in time
  $O(v\theta\rho^2)$. Therefore, the overall running time of
  ALSReuseColors is $O(e + v\log v + ev\log v + v\theta\rho^2)$.
  Since $e + v \log v < ev \log v$ holds for any directed graph
  $\mathcal G_c(\cdot)$ that models a wireless mesh network, the
  overall running time of ALSReuseColors simplifies to $O(ev\log v +
  v\theta\rho^2)$.
\end{proof}

\subsection{Discussion}

In this section, we have considered an STDMA wireless network with
uniform transmission power at all nodes and presented an algorithm for
point to point link scheduling under the protocol interference model.
The proposed algorithm, which is a modification of the
ArboricalLinkSchedule algorithm in
\cite{ramanathan_lloyd__scheduling_algorithms}, models the network by
a communication graph, partitions the communication graph into
edge-disjoint oriented graphs and colors each oriented graph
successively.  However, unlike
\cite{ramanathan_lloyd__scheduling_algorithms}, we reuse colors from
previously colored oriented graphs to color the current oriented
graph. The proposed algorithm results in around $26\%$ lower schedule
length than that of \cite{ramanathan_lloyd__scheduling_algorithms},
albeit at the cost of slightly higher computational
complexity\footnote{The computational complexity of
  ArboricalLinkSchedule is $O(ev \log v + v \theta^2 \rho)$
  \cite{ramanathan_lloyd__scheduling_algorithms}.}. Since schedules
are constructed only once offline and then used by the network for a
long period of time, our approach has the potential of providing
higher long-term network throughput.

For the rest of this chapter, we consider point to point link
scheduling under the physical interference model. The algorithm
developed in this section will be further refined to design a link
scheduling algorithm in the next section.

\section{A High Spatial Reuse Link Scheduling Algorithm}
\label{sec:high_spatial}

In this section, we propose a point to point link scheduling algorithm
based on the communication graph model of an STDMA wireless network as
well as SINR computations.  We adopt spatial reuse as the performance
metric, which has been motivated in Section
\ref{sec:performance_metric}.  We compare the performance of the
proposed algorithm with link scheduling algorithms which utilize a
communication graph model of the network.  We show that the proposed
algorithm achieves higher spatial reuse compared to existing
algorithms, without any increase in computational complexity.

\subsection{Problem Formulation}
\label{subsec:problem_formulation_cfls}

Our system model and notation are exactly as described in Section
\ref{sec:protocol_model}.  A link schedule is {\em feasible} if it
satisfies the following conditions:
\begin{enumerate}
\item 
Operational constraint (\ref{eq:operational_constraint}).
\item
Range constraint: Every receiver is within the communication range
of its intended transmitter, i.e.,
\begin{eqnarray}
D(t_{i,j},r_{i,j}) \leqslant R_c \;\;\forall\;\; i=1,\ldots,C 
 \;\;\forall\;\; j=1,\ldots,M_i.
\label{eq:range_constraint}
\end{eqnarray}
\end{enumerate}

A link schedule $\Psi(\cdot)$ is {\it exhaustive} if every pair of
nodes which are within communication range occur exactly twice in the
link schedule, once with one node being the transmitter and the other
node being the receiver, and during another time slot with the
transmitter-receiver roles interchanged.  Mathematically,
\begin{eqnarray}
D(j,k) \leqslant R_c \Rightarrow j \rightarrow k \in 
 \bigcup_{i=1}^C {\mathcal S}_i
  \;\;\mbox{and}\;\; k \rightarrow j \in \bigcup_{i=1}^C {\mathcal S}_i
  \;\; \forall \;\; 1 \leqslant j < k \leqslant N.
\label{eq:exhaustive_schedule}
\end{eqnarray}

Our aim is to design a low complexity conflict-free STDMA point to
point link scheduling algorithm that achieves high spatial reuse,
where spatial reuse is given by (\ref{eq:spatial_reuse}).  We only
consider STDMA link schedules which are feasible and
exhaustive\footnote{The set of edges in $\mathcal G_c(\cdot)$ to be
  scheduled is determined by a routing algorithm.  For simplicity, we
  only consider exhaustive schedules, i.e., schedules which assign
  exactly one time slot to every directed edge in $\mathcal
  G_c(\cdot)$.}. Thus, our schedules satisfy
(\ref{eq:operational_constraint}), (\ref{eq:conflict_free}),
(\ref{eq:range_constraint}) and (\ref{eq:exhaustive_schedule}).

\subsection{Motivation}
\label{subsec:motivation_cfls}

We briefly describe the essential features of STDMA link scheduling
algorithms.  An STDMA link scheduling algorithm is equivalent to
assigning a unique color to every edge in the communication graph,
such that transmitter-receiver pairs corresponding to communication
edges with the same color are simultaneously active in a particular
time slot, as described in Section \ref{sec:equivalence_coloring}.
The core of a typical link scheduling algorithm consists of the
following functions:
\begin{enumerate}

\item An order in which communication edges are considered for
  coloring.

\item A function which determines the set of all existing colors which
  can be assigned to the edge under consideration without violating
  the problem constraints.

\item A {\em BestColor} rule to determine which color to assign to the
  edge under consideration.
\end{enumerate}
The second function considers only operational and range constraints
in link scheduling algorithms based on the protocol interference model
(equivalently, based on the communication graph).  However, in the
link scheduling algorithm that we propose, SINR constraints are also
taken into account.

Algorithms based on the protocol interference model are inadequate to
design efficient link schedules. This is because the communication
graph ${\mathcal G}_c({\mathcal V},{\mathcal E}_c)$ is a crude
approximation of $\Phi(\cdot)$.  Even the two-tier graph ${\mathcal
  G}({\mathcal V},{\mathcal E}_c \,\cup\, {\mathcal E}_i)$, which is a
better approximation of $\Phi(\cdot)$, leads to low network
throughput, as argued in Section \ref{sec:limitations_protocol}.  On
the other hand, from $\Phi(\cdot)$ and ${\mathcal G}_c(\cdot)$, one
can exhaustively determine the STDMA schedule which yields the highest
spatial reuse. However, this is a combinatorial optimization problem
of prohibitive complexity $(O(|{\mathcal E}_c|^{|{\mathcal E}_c|}))$
and is thus computationally infeasible.

To overcome these problems, we propose a new algorithm for STDMA  
link scheduling under the realistic physical interference
model.  Our algorithm is based on the communication graph model
${\mathcal G}_c({\mathcal V},{\mathcal E}_c)$ as well as SINR
computations.  Motivated by techniques from matroid theory
\cite{lawler__combinatorial_optimization}, we develop a
computationally feasible algorithm with demonstrably high spatial
reuse.  The essence of our algorithm is to partition the set of
communication edges into subsets (forests) and color the edges in each
subset sequentially. The edges in each forest are considered in a
random order for coloring, since randomized algorithms are known to
outperform deterministic algorithms, especially when the
characteristics of the input are not known a priori
\cite{motwani_raghavan__randomized_algorithms}.

A similar matroid-based network partitioning technique is used in
\cite{brzezinski_zussman_modiano__enabling_distributed} to generate
high capacity subnetworks for a distributed throughput maximization
problem in wireless mesh networks.  Techniques from matroid theory
have also been employed to develop efficient heuristics for NP-hard
combinatorial optimization problems in fields such as distributed
computer systems \cite{ramalingom_thulasiraman_das__matroid_theoretic}
and linear network theory \cite{petersen__investigating_solvability}.

\subsection{ConflictFreeLinkSchedule Algorithm}
\label{subsec:cfls_algorithm}

We call the proposed point to point link scheduling algorithm as
ConflictFreeLinkSchedule (CFLS).  The algorithm considers the
communication graph ${\mathcal G}_c({\mathcal V},{\mathcal E}_c)$ and
SINR conditions and is explained in Algorithm \ref{algo:cfls}.

In Phase 1, we label all the vertices randomly. Specifically, if
${\mathcal G}_c(\cdot)$ has $v$ vertices, we perform a random
permutation of the sequence $(1,2,\ldots,v)$ and assign these labels
to vertices with indices $1,2,\ldots,v$ respectively.  $L(u)$ denotes
the label assigned to vertex $u$.

In Phase 2, the communication graph ${\mathcal G}_c(\cdot)$ is
decomposed into what are called out-oriented and in-oriented graphs
$T_1,T_2,\ldots,T_k$ \cite{ramanathan_lloyd__scheduling_algorithms}.
Each $T_i$ is a forest and every edge of ${\mathcal G}_c(\cdot)$ is in
exactly one of the $T_i$'s. This decomposition is achieved by
partitioning graph $G_c(\cdot)$, the undirected equivalent of
${\mathcal G}_c(\cdot)$, into undirected forests. The number of
forests can be minimized by using techniques from Matroid theory
(\cite{gabow_westermann__forests_frames}, $k$-forest problem).
However, this optimal decomposition requires extensive computation.
Hence, we adopt the faster albeit non-optimal approach of using
successive breadth first searches
\cite{introduction_algorithms__cormen_leiserson_rivest} to decompose
$G_c(\cdot)$ into undirected forests. Each undirected forest is
further mapped to two directed forests. In one forest, the edges in
every connected component point away from the root and every vertex
has at most one incoming edge, thus producing an out-oriented graph.
In the other forest, the edges in every connected component point
toward the root and every vertex has at most one outgoing edge, thus
producing an in-oriented graph.

In Phase 3, the oriented graphs are considered sequentially.  For each
oriented graph, vertices are considered in increasing order by label
and the unique edge associated with each vertex is colored using the
FirstConflictFreeColor (FCFC) function.

The FCFC function is explained in Algorithm
\ref{func:first_conflict_free}. For the edge under consideration $x$,
it discards any color that has an edge with a primary conflict with
$x$.  Among the residual set of colors, we choose the first color such
that the resulting SINRs at the receiver of $x$ and the receivers of
all co-colored edges are no less than the communication threshold
$\gamma_c$. If no such color is found, we assign a new color to $x$.
Hence, this function guarantees that the ensuing schedule is
conflict-free.

\begin{algorithm}
\caption{ConflictFreeLinkSchedule (CFLS)}
\label{algo:cfls}
\begin{algorithmic}[1]
\STATE {\bf input:} STDMA network $\Phi(\cdot)$, communication graph 
  ${\mathcal G}_c(\cdot)$
\STATE {\bf output:} A coloring $C: {\mathcal E}_c \rightarrow \{1,2,\ldots\}$
\STATE label the vertices of ${\mathcal G}_c$ randomly \COMMENT{Phase 1}
\STATE use successive breadth first searches to partition ${\mathcal G}_c$ into oriented graphs $T_i$, $1 \leqslant i \leqslant k$ \COMMENT{Phase 2}
\FOR[Phase 3 begins]{$i \leftarrow 1 \mbox{ to } k$}
\FOR{$j \leftarrow 1 \mbox{ to } n$}
\IF{$T_i$ is out-oriented}
\STATE let $x=(s,d)$ be such that $L(d)=j$
\ELSE
\STATE let $x=(s,d)$ be such that $L(s)=j$
\ENDIF
\STATE $C(x) \leftarrow \mbox{FirstConflictFreeColor}(x)$
\ENDFOR
\ENDFOR \COMMENT{Phase 3 ends}
\end{algorithmic}
\end{algorithm}

\begin{algorithm}
\caption{integer FirstConflictFreeColor($x$)}
\label{func:first_conflict_free}
\begin{algorithmic}[1]
\STATE {\bf input:} STDMA network $\Phi(\cdot)$, communication graph 
  ${\mathcal G}_c(\cdot)$
\STATE {\bf output:} A conflict-free color
\STATE ${\mathcal C} \leftarrow \mbox{set of existing colors}$
\STATE ${\mathcal C}_c \leftarrow \{C(h):h \in {\mathcal E}_c$, $h$ is colored, $x$ and $h$ have a primary edge conflict$\}$
\STATE ${\mathcal C}_{cf} = {\mathcal C} \setminus {\mathcal C}_c$
\FOR{$i \leftarrow 1 \mbox{ to } |{\mathcal C}_{cf}|$}
\STATE $r \leftarrow i^{th} \mbox{ color in } {\mathcal C}_{cf}$
\STATE $E_i \leftarrow \{h:h \in {\mathcal E}_c, C(h)=r\}$
\STATE $C(x) \leftarrow r$
\IF{SINR at all receivers of $E_i \cup \{x\}$ exceed $\gamma_c$}
\STATE return $r$
\ENDIF
\ENDFOR
\STATE return $|\mathcal C|+1$
\end{algorithmic}
\end{algorithm}

\subsection{Performance Results}
\label{subsec:performance_results}

\subsubsection{Simulation Model}
\label{sub2sec:simulation_model}

In the simulation experiments, the location of every node is generated
randomly, using a uniform distribution for its $X$ and $Y$
coordinates, in the deployment area. For a fair comparison of our
algorithm with the Truncated Graph-Based Scheduling Algorithm (TGSA)
\cite{behzad_rubin__performance_graph}, we assume that the deployment
region is a circular region of radius $R$. Thus, if $(X_j,Y_j)$ are
the Cartesian coordinates of $j^{th}$ node, $j=1,\ldots,N$, then $X_j
\sim U[-R,R]$ and $Y_j \sim U[-R,R]$ subject to $X_j^2 + Y_j^2
\leqslant R^2$. Equivalently, if $(R_j,\Theta_j)$ are the polar
coordinates of $j^{th}$ node, then $R_j^2 \sim U[0,R^2]$ and $\Theta_j
\sim U[0,2\pi]$.  After generating random positions for $N$ nodes, we
have complete information of $\Phi(\cdot)$. Using
(\ref{eq:communication_range}) and (\ref{eq:interference_range}), we
compute the communication and interference radii, and then map the
network $\Phi(\cdot)$ to the two-tier graph ${\mathcal G}({\mathcal
  V},{\mathcal E}_c \,\cup\, {\mathcal E}_i)$.  Once the link schedule
is computed by an algorithm, $\sigma$ is computed using
(\ref{eq:spatial_reuse}).  System parameters are chosen based on their
prototypical values in wireless mesh networks
\cite{kim_lim_hou__improving_spatial}.  For a given set of system
parameters, we calculate the average spatial reuse by averaging
$\sigma$ over 1000 randomly generated networks.  Keeping all other
parameters fixed, we observe the effect of increasing the number of
nodes $N$ on the average spatial reuse.

In our experiments, we compare the performance of the following algorithms:
\begin{itemize}
\item ArboricalLinkSchedule (ALS)
  \cite{ramanathan_lloyd__scheduling_algorithms},

\item Truncated Graph-Based Scheduling Algorithm\footnote{In Truncated
    Graph-Based Scheduling Algorithm, for the computation of optimal
    number of transmissions $M^*$, we follow the method described in
    \cite{behzad_rubin__performance_graph}. Since $0 < \xi <
    \frac{N_0}{P}$, we assume that $\xi = 0.9999\frac{N_0}{P}$ and
    compute successive Edmundson-Madansky (EM) upper bounds
    \cite{madansky__inequalities_stochastic},
    \cite{dokov_morton__higher_order}, till the difference between
    successive EM bounds is less than $0.3\%$. We have experimentally
    verified that only high values of $\xi$ lead to reasonable values
    for $M^*$, whereas low values of $\xi$, say $\xi =
    0.1\frac{N_0}{P}$, lead to the extremely conservative value of
    $M^*=1$ in most cases.}  (TGSA)
  \cite{behzad_rubin__performance_graph},

\item GreedyPhysical (GP)
  \cite{brar_blough_santi__computationally_efficient},

\item Proposed ConflictFreeLinkSchedule (CFLS).
\end{itemize}

\subsubsection{Performance Comparison under Path Loss Model}
\label{sub2sec:performance_pathloss}

In the first set of experiments (Experiment 1), we assume that $R=500$
m, $P=10$ mW, $\beta = 4$, $N_0=-90$ dBm, $\gamma_c = 20$ dB and
$\gamma_i=10$ dB \cite{kim_lim_hou__improving_spatial}.  Thus,
$R_c=100$ m and $R_i=177.8$ m.  We vary the number of nodes from 30 to
110 in steps of 5. Figure \ref{fig:expt1_path_loss} plots the average
spatial reuse vs. number of nodes for all the algorithms.

\begin{figure}[thbp]
  \centering
  \includegraphics[width=5in]{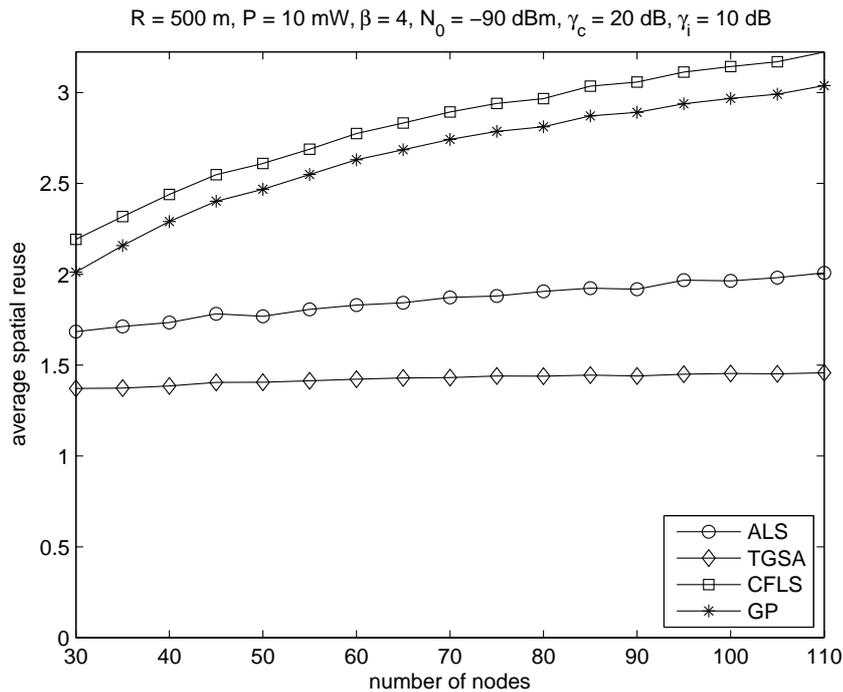}
  \caption{Spatial reuse vs. number of nodes for Experiment 1.}
  \label{fig:expt1_path_loss}
\end{figure}

In the second set of experiments (Experiment 2), we assume that
$R=700$ m, $P=15$ mW, $\beta = 4$, $N_0=-85$ dBm, $\gamma_c = 15$ dB
and $\gamma_i=7$ dB. Thus, $R_c=110.7$ m and $R_i=175.4$ m.  We vary
the number of nodes from 70 to 150 in steps of 5. Figure
\ref{fig:expt2_path_loss} plots the average spatial reuse vs. number
of nodes for all the algorithms.

\begin{figure}[thbp]
  \centering
  \includegraphics[width=5in]{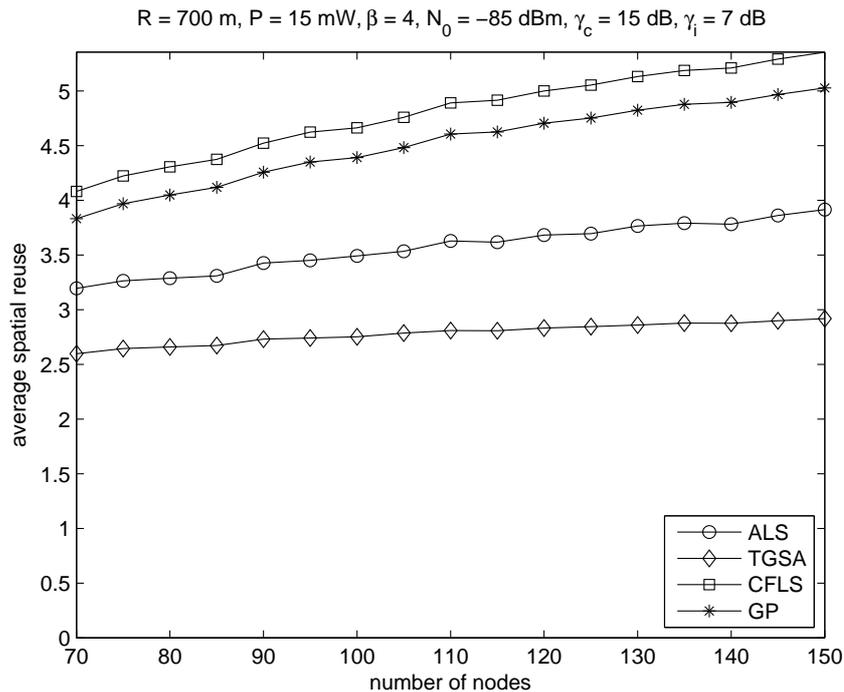}
  \caption{Spatial reuse vs. number of nodes for Experiment 2.}
  \label{fig:expt2_path_loss}
\end{figure}

For the ALS algorithm, we observe that spatial reuse increases very
slowly with increasing number of nodes.

For the TGSA algorithm, we observe that spatial reuse is $18$-$27\%$
lower than that of ALS and $30$-$55\%$ lower than that of GP.  A
plausible explanation for this behavior is as follows.  The basis for
TGSA is the computation of $M^*$, the optimal number of transmissions
in every slot \cite{behzad_rubin__performance_graph}.  $M^*$ is
determined by maximizing a lower bound on the expected number of
successful transmissions in a time slot.  Since the partitioning of a
maximal independent set of communication arcs into subsets of
cardinality at most $M^*$ is arbitrary and not geography-based, there
could be scenarios where the transmissions scheduled in a subset are
in the vicinity of each other, resulting in moderate to high
interference.  In essence, maximizing this lower bound does not
necessarily translate to maximizing the number of successful
transmissions in a time slot.  Also, due to its design, the TGSA
algorithm yields higher number of colors compared to ALS and GP.

Though the GP algorithm is based on communication graph and SINR
conditions, it yields slightly lower spatial reuse than CFLS.  A
possible reason for this observation is as follows. The GP algorithm
colors edges of the communication graph in the decreasing order of
interference number. The interference number of edge $e$ is the number
of edges $e_i$ such that, if $(e,e_i)$ are scheduled simultaneously,
then the SINR threshold condition (\ref{eq:conflict_free}) is violated
along one or both links. Edges with higher interference number tend to
be located towards the center of the deployment region.  Since these
edges are colored first, a large number of colors are utilized in the
initial stages of the algorithm, lead to potentially higher schedule
length and lower spatial reuse. A better technique would be
successively examine edges at the centre and the periphery, which is
achieved by the partition technique employed by CFLS.

For the proposed CFLS algorithm, we observe that spatial reuse
increases steadily with increasing number of nodes and is about $15\%$
higher than the spatial reuse of ALS, TGSA and GP.

\subsubsection{Performance Comparison under Realistic Conditions}
\label{sub2sec:performance_realistic}

In a realistic wireless environment, channel impairments like
multipath fading and shadowing affect the received SINR at a receiver
\cite{sklar__rayleigh_fading}.  In this section, we compare the
performance of the ALS, TGSA, GP and CFLS algorithms in a wireless
channel which experiences Rayleigh fading and lognormal shadowing.

In the absence of fading and shadowing, the SINR at receiver $r_{i,j}$
is given by (\ref{eq:definition_sinr}). We assume that every algorithm
(ALS, TGSA, GP and CFLS) considers only path loss in the channel prior
to constructing the two-tier graph ${\mathcal G}({\mathcal
  V},{\mathcal E}_c \,\cup\, {\mathcal E}_i)$ and computing the link
schedule.

However, for computing the average spatial reuse of each algorithm, we
take into account fading and shadowing channel gains between each pair
of nodes. More specifically, for computing the spatial reuse using
(\ref{eq:spatial_reuse}), the (actual) SINR at receiver $r_{i,j}$ is
given by
\begin{eqnarray}
{\mbox{SINR}}_{r_{i,j}} = 
 \frac{\frac{P}{D^{\beta}(t_{i,j},r_{i,j})}V(t_{i,j},r_{i,j})10^{W(t_{i,j},r_{i,j})}}{N_0+\sum_{\stackrel{k=1}{k\neq j}}^{M_i}\frac{P}{D^{\beta}(t_{i,k},r_{i,j})}V(t_{i,k},r_{i,j})10^{W(t_{i,k},r_{i,j})}},
\end{eqnarray}
where random variables $V(\cdot)$ and $W(\cdot)$ correspond to channel
gains due to Rayleigh fading and lognormal shadowing respectively. We
assume that $\{V(k,l)|1\leqslant k,l \leqslant N, k \neq l\}$ are
independent and identically distributed (i.i.d.) random variables with
probability density function (pdf)
\cite{tse_viswanath__fundamentals_wireless}
\begin{eqnarray}  
f_V(v) &=& \frac{1}{\sigma_V^2}e^{\frac{-v}{\sigma_V^2}}u(v),
\end{eqnarray}
where $u(\cdot)$ is the unit step function.  Also,
$\{W(k,l)|1\leqslant k,l \leqslant N, k\neq l\}$ are assumed to be
i.i.d. zero mean Gaussian random variables with pdf
\cite{goldsmith__wireless_communications}
\begin{eqnarray}
f_W(w) &=& \frac{1}{\sqrt{2\pi}\sigma_W}e^{\frac{-w^2}{2\sigma_W^2}}.
\end{eqnarray}
Random variables $V(\cdot)$ and $W(\cdot)$ are independent of each
other and also independent of the node locations.

The simulation model and experiments are exactly as described before.
In the simulations, we assume $\sigma_V^2 = \sigma_W^2 = 1$.  For
Experiment 1, Figure \ref{fig:expt1_fading} plots the average spatial
reuse vs. number of nodes for all the algorithms.  For Experiment 2,
Figure \ref{fig:expt2_fading} plots the average spatial reuse vs.
number of nodes for all the algorithms.

\begin{figure}[thbp]
  \centering
  \includegraphics[width=5in]{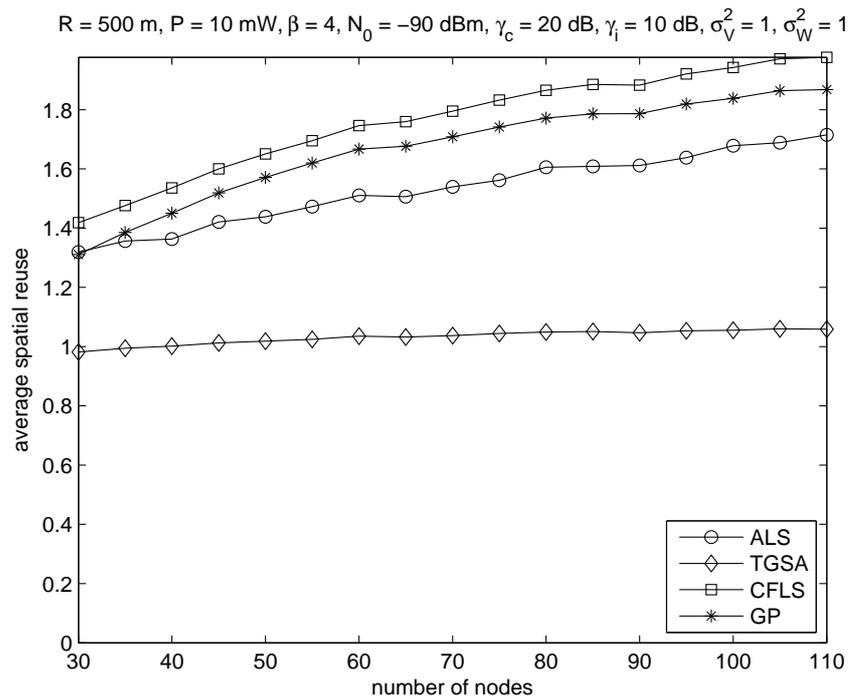}
  \caption{Spatial reuse vs. number of nodes for Experiment 1 under
    multipath fading and shadowing channel conditions.}
  \label{fig:expt1_fading}
\end{figure}

\begin{figure}[thbp]
  \centering
  \includegraphics[width=5in]{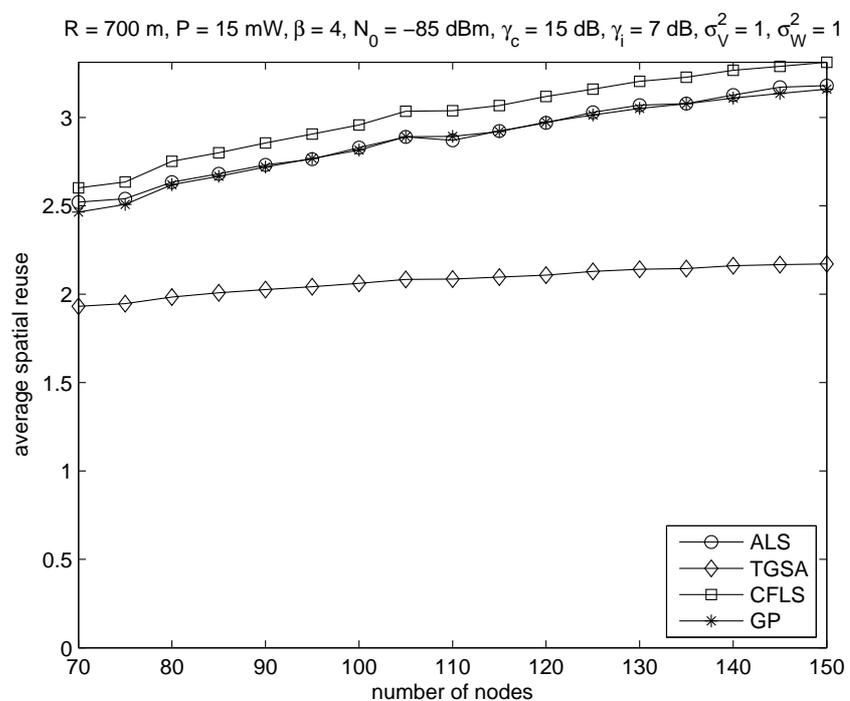}
  \caption{Spatial reuse vs. number of nodes for Experiment 2 under
    multipath fading and shadowing channel conditions.}
  \label{fig:expt2_fading}
\end{figure}

From Figures \ref{fig:expt1_path_loss}, \ref{fig:expt2_path_loss},
\ref{fig:expt1_fading} and \ref{fig:expt2_fading}, we observe that
spatial reuse decreases by $20$-$40\%$ in a channel experiencing
multipath fading and shadowing effects. A plausible explanation for
this observation is as follows. Since the channel gains between every
pair of nodes are independent of each other, it is reasonable to
assume that the interference power at a typical receiver remains
almost the same as in the non-fading case. This is because, even if
the power received from few unintended transmitters is low, the power
received from other unintended transmitters will be high (on average);
thus the interference power remains constant.  Consequently, the
change in SINR is determined by the change in received signal power
only. If the received signal power is higher compared to the
non-fading case, the transmission is anyway successful and spatial
reuse remains unchanged (see (\ref{eq:spatial_reuse})). However, if
the received signal power is lower, the transmission is now
unsuccessful and spatial reuse decreases. Hence, on average, the
spatial reuse decreases.

Finally, from Figures \ref{fig:expt1_fading} and
\ref{fig:expt2_fading}, we observe that the proposed CFLS algorithm
achieves $5$-$17\%$ higher spatial reuse than the ALS and GP
algorithms and $40$-$80\%$ higher spatial reuse than the TGSA
algorithm, under realistic wireless channel conditions.

\subsection{Analytical Results}
\label{subsec:cfls_complexity}

In this section, we derive upper bounds on the running time
(computational) complexity of ConflictFreeLinkSchedule algorithm.  We
use the following notation with respect to the communication graph
${\mathcal G}_c({\mathcal V},{\mathcal E}_c)$:
\begin{eqnarray*}
e &=& \mbox{number of communication edges},\\
v &=& \mbox{number of vertices},\\
\theta &=& \mbox{thickness of the graph}\\
&:=& \mbox{minimum number of planar graphs into which the undirected equivalent of ${\mathcal G}_c(\cdot)$}\\
&& \mbox{can be partitioned}.
\end{eqnarray*}

\begin{figure}
  \centering \subfigure[Experiment 1] {
    \label{fig:thickness_expt1}
    \includegraphics[width=5in]{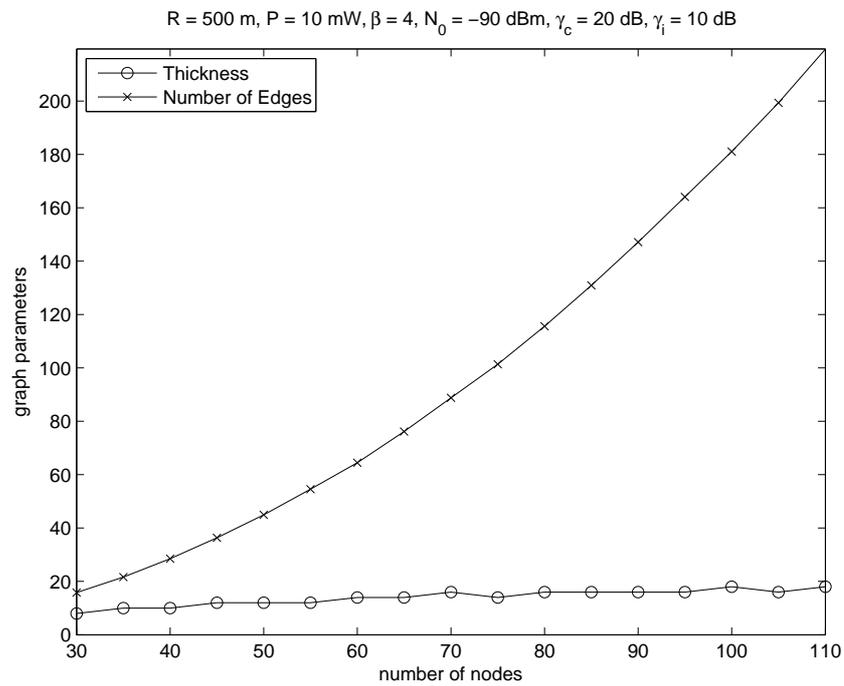}
  } \hspace{0.2in} \subfigure[Experiment 2] {
    \label{fig:thickness_expt2}
    \includegraphics[width=5in]{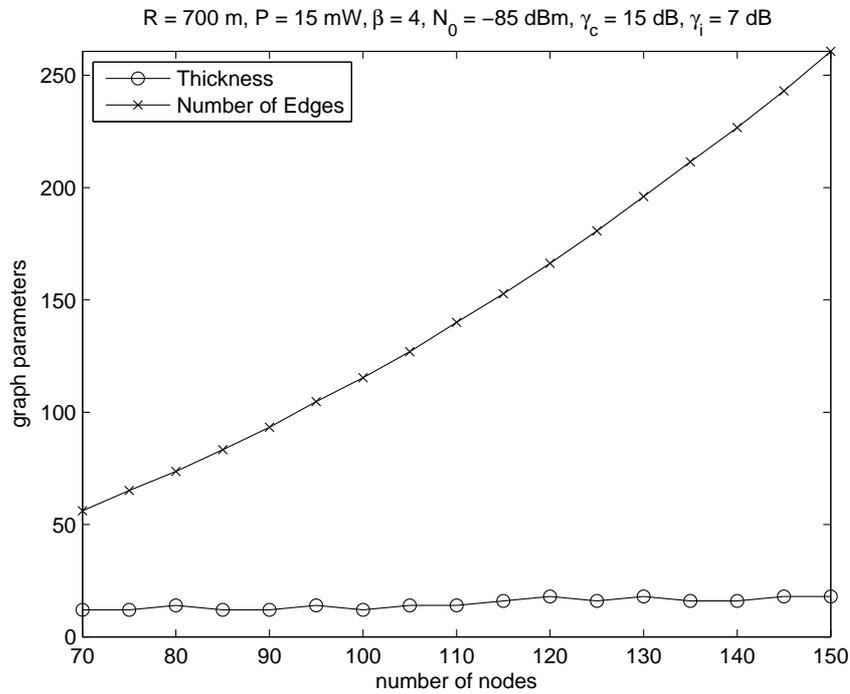}
  }
  \caption{Comparison of thickness and number of edges with number of
    vertices.}
  \label{fig:thickness_vs_vertices}
\end{figure}

Before we prove our results, it is instructive to observe Figure
\ref{fig:thickness_vs_vertices}, which shows the variation of $\theta$
and $e$ with $v$ for the two experiments described in Section
\ref{subsec:performance_results}.  Since determining the thickness of
a graph is NP-hard \cite{mutzel_odenthal__thickness_graphs}, each
value of $\theta$ in Figure \ref{fig:thickness_vs_vertices} is an
upper bound on the actual thickness based on the number of forests
into which the undirected equivalent of the communication graph has
been decomposed using successive breadth first searches.  We observe
that the graph thickness increases very slowly with the number of
vertices $(\theta \ll v)$, while the number of edges increases
super-linearly with the number of vertices.

\begin{lemma}
  An oriented graph $T$ can be colored using no more than $O(v)$
  colors using ConflictFreeLinkSchedule.
\label{lem:oriented_colors}
\end{lemma}
\begin{proof}
  Since an oriented graph with $v$ vertices has at most $v$ edges, the
  edges of $T$ can be colored with at most $v$ colors.
\end{proof}

\begin{lemma}
  For an oriented graph $T$, the running time of
  ConflictFreeLinkSchedule is $O(v^2)$.
\label{lem:oriented_cfls}
\end{lemma}
\begin{proof}
  Assuming that an element can be chosen randomly and uniformly from a
  finite set in unit time
  (\cite{motwani_raghavan__randomized_algorithms}, Chapter 1), the
  running time of Phase 1 can be shown to be $O(v)$.  Since there is
  only one oriented graph, Phase 2 runs in time $O(v)$. In Phase 3,
  the unique edge associated with the vertex under consideration is
  assigned a color using FirstConflictFreeColor. From Lemma
  \ref{lem:oriented_colors}, the size of the set of colors to be
  examined $|{\mathcal C}_c \cup {\mathcal C}_{cf}|$ is $O(v)$.  In
  FirstConflictFreeColor, the SINR is checked only once for every
  colored edge in the set $\bigcup_{i=1}^{|{\mathcal C}_{cf}|} E_i$
  and at most $v$ times for the edge under consideration $x$. With a
  careful implementation, FirstConflictFreeColor runs in time $O(v)$.
  So, the running time of Phase 3 is $O(v^2)$.  Thus, the total
  running time is $O(v^2)$.
\end{proof}

\begin{theorem}
  For an arbitrary graph $\mathcal G$, the running time of
  ConflictFreeLinkSchedule is $O(ev\log v + ev\theta)$.
\end{theorem}
\begin{proof}
  Assuming that an element can be chosen randomly and uniformly from a
  finite set in unit time
  \cite{motwani_raghavan__randomized_algorithms}, the running time of
  Phase 1 can be shown to be $O(v)$.  For Phase 2, the optimal
  partitioning technique of \cite{gabow_westermann__forests_frames}
  based on Matroids can be used to partition the communication graph
  ${\mathcal G}_c$ into at most $6\theta$ oriented graphs in time
  $O(ev\log v)$.  Thus, $k \leqslant 6\theta$ holds for Phase 3.  From
  Lemma \ref{lem:oriented_cfls}, it follows that the first oriented
  graph $T_1$ can be colored in time $O(v^2)$. However, consider the
  coloring of $j^{th}$ oriented graph $T_j$, where $2 \leqslant j
  \leqslant k$.  When coloring edge $x$ from $T_j$ using
  FirstConflictFreeColor, conflicts can occur not only with the
  colored edges of $T_j$, but also with the edges of the previously
  colored oriented graphs $T_1,T_2,\ldots,T_{j-1}$.  Hence, the
  worst-case size of the set of colors to be examined $|{\mathcal C}_c
  \cup {\mathcal C}_{cf}|$ is $O(e)$. Note that in
  FirstConflictFreeColor, the SINR is checked only once for every
  colored edge in the set $\bigcup_{i=1}^{|{\mathcal C}_{cf}|} E_i$
  and at most $e$ times for the edge under consideration $x$.  With a
  careful implementation, FirstConflictFreeColor runs in time $O(e)$.
  Hence, any subsequent oriented graph $T_j$ can be colored in time
  $O(ev)$. Thus, the running time of Phase 3 is $O(ev\theta)$.
  Therefore, the overall running time of ConflictFreeLinkSchedule is
  $O(ev\log v + ev\theta)$.
\end{proof}

\subsection{Discussion}
\label{subsec:cfls_discussion}

In this section, we have developed ConflictFreeLinkSchedule, a point
to point link scheduling algorithm for an STDMA wireless mesh network
under the physical interference model.  The performance of the
proposed algorithm is superior to those of existing link scheduling
algorithms for STDMA wireless networks with uniform power assignment.
A practical experimental modeling shows that, on average, the proposed
algorithm achieves $20\%$ higher spatial reuse than the
ArboricalLinkSchedule \cite{ramanathan_lloyd__scheduling_algorithms},
GreedyPhysical \cite{brar_blough_santi__computationally_efficient} and
Truncated Graph-Based Scheduling
\cite{behzad_rubin__performance_graph} algorithms.  Since link
schedules are constructed offline only once and then used by the
network for a long period of time, these improvements in performance
directly translate to higher long-term network throughput.

The computational complexity of ConflictFreeLinkSchedule is comparable
to the computational complexity of ArboricalLinkSchedule and is much
lower than the computational complexity of GreedyPhysical and
Truncated Graph-Based Scheduling algorithms. Thus, in cognizance of
spatial reuse as well as computational complexity,
ConflictFreeLinkSchedule appears to be a good candidate for efficient
STDMA link scheduling algorithms.

\clearpage{\pagestyle{empty}\cleardoublepage}

\chapter{Point to Point Link Scheduling based on SINR Graph Model}
\label{ch:sinr_graph}

In this chapter, we propound a somewhat different approach for point
to point link scheduling in an STDMA wireless network under the
physical interference model.  This approach is based on SINR graph
representation of the network wherein weights of edges correspond to
interferences between pairs of nodes and weights of vertices
correspond to normalized noise powers at receiving nodes. We develop a
novel link scheduling algorithm with polynomial time complexity and
improved performance in terms of spatial reuse.

The rest of the chapter is organized as follows.  We motivate our SINR
graph approach in Section \ref{sec:motivation_lgls}.  We describe the
proposed link scheduling algorithm and provide an illustrative example
in Section \ref{sec:algo_lgls}. We prove the correctness of the
algorithm and derive its computational complexity in Section
\ref{sec:complexity_lgls}. The performance of the proposed algorithm
is compared with existing link scheduling algorithms in Section
\ref{sec:performance_lgls}.  We discuss the implications of our work
in Section \ref{sec:conclusion_lgls}.

\section{Motivation}
\label{sec:motivation_lgls}

The system model, notation and problem formulation are exactly as
described in Section \ref{subsec:problem_formulation_cfls}.
Specifically, we seek a low complexity conflict-free point to point
link scheduling algorithm that achieves high spatial reuse.

In general, for the STDMA wireless network $\Phi(\cdot)$, the set of
links to be scheduled is determined by a routing algorithm.  For
simplicity, we only consider exhaustive link schedules, i.e., we
consider uniform load on all links.

Note that for point to point link schedules that are conflict-free,
i.e., for link schedules that satisfy
(\ref{eq:conflict_free}), the equation for spatial reuse
(\ref{eq:spatial_reuse}) reduces to
\begin{eqnarray}
\mbox{Spatial Reuse} &=& \sigma \;=\; \frac{e}{C},
\label{eq:spatial_reuse_conflict_free}
\end{eqnarray}
where $e$ denotes the number of directed edges in the communication
graph $\mathcal G_c(\mathcal V,\mathcal E_c)$ and $C$ denotes the
number of slots in the link schedule.  Therefore, for conflict-free
link schedules, maximizing spatial reuse is equivalent to minimizing
the number of colors, i.e., minimizing the schedule length.

To the best of our knowledge, there is no known polynomial time
algorithm that determines a provably optimal schedule (minimum length
schedule) for an STDMA wireless network with constrained transmission
power. Hence, the only recourse is to devise heuristics and show their
efficiency theoretically and experimentally. Towards this end, we
propose a heuristic based on an SINR graph representation of the
network.

Consider any directed graph $G(V,E)$, where $V$ is the set of vertices
and $E$ is the set of edges. The line graph of $G(V,E)$ is the graph
$G'(V',E')$ whose vertices are the edges of $G(\cdot)$, i.e., $V'=E$
\cite{west__graph_theory}.  The SINR graph that we consider in this
chapter is analogous to the concept of line graph in
\cite{west__graph_theory}.  However, unlike the line graph, we assume
that the SINR graph is a complete graph, i.e., for any two distinct
vertices $v_i'$, $v_j'$ $\in V'$, there is a directed edge from $v_i'$
to $v_j'$ in $E'$.

The crux of the proposed  link scheduling algorithm can
be understood by revisiting the condition for successful packet
reception under the physical interference model (Equation
\ref{eq:sinr_ge_gammac}), i.e.,
\begin{eqnarray}
\frac{\frac{P}{D^\beta(t_{i,j},r_{i,j})}}{N_0+\sum_{\stackrel{k=1}{k\neq j}}^{M_i} \frac{P}{D^\beta(t_{i,k},r_{i,j})}} \geqslant \gamma_c.
\label{eq:sinr_slot_agnostic}
\end{eqnarray}
Rearranging the terms in (\ref{eq:sinr_slot_agnostic}), we obtain
\begin{eqnarray}
\frac{N_0\gamma_c}{P}D^\beta(t_{i,j},r_{i,j}) 
 + \sum_{\stackrel{k=1}{k\neq j}}^{M_i}
 \gamma_c \frac{D^\beta(t_{i,j},r_{i,j})}{D^\beta(t_{i,k},r_{i,j})}
&\leqslant& 1.
\label{eq:sinr_gammac_manipulation}
\end{eqnarray}
Dropping time slot index $i$ for clarity, we obtain the
``equivalent'' SINR threshold condition
\begin{eqnarray}
\frac{N_0\gamma_c}{P}D^\beta(t_j,r_j) + \sum_{\stackrel{k=1}{k\neq j}}^{M}
 \gamma_c \frac{D^\beta(t_j,r_j)}{D^\beta(t_k,r_j)}
&\leqslant& 1,
\label{eq:sinr_gammac_rearranged}
\end{eqnarray}
where $t_j$, $r_j$ and $M$ can be interpreted as $j^{th}$ transmitter,
$j^{th}$ receiver and number of concurrent transmissions,
respectively, in a given time slot.  The terms appearing in
(\ref{eq:sinr_gammac_rearranged}) correspond to vertex and edge
weights in a special graph representation of the STDMA network, termed
as SINR graph. This idea will be elucidated further in Section
\ref{subsec:description_lgls}.

\section{SINRGraphLinkSchedule Algorithm}
\label{sec:algo_lgls}

In this section, we explain the proposed link scheduling algorithm
based on SINR graph representation of the STDMA network. We provide an
illustrative example to elucidate the intricacies of the proposed
algorithm.

\subsection{Description}
\label{subsec:description_lgls}

\renewcommand{\baselinestretch}{1.2}\Large\normalsize

\begin{algorithm}[tbhp]
\caption{SINRGraphLinkSchedule (SGLS)}
\label{algo:line_graph}
\begin{algorithmic}[1]
\STATE{\bf Input:} Communication graph ${\mathcal G}_c$(${\mathcal V}$,${\mathcal E_c}$), $\gamma_c$, $N_0$, $P$
\STATE{\bf Output:} A coloring ${\mathcal C}$: ${\mathcal E_c} \rightarrow \{1,2,\ldots\}$ 
\STATE ${\mathcal V'} \leftarrow {\mathcal E_c}$
\STATE Construct the directed complete graph $\mathcal G'(\mathcal V',\mathcal E')$
\FORALL{$e'_{ij} \in {\mathcal E'}$}
\IF {edges $i$ and $j$ have a common vertex in $\mathcal G_c(\cdot)$}
\STATE $w_{ij} \leftarrow 1$
\ELSE
\STATE $w_{ij} \leftarrow \gamma_c \frac{{D(t_j,r_j)}^\beta}{{D(t_i,r_j)}^\beta}$
\ENDIF
\ENDFOR
\FORALL{$e'_{ij} \in {\mathcal E'}$}
\STATE $w'_{ij} \leftarrow \max\{0,1-w_{ij}\}$
\ENDFOR
\FORALL{$v_j' \in {\mathcal V'}$}
\STATE ${\mathcal N}(v_j') \leftarrow \frac{N_0\gamma_c}{P} {{D(t_j,r_j)}^\beta}$
\ENDFOR
\STATE $p \leftarrow 0$; ${\mathcal V_{uc}'} \leftarrow {\mathcal V'}$
\WHILE {${\mathcal V_{uc}'} \neq \phi$}
\STATE $p \leftarrow p+1$; choose $v' \in \mathcal V_{uc}'$ randomly
\STATE ${\mathcal C}(v') \leftarrow p$; ${\mathcal V_{uc}'} \leftarrow {\mathcal V_{uc}'} \setminus \{v'\}$; ${\mathcal V_{c_p}'} \leftarrow \{v'\}$; $\psi \leftarrow \mbox{1}$
\WHILE {$\psi=1$ and $\mathcal V_{uc}' \neq \phi$ and $\max_{y' \in {\mathcal V}_{uc}'} \sum_{x' \in {\mathcal V}_{c_p}'} w'_{x'y'}+w'_{y'x'} > 0$}
\FORALL{$u' \in \mathcal V_{uc}'$ such that $\sum_{x' \in {\mathcal V}_{c_p}'} w'_{x'u'}+w'_{u'x'} > 0$}
\STATE $\varrho \leftarrow 1$
\FORALL {$v_c' \in {\mathcal V_{c_p}'}$}
\IF {$\sum_{v_1' \in {\mathcal V_{c_p}'} \setminus \{v_c'\} \cup \{u'\}} w'_{{v_1'}{v_c'}} \leqslant |{\mathcal V_{c_p}'}|+{\mathcal N}(v_c')-1$}
\STATE $\varrho \leftarrow 0$
\ENDIF
\ENDFOR
\IF {$\varrho=1$ and $\sum_{v_2' \in {\mathcal V_{c_p}'}} w'_{{v_2'}u'} > |{\mathcal V_{c_p}'}|+{\mathcal N}(u')-1$}
\STATE ${\mathcal C}(u') \leftarrow p$; ${\mathcal V_{c_p}'} \leftarrow {\mathcal V_{c_p}'} \cup \{u'\}$; ${\mathcal V_{uc}'} \leftarrow {\mathcal V_{uc}'} \setminus \{u'\}$
\ELSE
\STATE $\varrho \leftarrow 0$
\ENDIF
\ENDFOR
\IF{$\varrho=0$}
\STATE $\psi \leftarrow 0$
\ENDIF
\ENDWHILE
\ENDWHILE
\end{algorithmic}
\end{algorithm}

\renewcommand{\baselinestretch}{1.5}\Large\normalsize

The proposed link scheduling algorithm under the physical interference
model is SINRGraphLinkSchedule (SGLS), which considers the
communication graph $\mathcal G_c(\mathcal V,\mathcal E_c)$.

First, we construct a directed complete SINR graph ${\mathcal
  G'}({\mathcal V'},{\mathcal E'})$ that has the edges of $\mathcal
G_c(\cdot)$ as its vertices, i.e., $\mathcal V' = \mathcal E_c$.  Let
the edges of $\mathcal G_c(\cdot)$ and the corresponding vertices of
$\mathcal G'(\cdot)$ be labeled $1,2,\ldots,e$. Let $t_i$ and $r_i$
denote the transmitter and receiver respectively of edge $i$ in
$\mathcal G_c(\cdot)$.  For any two edges $i$ and $j$ in graph
$\mathcal G_c(\cdot)$, the {\em interference weight function} $w_{ij}$
is defined as:
\begin{eqnarray*}
w_{ij}
&:=&
\left\{
\begin{array}{ll}
1 & \mbox{if $i$ and $j$ have a common vertex}, \\
\gamma_c \frac{D(t_j,r_j)^\beta}{D(t_i,r_j)^\beta} & \mbox{otherwise}.
\end{array}
\right.
\end{eqnarray*}
The interference weight function $w_{ij}$ indicates the interference
energy at $r_j$ due to transmission from $t_i$ to $r_i$ scaled with
respect to the signal energy of $t_j$ at $r_j$. Note that the
interference weight function appears as a summand in the equivalent
SINR threshold condition (\ref{eq:sinr_gammac_rearranged}).

We then compute the {\em co-schedulability weight function} $w'$. For
any two edges $i$ and $j$ in $\mathcal G_c(\cdot)$, the weight of edge
$e'_{ij}$ in ${\mathcal G'}(\cdot)$ is given by
$w'_{ij}=\max\{0,1-w_{ij}\}$.  Since $w_{ij}$ and $w_{ji}$ represent
interferences among links $i$ and $j$ in the STDMA network
$\Phi(\cdot)$, $w'_{ij}$ and $w'_{ji}$ intuitively represent the
co-schedulability of links $i$ and $j$ in $\Phi(\cdot)$ (equivalently,
co-schedulability of vertices $i$ and $j$ in $\mathcal G'(\cdot)$).
For example, if $w_{ij}$ is greater than or equal to $1$, then the
interference at the receiver of link $j$ from the transmitter of link
$i$ is very high and these links cannot be scheduled simultaneously.
This will result in $w'_{ij}$ being equal to $0$ indicating that
vertices $i$ and $j$ in $\mathcal G'(\cdot)$ are not co-schedulable.
On the other hand, if $w_{ij}$ is slightly greater than $0$ $(0 <
w_{ij} \ll 1)$, $w'_{ij}$ will be slightly less than $1$ indicating
that the vertices $i$ and $j$ in $\mathcal G'(\cdot)$ are
co-schedulable.  Note that for the SINR graph $\mathcal G'(\cdot)$,
the weight of an edge refers to the value of co-schedulability
function for that edge.

Next, we determine the normalized noise power at the receiver of each
link of $\Phi(\cdot)$. This is tantamount to computing the normalized
noise power for each edge of $\mathcal G_c(\cdot)$, i.e., at each
vertex of ${\mathcal G'}(\cdot)$.  Note that the normalized noise
power function appears as a term in the equivalent SINR threshold
condition (\ref{eq:sinr_gammac_rearranged}).

Our objective is to color the vertices of ${\mathcal G'}(\cdot)$
(equivalently, edges of $\mathcal G_c(\cdot)$) using minimum number of
colors under the physical interference model, i.e., subject to the
condition that the SINR at the receiver of every link in $\Phi(\cdot)$
is no less than the communication threshold $\gamma_c$.  Equivalently,
for any ${\mathcal V_{cc}'} \subseteq {\mathcal V'}$, the coloring of
all vertices $v_i' \in {\mathcal V_{cc}'}$ with the same color is
defined to be {\em feasible} if
\begin{eqnarray}
\frac{\frac{P}{{D(t_{v_i'},r_{v_i'})}^\beta}}{N_0+\sum_{v_j'\in{\mathcal
V_{cc}'} \setminus \{v_i'\}}\frac{P}{{D(t_{v_j'},r_{v_i'})}^\beta}}
\geqslant \gamma_c \;\;\forall\;\; v_i' \in {\mathcal V_{cc}'}.
\label{eq:feasible_condition}
\end{eqnarray}
In the SINR graph $\mathcal G'(\cdot)$, this condition translates to
the sum of weights of edges incoming to a vertex from all co-colored
vertices being greater than the sum of the number of {\em remaining}
co-colored vertices and the normalized noise power minus a constant
factor (unity); this will be proved in Theorem
\ref{theo:correctness_lgls}.

Finally, we color vertices of $\mathcal G'(\cdot)$, i.e., edges of
$\mathcal G_c(\cdot)$, according to the following procedure.  Let
${\mathcal V_{uc}'}$ denote the set of uncolored vertices of $\mathcal
G'(\cdot)$.  Initially, $\mathcal V_{uc}'$ includes all vertices of
$\mathcal G'(\cdot)$.  First, we choose a vertex randomly from
${\mathcal V_{uc}'}$. This is assigned a new color, say $p$. Then, we
consider every vertex $u'$ from ${\mathcal V_{uc}'}$ such that the sum
of weights of all the edges between $u'$ and the vertices colored with
$p$ is positive.  Next, for each vertex colored with $p$, we check if
the sum of weights of all incoming edges is greater than the sum of
the number of vertices colored with $p$ and the normalized noise power
at that vertex minus a constant factor (unity).  If this inequality is
satisfied, we further check if the sum of weights of all edges
incoming to $u$ is greater than the sum of the number of vertices
colored with $p$ and the normalized noise power at $u'$ minus unity.
If this inequality is also satisfied, then vertex $u'$ is colored with
$p$. If any of these inequalities are not satisfied, vertex $u'$ is
colored with a new color. The algorithm exits when all the vertices
are colored.  The pseudocode of the algorithm is provided in Algorithm
\ref{algo:line_graph}.

\subsection{Example}
\label{subsec:example_lgls}

Consider the STDMA wireless network $\Phi(\cdot)$ whose deployment is
shown in Figure \ref{fig:deployment_lgls}. It consists of four labeled
nodes whose coordinates (in meters) are $1 \equiv (-40,5)$, $2 \equiv
(0,0)$, $3 \equiv (95,0)$ and $4 \equiv (135,0)$.  We use typical
values of system parameters in wireless networks
\cite{kim_lim_hou__improving_spatial}.  These values are shown in
Table \ref{tab:system_parameters_example}, which lead to $R_c = 100$
m.

\begin{figure}[thbp]
  \centering
  \includegraphics[width=6in]{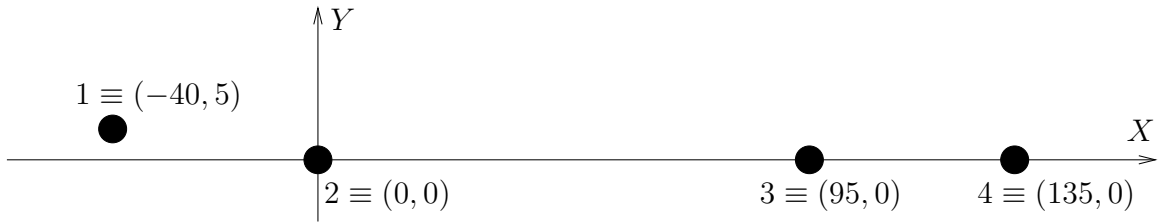}
  \caption{An STDMA wireless network with four nodes.}
  \label{fig:deployment_lgls}
\end{figure}

\begin{table}[tbhp]
  \centering
  \begin{tabular}{|l|l|l|} \hline
    Parameter & Symbol & Value \\ \hline
    transmission power & $P$ & 10 mW \\ \hline
    path loss exponent & $\beta$ & 4 \\ \hline
    noise power spectral density & $N_0$ & -90 dBm \\ \hline
    communication threshold & $\gamma_c$ & 20 dB \\ \hline
  \end{tabular}
  \caption{System parameters for the STDMA network shown in Figure 
    \ref{fig:deployment_lgls}.}
  \label{tab:system_parameters_example}
\end{table}

The communication graph model of the STDMA network is
shown in Figure \ref{fig:comm_graph_lgls}. The communication graph
$\mathcal G_c(\mathcal V,\mathcal E_c)$ consists of four vertices and six
directed edges. The vertex and edge sets are given by
\begin{eqnarray}
\mathcal V &=& \{v_1,v_2,v_3,v_4\}, \\
\mathcal E_c &=& \{(1,2),(2,1),(2,3),(3,2),(3,4),(4,3)\}.
\end{eqnarray}

\begin{figure}[thbp]
  \centering
  \includegraphics[width=5in]{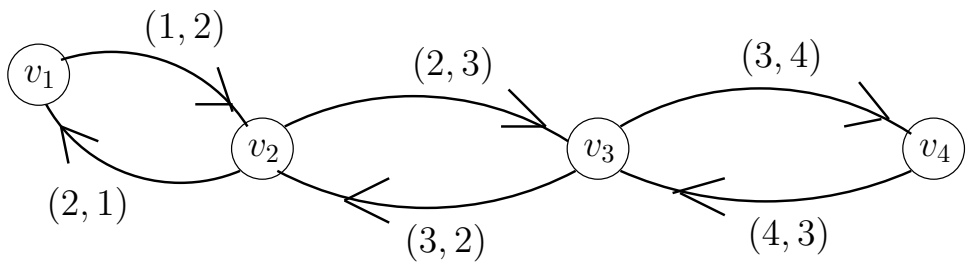}
  \caption{Communication graph model of STDMA network described by
    Figure \ref{fig:deployment_lgls} and Table
    \ref{tab:system_parameters_example}.}
  \label{fig:comm_graph_lgls}
\end{figure}

The SINR graph model of the communication graph $\mathcal G_c(\mathcal
V,\mathcal E_c)$ is shown in Figure \ref{fig:line_graph_lgls}.  The
SINR graph $\mathcal G'(\mathcal V',\mathcal E')$ is a complete graph
and consists of six vertices and thirty directed edges.  The vertex
set of the SINR graph is given by
\begin{eqnarray}
\mathcal V' &=& \{(1,2),(2,1),(2,3),(3,2),(3,4),(4,3)\}.
\end{eqnarray}
The edge set $\mathcal E'$ of the SINR graph is enumerated in Table
\ref{tab:weights_lgls}, along with the interference weight function
$w_{ij}$ and co-schedulability weight function $w_{ij}'$ for each edge
$i \rightarrow j \in \mathcal G'(\cdot)$. The normalized noise powers
at vertices of the SINR graph are enumerated in Table
\ref{tab:noise_powers_lgls}.

\begin{figure}[thbp]
  \centering
  \includegraphics[width=5in]{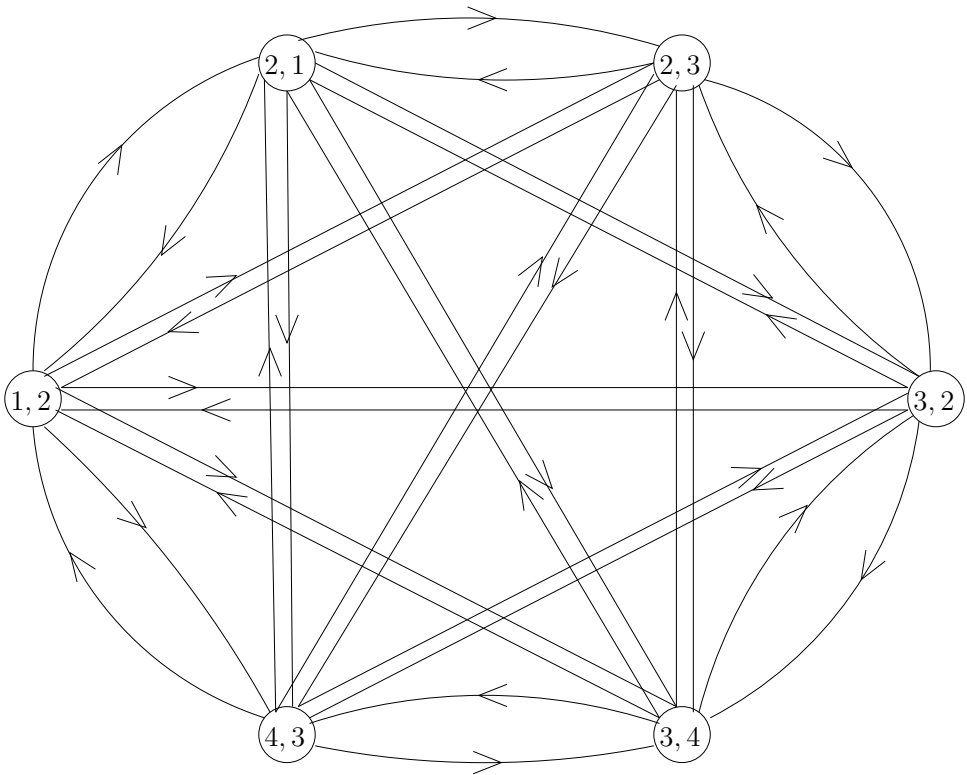}
  \caption{SINR graph model of communication graph shown in Figure
    \ref{fig:comm_graph_lgls}.}
  \label{fig:line_graph_lgls}
\end{figure}

\begin{table}[tbhp]
\begin{tabular}{|l|l|l||l|l|l|} \hline
Edge $i \rightarrow j$ of & & & Edge $i \rightarrow j$ of & & \\
SINR graph $\mathcal G'(\mathcal V',\mathcal E')$ & $w_{ij}$ & $w_{ij}'$ 
  & SINR graph $\mathcal G'(\mathcal V',\mathcal E')$ & $w_{ij}$ & $w_{ij}'$ 
    \\ \hline
$(1,2) \rightarrow (2,1)$ & 1 & 0 & $(3,2) \rightarrow (1,2)$ & 1 & 0 \\ \hline
$(2,1) \rightarrow (1,2)$ & 1 & 0 & 
  $(1,2) \rightarrow (3,4)$ & 0.2725 & 0.7275 \\ \hline
$(2,1) \rightarrow (2,3)$ & 1 & 0 & 
  $(3,4) \rightarrow (1,2)$ & 3.2420 & 0 \\ \hline
$(2,3) \rightarrow (2,1)$ & 1 & 0 & $(2,1) \rightarrow (3,2)$ & 1 & 0 \\ \hline
$(2,3) \rightarrow (3,2)$ & 1 & 0 & $(3,2) \rightarrow (2,1)$ & 1 & 0 \\ \hline
$(3,2) \rightarrow (2,3)$ & 1 & 0 & 
  $(2,1) \rightarrow (3,4)$ & 0.7707 & 0.2293 \\ \hline
$(3,2) \rightarrow (3,4)$ & 1 & 0 & 
  $(3,4) \rightarrow (2,1)$ & 0.7928 & 0.2072 \\ \hline
$(3,4) \rightarrow (3,2)$ & 1 & 0 & 
  $(2,1) \rightarrow (4,3)$ & 3.1430 & 0 \\ \hline
$(3,4) \rightarrow (4,3)$ & 1 & 0 & 
  $(4,3) \rightarrow (2,1)$ & 0.2811 & 0.7189 \\ \hline
$(4,3) \rightarrow (3,4)$ & 1 & 0 & $(2,3) \rightarrow (3,4)$ & 1 & 0 \\ \hline
$(4,3) \rightarrow (1,2)$ & 0.7950 & 0.2050 & 
  $(3,4) \rightarrow (2,3)$ & 1 & 0 \\ \hline
$(1,2) \rightarrow (4,3)$ & 0.7686 & 0.2314 &
  $(2,3) \rightarrow (4,3)$ & 1 & 0 \\ \hline
$(1,2) \rightarrow (2,3)$ & 1 & 0 & $(4,3) \rightarrow (2,3)$ & 1 & 0 \\ \hline
$(2,3) \rightarrow (1,2)$ & 1 & 0 & $(3,2) \rightarrow (4,3)$ & 1 & 0 \\ \hline
$(1,2) \rightarrow (3,2)$ & 1 & 0 & $(4,3) \rightarrow (3,2)$ & 1 & 0 \\ \hline
\end{tabular}
\caption{Interference and co-schedulability weight functions for 
  edges of SINR graph shown in Figure \ref{fig:line_graph_lgls}.}
\label{tab:weights_lgls}
\end{table}

\begin{table}
  \centering
  \begin{tabular}{|l|l|} \hline
    Vertex $v_j'$ of & \\
    SINR graph $\mathcal G'(\mathcal V',\mathcal E')$ & 
    $\mathcal N(v_j')$ \\ \hline
    (1,2) & 0.0264 \\ \hline
    (2,1) & 0.0264 \\ \hline
    (2,3) & 0.8145 \\ \hline
    (3,2) & 0.8145 \\ \hline
    (3,4) & 0.0256 \\ \hline
    (4,3) & 0.0256 \\ \hline
  \end{tabular}
  \caption{Normalized noise powers at vertices of SINR graph shown in 
    Figure \ref{fig:line_graph_lgls}.}
  \label{tab:noise_powers_lgls}
\end{table}

The truncated SINR graph $\mathcal G_t'(\mathcal V',\mathcal E_t')$ is
shown in Figure \ref{fig:truncated_line_graph}.  The truncated SINR
graph consists of all vertices of the SINR graph and only those edges
whose co-schedulability weight function is positive, i.e., $\mathcal
E_t' = \{(i,j): i,j \in \mathcal E_c \mbox{ and }w_{ij}' > 0\}$. The
values of the co-schedulability weight functions for all edges and the
normalized noise powers at all vertices are also shown in the figure.
We use the truncated SINR graph to explain the SGLS algorithm, since
edges having zero weight in the SINR graph do not play any role in the
SGLS algorithm.  Note that, in the truncated SINR graph, the weight of
an edge refers to the value of the co-schedulability weight function
for that edge.

\begin{figure}[thbp]
  \centering
  \includegraphics[width=5in]{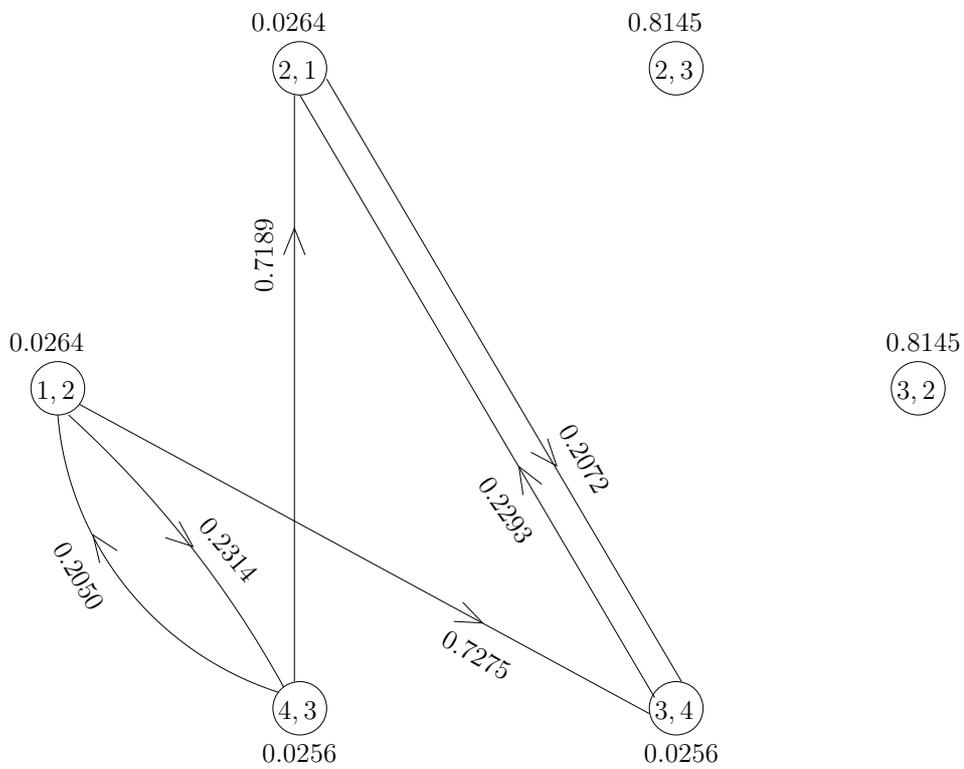}
  \caption{Truncated SINR graph derived from SINR graph shown in
    Figure \ref{fig:line_graph_lgls} and weight values given in Tables
    \ref{tab:weights_lgls} and \ref{tab:noise_powers_lgls}.}
  \label{fig:truncated_line_graph}
\end{figure}

Initially, the set of uncolored vertices is $\mathcal V_{uc}'=\{(1,2),
(2,1), (2,3), (3,2), (3,4), (4,3)\}$. In the first iteration, we
randomly choose $v'=(1,2)$ and assign it Color $1$ (say, red). So,
$\mathcal C(1,2)=1$. The set of uncolored vertices is $\mathcal
V_{uc}'=\{(2,1), (2,3), (3,2), (3,4), (4,3)\}$ and the set of vertices
colored 1 is $\mathcal V_{c_1}'=\{(1,2)\}$.  From the set of uncolored
vertices $\mathcal V_{uc}'$, we consider every vertex $u'$ such that
the sum of weights of edges from the presently colored vertex $(1,2)$
to $u'$ and from $u'$ to the presently colored vertex is positive.
From Figure \ref{fig:truncated_line_graph}, we obtain two candidates:
$u'=(3,4)$ and $u'=(4,3)$.  We first examine the candidate vertex
$(3,4)$.  We check if the weight of the edge from $(3,4)$ to the
presently colored vertex $(1,2)$ is no greater than the number of
vertices colored with the present color (red) plus the normalized
noise power at the colored vertex minus unity.  Our calculations show
that inequality holds $(0 < 0.0264)$ and candidate vertex $(3,4)$
cannot be assigned Color 1.  We next examine the candidate vertex
$(4,3)$.  We check if the weight of the edge from the candidate vertex
to $(1,2)$ is no greater than the number of vertices colored with the
present color plus the normalized noise power at $(1,2)$ minus unity.
Our calculations show that inequality does not hold $(0.2050
\not\leqslant 0.0264)$.  Furthermore, we check if the weight of the
edge from the presently colored vertex $(1,2)$ to the candidate vertex
$(4,3)$ is greater than the number of vertices colored red plus the
normalized noise power at $(4,3)$ minus unity.  The inequality holds
and hence the candidate vertex $(4,3)$ is assigned Color 1 (red). So,
$\mathcal C(4,3)=1$.  The set of uncolored vertices is $\mathcal
V_{uc}'=\{(2,1), (2,3), (3,2), (3,4)\}$ and the set of vertices
colored 1 is $\mathcal V_{c_1}'=\{(1,2),(4,3)\}$.  Again, from the set
of uncolored vertices $\mathcal V_{uc}'$, we consider every vertex
$u'$ such that the sum of weights of edges from the presently colored
vertices $\{(1,2),(4,3)\}$ to $u'$ and from $u'$ to the presently
colored vertices is positive.  From Figure
\ref{fig:truncated_line_graph}, the candidate vertices are $(2,1)$ and
$(3,4)$. Consider the candidate vertex $(2,1)$. For every vertex
$v_c'$ colored 1, we check if the sum of weights of edges from
remaining co-colored vertices and the candidate vertex to the colored
vertex is no greater than the number of co-colored vertices and the
normalized noise power at the colored vertex minus unity.  For the
colored vertex $(1,2)$, our calculations show that inequality holds
$(0.2050 \leqslant 1.0264)$. So, we discard $(2,1)$, consider the next
candidate vertex $(3,4)$ and perform an analogous comparison with
$v_c'=(1,2)$.  Since inequality holds in this case too $(0.2050
\leqslant 1.0264)$, we discard $(3,4)$ and proceed to the next
iteration. The set of vertices colored so far is shown in Figure
\ref{fig:sgls_1}.

\begin{figure}[thbp]
  \centering
  \includegraphics[width=5in]{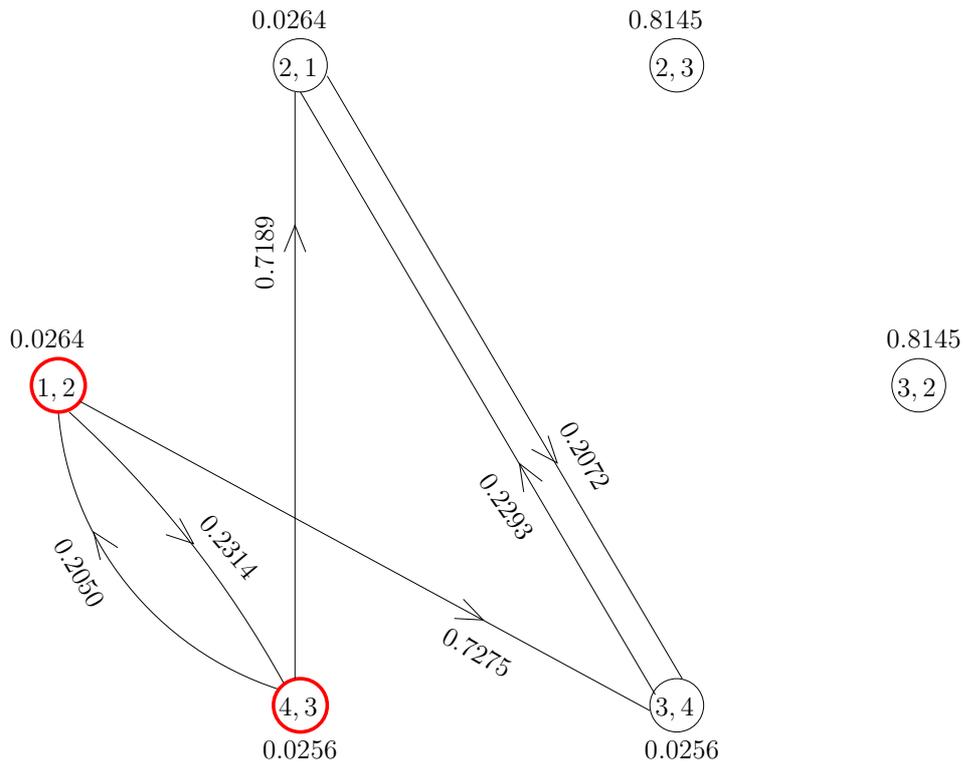}
  \caption{Coloring of vertices of truncated SINR graph after first
    iteration of SGLS algorithm.}
  \label{fig:sgls_1}
\end{figure}

In the second iteration, we randomly choose $v'=(2,3)$ and assign it
Color 2 (say, blue).  So, $\mathcal C(2,3)=2$.  The set of uncolored
vertices is $\mathcal V_{uc}'=\{(2,1), (3,2), (3,4)\}$ and the set of
vertices colored 2 is $\mathcal V_{c_2}'=\{(2,3)\}$.  From the set of
uncolored vertices $\mathcal V_{uc}'$, we consider every vertex $u'$
such that the sum of weights of edges from the presently colored
vertex $(2,3)$ to $u'$ and from $u'$ to the presently colored vertex
is positive.  From Figure \ref{fig:sgls_1}, no such vertex exists.
So, we proceed to the next iteration.  The vertices colored so far are
shown in Figure \ref{fig:sgls_2}.

\begin{figure}[thbp]
  \centering
  \includegraphics[width=5in]{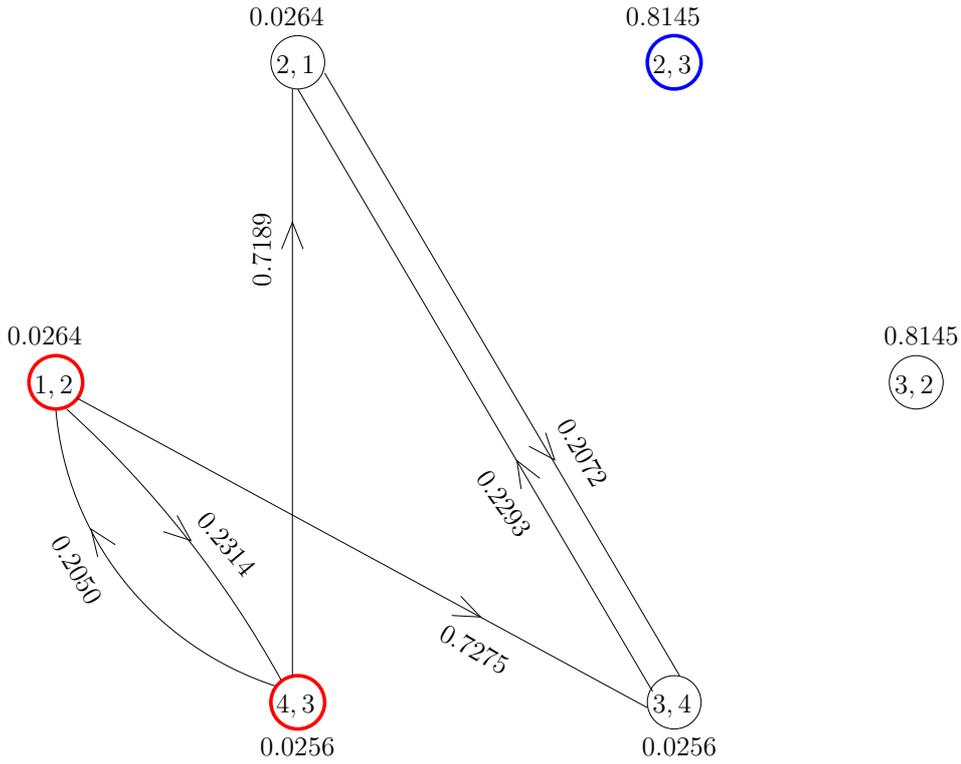}
  \caption{Coloring of vertices of truncated SINR graph after second
    iteration of SGLS algorithm.}
  \label{fig:sgls_2}
\end{figure}

In the third iteration, we randomly choose $v'=(3,4)$ and assign it
Color 3 (say, green).  So, $\mathcal C(3,4)=3$. The set of uncolored
vertices is $\mathcal V_{uc}'=\{(2,1), (3,2)\}$ and the set of
vertices colored 3 is $\mathcal V_{c_3}'=\{(3,4)\}$.  From the set of
uncolored vertices $\mathcal V_{uc}'$, we consider every vertex $u'$
such that the sum of weights of edges from the presently colored
vertex $(3,4)$ to $u'$ and from $u'$ to the presently colored vertex
is positive. From Figure \ref{fig:truncated_line_graph}, we obtain
$u'=(2,1)$ as the only candidate vertex.  Next, we check if the weight
of the edge from the candidate vertex $(2,1)$ to the presently colored
vertex $(3,4)$ is no greater than the number of vertices colored with
the present color (green) plus the normalized noise power at the
colored vertex minus unity.  Our calculations show that inequality
does not hold $(0.2293 \not\leqslant 0.0256)$.  So, we further check
if the weight of the edge from the presently colored vertex $(3,4)$ to
the candidate vertex $(2,1)$ exceeds the number of vertices colored
with the present color plus the normalized noise power at the
candidate vertex minus unity.  Since the inequality holds $(0.2072 >
0.0264)$, the candidate vertex $(2,1)$ is assigned Color 3 (green).
So, $\mathcal C(2,1)=3$.  The set of uncolored vertices is $\mathcal
V_{uc}'=\{(3,2)\}$ and the set of vertices colored green is $\mathcal
V_{c_3}'=\{(3,4),(2,1)\}$.  Next, from the set of uncolored vertices
$\mathcal V_{uc}'$, we choose that uncolored vertex $u'$ such that the
sum of weights of edges from $u'$ to the set of presently colored
vertices $\{(3,4),(2,1)\}$ and from $\{(3,4),(2,1)\}$ to $u'$ is
positive. From Figure \ref{fig:sgls_2}, no such vertex $u'$ exists.
So, we proceed to the next iteration. Figure \ref{fig:sgls_3} shows
the set of vertices colored so far.

\begin{figure}[thbp]
  \centering
  \includegraphics[width=5in]{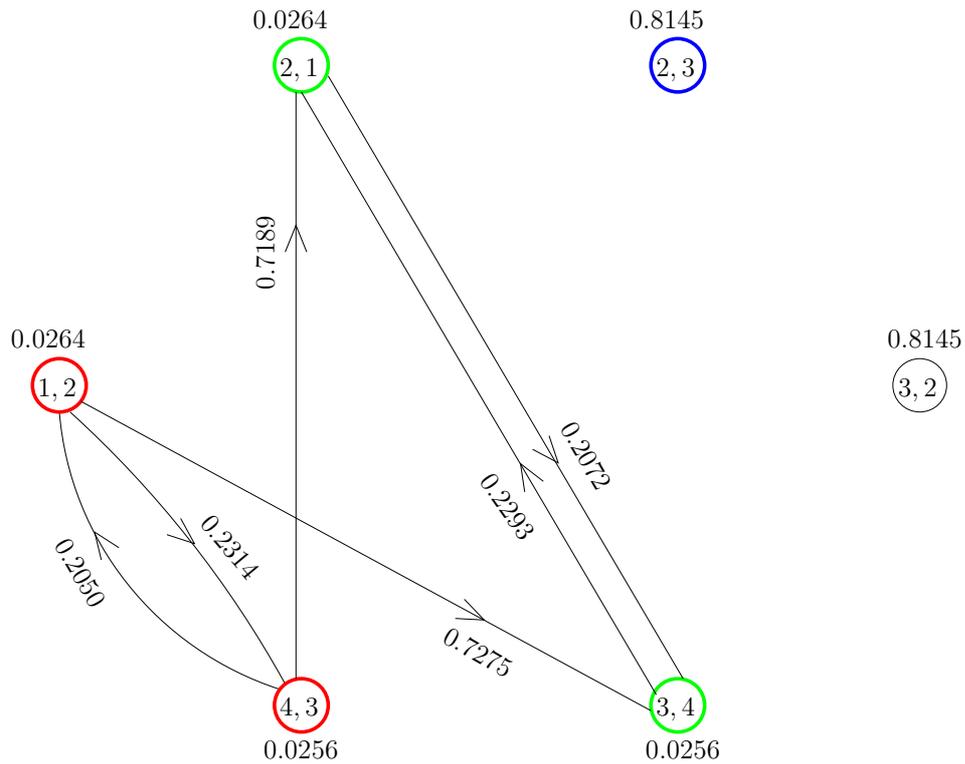}
  \caption{Coloring of vertices of truncated SINR graph after third
    iteration of SGLS algorithm.}
  \label{fig:sgls_3}
\end{figure}

In the fourth iteration, $(3,2)$ is the only uncolored vertex. So, we
choose $v'=(3,2)$ and assign it Color 4 (say, pink).  The set of
vertices colored 4 is $\mathcal V_{c_4}'=\{(3,2)\}$ and the set of
uncolored vertices is $\mathcal V_{uc}'=\phi$.  So, the algorithm
ends. The final coloring of vertices of the truncated SINR graph by
SGLS algorithm is shown in Figure \ref{fig:sgls_4}.

\begin{figure}[thbp]
  \centering
  \includegraphics[width=5in]{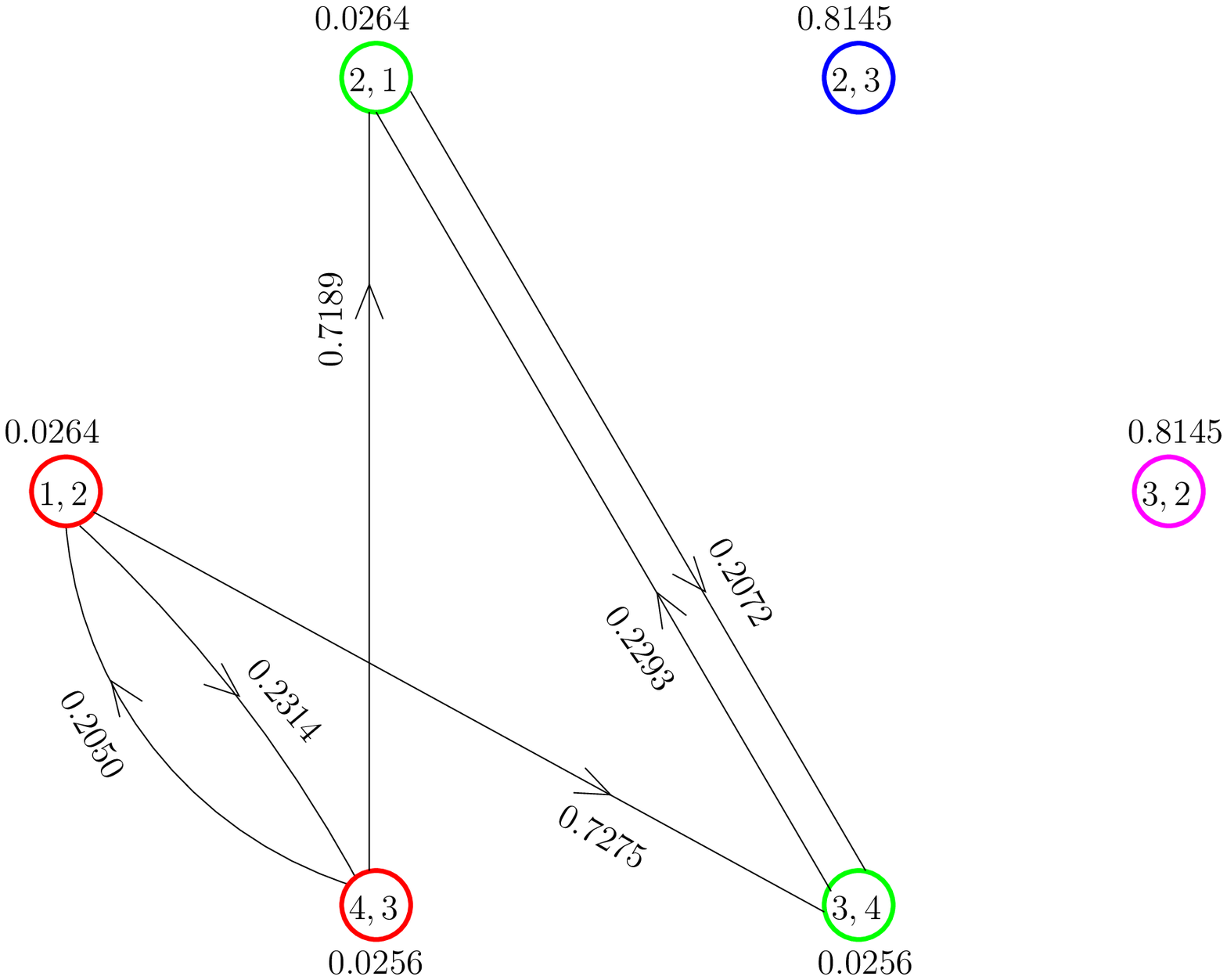}
  \caption{Coloring of vertices of truncated SINR graph after complete
    execution of SGLS algorithm.}
  \label{fig:sgls_4}
\end{figure}

The output of the SGLS algorithm is enumerated in Table
\ref{tab:output_lgls} and is also shown pictorially in Figure
\ref{fig:output_lgls}. The resulting link schedule is denoted by
$\Psi(\mathcal S_1,\mathcal S_2,\mathcal S_3, \mathcal S_4,\mathcal
S_5)$, where
\begin{eqnarray*}
\mathcal S_1 &=& \{1 \rightarrow 2, \; 4 \rightarrow 3\}, \\
\mathcal S_2 &=& \{2 \rightarrow 3\}, \\
\mathcal S_3 &=& \{3 \rightarrow 4, \; 2 \rightarrow 1\}, \\
\mathcal S_4 &=& \{3 \rightarrow 2\}.
\end{eqnarray*}

\begin{table}
  \centering
  \begin{tabular}{|l|l|l|} \hline
    Time slot & Color & Active (transmitter, receiver) pairs \\ \hline
    1 & red & (1,2), (4,3) \\ \hline
    2 & blue & (2,3) \\ \hline
    3 & green & (3,4), (2,1) \\ \hline
    4 & pink & (3,2) \\ \hline
  \end{tabular}
  \caption{Output of SGLS algorithm for STDMA network described by Figure
    \ref{fig:deployment_lgls} and Table \ref{tab:system_parameters_example}.}
  \label{tab:output_lgls}
\end{table}

\begin{figure}[thbp]
  \centering
  \includegraphics[width=5in]{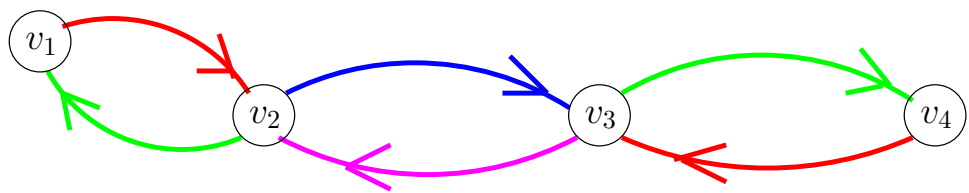}
  \caption{Output of SGLS algorithm for STDMA network described by
    Figure \ref{fig:deployment_lgls} and Table
    \ref{tab:system_parameters_example}.}
  \label{fig:output_lgls}
\end{figure}

Finally, we check if the link schedule enumerated in Table
\ref{tab:output_lgls} is conflict-free, i.e., if the SINR threshold
condition (\ref{eq:conflict_free}) is satisfied at every receiver for
the STDMA network described by Figure \ref{fig:deployment_lgls} and
Table \ref{tab:system_parameters_example}. Only one
transmitter-receiver pair is active during time slots 2 and 4.  Since
the receiver is within the communication range of its corresponding
transmitter for each of these time slots, the SINR threshold condition
is satisfied trivially for time slots 2 and 4. Two
transmitter-receiver pairs are active during time slots 1 and 3. In
time slot 1, the active transmitter-receiver pairs are $(1,2)$ and
$(4,3)$. Our computations show that the SINRs at Receivers 2 and 3 are
$20.85$ dB and $21$ dB, both of which exceed the communication
threshold of $20$ dB.  In time slot 3, the active transmitter-receiver
pairs are $(2,1)$ and $(3,4)$. Our computations show that the SINRs at
Receivers 1 and 4 are $20.87$ dB and $20.99$ dB respectively, both of
which exceed the communication threshold. This verifies that the SGLS
algorithm yields a conflict-free link schedule for the network
described by Figure \ref{fig:deployment_lgls} and Table
\ref{tab:system_parameters_example}.  Note that, from
(\ref{eq:spatial_reuse_conflict_free}), the spatial reuse of SGLS
algorithm for this network is $1.5$.

\section{Analytical Results}
\label{sec:complexity_lgls}

In this section, we prove the correctness of the SGLS algorithm and
derive its running time (computational) complexity. We follow the
notation of Algorithm \ref{algo:line_graph}.

\begin{theorem}
  For any ${\mathcal V_{cc}'} \subseteq {\mathcal V'}$, if
  \begin{eqnarray}
  \sum_{v_2'\in{\mathcal V_{cc}'} \setminus \{v_1'\}} w'_{v_2'v_1'} 
  &>& 
  |{\mathcal V_{cc}'}|+{\mathcal N}(v_1')-2 \;\;\forall\;\; 
  v_1' \in {\mathcal V_{cc}'}, 
  \end{eqnarray}
  then the coloring of all vertices of ${\mathcal V_{cc}'}$ with the
  same color is feasible.
\label{theo:correctness_lgls}
\end{theorem}

\begin{proof}
  Recall that $w'_{v_2'v_1'} = 0$ or $ 1-w_{v_2'v_1'}$ and that $0
  \leqslant w'_{v_2'v_1'} \leqslant 1$.  Suppose $w'_{v_3'v_1'}=0$ for
  some $v_1', v_3' \in {\mathcal V_{cc}'}$, $v_1' \neq v_3'$, then
  \begin{eqnarray*}
  \sum_{v_2'\in{\mathcal V_{cc}'} \setminus \{v_1'\}} w'_{v_2'v_1'}
  &=&
  \sum_{v_2'\in{\mathcal V_{cc}'} \setminus \{v_1',v_3'\} } w'_{v_2'v_1'},\\
  &\leqslant& \sum_{v_2'\in{\mathcal V_{cc}'} \setminus \{v_1',v_3'\} } 1, \\
  &=& |{\mathcal V_{cc}'} \setminus \{v_1',v_3'\}|,\\
  &=& |{\mathcal V_{cc}'}|-2,
  \end{eqnarray*}
  which contradicts the hypothesis since $\mathcal N(v_1')>0$. So, an
  edge connecting any two vertices in ${\mathcal V_{cc}'}$ must have
  positive weight.  Thus, $0 < w'_{v_2'v_1'} \leqslant 1$ $\forall$
  $v_1', v_2' \in {\mathcal V_{cc}'}$, $v_1' \neq v_2'$.
  Equivalently, $0 < 1-w_{v_2'v_1'} \leqslant 1$ $\forall$ $v_1', v_2'
  \in {\mathcal V_{cc}'}$, $v_1' \neq v_2'$.  If two vertices $v_1',
  v_2' \in {\mathcal V_{cc}'}$ (equivalently, edges $v_1', v_2' \in
  \mathcal G_c(\cdot)$) have a common vertex in $\mathcal G_c(\cdot)$,
  then $w_{v_2'v_1'}=1$, which is a contradiction. So, no two vertices
  in $\mathcal V_{cc}'$ have a common vertex in $\mathcal G_c(\cdot)$.
  From the hypothesis,
  \begin{eqnarray*}
  \sum_{v_2'\in{\mathcal V_{cc}'} \setminus \{v_1'\}} w'_{v_2'v_1'} 
  &>& |{\mathcal V_{cc}'}|+{\mathcal N}(v_1')-2 
    \;\;\forall\;\; v_1' \in {\mathcal V_{cc}'}, \\ 
   \Leftrightarrow 
   \sum_{v_2'\in{\mathcal V_{cc}'} \setminus \{v_1'\}} (1-w_{v_2'v_1'})
  &>& |{\mathcal V_{cc}'}|+{\mathcal N}(v_1')-2 
      \;\;\forall\;\; v_1' \in {\mathcal V_{cc}'}, \\ 
    \Leftrightarrow |{\mathcal V_{cc}'} 
     \setminus \{v_1'\}|-\sum_{v_2'\in{\mathcal V_{cc}'} \setminus 
     \{v_1'\}}w_{v_2'v_1'} 
  &>& |{\mathcal V_{cc}'}|+{\mathcal N}(v_1')-2 
      \;\;\forall\;\; v_1' \in {\mathcal V_{cc}'}, \\ 
    \Leftrightarrow 
    |{\mathcal V_{cc}'}|-1-\sum_{v_2'\in{\mathcal V_{cc}'} \setminus 
     \{v_1'\}}w_{v_2'v_1'} 
  &>& |{\mathcal V_{cc}'}|+{\mathcal N}(v_1')-2 
     \;\;\forall\;\; v_1' \in {\mathcal V_{cc}'}, \\ 
    \Leftrightarrow \sum_{v_2'\in{\mathcal V_{cc}'} 
    \setminus \{v_1'\}}w_{v_2'v_1'} + {\mathcal N}(v_1') 
  &<& 1 \;\;\forall\;\; v_1' \in {\mathcal V_{cc}'}, \\
     \Leftrightarrow \sum_{v_2'\in{\mathcal V_{cc}'} \setminus \{v_1'\}} 
       \gamma_c
     \frac{{D(t_{v_1'},r_{v_1'})}^\beta}{{D(t_{v_2'},r_{v_1'})}^\beta}
    + \frac{N_0 \gamma_c}{P}{D(t_{v_1'},r_{v_1'})}^\beta &<& 1 
      \;\;\forall\;\; v_1' \in {\mathcal V_{cc}'}, \\
   \Leftrightarrow
    \frac{\frac{P}{{D(t_{v_1'},r_{v_1'})}^\beta}}{N_0+\sum_{v_2'\in{\mathcal
    V_{cc}'} \setminus \{v_1'\}}\frac{P}{{D(t_{v_2'},r_{v_1'})}^\beta}}
  &>& \gamma_c \;\;\forall\;\; v_1' \in {\mathcal V_{cc}'}.
  \end{eqnarray*}
  Therefore, the SINR threshold condition
  (\ref{eq:feasible_condition}) is satisfied at the receivers of all
  vertices of ${\mathcal V_{cc}'}$.
\end{proof}

\noindent With respect to (w.r.t.) the communication graph $\mathcal
G_c(\mathcal V,\mathcal E_c)$, let:
\begin{eqnarray*}
e &=& \mbox{number of edges},\\
v &=& \mbox{number of vertices}.
\end{eqnarray*}

\begin{theorem}
  The running time complexity of SGLS algorithm is $O(e^2)$.
\label{theo:complexity_lgls}
\end{theorem}

\begin{proof}
  $|\mathcal V'| = |\mathcal E_c|=e$. Since $\mathcal G'(\cdot)$ is a
  directed complete graph, $|{\mathcal E'}|=e(e-1)=O(e^2)$.  Since the
  computation of $w_{ij}$ for given edges $i$ and $j$ of $\mathcal
  G'(\cdot)$ takes unit time, the computation of interference weight
  functions for all edges of $\mathcal G'(\cdot)$ takes $O(e^2)$ time.
  Similarly, the computation of co-schedulability weight functions for
  all edges of $\mathcal G'(\cdot)$ requires $O(e^2)$ time.  The
  computation of normalized noise powers at all vertices of $\mathcal
  G'(\cdot)$ takes $O(e)$ time.

  In $\mathcal G'(\cdot)$, let $C$ denote the total number of colors
  used to color all vertices and let $N_i$ denote the number of
  vertices assigned color $i$, i.e., $N_i =|\mathcal V_{c_i}'|$.
  Since $C$ can never exceed the number of vertices in $\mathcal
  G'(\cdot)$, i.e., the number of edges in $\mathcal G_c(\cdot)$, $C$
  is $O(e)$.  The time required by Lines 20-21 is $O(1)$, let it be
  $k_1$, where $k_1$ is a constant.  

  With a careful implementation of storing $\sum_{v_1' \in {\mathcal
      V_{c_p}'} \setminus \{v_c'\} \cup \{u'\}} w'_{{v_1'}{v_c'}}$
  $\forall$ $v_c' \in {\mathcal V_{c_p}'}$, Lines 26-28 take $O(1)$
  time. Thus, Lines 25-29 take $O(|\mathcal V_{c_p}'|)$ time, let it
  be equal to $k_2 |\mathcal V_{c_p}'|$, where $k_2$ is a constant.
  Along similar arguments, Lines 30-34 take $O(1)$ time, let it be
  equal to $k_3$, where $k_3$ is a constant. The time required by
  Lines 36-38 is $k_4$, where $k_4$ is a constant. Thus, the total
  running time of the coloring phase is
  \begin{eqnarray*}
    \tau &=& \sum_{i=1}^C \Big(k_1 + \sum_{j=1}^{N_i} \big(|\mathcal V_{uc}'|
    (k_2 |\mathcal V_{c_i}'| + k_3) + k_4 \big) \Big).
  \end{eqnarray*}
  Since $\mathcal V_{uc}', \mathcal V_{c_i}' \subseteq \mathcal V'$,
  it follows that $|\mathcal V_{uc}'|, |\mathcal V_{c_i}'| \leqslant
  |\mathcal V'|=e$.  Furthermore, for any color $i$, $\mathcal V_{uc}'
  \cup \mathcal V_{c_i}' \subseteq \mathcal V'$.  Thus, $|\mathcal
  V_{uc}'| |\mathcal V_{c_i}'| \leqslant \frac{e^2}{4}$. Therefore
  \begin{eqnarray*}
    \tau
    &\leqslant& 
    \sum_{i=1}^C k_1 + \sum_{i=1}^C \sum_{j=1}^{N_i} k_2 \frac{e^2}{4}
    + \sum_{i=1}^C \sum_{j=1}^{N_i} k_3 e 
    + \sum_{i=1}^C \sum_{j=1}^{N_i} k_4 \\
    &=& k_1 C + k_2 \frac{e^2}{4} (e) + k_3 e (e) + k_4 (e)\\
    &=& k_1 C + k_3 e^2 + \frac{k_2}{4} e^3 + k_4 e\\
    &=& O(e^3).
   \end{eqnarray*}
   Hence, the total running time complexity of SGLS algorithm is
   $O(e^3)$.
\end{proof}

\section{Performance Results}
\label{sec:performance_lgls}

In this section, we demonstrate the efficacy of SGLS algorithm via
simulations. To the best of our knowledge, for an STDMA network with
constrained transmission power, there is no existing work on link
scheduling that utilizes an SINR graph representation of the network.
However, for completeness, we compare the performance of SGLS
algorithm with the CFLS algorithm proposed in Chapter
\ref{ch:comm_graph}. Note that SGLS is based on SINR graph while CFLS
is based on communication graph and verifying SINR conditions.

In the simulation experiments, the location of every node is generated
randomly using a uniform distribution for its $X$ and $Y$ coordinates.
We assume that the deployment area is a circular region of radius $R$.
The values chosen for system parameters $P$, $\gamma_c$, $\beta$ and
$N_0$ are prototypical values of system parameters in wireless
networks \cite{kim_lim_hou__improving_spatial}.  After generating
random positions for $N$ nodes, we have complete information of
$\Phi(\cdot)$.  Once the link schedule $\Psi(\cdot)$ is computed by
every algorithm, $\sigma$ is computed using
(\ref{eq:spatial_reuse_conflict_free}).  For a given set of system
parameters, we calculate the average spatial reuse by averaging
$\sigma$ over 1000 randomly generated networks.  Keeping all other
parameters fixed, we observe the effect of increasing the number of
nodes $N$ on the average spatial reuse.

In the first experiment (Experiment 1), we assume that $R=500$ m,
$P=10$ mW, $\beta = 4$, $N_0=-90$ dBm and $\gamma_c = 20$ dB.  Thus,
$R_c=100$ m. We vary the number of nodes from 30 to 110 in steps of 5.
Figure \ref{fig:expt1_sgls} plots the average spatial reuse vs. number
of nodes for both the algorithms.

\begin{figure}[thbp]
  \centering
  \includegraphics[width=5in]{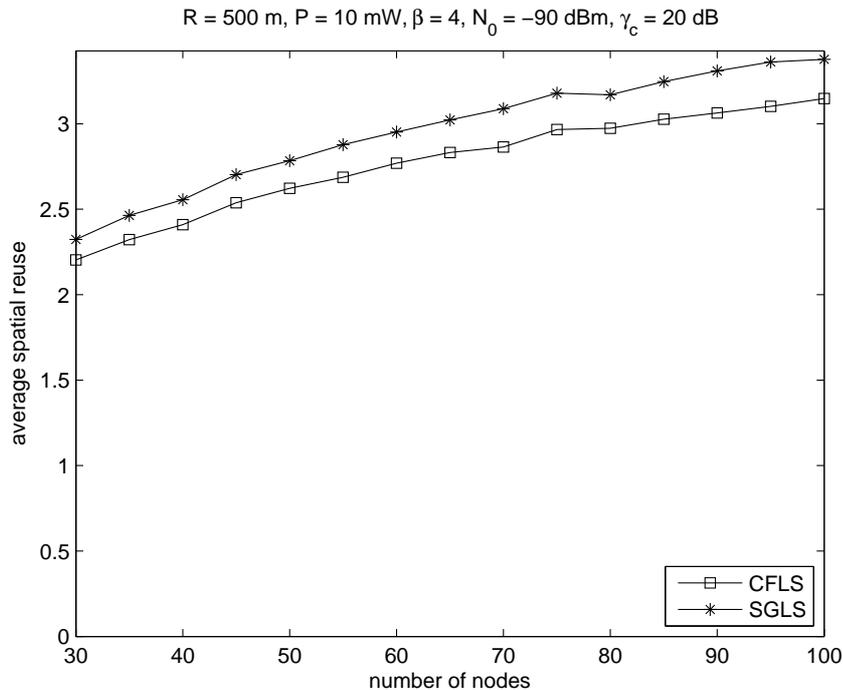}
  \caption{Spatial reuse vs. number of nodes for Experiment 1.}
  \label{fig:expt1_sgls}
\end{figure}

In the second experiment (Experiment 2), we assume that $R=700$ m,
$P=15$ mW, $\beta = 4$, $N_0=-85$ dBm and $\gamma_c = 15$ dB.  Thus,
$R_c=110.7$ m.  We vary the number of nodes from 70 to 150 in steps of
10. Figure \ref{fig:expt2_sgls} plots the average spatial reuse vs.
number of nodes for both the algorithms.

\begin{figure}[thbp]
  \centering
  \includegraphics[width=5in]{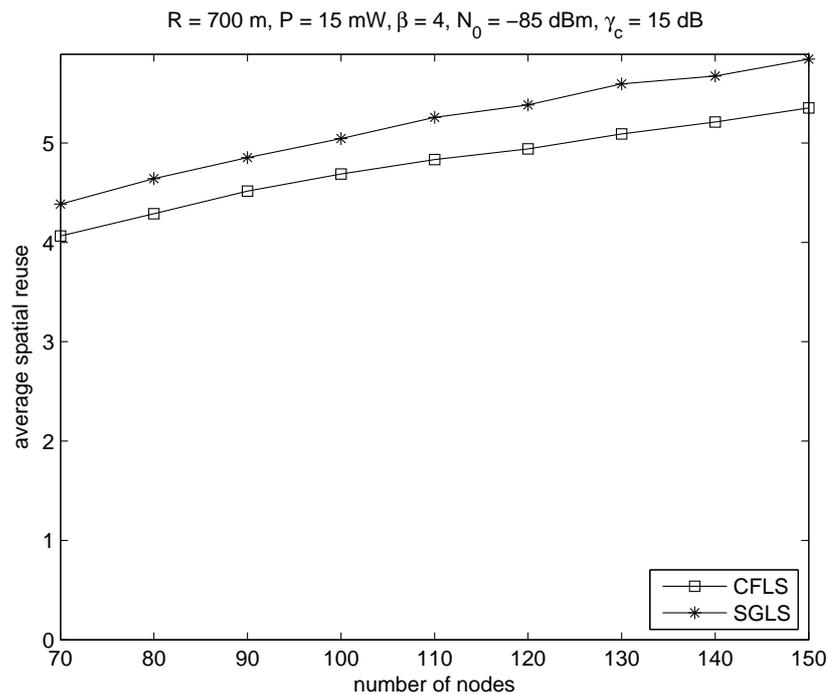}
  \caption{Spatial reuse vs. number of nodes for Experiment 2.}
  \label{fig:expt2_sgls}
\end{figure}

From the figures, we observe that SGLS achieves 5-10\% higher spatial
reuse than CFLS. However, this improvement in performance is obtained
at the cost of higher computational complexity.

\section{Discussion}
\label{sec:conclusion_lgls}

In this chapter, we have proposed a novel point to point link
scheduling algorithm based on an SINR graph representation of an STDMA
wireless network under the physical interference model.  Our results
demonstrate that the spatial reuse for the proposed algorithm is
higher than that of the ConflictFreeLinkSchedule algorithm.  This is
due to the fact that we have embedded interference conditions between
pairs of nodes into the edge weights and normalized noise powers at
receiver nodes into vertex weights of the SINR graph and consequently
determined a conflict-free schedule. Our approach has the potential to
scale with the number of nodes in the network.

\clearpage{\pagestyle{empty}\cleardoublepage}

\chapter{Point to Multipoint Link Scheduling: A Hybrid Approach}
\label{ch:broadcastschedule}

In this chapter, we investigate point to multipoint link scheduling in
STDMA wireless networks.  We generalize the definition of spatial
reuse introduced in Chapter \ref{ch:framework_link} for point to
multipoint link scheduling. We propose a ``hybrid'' link scheduling
algorithm based on a communication graph representation of the network
and SINR conditions.  We demonstrate that the proposed algorithm
achieves higher spatial reuse than existing algorithms, without any
increase in running time complexity.

The rest of this chapter is organized as follows. In Section
\ref{sec:system_multipoint}, we describe our system model.  We
describe point to multipoint link scheduling based on the protocol
interference model in Section \ref{sec:protocol_based_broadcast} and
describe its limitations in Section
\ref{sec:limitations_protocol_broadcast}.  In Section
\ref{sec:problem_formulation_broadcast}, we introduce spatial reuse as
our performance metric and formulate the problem. In Section
\ref{sec:mass_algorithm}, we describe the proposed link scheduling
algorithm.  We evaluate its performance in Section
\ref{sec:performance_mass} and derive its computational complexity in
Section \ref{sec:complexity_mass}.  We discuss the implications of our
work in Section \ref{sec:conclusions_mass}.

\section{System Model}
\label{sec:system_multipoint}

Our system model and notations are exactly as described in Section
\ref{sec:system_model_stdma}. However, we redefine and introduce terms
that are applicable to point to multipoint link scheduling.

If node $k$ is within node $j$'s communication range, then $k$ is
defined as a {\it neighbor} of $j$, since $k$ can decode $j$'s packet
correctly (subject to Equation \ref{eq:sinr_ge_gammac}).  Note that if
node $k$ is outside node $j$'s communication range, then it can never
decode $j$'s packet correctly (from Equation \ref{eq:sinr_ge_gammac}).
The number of neighbors of node $j$ is denoted by $\eta(j)$.

A point to multipoint link schedule for an STDMA wireless network
$\Phi(\cdot)$ is a mapping from the set of nodes to time slots.  Let
$C$ denote the number of time slots in a point to multipoint link
schedule.  For a given time slot $i$, $j^{th}$ point to multipoint
transmission is denoted by $\{t_{i,j}\rightarrow
\{r_{i,j,1},r_{i,j,2},\ldots,r_{i,j,\eta(t_{i,j})}\}\}$, where
$t_{i,j}$ denotes the index of the node which transmits a packet and
$r_{i,j,1},r_{i,j,2},\ldots,r_{i,j,\eta(t_{i,j})}$ denote the indices
of neighboring nodes (neighbors of $t_{i,j}$) that receive the packet.
Note that $r_{i,j,k}$ denotes $k^{th}$ receiver of $j^{th}$
transmission in time slot $i$.  Let $M_i$ denote the number of
concurrent point to multipoint transmissions in time slot $i$.  A
point to multipoint link schedule for an STDMA network $\Phi(\cdot)$
is denoted by $\Omega(\mathcal B_1,\cdots,\mathcal B_C)$, where
\begin{eqnarray*}
{\mathcal B}_i 
&:=& \{t_{i,1}\rightarrow \{r_{i,1,1},r_{i,1,2},\ldots,r_{i,1,\eta(t_{i,1})}\},
 \cdots,t_{i,M_i}\rightarrow \{r_{i,M_i,1},r_{i,M_i,2},\ldots,r_{i,M_i,\eta(t_{i,M_i})}\}\} \\
&=& 
\mbox{set of concurrent point to multipoint transmissions in time slot $i$}.
\end{eqnarray*}

\noindent Every point to multipoint schedule $\Omega(\cdot)$ must
satisfy the following:
\begin{enumerate}
\item Operational constraints:
  \begin{enumerate}
  \item A node cannot transmit and receive in the same time slot, i.e.,
    \begin{eqnarray}
      \{t_{i,j}\} \cap \{r_{i,k,1}, \ldots, r_{i,k,\eta(t_{i,k})}\} &=& \phi
      \;\; \forall \;\; i=1,\ldots,C \;\; \forall \;\; j \neq k.
      \label{eq:op_constraint_broadcast_1}
    \end{eqnarray} 

  \item A node cannot receive from multiple transmitters in the same
    time slot, i.e.,
    \begin{multline}
      \{r_{i,j,1}, \ldots, r_{i,j,\eta(t_{i,j})}\} \cap \{r_{i,k,1},
      \ldots, r_{i,k,\eta(t_{i,k})}\} = \phi \;\; \forall \;\;
      i=1,\ldots,C \;\; \forall \;\; j \neq k.
      \label{eq:op_constraint_broadcast_2}
    \end{multline}
  \end{enumerate}

\item Range constraint: Every receiver is within the communication
  range of its intended transmitter, i.e.,
  \begin{multline}
    D(t_{i,j},r_{i,j,k}) \leqslant R_c \;\forall\; i=1,\ldots,C
    \;\forall \; j=1,\ldots,M_i \;\forall\; k=1,\ldots,\eta(t_{i,j}).
    \label{eq:range_constraint_broadcast}
  \end{multline}
\end{enumerate}

\begin{figure}
\centering
\subfigure[An STDMA wireless network with six nodes.]
{
\label{fig:deployment_stdma_broadcast}
\includegraphics[width=5.5in]{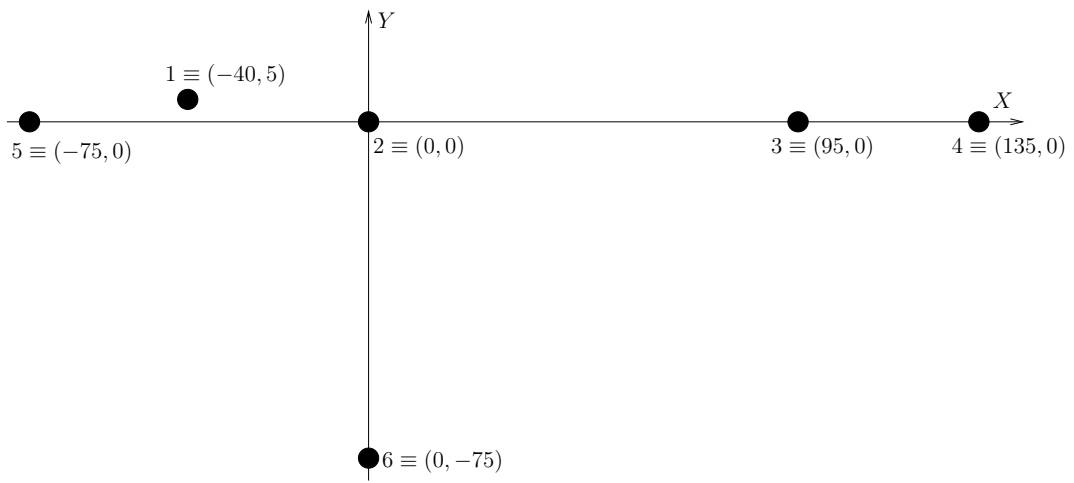}
}\\
\subfigure[A point to multipoint link schedule for the network shown in Figure \ref{fig:deployment_stdma_broadcast}.]
{
\label{fig:multipoint_link_schedule}
\includegraphics[width=5.5in]{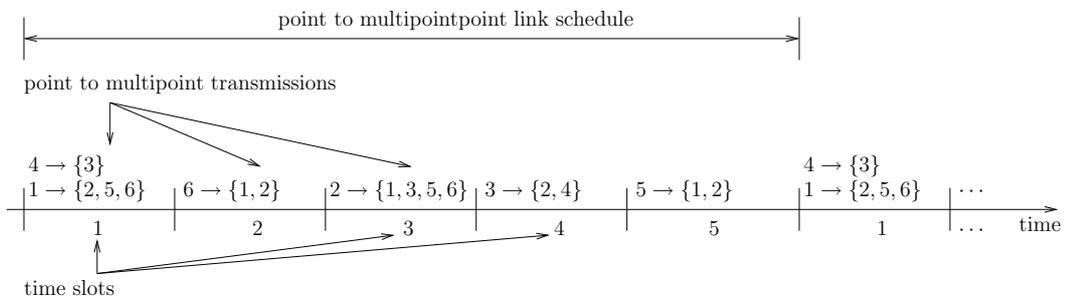}
}
\caption{Example of STDMA network and point to multipoint link schedule.}
\label{fig:deployment_broadcast_schedule}
\end{figure}

For an example, consider the STDMA wireless network $\Phi(\cdot)$
shown in Figure \ref{fig:deployment_stdma_broadcast}. It consists of
six nodes whose coordinates (in meters) are $1 \equiv (-40,5)$, $2
\equiv (0,0)$, $3 \equiv (95,0)$, $4 \equiv (135,0)$, $5 \equiv
(-75,0)$ and $6 \equiv (0,-75)$. One of the possible point to
multipoint link schedules for this STDMA network is shown in Figure
\ref{fig:multipoint_link_schedule}.  The schedule length is $C=5$ time
slots and the schedule is defined by $\Omega(\mathcal B_1,\mathcal
B_2,\mathcal B_3,\mathcal B_4,\mathcal B_5)$, where
\begin{eqnarray*}
\mathcal B_1 
&=& \big\{t_{1,1} \rightarrow \{r_{1,1,1},r_{1,1,2},r_{1,1,3}\}, \;
 t_{1,2} \rightarrow \{r_{1,2,1}\} \big\} \\
&=& \big\{1 \rightarrow \{2,5,6\}, \; 4 \rightarrow \{3\} \big\}, \\
\mathcal B_2
&=& \big\{ t_{2,1} \rightarrow \{r_{2,1,1},r_{2,1,2}\} \big\} \\
&=& \big\{ 6 \rightarrow \{1,2\} \big\}, \\
\mathcal B_3
&=& \big\{ t_{3,1} \rightarrow 
  \{r_{3,1,1},r_{3,1,2},r_{3,1,3},r_{3,1,4}\} \big\} \\
&=& \big\{ 2 \rightarrow \{1,3,5,6\} \big\}, \\
\mathcal B_4
&=& \big\{ t_{4,1} \rightarrow 
  \{r_{4,1,1},r_{4,1,2}\} \big\} \\
&=& \big\{ 3 \rightarrow \{2,4\} \big\}, \\
\mathcal B_5
&=& \big\{ t_{5,1} \rightarrow \{r_{5,1,1},r_{5,1,2}\} \big\} \\
&=& \big\{ 5 \rightarrow \{1,2\} \big\}.
\end{eqnarray*}
After 5 time slots, the schedule repeats periodically, as shown in
Figure \ref{fig:multipoint_link_schedule}.

A point to multipoint link scheduling algorithm is a set of rules that
is used to determine a schedule $\Omega(\cdot)$. Typically, a
scheduling algorithm is required to satisfy certain objectives.

Consider $k^{th}$ receiver of $j^{th}$ transmission in time slot $i$,
i.e., receiver $r_{i,j,k}$. The power received at $r_{i,j,k}$ from its
intended transmitter $t_{i,j}$ (signal power) is
$\frac{P}{D^{\beta}(t_{i,j},r_{i,j,k})}$. The power received at
$r_{i,j,k}$ from its unintended transmitters (interference power) is
$\sum_{\stackrel{l=1}{l\neq
    j}}^{M_i}\frac{P}{D^{\beta}(t_{i,l},r_{i,j,k})}$. Thus, the SINR
at receiver $r_{i,j,k}$ is given by
\begin{eqnarray}
{\mbox{SINR}}_{r_{i,j,k}} 
&=& 
 \frac{\frac{P}{D^{\beta}(t_{i,j},r_{i,j,k})}}{N_0+\sum_{\stackrel{l=1}{l\neq j}}^{M_i}\frac{P}{D^{\beta}(t_{i,l},r_{i,j,k})}} .
\label{eq:sinr_broadcast}
\end{eqnarray}

According to the physical interference model
\cite{gupta_kumar__capacity_wireless}, receiver $r_{i,j,k}$ can
successfully decode the packet transmitted by $t_{i,j}$ if the SINR at
$r_{i,j,k}$ is no less than the communication threshold $\gamma_c$,
i.e.,
\begin{eqnarray}
{\mbox{SINR}}_{r_{i,j,k}} &\geqslant& \gamma_c .
\label{eq:sinr_gammac_broadcast}
\end{eqnarray}

A link schedule $\Omega(\cdot)$ is {\it exhaustive} if every two nodes
$j$, $k$ who are neighbors of each other are included in the schedule
exactly twice, once with $j$ being a transmitter and $k$ being one of
its receivers, and during another time slot with $k$ being a
transmitter and $j$ being one of its receivers.  


\section{Equivalence of Link Scheduling and Graph Vertex Coloring}
\label{sec:protocol_based_broadcast}

In this section, we describe the equivalence between a point to
multipoint link schedule for an STDMA wireless network and the coloring of
vertices of the communication graph representation (see Section
\ref{sec:equivalence_coloring}) of the network.

\begin{table}[tbhp]
\centering
\begin{tabular}{|l|l|l|} \hline
  Parameter & Symbol & Value \\ \hline
  transmission power & $P$ & 10 mW \\ \hline
  path loss exponent & $\beta$ & 4 \\ \hline
  noise power spectral density & $N_0$ & -90 dBm \\ \hline
  communication threshold & $\gamma_c$ & 20 dB \\ \hline
  interference threshold & $\gamma_i$ & 10 dB \\ \hline
\end{tabular}
\caption{System parameters for STDMA networks shown in Figures
  \ref{fig:deployment_stdma_broadcast} and 
  \ref{fig:deploy_broadcast_high_intf}.}
\label{tab:sys_parameters_broadcast}
\end{table}

\begin{figure}[thbp]
  \centering
  \includegraphics[width=5.5in]{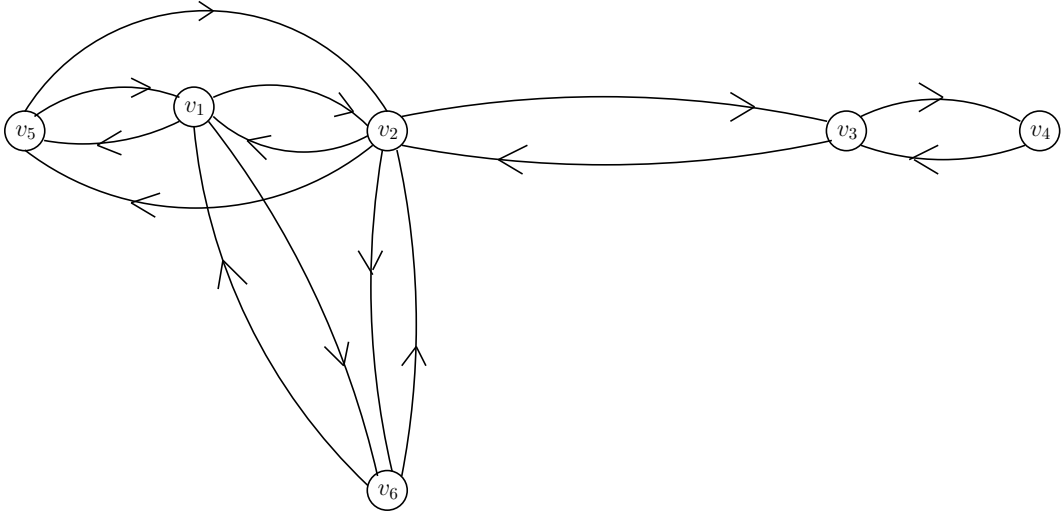}
  \caption{Communication graph model of STDMA network described by
    Figure \ref{fig:deployment_stdma_broadcast} and Table
    \ref{tab:sys_parameters_broadcast}.}
  \label{fig:communication_graph_multipoint}
\end{figure}

Consider the STDMA wireless network $\Phi(\cdot)$ whose deployment is
shown in Figure \ref{fig:deployment_stdma_broadcast}.  The system
parameters for this network are given in Table
\ref{tab:sys_parameters_broadcast}. From
(\ref{eq:communication_range}), we obtain $R_c=100$ m.  The
corresponding communication graph representation $\mathcal
G_c(\mathcal V,\mathcal E_c)$ is shown in Figure
\ref{fig:communication_graph_multipoint}.  The communication graph
comprises of 6 vertices and 14 directed communication edges.  The
vertex and communication edge sets are given by
\begin{eqnarray}
\mathcal V
&=& \{v_1,v_2,v_3,v_4,v_5,v_6\} \label{eq:set_vertices}, \\
\mathcal E_c
&=& \{
v_1 \stackrel{c}{\rightarrow} v_2,
v_2 \stackrel{c}{\rightarrow} v_1,
v_1 \stackrel{c}{\rightarrow} v_5,
v_5 \stackrel{c}{\rightarrow} v_1,
v_1 \stackrel{c}{\rightarrow} v_6,
v_6 \stackrel{c}{\rightarrow} v_1,
v_2 \stackrel{c}{\rightarrow} v_5, \nonumber \\
&&
v_5 \stackrel{c}{\rightarrow} v_2,
v_2 \stackrel{c}{\rightarrow} v_6, 
v_6 \stackrel{c}{\rightarrow} v_2,
v_2 \stackrel{c}{\rightarrow} v_3,
v_3 \stackrel{c}{\rightarrow} v_2,
v_3 \stackrel{c}{\rightarrow} v_4,
v_4 \stackrel{c}{\rightarrow} v_3
\}. \label{eq:set_comm_edges}
\end{eqnarray}

\begin{figure}[thbp]
\centering
\includegraphics[width=5.5in]{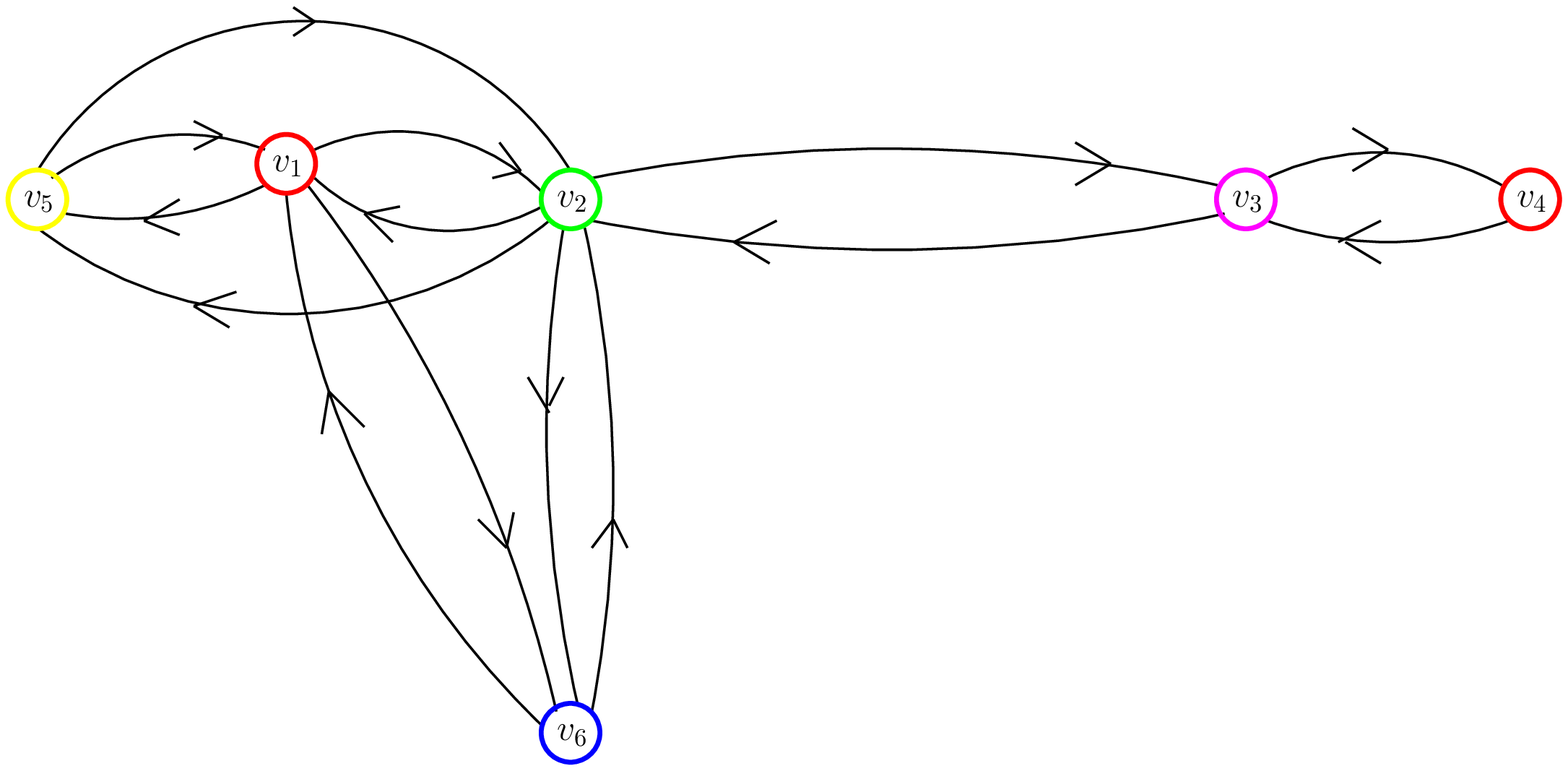}
\caption{Vertex coloring of communication graph shown in Figure
  \ref{fig:communication_graph_multipoint} corresponding to the link
  schedule shown in Figure \ref{fig:multipoint_link_schedule}.}
\label{fig:comm_graph_bcast_colored}
\end{figure}

Given the above representation of an STDMA network, a point to
multipoint link schedule $\Omega(\cdot)$ can be considered as
equivalent to assigning a unique color to every vertex in the
communication graph, such that nodes with the same color transmit
simultaneously in a particular time slot. For the example network
considered, the link schedule shown in Figure
\ref{fig:multipoint_link_schedule} corresponds to the coloring of the
vertices of the communication graph shown in Figure
\ref{fig:comm_graph_bcast_colored}. Time slots 1, 2, 3, 4 and 5 in
$\Omega(\cdot)$ correspond to colors red, blue, green, magenta and
yellow in $\mathcal V$ respectively. Note that a coloring algorithm
that uses the least number of colors also minimizes the schedule
length.

Algorithms for assigning nodes to time slots (equivalently, colors)
require that two vertices $v_i$, $v_j$ can be colored the same if and
only if:
\begin{enumerate}
\item 
  edge $v_i \stackrel{c}{\rightarrow} v_j \not\in {\mathcal E}_c$ and
  edge $v_j \stackrel{c}{\rightarrow} v_i \not\in {\mathcal E}_c$,
  i.e., there is no {\it primary vertex conflict,} and
  \label{primary_vertex_conflict}

\item
  there is no vertex $v_k$ such that $v_i \stackrel{c}{\rightarrow}
  v_k \in {\mathcal E}_c$ and $v_j \stackrel{c}{\rightarrow} v_k \in
  {\mathcal E}_c$ , i.e., there is no {\it secondary vertex conflict}.
  \label{secondary_vertex_conflict}
\end{enumerate}
These criteria are based on the operational constraints
(\ref{eq:op_constraint_broadcast_1}) and
(\ref{eq:op_constraint_broadcast_2}).

Algorithms based on the protocol interference model represent the
network by a communication graph and utilize various graph coloring
methodologies to devise heuristics which yield a minimum length
schedule. Hence, such algorithms have the merit of low computational
complexity. However, recent research suggests that these algorithms
yield low network throughput. This aspect is elaborated in the
following section.

\section{Limitations of Algorithms based on Protocol Interference
  Model}
\label{sec:limitations_protocol_broadcast}

In this section, we illustrate that algorithms based on the protocol
interference model can result in schedules that yield low network
throughput. Note that the limitations of point to multipoint link
scheduling algorithm are similar to those of point to point link
scheduling algorithms described in Section
\ref{sec:limitations_protocol}.

With the intent of maximizing the throughput of an STDMA network,
algorithms based on the protocol interference model transform the
scheduling problem to a vertex coloring problem for the communication
graph representation of the network.  For example, the
BroadcastSchedule algorithm
\cite{ramanathan_lloyd__scheduling_algorithms} works in two phases.
In Phase 1, the vertices of the communication graph are labeled using
the labeler function (Algorithm \ref{func:labeler2}, Section
\ref{sec:improvement_als}).  In Phase 2, vertices are considered in
increasing order of label. For the vertex $u$ under consideration, it
discards any color that leads to primary or secondary vertex conflicts
with $u$. The least color among the residual set of non-conflicting
colors is used to color vertex $u$. If no non-conflicting color
exists, vertex $u$ is colored with a new color.

The simplification of the link scheduling problem in a wireless
network as a vertex coloring problem on the communication graph can
result in schedules that violate the SINR threshold condition
(\ref{eq:sinr_gammac_broadcast}).  Specifically, algorithms based on
the protocol interference model do not necessarily maximize the
throughput of an STDMA network because:
\begin{enumerate}

\item They can result in high cumulative interference at a receiver,
  due to hard-thresholding based on communication radius.  This is
  because the SINR at receiver $r_{i,j,k}$ decreases with an increase
  in the number of concurrent transmissions $M_i$, while $R_c$ has
  been defined for a single transmission only.

  \begin{figure}[thbp]
    \centering
    \includegraphics[width=5.5in]{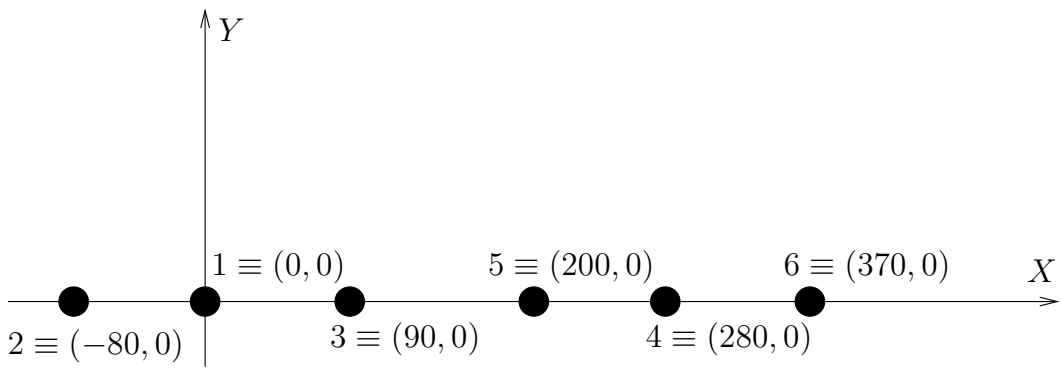}
    \caption{An STDMA wireless network with six nodes.}
    \label{fig:deploy_broadcast_high_intf}
  \end{figure}

  \begin{figure}[thbp]
    \centering
    \includegraphics[width=6in]{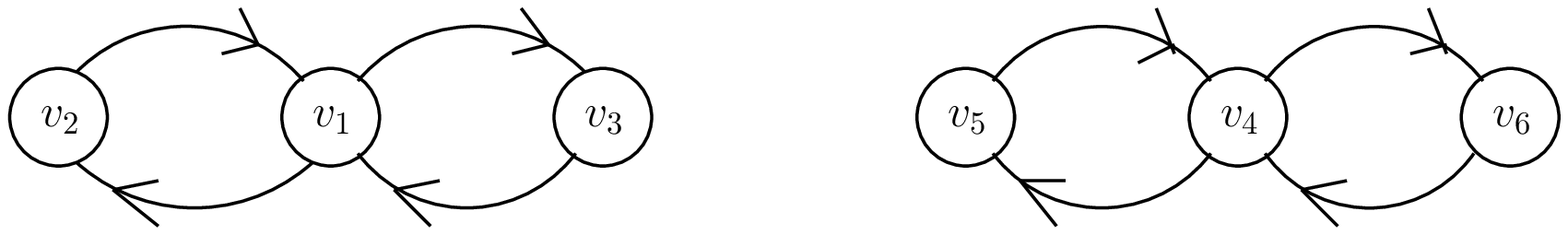}
    \caption{Communication graph model of STDMA network described by
      Figure \ref{fig:deploy_broadcast_high_intf} and Table
      \ref{tab:sys_parameters_broadcast}.}
    \label{fig:comm_graph_broadcast}
  \end{figure}

  \begin{figure}[thbp]
    \centering
    \includegraphics[width=6in]{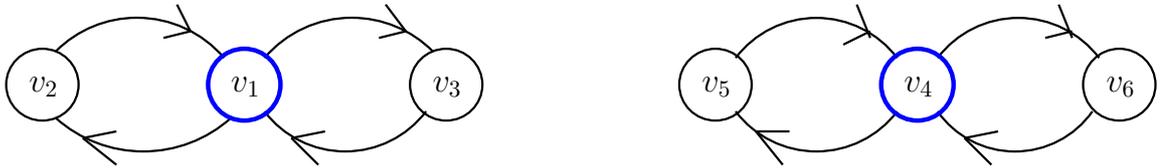}
    \caption{Coloring of vertices $v_1$ and $v_4$ of graph shown in Figure
      \ref{fig:deploy_broadcast_high_intf}.}
    \label{fig:comm_graph_broadcast_colored}
  \end{figure}

  \begin{figure}[thbp]
    \centering
    \includegraphics[width=6.2in]{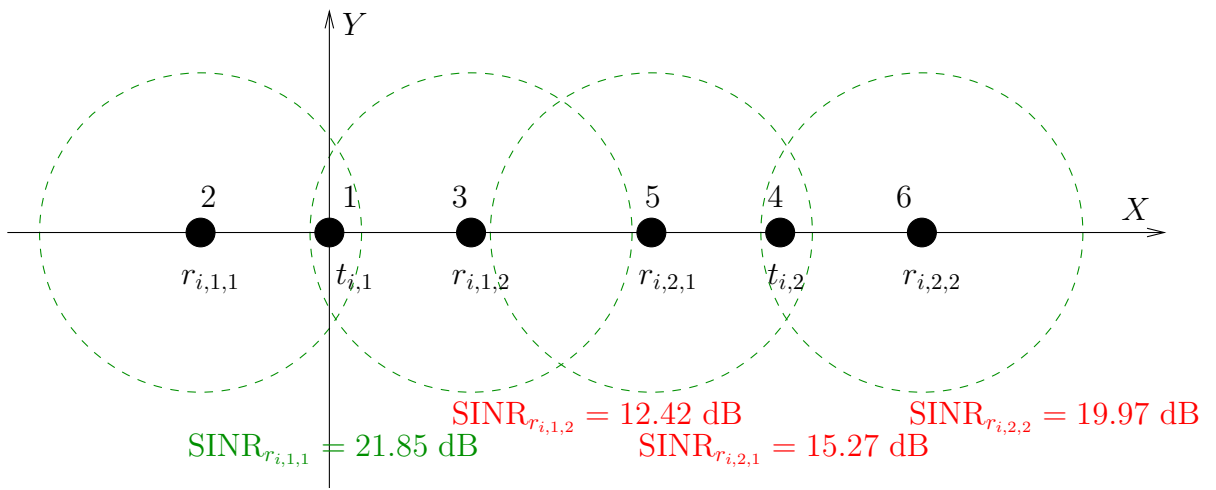}
    \caption{Point to multipoint link scheduling algorithms based on
      protocol interference model can lead to high interference.}
    \label{fig:high_interference_broadcast}
  \end{figure}

  For example, consider the STDMA wireless network whose deployment is
  shown in Figure \ref{fig:deploy_broadcast_high_intf}.  The network
  consists of six nodes whose coordinates (in meters) are $1 \equiv
  (0,0)$, $2 \equiv (-80,0)$, $3 \equiv (90,0)$, $4 \equiv (280,0)$,
  $5 \equiv (200,0)$ and $6 \equiv (370,0)$. The system parameters are
  shown in Table \ref{tab:sys_parameters_broadcast}, which yields
  $R_c=100$ m. The communication graph model of the STDMA network is
  shown in Figure \ref{fig:comm_graph_broadcast}.  Consider the
  transmission requests $1 \rightarrow \{2,3\}$ and $4 \rightarrow
  \{5,6\}$, which correspond to vertices $v_1$ and $v_4$ of the graph
  shown in Figure \ref{fig:comm_graph_broadcast}.  Note that vertices
  $v_1$ and $v_4$ do not have primary or secondary vertex conflicts.
  So, to minimize the number of colors, such an algorithm will color
  these vertices with the same color, as shown in Figure
  \ref{fig:comm_graph_broadcast_colored}.  Equivalently, transmissions
  $1 \rightarrow \{2,3\}$ and $4 \rightarrow \{5,6\}$ will be
  scheduled in the same time slot, say time slot $i$.  However, our
  computations show that the SINRs at receivers $r_{i,1,1}$,
  $r_{i,1,2}$, $r_{i,2,1}$ and $r_{i,2,2}$ are $21.85$ dB, $12.42$ dB,
  $15.27$ dB and $19.97$ dB respectively.  Figure
  \ref{fig:high_interference_broadcast} shows the nodes of the network
  along with the labeled transmitter-receivers sets, receiver-centric
  communication zones and SINRs at the receivers.  From the SINR
  threshold condition (\ref{eq:sinr_gammac_broadcast}), transmission
  $t_{i,1} \rightarrow r_{i,1,1}$ is successful, while transmissions
  $t_{i,1} \rightarrow r_{i,1,2}$, $t_{i,2} \rightarrow r_{i,2,1}$ and
  $t_{i,2} \rightarrow r_{i,2,2}$ are unsuccessful. This leads to low
  network throughput.

\item Moreover, these algorithms are not aware of the topology of the
  network, i.e., they determine a link schedule without being
  cognizant of the exact positions of the transmitters and receivers.

\end{enumerate}

As argued above, point to multipoint link scheduling algorithms based
on the protocol interference model can result in low network
throughput.  In essence, algorithms that construct an approximate
model of an STDMA network (communication graph) and concentrate on
minimizing the schedule length do not necessarily maximize network
throughput.  This observation is developed into a proposal for an
appropriate performance metric in Section
\ref{sec:problem_formulation_broadcast}.

\section{Problem Formulation}
\label{sec:problem_formulation_broadcast}

In this section, we motivate the need for a performance metric that
takes into account the SINR threshold condition
(\ref{eq:sinr_gammac_broadcast}) as the criterion for successful
packet reception. Analogous to the notion of spatial reuse, we propose
a performance metric for point to multipoint link scheduling, which is
also termed as spatial reuse. We argue that spatial reuse is directly
proportional to the number of successful point to multipoint
transmissions.  Finally, we formulate the scheduling problem from a
perspective of maximizing spatial reuse.

Algorithms based on the protocol interference model are inadequate to
design efficient point to multipoint link schedules. This is because
these algorithms are entirely based on the communication graph
$\mathcal G_c(\mathcal V,\mathcal E_c)$, which is a crude
approximation of $\Phi(\cdot)$, and can lead to low network
throughput, as argued in Section
\ref{sec:limitations_protocol_broadcast}.  On the other hand, from
$\Phi(\cdot)$ and ${\mathcal G}_c(\cdot)$, one can exhaustively
determine the link schedule $\Omega(\cdot)$ which yields highest
network throughput according to the physical interference model.
However, this is a combinatorial optimization problem of prohibitive
complexity $(O(|\mathcal V|^{|\mathcal V|}))$ and is thus
computationally infeasible.

To overcome these problems, we propose a point to multipoint link
scheduling algorithm for STDMA wireless networks under the physical
interference model.  Our algorithm is based on the communication graph
model ${\mathcal G}_c({\mathcal V},{\mathcal E}_c)$ as well as SINR
computations.

To evaluate the performance of our algorithm and compare it with
existing link scheduling algorithms, we define the notion of
spatial reuse.  Consider the point to multipoint link schedule
$\Omega(\cdot)$ for the STDMA network $\Phi(\cdot)$.  Under the
physical interference model, transmission $t_{i,j} \rightarrow
r_{i,j,k}$ is successful if and only if
(\ref{eq:sinr_gammac_broadcast}) is satisfied.  The {\it spatial
  reuse} of the link schedule $\Omega(\cdot)$ is defined as the
average number of successful point to multipoint transmissions per
time slot.  Thus
\begin{eqnarray}
\mbox{Spatial Reuse} \;=\; \varsigma &:=&
 \frac{\sum_{i=1}^C\sum_{j=1}^{M_i}\frac{\sum_{k=1}^{\eta(t_{i,j})}I({\mbox{\scriptsize SINR}}_{r_{i,j,k}}\geqslant\gamma_c)}{\eta(t_{i,j})}}{C},
\label{eq:spatial_reuse_broadcast}
\end{eqnarray}
where $I(A)$ denote the indicator function for event $A$, i.e.,
$I(A)=1$ if event $A$ occurs, $I(A)=0$ if event $A$ does not occur.
Note that in (\ref{eq:spatial_reuse_broadcast}), the number of nodes
that successfully receive a transmitted packet is normalized by the
number of neighbors of the transmitting node.  A high value of spatial
reuse corresponds to high network throughput.

The essence of STDMA is to have a reasonably large number of
simultaneous and successful transmissions.  For an STDMA wireless
network which is operational for a long period of time, say $L$ time
slots, the total number of successful point to multipoint
transmissions is $L\varsigma$.  Thus, a high value of spatial reuse
directly translates to higher network throughput and the number of
colors $C$ is relatively unimportant.  Hence, spatial reuse turns out
to be a crucial metric for the comparison of various STDMA link
scheduling algorithms.

Our goal is to design a low complexity point to multipoint link
scheduling algorithm that achieves high spatial reuse, where spatial
reuse is given by (\ref{eq:spatial_reuse_broadcast}). We only consider
link schedules that are feasible and exhaustive.

\section{MaxAverageSINRSchedule Algorithm}
\label{sec:mass_algorithm}

Our proposed point to multipoint link scheduling scheduling algorithm
under the physical interference model is MaxAverageSINRSchedule
(MASS), which considers the communication graph ${\mathcal
  G}_c({\mathcal V},{\mathcal E}_c)$ and is described in Algorithm
\ref{algo:mass}.

\begin{algorithm}
\caption{MaxAverageSINRSchedule (MASS)}
\label{algo:mass}
\begin{algorithmic}[1]
\STATE {\bf input:} STDMA wireless network $\Phi(\cdot)$, communication graph 
  $\mathcal G_c(\mathcal V, \mathcal E_c)$
\STATE {\bf output:} A coloring $C: {\mathcal V} \rightarrow \{1,2,\ldots\}$
\STATE label the vertices of ${\mathcal G}_c$ randomly \COMMENT{Phase 1} 
\FOR[Phase 2 begins]{$j \leftarrow 1 \mbox{ to } n$}
\STATE let $u$ be such that $L(u)=j$
\STATE $C(u) \leftarrow \mbox{MaxAverageSINRColor}(u)$
\ENDFOR \COMMENT{Phase 2 ends}
\end{algorithmic}
\end{algorithm}

In Phase 1, we label all the vertices randomly\footnote{Randomized
  algorithms are known to outperform deterministic algorithms,
  especially when the characteristics of the input are not known a
  priori \cite{motwani_raghavan__randomized_algorithms}.}.
Specifically, if $\mathcal G_c(\cdot)$ has $v$ vertices, we perform a
random permutation of the sequence $(1,2,\ldots,v)$ and assign these
labels to vertices with indices $1,2,\ldots,v$ respectively. Let
$L(u)$ denote the label assigned to vertex $u$.

\begin{algorithm}
\caption{integer MaxAverageSINRColor($u$)}
\label{func:masc}
\begin{algorithmic}[1]
\STATE {\bf input:} STDMA wireless network $\Phi(\cdot)$, communication graph 
  $\mathcal G_c(\mathcal V,\mathcal E_c)$
\STATE {\bf output:} A non-conflicting color
\STATE ${\mathcal C} \leftarrow \mbox{set of existing colors}$
\STATE ${\mathcal C}_p \leftarrow \{C(x):\mbox{$x$ is colored and is a neighbor of $u$}\}$
\STATE ${\mathcal C}_s \leftarrow \{C(x):\mbox{$x$ is colored and is two hops away from $u$}\}$
\STATE ${\mathcal C}_{nc} = {\mathcal C} \setminus \{{\mathcal C}_p \cup {\mathcal C}_s\}$
\IF{${\mathcal C}_{nc} \neq \phi$}
\STATE $r \leftarrow \mbox{color in ${\mathcal C}_{nc}$ which results in 
 maximum average}$ $\mbox{SINR at neighbors of $u$}$
\IF{$\mbox{maximum average SINR} \geqslant \gamma_c$}
\STATE return $r$
\ENDIF
\ENDIF
\STATE return $|\mathcal C|+1$
\end{algorithmic}
\end{algorithm}

In Phase 2, the vertices are examined in increasing order of
label\footnote{In essence, the vertices are scanned in a random order,
  since labeling is random.} and the MaxAverageSINRColor (MASC)
function is used to assign a color to the vertex under consideration.
The MASC function is explained in Algorithm \ref{func:masc}.  It
begins by discarding all colors that have a primary or secondary
vertex conflict with $u$, the vertex under consideration. Among the
set of non-conflicting colors ${\mathcal C}_{nc}$, it chooses that
color for $u$ which results in the maximum value of average SINR at
the neighbors of $u$, provided this value exceeds the communication
threshold.  Intuitively, the average SINR is also a measure of the
average distance of every neighbor of $u$ from all co-colored
transmitters.  The higher the average SINR, the higher is this average
distance. We choose that color which results in the maximum average
SINR at the neighbors of $u$, so that the additional interference at
the neighbors of all co-colored transmitters is kept low. If no such
color is found, it assigns a new color to $u$.

\section{Performance Results}
\label{sec:performance_mass}

In this section, we describe our simulation model. We compare the
performance of the proposed algorithm with existing point to
multipoint link scheduling algorithms.

In our simulation experiments, the location of every node is generated
randomly in a circular region of radius $R$. If $(X_j,Y_j)$ are the
Cartesian coordinates of node $j$, then $X_j \sim U[-R,R]$ and $Y_j
\sim U[-R,R]$ subject to $X_j^2 + Y_j^2 \leqslant R^2$.  Equivalently,
if $(R_j,\Theta_j)$ are the polar coordinates of node $j$, then $R_j^2
\sim U[0,R^2]$ and $\Theta_j \sim U[0,2\pi]$.  Using
(\ref{eq:communication_range}) and (\ref{eq:interference_range}), we
compute $R_c$ and $R_i$, and then map the STDMA network $\Phi(\cdot)$
to the communication graph $\mathcal G(\mathcal V,\mathcal E_c)$.
Once the link schedule $\Omega(\cdot)$ is computed by every algorithm,
the spatial reuse $\varsigma$ is computed using
(\ref{eq:spatial_reuse_broadcast}).  We use two sets of prototypical
values of system parameters in wireless networks
\cite{kim_lim_hou__improving_spatial}.  For a given set of system
parameters, we calculate the average spatial reuse by averaging $\varsigma$
over 1000 randomly generated networks.  Keeping all other parameters
fixed, we observe the effect of increasing the number of nodes on the
average spatial reuse.

In our experiments, we compare the performance of the following
algorithms:
\begin{enumerate}
\item BroadcastSchedule (BS) \cite{ramanathan_lloyd__scheduling_algorithms}
\item MaxAverageSINRSchedule (MASS)
\end{enumerate}

In our first set of experiments (Experiment 1), we assume that $R=500$
m, $P=10$ mW, $\beta = 4$, $N_0=-90$ dBm, $\gamma_c = 20$ dB and
$\gamma_i=10$ dB. Thus, $R_c=100$ m and $R_i=177.8$ m.  We vary the
number of nodes from 30 to 110 in steps of 5. Figure
\ref{fig:expt1_broadcast} plots the average spatial reuse vs. number of nodes
for both the algorithms.

\begin{figure}[thbp]
\centering
\includegraphics[width=5in]{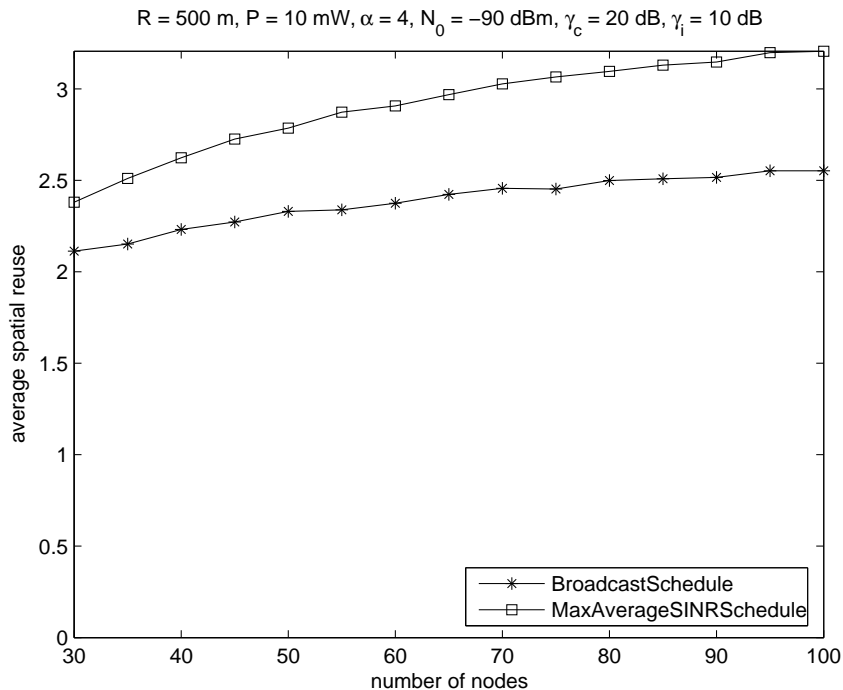}
\caption{Average spatial reuse vs. number of nodes for Experiment 1.}
\label{fig:expt1_broadcast}
\end{figure}

In our second set of experiments (Experiment 2), we assume that
$R=700$ m, $P=15$ mW, $\beta = 4$, $N_0=-85$ dBm, $\gamma_c = 15$ dB
and $\gamma_i=7$ dB. Thus, $R_c=110.7$ m and $R_i=175.4$ m.  We vary
the number of nodes from 70 to 150 in steps of 5. Figure
\ref{fig:expt2_broadcast} plots the average spatial reuse vs. number
of nodes for both the algorithms.

\begin{figure}[thbp]
\centering
\includegraphics[width=5in]{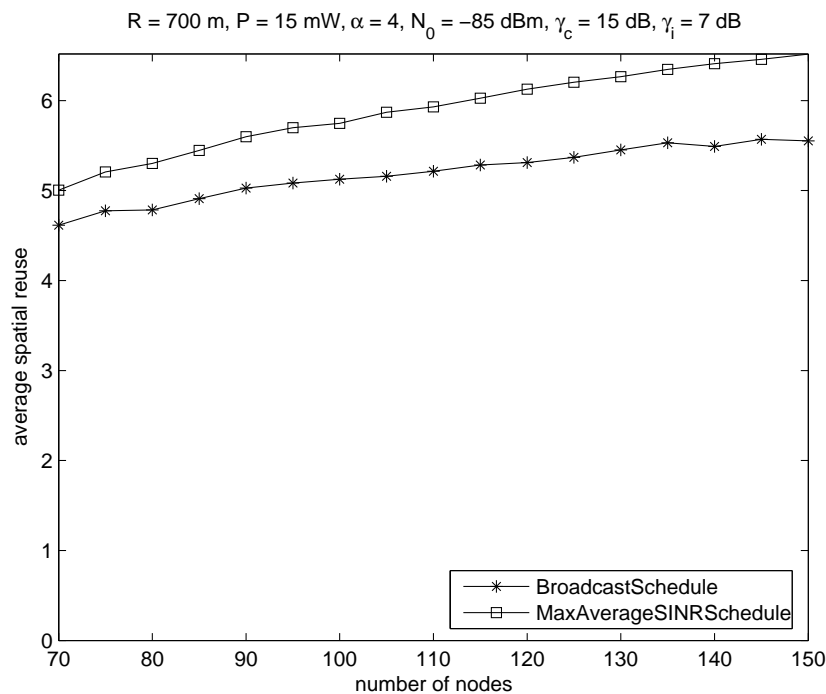}
\caption{Average spatial reuse vs. number of nodes for Experiment 2.}
\label{fig:expt2_broadcast}
\end{figure}

From Figures \ref{fig:expt1_broadcast} and \ref{fig:expt2_broadcast},
we observe that average spatial reuse increases with the number of
nodes for both the algorithms.  The MASS algorithm consistently yields
higher average spatial reuse compared to BS.  The spatial reuse of
MASS is about $15\%$ higher than BS in Experiment 1 and $4\%$ higher
in Experiment 2.  This improvement in performance translates to
substantially higher network throughput.

Also, an increase in the number of nodes in a given geographical area
leads to an increase in the number of vertices having a primary or
secondary vertex conflict with a given vertex. Hence, the number of
non-conflicting colors for a given vertex also decreases. From this
reduced set of non-conflicting colors, BroadcastSchedule chooses a
color randomly, while MaxAverageSINRSchedule chooses a color based on
SINR conditions. Since spatial reuse takes SINR threshold conditions
into account, the gap between average spatial reuse values increases
with number of nodes in Figures \ref{fig:expt1_broadcast} and
\ref{fig:expt2_broadcast}.

\section{Analytical Result}
\label{sec:complexity_mass}

In this section, we derive an upper bound on the running time
(computational) complexity of the MaxAverageSINRSchedule algorithm.
Let $v$ denote the number of vertices of the communication graph
${\mathcal G}_c({\mathcal V},{\mathcal E}_c)$.

\begin{theorem}
  The running time of MaxAverageSINRSchedule is $O(v^2)$.
\end{theorem}

\begin{proof}
  Assuming that an element can be chosen randomly and uniformly from a
  finite set in unit time
  (\cite{motwani_raghavan__randomized_algorithms}, Chapter 1), the
  running time of Phase 1 can be shown to be $O(v)$.  In Phase 2, the
  vertex under consideration is assigned a color using
  MaxAverageSINRColor. The worst-case size of the set of colors to be
  examined $|{\mathcal C}_{nc} \cup {\mathcal C}_p \cup {\mathcal
    C}_s|$ is $O(v)$.  With a careful implementation,
  MaxAverageSINRColor runs in time proportional to $|{\mathcal
    C}_{nc}|$, i.e., $O(v)$. Thus, the running time of Phase 2 is
  $O(v^2)$. Finally, the overall running time of
  MaxAverageSINRSchedule is $O(v^2)$.
\end{proof}

\section{Discussion}
\label{sec:conclusions_mass}

In this chapter, we have developed a point to multipoint link
scheduling algorithm for STDMA wireless networks under the physical
interference model, namely MaxAverageSINRSchedule.  The performance of
our algorithm is superior to existing algorithms.  A practical
experimental modeling shows that, on an average, our algorithm
achieves $15\%$ higher spatial reuse than the BroadcastSchedule
algorithm \cite{ramanathan_lloyd__scheduling_algorithms}.  Since link
schedules are constructed offline only once and then used by the
network for a long period of time, this improvement in performance
directly translates to higher network throughput.  The computational
complexity of MaxAverageSINRSchedule is also comparable to the
computational complexity of BroadcastSchedule. Therefore,
MaxAverageSINRSchedule is a good candidate for efficient STDMA point
to multipoint link scheduling algorithms.

\clearpage{\pagestyle{empty}\cleardoublepage}

\chapter{A Review of Random Access Algorithms for Wireless Networks}
\label{ch:review_random}

The MAC problem or multipoint to point problem is present in all
communication networks, both wired and wireless. Multiple nodes
(users) can access a single channel simultaneously to communicate with
each other or a common receiver -- the challenge is to design
efficient channel access algorithms to achieve the desired performance
in terms of throughput and delay. Several solutions to the MAC problem
have been proposed depending on source traffic characteristics,
channel models and Quality of Service (QoS) requirements of the users.

MAC protocols can be broadly classified into two types: fixed resource
allocation protocols and random access protocols.  Fixed resource
allocation protocols such as Time Division Multiple Access (TDMA),
Frequency Division Multiple Access (FDMA) and Code Division Multiple
Access (CDMA) assign orthogonal or near-orthogonal channels to every
user and are mostly implemented in voice-dominant wireless cellular
networks.  These protocols typically require the presence of a central
entity (base station) to perform channel allocation and admission
control, i.e., they are highly centralized. Though fixed resource
allocation protocols are contention-free and can multiplex users with
similar traffic characteristics easily, they suffer from low
throughput and high channel access delay when the traffic is bursty
and there are large number of users.  On the other hand, in random
access protocols, users vary their transmission probabilities or
transmission times based on limited channel feedback, i.e., random
access protocols are highly distributed.  Random access protocols are
more suitable for scenarios wherein many users with varied traffic
requirements have to be multiplexed, i.e., the traffic is bursty.

Random access algorithms for satellite communications, multidrop
telephone lines and multitap bus (``traditional random access
algorithms'') have been well studied for the past four decades. These
algorithms can be broadly classified into three categories: ALOHA
\cite{abramson__aloha_system}, \cite{roberts__aloha_packet}, Carrier
Sense Multiple Access \cite{kleinrock_tobagi__packet_switching} and
tree (or stack or splitting) algorithms
\cite{capetanakis__tree_algorithms}. Traditional random access
algorithms have been implemented in practical systems.  For example,
ALOHA is used in most cellular networks to request channel access and
also in satellite communication networks.  Carrier Sense Multiple
Access with Collision Detection (CSMA/CD) is used to resolve
contentions in Local Area Networks (LANs).

On the other hand, random access algorithms that incorporate physical
layer characteristics such as SINR and channel variations have only
been studied recently. These algorithms, which have been primarily
proposed for wireless networks, can be broadly classified into three
categories: algorithms based on signal processing and diversity
techniques, channel-aware ALOHA algorithms based on adapting the
retransmission probabilities of contending users and ``tree-like''
algorithms based on adapting the set of contending users.  Existing
random access algorithms, such as Carrier Sense Multiple Access with
Collision Avoidance (CSMA/CA ), are not channel-aware and can lead to
low throughput.  Thus, the design of physical layer aware random
access algorithms can be a potential step towards achieving higher
data rates in future wireless networks.

The organization of this chapter is as follows. Section
\ref{sec:traditional_random_access} provides a summary of traditional
random access algorithms along with the canonical system model,
performance metrics and well-known random access techniques such as
ALOHA and tree algorithms.  This helps us understand channel-aware
generalizations of these algorithms. In Section
\ref{sec:sigpro_random_access}, we review research papers which employ
signal processing and diversity techniques to correctly decode packets
in random access wireless networks. We critically review some of the
research which focus on channel-aware ALOHA and tree-like algorithms
for wireless networks in Sections \ref{sec:aloha_wireless} and
\ref{sec:splitting_wireless} respectively. Finally, we motivate the
use of variable transmission power to increase the throughput of
random access wireless networks in Section \ref{sec:power_controlled}.

\section{Traditional Random Access Algorithms}
\label{sec:traditional_random_access}

In this section, we describe the idealized slotted multiaccess model,
which can be used to represent various multiaccess media such as
satellite channels, multidrop telephone lines and multitap bus.  We
explain traditional random access algorithms such as ALOHA and tree
algorithms. We also describe the performance metrics used to analyze
and evaluate random access algorithms, namely, throughput, delay and
stability.

Consider an idealized slotted multiaccess system with $m$ transmitting
nodes and one receiver. The assumptions of the model are
\cite{bertsekas_gallager__data_networks}:
\begin{enumerate}

\item Slotted system: All transmitted packets have the same length and
  each packet requires one time unit, called a slot, for transmission.

\item One of the following is usually assumed:
  \begin{enumerate}
  \item Poisson arrivals: Packets arrive at each of the $m$ nodes
    according to an independent Poisson process. Let $\lambda$ be the
    overall arrival rate to the system and let $\frac{\lambda}{m}$ be
    the arrival rate at each transmitting node.
  \item Backlogged model: Every node always has a packet to transmit.
    Once a node transmits a packet successfully, a new packet is
    generated and awaits transmission. \label{asmp:backlogged_model}
  \end{enumerate}

\item Collision or perfect reception: If two or more nodes transmit a
  packet in a given slot, then there is a collision and the receiver
  obtains no information about the contents or the sources of
  transmitted packets. If only one node transmits a packet in a given
  slot, the packet is correctly received.

\item $\{0,1,e\}$ immediate feedback: At the end of each slot, every
  node obtains feedback from the receiver specifying whether 0 packet,
  1 packet or more than one packet ($e$ denotes error) were
  transmitted in that slot. \label{asmp:collision_model}

\item Retransmission of collisions: Each packet involved in a
  collision must be retransmitted in some later slot, with further
  such retransmissions until the packet is successfully received. A
  node with a packet that must be retransmitted is said to be
  backlogged.

\item Only one of the following is assumed:
  \begin{enumerate}
  \item No buffering: If one packet at a node is currently waiting for
    transmission or colliding with another packet during transmission,
    new arrivals at that node are discarded and never transmitted.
  \item Infinite set of nodes: The system has an infinite set of nodes
    and each new packet arrives at a new node.
    \label{asmp:infinite_nodes}
  \end{enumerate}

\end{enumerate}

For the analysis and performance evaluation of random access
algorithms, the metrics of interest are:
\begin{enumerate}

\item Delay: Index packets as $1,2,3,\ldots$ according to their
  arrival instants. Let $D_j$ denote the delay experienced by $j^{th}$
  packet. Then the average packet delay is defined as
  \begin{eqnarray}
    \mathcal D &=& \lim_{m \rightarrow \infty} 
      E\left[\frac{1}{m} \sum_{j=1}^m D_j\right].
  \end{eqnarray}

\item Throughput: The following are the two most common definitions of
  throughput:
  \begin{enumerate}
  \item Throughput is the supremum of input packet arrival rates
    $\lambda$ such that the packet delay remains bounded, i.e.,
    \begin{eqnarray} {\mathcal T}_1 &=& \sup_{{\mathcal D} < \infty }
      \lambda .
    \end{eqnarray}

  \item Let $n(t)$ denote the number of packets successfully transmitted
    in $[0,t]$. Define
    \begin{eqnarray*}
      T(\lambda) &=&
      \left\{
        \begin{array}{ll}
          \lim_{t \rightarrow \infty} E\left[\frac{n(t)}{t}\right] & 
          \mbox{if} \;\; {\mathcal D} < \infty, \\
          0 & \mbox{otherwise}.
        \end{array}
      \right.
    \end{eqnarray*}
    Throughput is then defined as
    \begin{eqnarray} 
      \mathcal T_2 &=& \sup_{\lambda} T(\lambda).
    \end{eqnarray}

\end{enumerate}

\item Stability: A random access algorithm is stable if the throughput
  $\mathcal T > 0$ and unstable if $\mathcal T = 0$.

\end{enumerate}

The research of random access algorithms began with the unslotted
ALOHA (pure ALOHA) algorithm proposed by Abramson
\cite{abramson__aloha_system}. Each node, upon receiving a packet,
transmits it immediately rather than waiting for a slot boundary.  If
a packet is involved in a collision, it is retransmitted after a
random delay. It can be shown that unslotted ALOHA achieves a maximum
throughput of $\frac{1}{2e} \approx 0.1839$
\cite{bertsekas_gallager__data_networks}.  An advantage of unslotted
ALOHA is that it can be used with variable-length packets.

Slotted ALOHA is a variation by Roberts \cite{roberts__aloha_packet}
of the original unslotted ALOHA protocol proposed by Abramson.  Each
node simply transmits a newly arriving packet in the first slot after
the packet arrival. When a collision occurs, every node sending a
colliding packet discovers the collision at the end of the slot and
becomes backlogged. Backlogged nodes wait for a random number of slots
before retransmitting.  The maximum throughput of slotted ALOHA can be
shown to be $\frac{1}{e} \approx 0.3678$
\cite{bertsekas_gallager__data_networks}.
Drift-analytic\footnote{Drift in state $n$ is defined as the expected
  change in backlog over one time-slot, starting in state $n$.}
methods reveal that slotted ALOHA is unstable.  To stabilize ALOHA,
some techniques estimate $n$ or $p_r$, so as to maintain the attempt
rate $G(n)$ at 1, resulting in a maximum stable throughput of
$\frac{1}{e}$ \cite{hajek_loon__decentralized_dynamic},
\cite{rivest__network_control}.  Unlike unslotted ALOHA, slotted ALOHA
cannot be easily used with variable-sized packets. In slotted ALOHA,
long packets must be broken up to fit into slots and short packets
must be padded out to fill up slots.

Keeping the random access spirit of the ALOHA protocol, researchers
attempted to design more efficient protocols. A highly successful
approach consists of improving the control of the channel by carrier
sensing, i.e., the Carrier Sense Multiple Access (CSMA) technique. In
\cite{bertsekas_gallager__data_networks}, the authors show that CSMA
outperforms ALOHA. Research has shown that CSMA based protocols can
achieve a throughput close to 0.9
\cite{takagi_kleinrock__throughput_analysis}.  The Ethernet protocol,
which is used to connect computers on a wired LAN, utilizes Carrier
Sense Multiple Access with Collision Detection (CSMA/CD).

In splitting algorithms, the set of colliding nodes splits into
subsets, one of which transmits in the next slot. For a given
colliding node, the choice of its subset depends on a pre-determined
rule such as, the outcome of tossing an unbiased coin, a function of
its arrival time or a function of its node identifier.  If the
collision is not resolved, a further splitting into subsets takes
place. The algorithm proceeds recursively until all collisions are
resolved.

\begin{figure}[thbp]
  \centering
  \includegraphics[width=6in]{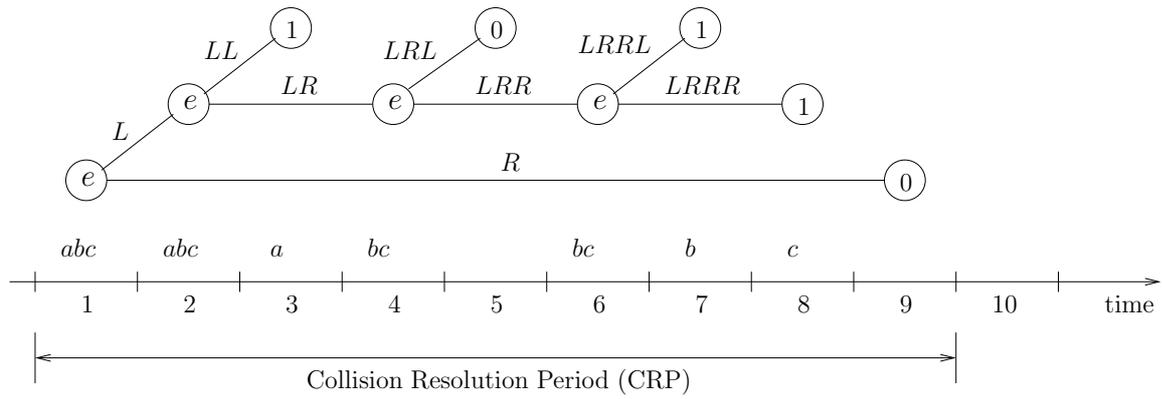}
  \caption{Basic Tree Algorithm for three nodes $a$, $b$ and $c$.}
  \label{fig:crp_bta}
\end{figure}

\begin{table}[tbhp]
\centering
\begin{tabular}{|l|l|l|l|} \hline
Slot & Transmitting set & Waiting sets & Feedback \\ \hline
1 & $U$    & $\phi$   & $e$ \\
2 & $L$    & $R$      & $e$ \\
3 & $LL$   & $LR,R$   & 1 \\
4 & $LR$   & $R$      & $e$ \\
5 & $LRL$  & $LRR,R$  & 0 \\
6 & $LRR$  & $R$      & $e$ \\
7 & $LRRL$ & $LRRR,R$ & 1 \\
8 & $LRRR$ & $R$      & 1 \\
9 & $R$    & $\phi$   & 0 \\ \hline
\end{tabular}
\caption{Transmitting and waiting sets for basic tree algorithm shown in Figure
  \ref{fig:crp_bta}.}
\label{tab:sets_bta}
\end{table}

\begin{figure}[thbp]
  \centering
  \includegraphics[width=6in]{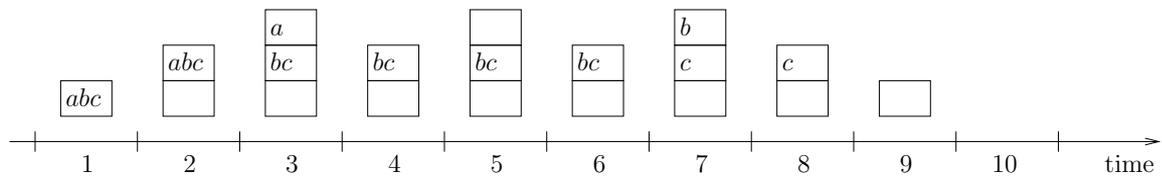}
  \caption{Stack representation of transmitting and waiting nodes for
    basic tree algorithm shown in Figure \ref{fig:crp_bta}.}
  \label{fig:stack_bta}
\end{figure}

In the Basic Tree Algorithm (BTA) \cite{capetanakis__tree_algorithms},
when a collision occurs, say in $k^{th}$ slot, all nodes not involved
in the collision go into a waiting mode, and all those involved in the
collision split into two subsets, according to the pre-determined
rule.  The first subset (``left'' subset) transmits in slot $k+1$, and
if that slot is idle or successful, the second subset (``right''
subset) transmits in slot $k+2$.  Alternatively, if another collision
occurs in slot $k+1$, the first of the two subsets splits again, and
the second subset waits for the resolution of that collision.  Figure
\ref{fig:crp_bta} exemplifies the operation of BTA for three nodes
$a$, $b$ and $c$.  Observe the binary tree structure of the sets of
transmitting and waiting nodes in the figure.  The transmitting and
waiting sets in terms of subtrees of this binary tree are shown in
Table \ref{tab:sets_bta}, where $U=\{a,b,c\}$ denotes the set of all
nodes that were involved in the initial collision.  The labeling of
the subtrees is recursive; for example, $LR$ denotes the right subtree
of the left subtree of the original binary tree.  The transmission
order corresponds to that of a stack, as shown in Figure
\ref{fig:stack_bta}.  In each slot, the stack is popped and all the
nodes that were at the top of the stack transmit their packets. In
case of a collision, the stack is pushed with nodes that join the
right subset and then pushed again with nodes that join the left
subset. In case of a success or idle, no push operations are performed
on the stack.  A Collision Resolution Period (CRP) is defined to be
completed when a success or idle occurs and there are no remaining
elements on the stack. In Figure \ref{fig:crp_bta}, the length of the
CRP is 9 slots.

During the operation of BTA, many new packets might arrive while a
collision is being resolved. To solve this problem, at the end of a
CRP, the set of nodes with new arrivals is immediately split into $j$
subsets, where $j$ is chosen so that the expected number of packets
per subset is slightly greater than 1.  The maximum throughput,
optimized over the choice of $j$ as a function of expected number of
waiting packets, is 0.43 packets per slot
\cite{capetanakis__tree_algorithms}. 

\begin{figure}[thbp]
  \centering
  \includegraphics[width=6in]{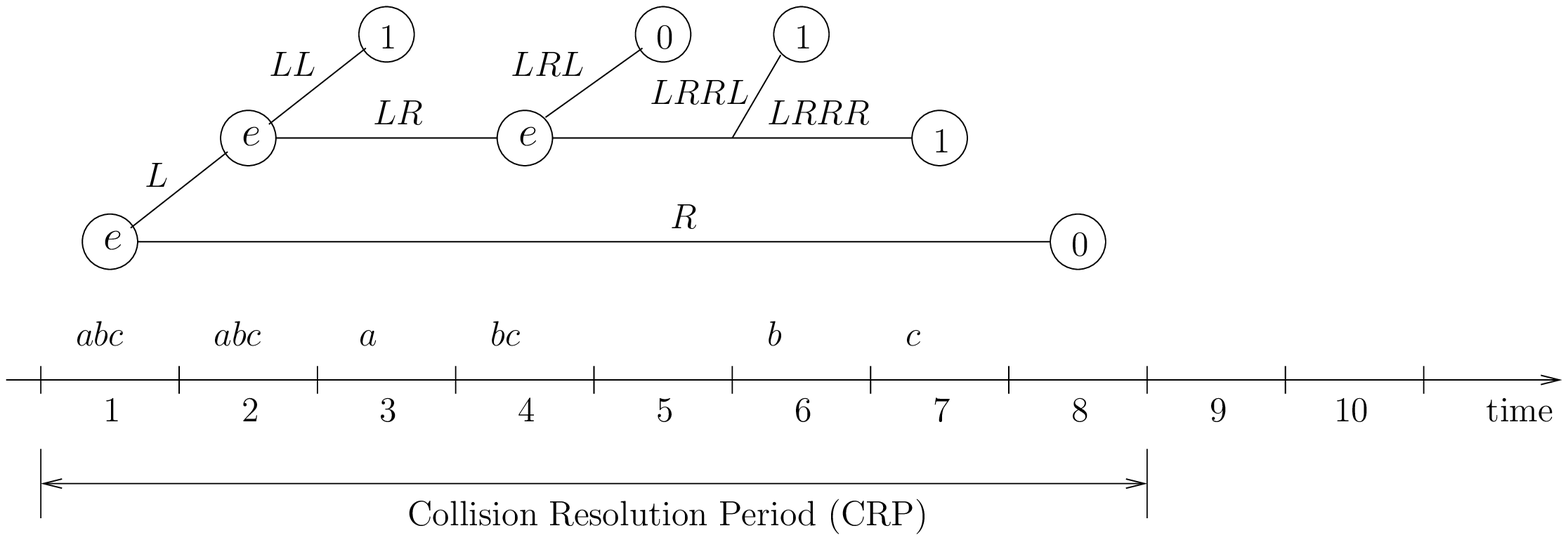}
  \caption{Modified Tree Algorithm for three nodes $a$, $b$ and $c$.}
  \label{fig:crp_mta}
\end{figure}

\begin{table}[tbhp]
\centering
\begin{tabular}{|l|l|l|l|} \hline
Slot & Transmitting set & Waiting sets & Feedback \\ \hline
1 & $U$    & $\phi$   & $e$ \\
2 & $L$    & $R$      & $e$ \\
3 & $LL$   & $LR,R$   & 1 \\
4 & $LR$   & $R$      & $e$ \\
5 & $LRL$  & $LRR,R$  & 0 \\
6 & $LRRL$ & $LRRR,R$ & 1 \\
7 & $LRRR$ & $R$      & 1 \\
8 & $R$    & $\phi$   & 0 \\ \hline
\end{tabular}
\caption{Transmitting and waiting sets for modified tree algorithm shown in 
  Figure \ref{fig:crp_mta}.}
\label{tab:sets_mta}
\end{table}

\begin{figure}[thbp]
  \centering
  \includegraphics[width=6in]{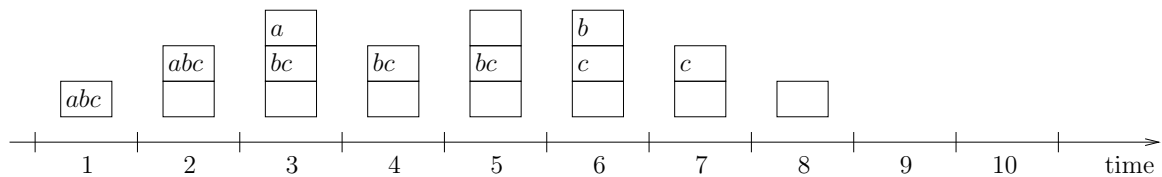}
  \caption{Stack representation of transmitting and waiting nodes for
    modified tree algorithm shown in Figure \ref{fig:crp_mta}.}
  \label{fig:stack_mta}
\end{figure}

There exist variants of BTA that yield higher throughput.  For
example, in Modified Tree Algorithm (MTA), if a collision in slot $k$
is followed by an idle in slot $k+1$, then nodes which collided in
slot $k$ refrain from transmitting in slot $k+2$. Instead, they
further split into two subsets, one of which transmits in slot $k+2$.
As an example, the operation of MTA for three nodes $a$, $b$ and $c$
is shown in Figure \ref{fig:crp_mta}.  Observe that the length of the
CRP is 8 slots.  For this example, the transmitting and waiting sets
of subtrees are shown in Table \ref{tab:sets_mta}, while the
corresponding stack representation is shown in Figure
\ref{fig:stack_mta}.  If an idle occurs in the current slot and a
collision occurred in the previous slot (see Slot 5 in Figure
\ref{fig:stack_mta}), then the stack is popped a second time but the
nodes at the top of the stack are not transmitted. Instead, these
nodes split into right and left subsets and these subsets are pushed
on the stack.  This leads to fewer collisions and higher throughput
compared to BTA.  The maximum stable throughput of MTA is 0.46 packets
per slot \cite{massey__conflict_resolution}.

In First Come First Serve (FCFS) splitting algorithm
\cite{bertsekas_gallager__data_networks}, nodes involved in a
collision split into two subsets based on the arrival times of
collided packets. Using this approach, each subset consists of all
packets that arrived in some given interval, and when a collision
occurs, that interval will be split into two smaller intervals.  By
always transmitting packets that arrived in the earlier interval
first, the algorithm transmits successful packets in the order of
their arrival. The FCFS algorithm is stable for $\lambda < 0.4871$
\cite{bertsekas_gallager__data_networks}.  Conflict resolution
protocols based on tree algorithms have provable stability properties
\cite{molle_polyzos__conflict_resolution}.

We should point out that the random access algorithm proposed in
Chapter \ref{ch:powercontrolled} has a ``tree structure'' analogous to
that of MTA. The detailed explanations of BTA and MTA provide a basic
background to understand the dynamics of the proposed algorithm.

So far, we have summarized the methodology of traditional random
access algorithms.  In subsequent sections, we will focus on random
access algorithms that are better suited for wireless networks such as
WLANs and Wireless Metropolitan Area Networks (WMANs).

\section{Signal Processing in Random Access}
\label{sec:sigpro_random_access}

The collision model (Section \ref{sec:traditional_random_access},
Assumption \ref{asmp:collision_model}) is simple in that the onus of
scheduling packets is left entirely to the MAC layer. On the contrary,
physical layer techniques like multipacket reception, capture and
network-assisted diversity are able to correctly decode packets from
collisions by means such as coding and signal processing.  These
techniques are potential steps towards alleviating the burden of
decoding packets from the MAC layer to the physical layer
\cite{tong_naware__signal_processing}.  In this section, we review
representative research work which exploits signal processing and
diversity techniques to correctly decode the received packets in
random access wireless networks.

With the advent of multiaccess techniques such as CDMA and Multiuser
Detection \cite{verdu__multiuser_detection}, the first fundamental
change in the collision model has been propounded in
\cite{ghez_verdu__stability_properties}. The authors offer the
generalization that, in the presence of simultaneous transmissions,
the reception can be described by conditional probabilities instead of
deterministic failure. They propose the MultiPacket Reception (MPR)
model defined by the matrix
\begin{eqnarray}
  {\mathbf C} &=&
  \left[ 
    \begin{array}{cccc}
    C_{10} & C_{11} & & \\
    C_{20} & C_{21} & C_{22} & \\
    \vdots & & & \ddots
    \end{array}
  \right],
\end{eqnarray}
where $C_{ij}$ is the conditional probability that, given $i$ users
transmit, $j$ out of $i$ transmissions are successful. Given $k$ users
transmit at the same time, the average number of successfully received
packets is given by
\begin{eqnarray}
  C_k &=& \sum_{j=0}^k jC_{kj}.
\end{eqnarray}
They show that ALOHA under MPR achieves stable throughput $\lim_{k
  \rightarrow \infty}C_k$ assuming that the limit exists.  The
stability and delay of finite-user slotted ALOHA with multipacket
reception has been analyzed in \cite{naware_mergen__stability_delay}.

In \cite{hajek_krishna__capture_probability}, the authors analyze the
probability of capture in a multipoint to point wireless network.
Analogous to the physical interference model, the capture model
assumes that if a user's SINR exceeds a threshold $\gamma$, then that
user's packet will be successfully received.  They consider a
realistic multiplicative propagation model in which the received power
is obtained by multiplying the transmitted power by independent random
variables representing fading, shadowing and path loss effects. To
model the near-far effect, they assume that the distance $r$ of a
mobile station from the base station is a random variable with
distribution function $F_R(r)$. They show that, under broad
conditions, the roll-off parameter $\delta$ of the distribution of
power received from a mobile station is determined by the path loss
exponent and $F_R(r)$. Additionally, $\delta$ is insensitive to other
effects such as Rayleigh or Rician fading and log-normal shadowing.
Finally, they show that in the limit of a large number of
transmitters, the probability of capture is determined by the power
capture threshold $\gamma$ and $\delta$.  Though the analysis provided
in \cite{hajek_krishna__capture_probability} is mathematically robust,
the authors do not describe any multiple access algorithm which
achieves high throughput in wireless networks under the capture model,
i.e., their result is more existential than constructive.

In \cite{tsatsanis_zhang_banerjee__network_assisted}, the authors
propose Network-Assisted Diversity Multiple Access (NDMA), a technique
for resolving collisions in wireless networks. They consider a
wireless slotted random access network with Rayleigh fading and
Additive White Gaussian Noise (AWGN). In NDMA, if $k$ users collide in
a given slot, they repeat their transmissions $k-1$ times so that $k$
copies of the collided packets are received.  Using signal separation
principles, the receiver resolves a $k \times k$ source mixing problem
to extract the signals of individual users, without incurring any
penalty in throughput. The protocol has been extended to blind user
detection \cite{zhang_sidiropoulos__collision_resolution} and has
provable stability \cite{dimic_sidiropoulos__wireless_networks}.  A
good review of NDMA protocols is given in
\cite{dimic_sidiropoulos__medium_access}.

An alternative to employing signal processing techniques in random
access wireless networks is to appropriately model the wireless
channel and modify the well-researched ALOHA protocol. We review such
research work in the next section.

\section{Channel-Aware ALOHA Algorithms}
\label{sec:aloha_wireless}

In this section, we review representative research work whose central
theme is to adapt the retransmission probabilities of users in random
access wireless networks.  In other words, we review research work
which develops channel-aware ALOHA algorithms for wireless networks.

In \cite{qin_berry__exploiting_multiuser}, the authors develop a
channel-aware ALOHA protocol for wireless networks.  They assume a
slotted system, block fading, $\{0,1,e\}$ feedback and a backlogged
model (Section \ref{sec:traditional_random_access}, Assumption
\ref{asmp:backlogged_model}).  They develop a distributed random
access protocol in which each node only has knowledge of its own
channel gain and nodes have long-term power constraints. A node
transmits only if its channel gain exceeds $H_0$.  For a system with
$n$ nodes, the authors show that the optimum transmission probability
is $\frac{\alpha(n)}{n}$, where $\alpha(n) \in (0,1]$ and $\alpha(n)
\rightarrow 1$ as $n \rightarrow \infty$.  Asymptotically, the ratio
of the throughput of channel-aware ALOHA to the throughput of a
centralized scheduler (which has knowledge of channel gains of all
nodes) is shown to be $\frac{1}{e}$.

Opportunistic ALOHA algorithms for wireless networks have been studied
in \cite{adireddy_tong__exploiting_decentralized}.  The authors
consider a general reception model which encompasses $\{0,1,e\}$
feedback, capture as well as multipacket reception.  Under the
assumption that the Channel State Information (CSI) is known to each
user, they propose a variant of slotted ALOHA, where the transmission
probability is allowed to be a function of the CSI.  The maximum
throughput for the finite-user infinite-buffer model is derived.
Finally, the theory is applied to CDMA networks with Linear Minimum
Mean Square Error (LMMSE) receiver and matched filters.

The performance of slotted ALOHA in a wireless network with multiple
destinations under the physical interference model is evaluated in
\cite{nguyen_wieselthier__capture_wireless}. A packet is successful
only if it is captured at the receiver of its intended destination.
The authors assume Poisson packet generation, $\{0,1,e\}$ feedback and
circularly symmetric Gaussian distribution of users around each
destination.  They use a modified version of Rivest's pseudo-Bayesian
estimator \cite{bertsekas_gallager__data_networks} to estimate the
backlog. Their simulation results demonstrate the effect of arrival
rate, capture threshold, variance of user distribution and number of
destinations on the throughput and energy efficiency per destination.

In \cite{gau__performance_analysis}, the author analyzes the
throughput of slotted ALOHA in a multipoint to point wireless ad hoc
network under the physical interference model.  The cluster head
employs reverse link power control, similar to IS-95 CDMA systems
\cite{lee_miller__cdma_systems}, to ensure that equal power is
received from all nodes who attempt transmission in a time slot. The
wireless channel is modeled as a multipacket reception channel.
Assuming that one new packet arrives at each node in every time slot,
the state of the system is characterized by a discrete time Markov
chain with a steady state distribution.  Finally, the author describes
a technique to compute the network throughput.

In \cite{baccelli_blaszczyszyn__aloha_protocol}, the authors introduce
spatial reuse slotted ALOHA, a random access protocol for random
homogeneous mobile wireless networks. The occurrence of a collision is
determined by the SINR at a receiver, i.e., the authors employ the
physical interference model.  They assume that nodes are randomly
placed in a two-dimensional plane according to a Poisson point process
and each node chooses a random destination at some finite distance.
The powers at which stations can transmit are assumed to be
independent and identically distributed (i.i.d.) and the wireless
channel is characterized by its propagation path loss. Nodes move
according to the random waypoint mobility model
\cite{navidi_camp__stationary_distributions}. The authors characterize
the interference process using tools from stochastic geometry.
Subsequently, they determine the probability of channel access that
maximizes the expected projected distance traversed per hop towards a
destination, termed as ``spatial density of progress''.  Under the
assumption that there is some non-degenerate node mobility, the
authors show that the spatial density of progress is proportional to
the square-root of the density of the nodes. Though the authors
present a distributed ALOHA protocol and address certain
implementation issues, their model does not represent real-world
scenarios. Practical deployments of wireless nodes are better modeled
by a uniform distribution in a finite plane rather than a Poisson
point process in an infinite plane. Also, most of their results do not
hold for static wireless networks (say, wireless mesh networks) since
ergodicity assumptions no longer hold. Finally, their proposed routing
protocol requires every node to have knowledge of locations and MAC
states (receiver or transmitter) of all other nodes, which requires a
lot of message passing between nodes (especially with mobile nodes)
and is thus not scalable.

Instead of adapting the transmission probabilities of users in random
access wireless networks, one can also adapt the transmission times of
users based on the channel state and feedback from the receiver.  Such
techniques, which can be broadly termed as splitting algorithms or
tree-like algorithms for wireless networks, are reviewed in the next
section.

\section{Splitting Algorithms}
\label{sec:splitting_wireless}

In this section, we review representative research work on random
access algorithms whose main idea is to adapt the set of contending
users based on feedback from the channel or the common receiver.  In
such work, the authors develop and analyze splitting (or tree or
stack) algorithms for various models of the wireless channel and
evaluate the performance of their algorithms via simulations.

In \cite{qin_berry__opportunistic_splitting}, the authors propose an
opportunistic splitting algorithm for a multipoint to point wireless
network. They assume a slotted system, block fading channel and
$\{0,1,e\}$ feedback. Assuming that each user only knows its own
channel gain and the number of backlogged users, the authors propose a
distributed splitting algorithm to determine the user with the best
channel gain over a sequence of mini-slots. The algorithm determines a
lower threshold $H_l$ and a higher threshold $H_h$ for each mini-slot,
such that only users whose channel gains lie between between $H_l$ and
$H_h$ are allowed to transmit their packets. Based on results from
``partitioning a sample with binary type questions''
\cite{arrow_pesotchinsky__partitioning_sample}, they show that the
average number of mini-slots required to determine the user with the
best channel is 2.5, independent of the number of users and the fading
distribution. However, their algorithm is impractical because it
assumes that every user can accurately estimate the number of
backlogged users.

In \cite{yu_giannakis__sicta_0693}, the authors consider a random
access network with infinite users, Poisson arrivals and $\{0,k,e\}$
immediate feedback, where $k$ is any positive integer.  In contrast to
standard tree algorithms (BTA, MTA, FCFS) that discard collided
packets (Section \ref{sec:traditional_random_access}, Assumption
\ref{asmp:collision_model}), they propose an algorithm that stores
collided packets. The receiver extracts information from the collided
packets by relying on successive interference cancellation techniques
(\cite{verdu__multiuser_detection}, Chapter 7) and the tree structure
of a collision resolution algorithm. Though their algorithm achieves a
stable throughput of 0.693, it requires infinite storage and increased
input voltage range at the receiver, which are not feasible in
practical systems.

In \cite{tsybakov__packet_multiple}, the author considers a multipoint
to point wireless channel with and without capture and MPR.  The
channel provides Empty(E)/Non-Empty(NE) feedback to all active users
and `success' feedback to successful users only. The users do not need
to know the starting times and ending times of collision resolution
periods.  For such a channel with E/NE binary feedback, the author
proposes and analyzes a stack multiple access algorithm that is
limited sensing and does not require any frame synchronization. The
author considers two models for capture, namely Rayleigh fading with
incoherent and coherent combining of joint interference power.  For
MPR, the author assumes a maximum of two successes during a collision.
The maximum throughput of the algorithm is numerically evaluated be to
0.6548 when capture and MPR are present, and 0.2891 when both effects
are absent.  Though a novel splitting algorithm is proposed in
\cite{tsybakov__packet_multiple}, the author does not take into
account throughput gains possible by varying transmission powers of
users.

So far, we have reviewed research papers that either utilize signal
processing techniques or adapt transmission probabilities or
transmission times to increase the throughput in random access
wireless networks. The throughput can be further increased by allowing
users to use variable transmission powers.  We review research papers
which employ this idea in the next section.

\section{Towards Power Controlled Random Access}
\label{sec:power_controlled}

In this section, we review representative research papers which focus
on power control techniques in random access wireless networks. We
then motivate the use of variable transmission power to increase the
throughput in random access wireless networks.

In \cite{dua__power_controlled}, the author considers a time-slotted
CDMA-based wireless network wherein a finite number of nodes
communicate with a common receiver. The author formulates the problem
of determining the set of nodes that can transmit in each slot along
with their corresponding transmission powers, subject to constraints
on maximum transmission power and the SINRs of all transmissions
exceeding the communication threshold. Due to its NP-hard nature, the
problem is relaxed to a case wherein a node transmits with a certain
probability in each slot. Equivalently, the problem of joint power
control and link scheduling is transformed to a problem of power
controlled random access, wherein the objective is to determine the
probability of transmission $\Delta_i$ and transmission power $P_i$
for each node $i$, subject to constraints on maximum transmission
power and the ``expected SINR'' exceeding the communication threshold.
The author seeks to minimize a weighted sum of the maximum
transmission power and maximum reciprocal probability, i.e., minimize
$(\max_i P_i + \lambda \max_i \frac{1}{\Delta_i})$. This convex
optimization problem is solved using techniques from geometric
programming \cite{duffin_peterson_zener__geometric_programming}.
Finally, the author derives the probability of outage\footnote{An
  outage occurs on a link if the received (actual) SINR on the link is
  less than the communication threshold.} and delay distribution of
buffered packets and demonstrates the efficacy of the schemes via
simulations.

In \cite{sagduyu_ephremides__power_control}, the authors investigate
transmission power control and rate adaptation in random access
wireless networks using game theoretic techniques. They consider
multiple transmitters sharing a time-slotted channel to communicate
equal-length packets with a common receiver. A user's packet is
successfully received if the SINR at the receiver is no less than the
communication threshold, i.e., the authors employ the physical
interference model. The random access problem is formulated as a game
wherein each user selects its strategy (transmit or wait) at each
stage of the game in a non-cooperative (independent) or cooperative
manner. The authors evaluate equilibrium strategies for
non-cooperative and cooperative symmetric random access games.
Finally, the authors describe distributed power control and rate
adaptation games for non-cooperative users for a collision channel
with power-based capture. Their numerical results demonstrate improved
expected user utilities when power control and rate adaptation are
incorporated, at the expense of increased computational complexity.
Though the authors propose a distributed random access algorithm based
on game theoretic techniques, their algorithm is impractical because
it assumes that every user knows $n$, the number of backlogged users,
in each slot. However, in practice, $n$ can only be estimated using
techniques such as Rivest's pseudo-Bayesian algorithm
\cite{rivest__network_control}.

Though researchers have addressed the problem of random access in
wireless networks by considering various channel models, different
types of feedback and realistic criteria for successful packet
reception, only few of them exploit the idea that throughput gains are
achievable in a random access wireless network by varying transmission
powers of users. In general, varying the transmission powers of users
leads to higher long-term average power. However, there exist wireless
networks whose users do not have stringent energy requirements. For
such scenarios, it would be useful to investigate the throughput gains
achievable in the network by varying the transmission powers of users.

We envisage developing a power controlled random access algorithm for
wireless networks under the physical interference model.  We seek an
algorithm that yields higher throughput than traditional random access
algorithms.  In cognizance of these requirements, we propose a power
controlled splitting algorithm for wireless networks in Chapter
\ref{ch:powercontrolled}. The algorithm is so designed that successful
packets are transmitted in the order of their arrivals, i.e., in an
FCFS manner.

In the system model considered in Chapter \ref{ch:powercontrolled}, if
multiple transmissions occur, the receiver can decode a certain user's
packet correctly only if the received SINR exceeds a threshold, i.e.,
we consider a channel with power-based capture. The notion of capture
has been addressed previously, though in different contexts
\cite{nguyen_wieselthier__capture_wireless},
\cite{tsybakov__packet_multiple}, \cite{zorzi__mobile_radio}.
However, in Chapter \ref{ch:powercontrolled}, we motivate the idea
that a user can transmit at variable power levels to increase the
chances of capture. Moreover, unlike
\cite{nguyen_wieselthier__capture_wireless},
\cite{tsybakov__packet_multiple}, \cite{zorzi__mobile_radio}, we
assume $\{0,1,c,e\}$ feedback, where $0$, $1$ and $e$ denote idle,
success and error respectively (Section
\ref{sec:traditional_random_access}, Assumption
\ref{asmp:collision_model}), and $c$ denotes capture in the presence
of multiple transmissions. Note that the system model considered in
Chapter 7 is different from those considered in existing works on
splitting algorithms for wireless networks. For example, in
\cite{tsybakov__packet_multiple}, the author proposes a novel
splitting algorithm, but does not take into account throughput gains
possible by varying the transmission power.  Though the authors of
\cite{qin_berry__opportunistic_splitting} propose a splitting
algorithm to determine the user with the best channel gain, their
algorithm is impractical because it assumes that every user can
accurately estimate the number of backlogged users.

To the best of our knowledge, there is no existing work on variable
power splitting algorithms for a wireless network under the physical
interference model.  The specification of the proposed algorithm along
with its performance analysis and evaluation constitute the subject
matter of the next chapter.

\clearpage{\pagestyle{empty}\cleardoublepage}

\chapter{Power Controlled FCFS Splitting Algorithm for Wireless Networks}
\label{ch:powercontrolled}

In this chapter, we propound a random access algorithm that
incorporates variable transmission powers in a multipoint to point
wireless network.  Specifically, we investigate random access in
wireless networks under the physical interference model wherein the
receiver is capable of power-based capture, i.e., a packet can be
decoded correctly in the presence of multiple transmissions if the
received SINR exceeds the communication threshold. We propose an
interval splitting algorithm that varies the transmission powers of
users based on channel feedback. We derive the maximum stable
throughput of the proposed algorithm and demonstrate that it achieves
better performance than the FCFS splitting algorithm
\cite{bertsekas_gallager__data_networks} with uniform transmission
power.

The rest of the chapter is organized as follows. We describe our
system model in Section \ref{sec:system_pcfcfs} and motivate variable
control of transmission powers of contending users in Section
\ref{sec:motivation_pcfcfs}. We describe the proposed random access
algorithm and provide two illustrative examples in Section
\ref{sec:algorithm_pcfcfs}. We model the algorithm dynamics by a
Markov chain and derive its maximum stable throughput in Section
\ref{sec:throughput_pcfcfs}. The performance of the proposed algorithm
is evaluated in Section \ref{sec:results_pcfcfs}. We conclude in
Section \ref{sec:conclusions_pcfcfs}.

\section{System Model}
\label{sec:system_pcfcfs}

Consider a multipoint to point wireless network. We assume the
following:

\begin{enumerate}

\item Slotted system: Users (nodes) transmit fixed-length packets to a
  common receiver over a time-slotted channel. All users are
  synchronized such that the reception of a packet starts at an
  integer time and ends before the next integer time.

\item Poisson arrivals: The packet arrival process is Poisson
  distributed with overall rate $\lambda$, and each packet arrives to
  a new user that has never been assigned a packet before.  After a
  user successfully transmits its packet, that user ceases to exist
  and does not contend for channel access in future time slots.

\item Channel model: The wireless channel is modeled by the path loss
  propagation model.  The received signal power at a distance $D$ from
  the transmitter is given by $\frac{P}{D^\beta}$, where $P$ is the
  transmission power and $\beta$ is the path loss factor. We do not
  consider fading and shadowing effects.

\item Power-based capture: According to the physical interference
  model \cite{gupta_kumar__capacity_wireless}, a packet transmission
  from transmitter $t_{i,j}$ to receiver $r$ in $i^{th}$ time slot is
  successful if and only if the SINR at receiver $r$ is greater than
  or equal to the communication threshold $\gamma_c$\footnote{In
    literature, $\gamma_c$ is also referred to as capture ratio
    \cite{tsybakov__packet_multiple}, capture threshold
    \cite{nguyen_wieselthier__capture_wireless} and power ratio
    threshold \cite{hajek_krishna__capture_probability}.}, i.e.,
  \begin{eqnarray}
    \frac{\frac{P_{i,j}}{D^\beta(t_{i,j},r)}}{N_0+\sum_{\stackrel{k=1}{k\neq j}}^{M_i}\frac{P_{i,k}}{D^\beta(t_{i,k},r)}} \geqslant \gamma_c,
  \label{eq:sinr_threshold_pcfcfs}
  \end{eqnarray}
  where
  \begin{eqnarray*}
    M_i &=& \mbox{number of concurrent transmitters in $i^{th}$ time slot}, \\
    t_{i,j} &=& \mbox{$j^{th}$ transmitter in $i^{th}$ time slot}
    \;\;(j=1,2,\ldots,M_i),\\
    D(t_{i,j},r) &=& \mbox{Euclidean distance between $t_{i,j}$ and $r$}, \\
    P_{i,j} &=& \mbox{transmission power of $t_{i,j}$}, \\
    N_0 &=& \mbox{thermal noise power spectral density}.
  \end{eqnarray*}

\item $\{0,1,c,e\}$ immediate feedback: By the end of each slot, users
  are informed of the feedback from the receiver immediately and
  without any error. The feedback is one of:
  \begin{enumerate}
  \item idle $(0)$: when no packet transmission occurs,
  \item perfect reception $(1)$: when one packet transmission occurs
    and is received successfully,
  \item capture $(c)$: when multiple packet transmissions occur and
    only one packet is received successfully, or
  \item collision $(e)$: when multiple packet transmissions occur and
    no packet reception is successful.
  \end{enumerate}
  The receiver can distinguish between $1$ and $c$ by using energy
  detectors \cite{yu_giannakis__sicta_0693},
  \cite{georgiadis_kazakos__collision_resolution}.  Thus, by the end
  of every slot, only two bits are required to provide feedback from
  the receiver to all users. Note that two bits are required to
  provide feedback even for the classical $\{0,1,e\}$ feedback model.
  Thus, our $\{0,1,c,e\}$ immediate feedback assumption does not
  increase the number of bits required for feedback..

\item Gated Channel Access Algorithm (CAA): New packets are
  transmitted in the first available slot after previous conflicts are
  resolved. The time interval from the slot where an initial collision
  occurs up to and including the slot in which all users recognize
  that all packets involved in the collision have been successfully
  received, is called a Collision Resolution Period (CRP).  Thus, new
  arrivals are inhibited from transmission during the CRP.

\item Equal distances: We assume that each user is at the same
  distance $D$ from the common receiver.

\end{enumerate}

\section{Motivation and Problem Formulation}
\label{sec:motivation_pcfcfs}

The maximum stable throughput of the well-known FCFS splitting
algorithm is 0.4871 \cite{bertsekas_gallager__data_networks}, which is
the highest throughput amongst a wide class of random access
algorithms for wired networks. However, in a wireless network,
transmission power of a node provides an extra degree of freedom, and
higher throughputs are achievable.

Consider a scenario wherein all contending nodes transmit with equal
power $P$ in a given time slot. When only one node transmits, its
packet is successfully received if the SINR threshold condition
(\ref{eq:sinr_threshold_pcfcfs}) is satisfied, i.e.,
\begin{eqnarray}
P \geqslant \gamma_c N_0 D^\beta.
\end{eqnarray}
When $M$ nodes transmit concurrently with equal power $P$, where $M
\geqslant 2$, the SINR corresponding to $i^{th}$ transmission is given
by
\begin{eqnarray}
\mbox{SINR}_i &=& \frac{\frac{P}{D^\beta}}{N_0+(M-1)\frac{P}{D^\beta}},
\label{eq:sinr_equal_power}
\end{eqnarray}
a quantity which is always less than 1. Since $\gamma_c>1$ for all
practical narrowband communication receivers
\cite{hajek_krishna__capture_probability}, $\mbox{SINR}_i < \gamma_c$
$\forall$ $i$ and all $M$ transmissions are unsuccessful\footnote{For
  a spread spectrum CDMA system with processing gain $L$,
  (\ref{eq:sinr_equal_power}) gets modified to $\mbox{SINR}_i =
  \frac{\frac{P}{D^\beta}}{N_0 + \frac{I}{L}(M-1)\frac{P}{D^\beta}}$
  \cite{sagduyu_ephremides__power_control}. For such a wideband
  system, $\gamma_c < 1$, and more than one packet can be decoded
  correctly in the presence of multiple transmissions. However, in
  this thesis, we consider narrowband systems only.}.  Thus, when
multiple nodes transmit with equal power, a collision occurs
irrespective of the transmission power $P$.

However, the above situation can be circumvented by varying
transmission powers of users in some special cases. With relatively
small attempt rates, when a collision occurs, it is most likely
between only two packets \cite{bertsekas_gallager__data_networks}. In
this case, if the receiver is capable of power-based capture, a
collision between two nodes can be avoided by using different
transmission powers.  Specifically, one of the nodes, say $N_1$,
transmits with minimum power $P_1$ such that, if it were the only node
transmitting in that time slot, then its packet transmission will be
successful. From (\ref{eq:sinr_threshold_pcfcfs}), the required
nominal power is
\begin{eqnarray}
P_1 &=& \gamma_c N_0 D^\beta.
\label{eq:nominal_power}
\end{eqnarray}
The other node, say $N_2$, transmits with minimum power $P_2$ such
that if there is exactly one other node transmitting at nominal power
$P_1$, then the packet transmitted by $N_2$ will be successful. From
(\ref{eq:sinr_threshold_pcfcfs}) and (\ref{eq:nominal_power}), we
obtain
\begin{eqnarray}
\frac{\frac{P_2}{D^\beta}}{N_0+\frac{P_1}{D^\beta}} 
 &=& \gamma_c, \nonumber \\
P_2 &=& \gamma_c (N_0 D^\beta + P_1), \nonumber \\
\therefore P_2
 &=& \gamma_c (1+\gamma_c) N_0 D^\beta.
\label{eq:higher_power}
\end{eqnarray}
Note that $\frac{P_2}{P_1} = 1+\gamma_c$.  We do not consider more
than two power levels for the following reasons:
\begin{enumerate}
\item 
  it complicates the power-control algorithm, and
\item 
  most mobile wireless devices have constraints on peak transmission
  power.
\end{enumerate}
Note that the above power control technique converts some collisions
into ``captures''. Thus, it has the potential of increasing the
throughput of random access algorithms employing uniform transmission
power.

We seek to design a distributed algorithm incorporating this power
control technique, while still ensuring that the algorithm transmits
successful packets in the order of their arrival, i.e., in an FCFS
manner\footnote{Since successful packets are transmitted in an FCFS
  manner, the delay experienced by a packet will not be significantly
  higher than the average packet delay. Thus, from a QoS perspective,
  FCFS transmission of packets not only guarantees average delay
  bounds, but also ensures fairness of user packets.}.

\section{PCFCFS Interval Splitting Algorithm}
\label{sec:algorithm_pcfcfs}

In this section, we present an algorithmic description of the proposed
Power Controlled First Come First Serve (PCFCFS) splitting algorithm.
We also explain the behavior of the proposed algorithm by providing
two illustrative examples.

\subsection{Description}
\label{subsec:description_pcfcfs}

We first describe the notation.  Slot $k$ is defined to be the time
interval $[k,k+1)$. At each integer time $k$ ($k \geqslant 1$), the
algorithm specifies the packets to be transmitted in slot $k$ to be
the set of all packets that arrived in an earlier interval
$[T(k),T(k)+\phi(k))$, which is defined as the {\em allocation
  interval} for slot $k$.  The maximum size of the allocation interval
is denoted by $\phi_0$, a parameter which will be optimized for
maximum throughput in Section \ref{sec:throughput_pcfcfs}.  Packets
are indexed as $1,2,\ldots$ in the order of their arrival.  Since the
arrival times are Poisson distributed with rate $\lambda$, the
inter-arrival times are exponentially distributed with mean
$\frac{1}{\lambda}$.  Let $a_i$ denote the arrival time of $i^{th}$
packet. Using the memoryless property of the exponential distribution
(and without loss of generality), we assume that $a_1=0$.  The
transmission power of $i^{th}$ packet in slot $k$ is denoted by
$P_i(k)$, where $P_i(k) \in \{0,P_1,P_2\}$. Note that, if $P_i(k)=0$,
then $i^{th}$ packet is not transmitted in slot $k$.

Algorithm \ref{algo:pcfcfs} describes the proposed Power Controlled
First Come First Serve (PCFCFS) splitting algorithm, which is the set
of rules by which the users compute allocation interval parameters
$\{T(k+1), \phi(k+1), \sigma(k+1)\}$ and transmission power $P_i(k+1)$
for slot $k+1$ in terms of the feedback and allocation interval
parameters for slot $k$. In our algorithm, every allocation interval
is tagged as a ``left'' $(\mathcal L)$ or ``right'' $(\mathcal R)$
interval.  $\sigma(k)$ denotes the tag ($\mathcal L$ or $\mathcal R$)
of allocation interval $[T(k),T(k)+\phi(k))$ in slot $k$.  Moreover,
whenever an allocation interval is split, it is always split into two
equal-sized subintervals, and these subintervals $(\mathcal L,
\mathcal R)$ are said to {\em correspond} to each other.

\renewcommand{\baselinestretch}{1.2}\Large\normalsize

\begin{algorithm}
\caption{PCFCFS splitting algorithm}
\label{algo:pcfcfs}
\begin{algorithmic}[1]
\STATE {\bf input:} $\phi_0$, $P_1$, $P_2$, 
  arrivals $a_1,a_2,a_3,\ldots$ in $[0,\tau)$ \COMMENT{Phase 1 begins}
\STATE $T(1) \leftarrow 0$
\STATE $\phi(1) \leftarrow \min(\phi_0,1)$
\STATE $\sigma(1) = \mathcal R$
\STATE $\mbox{feedback} = 0$ \COMMENT{Phase 1 ends}
\FOR[Phase 2 begins]{$k \leftarrow 1 \mbox{ to } \tau$}
 \IF[Phase 2a begins]{$\mbox{feedback} \neq c$}
  \FORALL{$i$ such that $\scriptstyle T(k) \leqslant a_i < T(k)+\frac{\phi(k)}{2}$}
   \STATE $P_i(k)=P_2$
  \ENDFOR
  \FORALL{$i$ such that $\scriptstyle T(k)+\frac{\phi(k)}{2} \leqslant a_i < T(k)+\phi(k)$}
   \STATE $P_i(k)=P_1$
  \ENDFOR
 \ENDIF \COMMENT{Phase 2a ends}
 \STATE transmit packets whose arrivals times lie in $[T(k),T(k)+\phi(k))$ 
   and obtain channel feedback \COMMENT{Phase 2b begins}
 \IF{$\mbox{feedback} = e$}
  \STATE $T(k+1)\leftarrow T(k)$
  \STATE $\phi(k+1)\leftarrow\frac{\phi(k)}{2}$
  \STATE $\sigma(k+1) \leftarrow \mathcal L$
 \ELSIF{$\mbox{feedback} = c$}
  \STATE $T(k+1) \leftarrow T(k) + \frac{\phi(k)}{2}$
  \STATE $\phi(k+1) \leftarrow \frac{\phi(k)}{2}$
  \STATE $\sigma(k+1) \leftarrow \mathcal R$
 \ELSIF{$\mbox{feedback} = 1$ and $\sigma(k) = \mathcal L$}
  \STATE $T(k+1) \leftarrow T(k) + \phi(k)$
  \STATE $\phi(k+1) \leftarrow \phi(k)$
  \STATE $\sigma(k+1) \leftarrow \mathcal R$
 \ELSIF{$\mbox{feedback} = 0$ and $\sigma(k) = \mathcal L$}
  \STATE $T(k+1) \leftarrow T(k) + \phi(k)$
  \STATE $\phi(k+1) \leftarrow \frac{\phi(k)}{2}$
  \STATE $\sigma(k+1) \leftarrow \mathcal L$
 \ELSE
  \STATE $T(k+1) \leftarrow T(k) + \phi(k)$
  \STATE $\phi(k+1) = \min(\phi_0,k-T(k))$
  \STATE $\sigma(k+1) \leftarrow \mathcal R$
 \ENDIF \COMMENT{Phase 2b ends}
\ENDFOR \COMMENT{Phase 2 ends}
\end{algorithmic}
\end{algorithm}

\renewcommand{\baselinestretch}{1.5}\Large\normalsize

In Phase 1 of the algorithm, we initialize various quantities.  $\tau$
denotes the number of slots for which the algorithm operates; ideally
$\tau \rightarrow \infty$.  By convention, the initial allocation
interval is $[0,\min(\phi_0,1))$, which is a right interval ($\mathcal
R$). The initial channel feedback is assumed to be idle (0).

In Phase 2 of the algorithm, we determine power levels, obtain channel
feedback and compute allocation interval parameters for each
successive slot $k$.  In Phase 2a, all users whose arrival times lie
in the left half of the current allocation interval transmit with
higher power $P_2$, while all users whose arrival times lie in the
right half of the current allocation interval transmit with nominal
power $P_1$.  However, if a capture occurred in the previous slot
$k-1$, all users in the current allocation interval transmit with
nominal power $P_1$.  Therefore, our algorithm always transmits
successful packets in an FCFS manner.  In Phase 2b, the allocation
interval parameters are modulated based on the channel feedback. More
specifically, if a collision occurs, then the left half of the current
allocation interval becomes the new allocation interval.  If a capture
occurs, then the right half of the current allocation interval becomes
the new allocation interval.  If a success occurs and the current
allocation interval is tagged as a left interval, then the
corresponding right interval becomes the new allocation interval. If
an idle occurs and the current allocation interval is tagged as a left
interval, then the left half of the corresponding right interval
becomes the new allocation interval.  Otherwise, if a success or an
idle occurs and the current allocation interval is tagged as a right
interval, the waiting interval truncated to length $\phi_0$ becomes
the new allocation interval, and a new Collision Resolution Period
(CRP) begins in the next time slot $k+1$. Note that the transmit power
levels in PCFCFS are variable and based on channel feedback, i.e.,
they are adaptive.

\subsection{Examples}
\label{subsec:examples_pcfcfs}

\begin{figure}[thbp]
\centering
\includegraphics[width=6in]{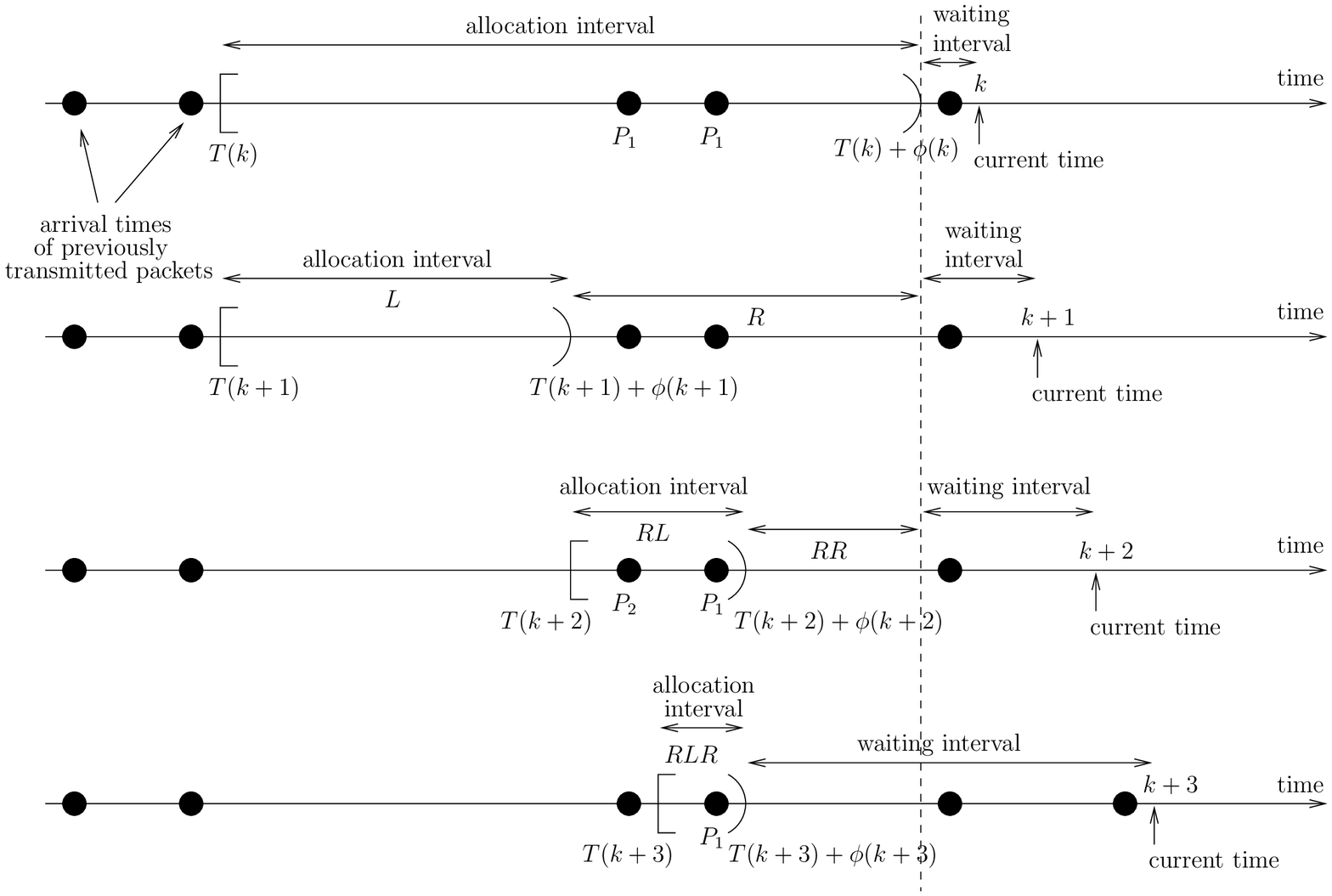}
\caption{PCFCFS splitting algorithm illustrating a collision followed
  by an idle.}
\label{fig:pcfcfs_example1}
\end{figure}

\begin{figure}[thbp]
\centering
\includegraphics[width=6in]{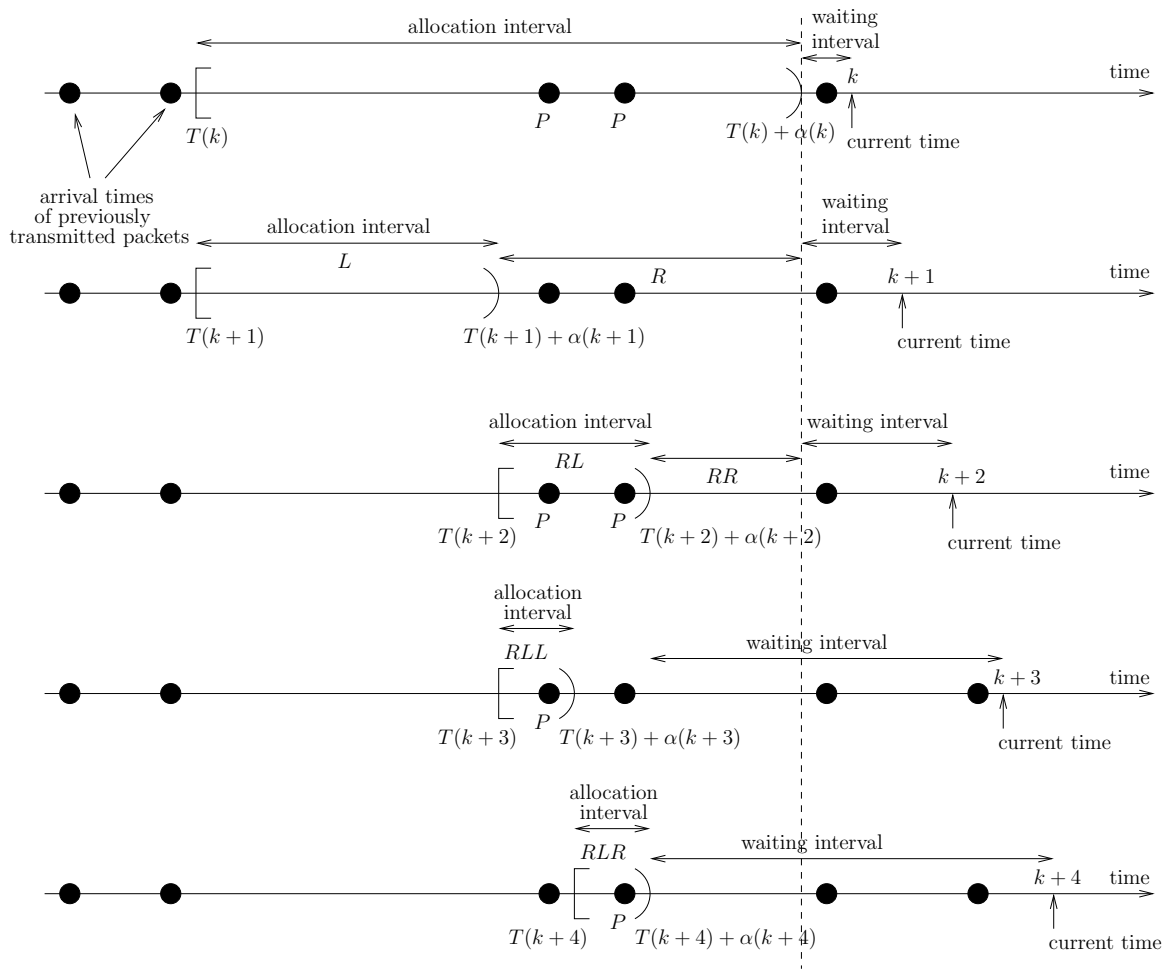}
\caption{FCFS splitting algorithm illustrating a collision followed by
  an idle.}
\label{fig:fcfs_example1}
\end{figure}

To illustrate the rules of the PCFCFS splitting algorithm for a single
CRP, consider the example shown in Figure \ref{fig:pcfcfs_example1}.
In slot $k$, the allocation interval has no node in its left half and
two nodes in its right half. Both these nodes transmit with nominal
power $P_1$ and a collision occurs. So, the allocation interval is
split, with the left interval $L$ being the allocation interval for
slot $k+1$.  An idle occurs in slot $k+1$. Next, the right subinterval
$R$ is further split, with $RL$ being the allocation interval for slot
$k+2$. The left node in $RL$ transmits with higher power $P_2$, while
the right node in $RL$ transmits with nominal power $P_1$, resulting
in a capture of the packet transmitted by the left node.  The
allocation interval is further split, with $RLR$ forming the
allocation interval for slot $k+3$.  Since a capture occurred in $RL$
in slot $k+2$, the corresponding right subinterval $RR$ is returned to
the waiting interval in slot $k+3$.  Post-capture, the lone node in
$RLR$ transmits with nominal power $P_1$, resulting in a success and
completing the CRP. For the same sequence of arrival times, the
behavior of the FCFS algorithm with uniform transmission power $P$ is
shown in Figure \ref{fig:fcfs_example1}, where $\alpha(k)$ denotes the
length of the allocation interval in slot $k$.  Note that the FCFS
algorithm requires 5 slots to resolve the collisions, while the
proposed PCFCFS algorithm requires only 4 slots.

\begin{figure}[thbp]
\centering
\includegraphics[width=6in]{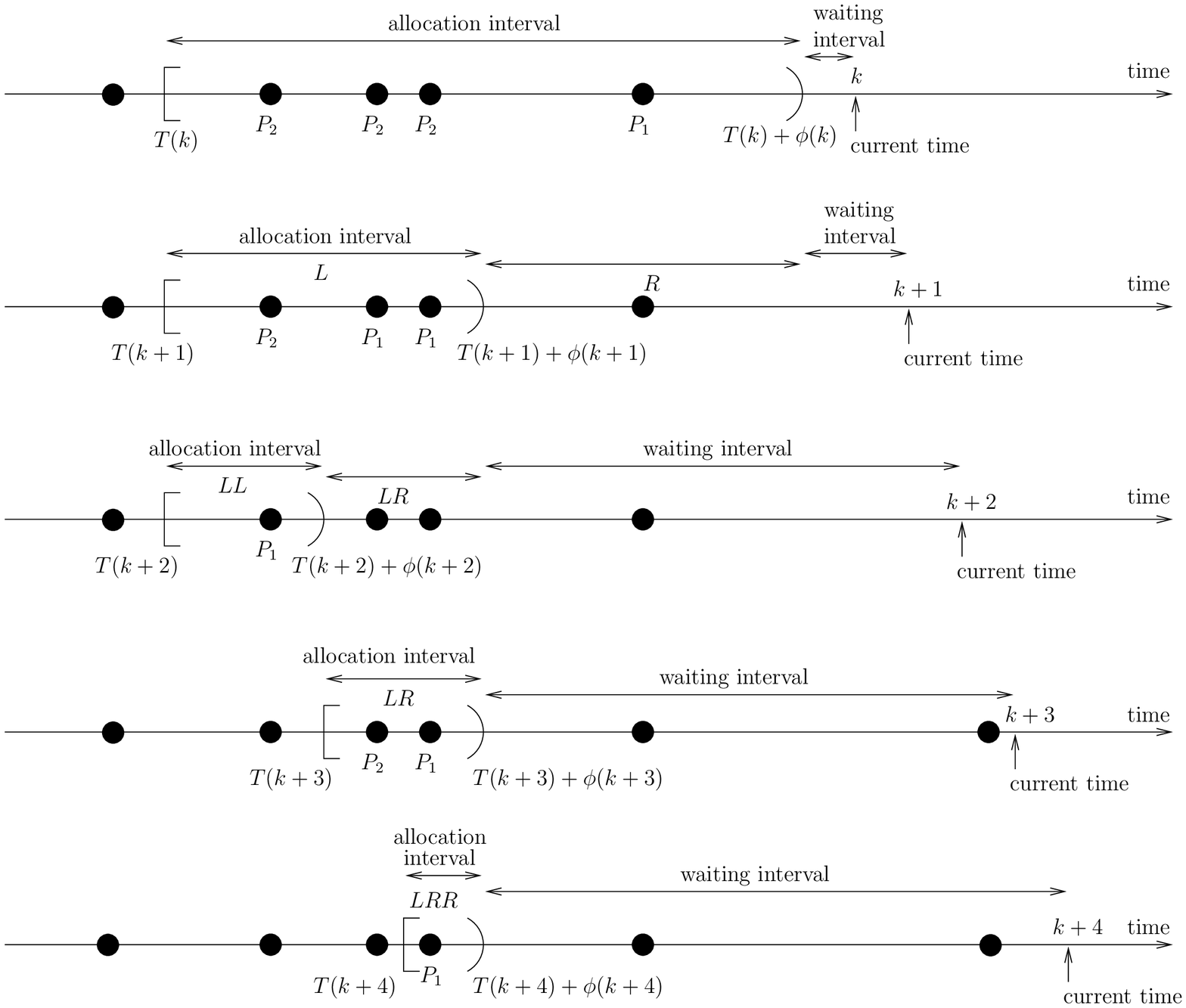}
\caption{PCFCFS splitting algorithm illustrating a collision followed
  by another collision.}
\label{fig:pcfcfs_example2}
\end{figure}

\begin{figure}[thbp]
\centering
\includegraphics[width=6in]{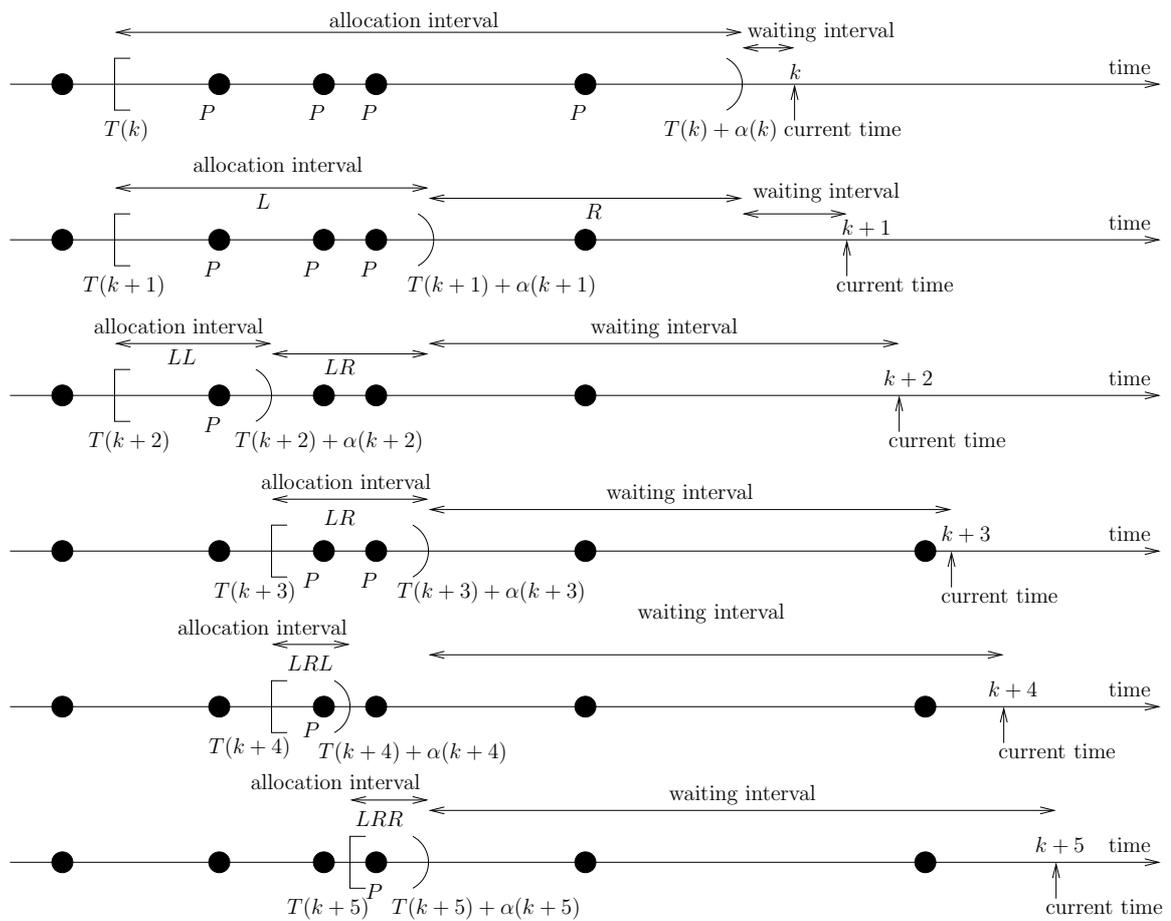}
\caption{FCFS splitting algorithm illustrating a collision followed
  by another collision.}
\label{fig:fcfs_example2}
\end{figure}

To further illustrate the rules of the PCFCFS splitting algorithm for
a single CRP, consider the example shown in Figure
\ref{fig:pcfcfs_example2}.  In slot $k$, the allocation interval has
three nodes in its left half and one node in its right half. All `left
half' nodes transmit with higher power $P_2$, while the `right half'
node transmits with nominal power $P_1$, leading to a collision.  So,
the allocation interval is split, with the left interval $L$ being the
allocation interval for slot $k+1$. In slot $k+1$, the allocation
interval has one node in its left half, which transmits with higher
power $P_2$, and two nodes in its right half, which transmit with
nominal power $P_1$. Hence, a collision occurs, and the allocation
interval $L$ is split into two equal sized subintervals $LL$ and $LR$,
with $LL$ being the allocation interval for slot $k+2$. Since a
collision is followed by another collision, the right interval $R$ is
returned to the waiting interval in slot $k+2$. In slot $k+2$, there
is only one node in the allocation interval. Since this lone node lies
in the right half of the allocation interval, it transmits with
nominal power $P_1$, leading to a success. Thus, $LR$ becomes the
allocation interval for slot $k+3$. For this allocation interval, the
node in the left half transmits with higher power $P_2$ and the node
in the right half transmits with nominal power $P_1$, resulting in a
capture of the packet transmitted by the former node.  Consequently,
$LRR$ becomes the new allocation interval for slot $k+4$.  Finally, in
slot $k+4$, the lone node transmits with nominal power $P_1$, leading
to a deterministic success and completing the CRP.  For the same
sequence of arrival times, the behavior of the FCFS algorithm with
uniform transmission power $P$ is shown in Figure
\ref{fig:fcfs_example2}.  Note that the FCFS algorithm requires 6
slots to resolve the collisions, while the proposed PCFCFS algorithm
requires only 5 slots.

\section{Throughput Analysis}
\label{sec:throughput_pcfcfs}

\begin{figure}[thbp]
\centering
\includegraphics[width=6.5in]{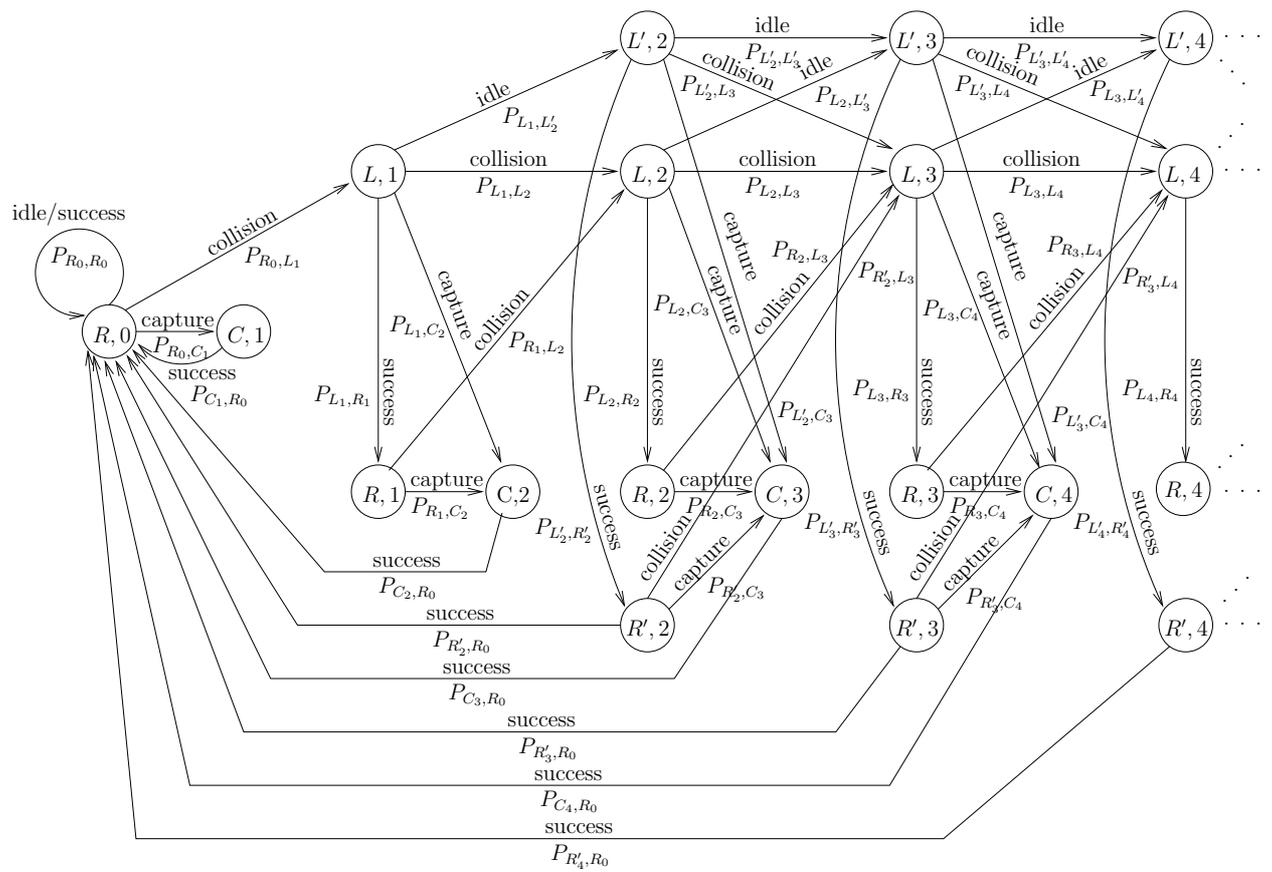}
\caption{Discrete Time Markov Chain representing a CRP of PCFCFS splitting algorithm.}
\label{fig:dtmc_pcfcfs}
\end{figure}

The evolution of a CRP can be represented by the Discrete Time Markov
Chain (DTMC) shown in Figure \ref{fig:dtmc_pcfcfs}.  Every state in
the DTMC is a pair $(\sigma,i)$, where $\sigma$ is the status $\{L,
L', R, R', C\}$ and $i$ is the number of times the original allocation
interval (of length $\phi_0$) has been split.  State $(R,0)$
corresponds to the initial slot of a CRP. If an idle or a success
occurs, the CRP ends immediately and a new CRP begins in the next
slot.  If a capture occurs, a transition occurs to state $(C,1)$,
where $C$ indicates that capture has occurred in the allocation
interval.  If a collision occurs in $(R,0)$, a transition occurs to
state $(L,1)$.  Each subsequent idle in a left allocation interval
generates one additional split with a smaller left allocation
interval, corresponding to a transition to $(L',i+1)$, where $L'$
indicates that the current left allocation interval has been reached
after a collision (in some time slot) followed by one or more idles.
A collision in an allocation interval generates one additional split
with a smaller left allocation interval, corresponding to a transition
to $(L,i+1)$, where $L$ indicates that the current left allocation
interval has been reached just after a collision.  A capture in an
allocation interval generates an additional split with a smaller right
allocation interval and corresponds to a transition to $(C,i+1)$. This
is followed by a success from $(C,i+1)$ to $(R,0)$, thus ending the
CRP.  A success in a left allocation interval leads to the
corresponding right allocation interval with no additional split,
which causes a transition from $(L,i)$ to $(R,i)$, or $(L',i)$ to
$(R',i)$.  A success in $(R',i)$ causes a transition to $(R,0)$, thus
ending the CRP.  It can be easily verified that the states and
transitions in Figure \ref{fig:dtmc_pcfcfs} constitute a Markov chain,
i.e., each transition from every state is independent of the path used
to reach the given state.

We now analyze a single CRP. Assume that the size of the initial
allocation interval is $\phi_0$ (corresponding to state $(R,0)$).
Each splitting of the allocation interval halves this, so that states
$(L,i)$, $(L',i)$, $(R,i)$, $(R',i)$ and $(C,i)$ in Figure
\ref{fig:dtmc_pcfcfs} correspond to allocation intervals of size
$2^{-i}\phi_0$. Since the arrival process is Poisson with rate
$\lambda$, the number of packets in the original allocation interval
is a Poisson random variable (r.v.) with mean $\lambda\phi_0$.
Consequently, the a priori distributions on the number of packets in
disjoint subintervals are independent and Poisson. Define $G_i$ as the
expected number of packets in an interval that has been split $i$
times. Thus
\begin{eqnarray}
G_i &=& 2^{-i}\lambda\phi_0 \;=\; 2^{-i}G_0 \;\;\;\forall\;\;\; i \geqslant 0,
 \label{eq:expected_packets} \\
\therefore G_i &=& \frac{1}{2}G_{i-1} \;\;\;\forall\;\;\; i \geqslant 1.
 \label{eq:expected_half}
\end{eqnarray}

We view $(R,0)$ as the starting state as well as the final state. For
brevity in notation, the transition probability from state $(A,i)$ to
state $(B,j)$ is denoted by $P_{A_i,B_j}$, where $A,B \in
\{L,L',R,R',C\}$ and $i,j \in \{0\} \cup {\mathbb Z}^+$ (see Figure
\ref{fig:dtmc_pcfcfs}). For example, the transition probability from
$(L,1)$ to $(C,2)$ is denoted by $P_{L_1,C_2}$.

\begin{figure}[thbp]
\centering
\includegraphics[width=6in]{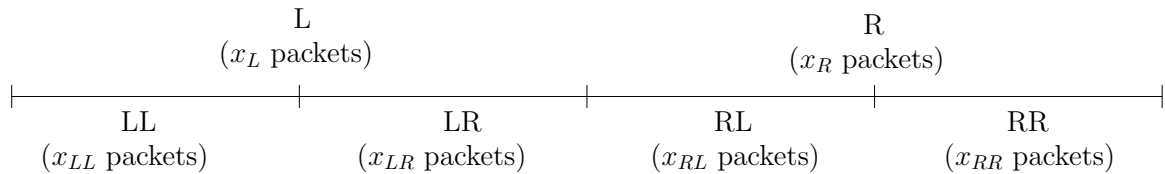}
\caption{Notation for number of packets in left and right subintervals
of the original allocation interval.}
\label{fig:notation_LR}
\end{figure}

$P_{R_0,R_0}$ is the probability of an idle or success in the first
slot of the CRP. Since the number of packets in the initial allocation
interval is Poisson with mean $G_0$, the probability of 0 or 1 packet
is
\begin{eqnarray}
P_{R_0,R_0} &=& (1+G_0)e^{-G_0}.
\label{eq:probability_R0_R0}
\end{eqnarray}

$P_{R_0,C_1}$ is the probability of capture in the first slot of a
CRP. Let $x_L$ and $x_R$ denote the number of packets in the left and
right halves of the original allocation interval respectively, as
shown in Figure \ref{fig:notation_LR}.  Capture occurs if and only if
$x_L=1$ and $x_R=1$.  $x_L$ and $x_R$ are independent Poisson r.v.s of
mean $G_1$ each. Thus
\begin{eqnarray}
P_{R_0,C_1}
 &=& \Pr(x_L=1,x_R=1), \nonumber \\
 &=& \Pr(x_L=1)\Pr(x_R=1), \nonumber \\
 &=& G_1^2 e^{-2G_1}, \nonumber \\
\therefore P_{R_0,C_1}
 &=& \frac{G_0^2}{4}e^{-G_0}.
\label{eq:probability_R0_C1}
\end{eqnarray}

State $(L,1)$ is entered after collision in state $(R,0)$.  Using
(\ref{eq:probability_R0_R0}) and (\ref{eq:probability_R0_C1}), this
occurs with probability
\begin{eqnarray}
P_{R_0,L_1} 
 &=& 1 - P_{R_0,R_0} - P_{R_0,C_1}, \nonumber \\
\therefore P_{R_0,L_1} 
 &=& 1 - \Big(1+G_0+\frac{G_0^2}{4}\Big)e^{-G_0}.
\label{eq:probability_R0_L1}
\end{eqnarray}

Since a capture is always followed by a deterministic success,
\begin{eqnarray}
P_{C_i,R_0} &=& 1 \;\;\forall\;\; i \geqslant 1.
\label{eq:probability_Ci_R0}
\end{eqnarray}

\begin{lemma}
The outgoing transition probabilities from $(L,i)$, where $i \geqslant1$, 
are given by
\begin{eqnarray}
P_{L_i,R_i} 
&=& \frac{(1-e^{-G_i}-G_ie^{-G_i})G_ie^{-G_i}}{1-\Big(1+G_{i-1}+\frac{G_{i-1}^2}{4}\Big)e^{-G_{i-1}}},  \\
P_{L_i,L_{i+1}'}
&=& \frac{(1-e^{-G_i}-G_ie^{-G_i})e^{-G_i}}{1-\Big(1+G_{i-1}+\frac{G_{i-1}^2}{4}\Big)e^{-G_{i-1}}},  \\
P_{L_i,C_{i+1}}
&=& \frac{\frac{G_i^2}{4}e^{-G_i}}{1-\Big(1+G_{i-1}+\frac{G_{i-1}^2}{4}\Big)e^{-G_{i-1}}}, \label{eq:probability_Li_Ciplus1} \\
P_{L_i,L_{i+1}}
&=& \frac{1-(1+G_i+\frac{G_i^2}{4})e^{-G_i}}{1-\Big(1+G_{i-1}+\frac{G_{i-1}^2}{4}\Big)e^{-G_{i-1}}}. \label{eq:probability_Li_Liplus1}
\end{eqnarray}
\label{lem:outgoing_probabilities_Li}
\end{lemma}

\begin{proof}
  Refer to Figure \ref{fig:dtmc_pcfcfs}.  For $i=1$, $(L,i)$ is
  entered only via a collision in $(R,i-1)$.  For $i=2$, $(L,i)$ is
  entered only via a collision in $(L,i-1)$ or $(R,i-1)$.  For
  $i\geqslant 3$, $(L,i)$ is entered only via a collision in
  $(L',i-1)$, $(L,i-1)$, $(R,i-1)$ or $(R',i-1)$. In every case, a
  subinterval $Y$ is split into $YL$ and $YR$, and $YL$ becomes the
  new allocation interval. Let $x_{YL}$ and $x_{YR}$ denote the number
  of packets in $YL$ and $YR$ respectively.  A priori, $x_{YL}$ and
  $x_{YR}$ are independent Poisson r.v.s of mean $G_i$ each. The event
  that a collision occurred in the previous state is
  $\{x_{YL}+x_{YR}\geqslant 2\} \cap \{x_{YL}=x_{YR}=1\}^c =: \mathcal
  C_Y$.  Note that $x_{YL}+x_{YR}=x_Y$ is a Poisson r.v. of mean
  $G_{i-1}$.  From (\ref{eq:expected_half}), $G_i =
  \frac{1}{2}G_{i-1}$ $\forall$ $i \geqslant 1$.  The probability of
  success in $(L,i)$ is the probability that $x_{YL}=1$ conditional on
  $\mathcal C_Y$, i.e.,
  \begin{eqnarray}
    P_{L_i,R_i}
    &=& \Pr(x_{YL}=1|\mathcal C_Y), \nonumber \\
    &=& \frac{\Pr(\mathcal C_Y|x_{YL}=1)\Pr(x_{YL}=1)}{\Pr(\mathcal C_Y)}, \nonumber \\
    &=& \frac{\Pr(\{x_{YR}=1\}\cap \{x_{YR}=1\}^c)\Pr(x_{YL}=1)}{\Pr(\mathcal C_Y)}, \nonumber \\
    &=& \frac{\Pr(x_{YR}\geqslant 2)\Pr(x_{YL}=1)}{\Pr(\{x_{YL}+x_{YR}\geqslant 2\} \cap \{x_{YL}=x_{YR}=1\}^c)}, \nonumber \\
    &=& \frac{\Pr(x_{YR}\geqslant 2)\Pr(x_{YL}=1)}{\Pr(x_Y\geqslant 2)-\Pr(x_{YL}=1)\Pr(x_{YR}=1)}, \nonumber \\
    &=& \frac{(1-e^{-G_i}-G_ie^{-G_i})G_ie^{-G_i}}{1-e^{-G_{i-1}}-G_{i-1}e^{-G_{i-1}}-G_i^2e^{-2G_i}}, \nonumber \\
    \therefore P_{L_i,R_i}
    &=& \frac{(1-e^{-G_i}-G_ie^{-G_i})G_ie^{-G_i}}{1-\Big(1+G_{i-1}+\frac{G_{i-1}^2}{4}\Big)e^{-G_{i-1}}}. \label{eq:prob_Li_Ri}
  \end{eqnarray}

  The probability of idle in $(L,i)$ is the probability that
  $x_{YL}=0$ conditional on $\mathcal C_Y$, i.e.,
  \begin{eqnarray}
    P_{L_i,L_{i+1}'}
    &=& \Pr(x_{YL}=0|\mathcal C_Y), \nonumber \\
    &=& \frac{\Pr(\mathcal C_Y|x_{YL}=0)\Pr(x_{YL}=0)}{\Pr(\mathcal C_Y)}, \nonumber \\
    &=& \frac{\Pr(\{x_{YR}\geqslant 2\}\cap \{x_{YR}=1\}^c)\Pr(x_{YL}=0)}{\Pr(\mathcal C_Y)}, \nonumber \\
    &=& \frac{\Pr(x_{YR}\geqslant 2)\Pr(x_{YL}=0)}{\Pr(\{x_{YL}+x_{YR}\geqslant 2\} \cap \{x_{YL}=x_{YR}=1\}^c)}, \nonumber \\
    &=& \frac{\Pr(x_{YR}\geqslant 2)\Pr(x_{YL}=0)}{\Pr(x_Y\geqslant 2)-\Pr(x_{YL}=1)\Pr(x_{YR}=1)}, \nonumber \\
    &=& \frac{(1-e^{-G_i}-G_ie^{-G_i})e^{-G_i}}{1-e^{-G_{i-1}}-G_{i-1}e^{-G_{i-1}}-G_i^2e^{-2G_i}}, \nonumber \\
    \therefore P_{L_i,L_{i+1}'}
    &=& \frac{(1-e^{-G_i}-G_ie^{-G_i})e^{-G_i}}{1-\Big(1+G_{i-1}+\frac{G_{i-1}^2}{4}\Big)e^{-G_{i-1}}}. \label{eq:prob_Li_Liplus1prime}
  \end{eqnarray}

  Let $x_{YLL}$ and $x_{YLR}$ denote the number of packets in $YLL$
  and $YLR$ respectively.  $x_{YLL}$ and $x_{YLR}$ are independent
  Poisson r.v.s of mean $G_{i+1}$ each, and $x_{YLL}+x_{YLR}=x_{YL}$.
  The probability of capture in $(L,i)$ is the probability that
  $x_{YLL}=1$ and $x_{YLR}=1$ conditional on $\mathcal C_Y$, i.e.,
  \begin{eqnarray}
    P_{L_i,C_{i+1}}
    &=& \Pr(x_{YLL}=1,x_{YLR}=1|\mathcal C_Y), \nonumber \\
    &=& \frac{\Pr(\mathcal C_Y|x_{YLL}=1,x_{YLR}=1)\Pr(x_{YLL}=1,x_{YLR}=1)}{\Pr(\{x_{YL}+x_{YR}\geqslant 2\}\cap\{x_{YL}=x_{YR}=1\}^c)}, \nonumber \\
    &=& \frac{\Pr(\mathcal C_Y|x_{YL}=2)\Pr(x_{YLL}=1)\Pr(x_{YLR}=1)}{\Pr(\{x_{YL}+x_{YR}\geqslant 2\}\cap\{x_{YL}=x_{YR}=1\}^c)}, \nonumber \\
    &=& \frac{\Pr(x_{YR}\geqslant 0)\Pr(x_{YLL}=1)\Pr(x_{YLR}=1)}{\Pr(x_Y\geqslant 2)-\Pr(x_{YL}=1)\Pr(x_{YR}=1)}, \nonumber \\
    &=& \frac{1.G_{i+1}^2e^{-2G_{i+1}}}{1-e^{-G_{i-1}}-G_{i-1}e^{-G_{i-1}}-G_i^2e^{-2G_i}}, \nonumber \\
    \therefore P_{L_i,C_{i+1}}
    &=& \frac{\frac{G_i^2}{4}e^{-G_i}}{1-\Big(1+G_{i-1}+\frac{G_{i-1}^2}{4}\Big)e^{-G_{i-1}}}. \label{eq:prob_Li_Ciplus1}
  \end{eqnarray}

  From (\ref{eq:prob_Li_Ri}), (\ref{eq:prob_Li_Liplus1prime}) and
  (\ref{eq:prob_Li_Ciplus1}), we obtain
  \begin{eqnarray}
    P_{L_i,L_{i+1}}
    &=& 1 - P_{L_i,R_i} - P_{L_i,L_{i+1}'} - P_{L_i,C_{i+1}}, \nonumber \\
    \therefore P_{L_i,L_{i+1}}
    &=& \frac{1-\Big(1+G_i+\frac{G_i^2}{4}\Big)e^{-G_i}}{1-\Big(1+G_{i-1}+\frac{G_{i-1}^2}{4}\Big)e^{-G_{i-1}}}.
  \end{eqnarray}
\end{proof}

\begin{lemma}
The outgoing transition probabilities from $(R,i)$ are given by
\begin{eqnarray}
P_{R_i,C_{i+1}}
&=& \frac{\frac{G_i^2}{4}e^{-G_i}}{1-(1+G_i)e^{-G_i}}
 \;\; \forall \;\; i \geqslant 1, \label{eq:probability_Ri_Ciplus1} \\
P_{R_i,L_{i+1}}
&=& \frac{1-\big(1+G_i+\frac{G_i^2}{4}\big)e^{-G_i}}{1-(1+G_i)e^{-G_i}}
 \;\; \forall \;\; i \geqslant 1. \label{eq:probability_Ri_Liplus1}
\end{eqnarray}
\label{lem:outgoing_probabilities_Ri}
\end{lemma}

\begin{proof}
  Refer to Figure \ref{fig:dtmc_pcfcfs}.  For $i \geqslant 1$, $(R,i)$
  is entered only via a success in $(L,i)$.  Recall that $(L,i)$ was
  entered only via a collision from a previous state. We use the
  notation introduced in the proof of Lemma
  \ref{lem:outgoing_probabilities_Li}. Define the event
  \begin{eqnarray}
    \mathcal S_{YL} 
    &:=& \mathcal C_Y \cap \{x_{YL}=1\}, \nonumber \\
    &=& \{x_{YL}+x_{YR}\geqslant 2\}\cap \{x_{YL}=x_{YR}=1\}^c \cap \{x_{YL}=1\}, \nonumber \\
    &=& \{x_{YR}\geqslant 1\} \cap \{x_{YR}=1\}^c \cap \{x_{YL}=1\}, \nonumber \\
    \therefore \mathcal S_{YL}
    &=& \{x_{YR}\geqslant 2\} \cap \{x_{YL}=1\}.
  \end{eqnarray}

  Let $x_{YRL}$ and $x_{YRR}$ denote the number of packets in $YRL$
  and $YRR$ respectively.  $x_{YRL}$ and $x_{YRR}$ are independent
  Poisson r.v.s of mean $G_{i+1}$ each. Since $x_{YR}\geqslant 2$, a
  success or an idle can never occur in state $(R,i)$. Note that
  $x_{YR}=x_{YRL}+x_{YRR}$.  The probability of capture in state
  $(R,i)$ is the probability that $x_{YRL}=1$ and $x_{YRR}=1$
  conditional on $\mathcal S_{YL}$, i.e.,
  \begin{eqnarray}
    P_{R_i,C_{i+1}}
    &=& \Pr(x_{YRL}=1,x_{YRR}=1|x_{YR}\geqslant 2,x_{YL}=1), \nonumber \\
    &=& \Pr(x_{YRL}=1,x_{YRR}=1|x_{YR}\geqslant 2), \nonumber \\
    &=& \frac{\Pr(x_{YR}\geqslant 2|x_{YRL}=1,x_{YRR}=1)\Pr(x_{YRL}=1,x_{YRR}=1)}{\Pr(x_{YR}\geqslant 2)}, \nonumber \\
    &=& \frac{\Pr(x_{YRL}+x_{YRR}\geqslant 2|x_{YRL}=1,x_{YRR}=1)\Pr(x_{YRL}=1,x_{YRR}=1)}{\Pr(x_{YR}\geqslant 2)}, \nonumber \\
    &=& \frac{1.\Pr(x_{YRL}=1)\Pr(x_{YRR}=1)}{\Pr(x_{YR}\geqslant 2)}, \nonumber \\
    &=& \frac{G_{i+1}^2e^{-2G_{i+1}}}{1-e^{-G_i}-G_ie^{-G_i}}, \nonumber \\
    \therefore 
      P_{R_i,C_{i+1}} &=& \frac{\frac{G_i^2}{4}e^{-G_i}}{1-(1+G_i)e^{-G_i}}.
          \label{eq:prob_Ri_Ciplus1}
  \end{eqnarray}

  From (\ref{eq:prob_Ri_Ciplus1}), we obtain
  \begin{eqnarray}
    P_{R_i,L_{i+1}}
    &=& 1 - P_{R_i,C_{i+1}}, \nonumber \\
    &=& 1 - \frac{\frac{G_i^2}{4}e^{-G_i}}{1-(1+G_i)e^{-G_i}}, \nonumber \\
    \therefore P_{R_i,L_{i+1}}
    &=& \frac{1-\big(1+G_i+\frac{G_i^2}{4}\big)e^{-G_i}}{1-(1+G_i)e^{-G_i}}.
  \end{eqnarray}
\end{proof}

\begin{lemma}
The outgoing transition probabilities from $(L',i)$ are given by
\begin{eqnarray}
P_{L_i',R_i'} 
 &=& \frac{(1-e^{-G_i})G_ie^{-G_i}}{1-(1+G_{i-1})e^{-G_{i-1}}}
  \;\; \forall \;\; i \geqslant 2, \label{eq:probability_Liprime_Riprime} \\
P_{L_i',L_{i+1}'}
 &=& \frac{(1-e^{-G_i}-G_ie^{-G_i})e^{-G_i}}{1-(1+G_{i-1})e^{-G_{i-1}}} \;\; \forall \;\; i \geqslant 2, \label{eq:probability_Liprime_Liplus1prime} \\
P_{L_i',C_{i+1}}
 &=& \frac{\frac{G_i^2}{4}e^{-G_i}}{1-(1+G_{i-1})e^{-G_{i-1}}} 
  \;\; \forall \;\; i \geqslant 2, \label{eq:probability_Liprime_Ciplus1} \\
P_{L_i',L_{i+1}}
 &=&\frac{1-\big(1+G_i+\frac{G_i^2}{4}\big)e^{-G_i}}{1-(1+G_{i-1})e^{-G_{i-1}}}
  \;\; \forall \;\; i \geqslant 2. \label{eq:probability_Liprime_Liplus1}
\end{eqnarray}
\label{lem:outgoing_probabilities_Liprime}
\end{lemma}

\begin{proof}
  Refer to Figure \ref{fig:dtmc_pcfcfs}.  For $i=2$, $(L',i)$ is
  entered only by an idle in $(L,i-1)$.  For $i \geqslant 3$, state
  $(L',i)$ is entered by an idle in $(L',i-1)$ or an idle in
  $(L,i-1)$.  In every case, a residual right subinterval, say $Z$, is
  split into $ZL$ and $ZR$, and $ZL$ becomes the new allocation
  interval. Note that $(L',i)$ can be entered if and only if there is
  a collision (in some time slot) followed by one or more idles.
  Therefore, $Z$ must contain at least two packets.  Let $x_{ZL}$ and
  $x_{ZR}$ denote the number of packets in $ZL$ and $ZR$ respectively.
  A priori, $x_{ZL}$ and $x_{ZR}$ are independent Poisson r.v.s of
  mean $G_i$ each.  Let $x_Z$ denote the number of packets in $Z$.
  Thus $x_Z=x_{ZL}+x_{ZR}$, $x_Z$ is a Poisson r.v. of mean $G_{i-1}$
  and $x_Z \geqslant 2$.

  The probability of success in $(L',i)$ is the probability that
  $x_{ZL}=1$ conditional on $x_Z \geqslant 2$, i.e.,
  \begin{eqnarray}
    P_{L_i',R_i'}
    &=& \Pr(x_{ZL}=1|x_Z\geqslant 2), \nonumber \\
    &=& \frac{\Pr(x_Z\geqslant 2|x_{ZL}=1)\Pr(x_{ZL}=1)}{\Pr(x_Z\geqslant 2)}, \nonumber \\
    &=& \frac{\Pr(x_{ZL}+x_{ZR}\geqslant 2|x_{ZL}=1)\Pr(x_{ZL}=1)}{\Pr(x_Z\geqslant 2)}, \nonumber \\
    &=& \frac{\Pr(x_{ZR}\geqslant 1)\Pr(x_{ZL}=1)}{\Pr(x_Z\geqslant 2)}, \nonumber \\
    \therefore P_{L_i',R_i'}
    &=& \frac{(1-e^{-G_i})G_ie^{-G_i}}{1-(1+G_{i-1})e^{-G_{i-1}}} .
          \label{eq:prob_Liprime_Riprime}
  \end{eqnarray}

  The probability of idle in $(L',i)$ is the probability that
  $x_{ZL}=0$ conditional on $x_Z \geqslant 2$, i.e.,
  \begin{eqnarray}
    P_{L_i',L_{i+1}'}
    &=& \Pr(x_{ZL}=0|x_Z\geqslant 2), \nonumber \\
    &=& \frac{\Pr(x_Z\geqslant 2|x_{ZL}=0)\Pr(x_{ZL}=0)}{\Pr(x_Z\geqslant 2)}, \nonumber \\
    &=& \frac{\Pr(x_{ZL}+x_{ZR}\geqslant 2|x_{ZL}=0)\Pr(x_{ZL}=0)}{\Pr(x_Z\geqslant 2)}, \nonumber \\
    &=& \frac{\Pr(x_{ZR}\geqslant 2)\Pr(x_{ZL}=0)}{\Pr(x_Z\geqslant 2)}, \nonumber \\
    \therefore P_{L_i',L_{i+1}'}
    &=& \frac{(1-e^{-G_i}-G_ie^{-G_i})e^{-G_i}}{1-(1+G_{i-1})e^{-G_{i-1}}}.
          \label{eq:prob_Liprime_Liplus1prime}
  \end{eqnarray}

  Let $x_{ZLL}$ and $x_{ZLR}$ denote the number of packets in $ZLL$
  and $ZLR$ respectively. A priori,$x_{ZLL}$ and $x_{ZLR}$ are
  independent Poisson r.v.s of mean $G_{i+1}$ each.  The probability
  of capture in $(L',i)$ is the probability that $x_{ZLL}=1$ and
  $x_{ZLR}=1$ conditional on $x_Z \geqslant 2$, i.e.,

  \begin{eqnarray}
    P_{L_i',C_{i+1}}
    &=& \Pr(x_{ZLL}=1,x_{ZLR}=1|x_Z\geqslant 2), \nonumber \\
    &=& \frac{\Pr(x_Z\geqslant 2|x_{ZLL}=1,x_{ZLR}=1)\Pr(x_{ZLL}=1,x_{ZLR}=1)}{\Pr(x_Z\geqslant 2)}, \nonumber \\
    &=& \frac{1.\Pr(x_{ZLL}=1)\Pr(x_{ZLR}=1)}{\Pr(x_Z\geqslant 2)}, \nonumber \\
    &=& \frac{G_{i+1}^2e^{-2G_{i+1}}}{1-e^{-G_{i-1}}-G_{i-1}e^{-G_{i-1}}}, \nonumber \\
    \therefore P_{L_i',C_{i+1}}
    &=& \frac{\frac{G_i^2}{4}e^{-G_i}}{1-(1+G_{i-1})e^{-G_{i-1}}}.
          \label{eq:prob_Liprime_Ciplus1}
  \end{eqnarray}

  From (\ref{eq:prob_Liprime_Riprime}),
  (\ref{eq:prob_Liprime_Liplus1prime}) and
  (\ref{eq:prob_Liprime_Ciplus1}), we obtain
  \begin{eqnarray}
    P_{L_i',L_{i+1}}
    &=& 1 - P_{L_i',R_i'} - P_{L_i',L_{i+1}'} - P_{L_i',C_{i+1}'}, \\
    \therefore P_{L_i',L_{i+1}}
    &=& \frac{1-\big(1+G_i+\frac{G_i^2}{4}\big)e^{-G_i}}{1-(1+G_{i-1})e^{-G_{i-1}}}.
  \end{eqnarray}
\end{proof}

\begin{lemma}
The outgoing transition probabilities from $(R',i)$ are given by
\begin{eqnarray}
P_{R_i',R_0} 
 &=& \frac{G_ie^{-G_i}}{1-e^{-G_i}} 
  \;\; \forall \;\; i \geqslant 2, \label{eq:probability_Riprime_R0} \\
P_{R_i',C_{i+1}}
 &=& \frac{\frac{G_i^2}{4}e^{-G_i}}{1-e^{-G_i}}
  \;\; \forall \;\; i \geqslant 2, \label{eq:probability_Riprime_Ciplus1} \\
P_{R_i',L_{i+1}}
 &=& \frac{1-\big(1+G_i+\frac{G_i^2}{4}\big)e^{-G_i}}{1-e^{-G_i}}
  \;\; \forall \;\; i \geqslant 2. \label{eq:probability_Riprime_Liplus1}
\end{eqnarray}
\label{lem:outgoing_probabilities_Riprime}
\end{lemma}

\begin{proof}
  Refer to Figure \ref{fig:dtmc_pcfcfs}.  For $i \geqslant 2$, state
  $(R',i)$ is entered if and only if a success occurs in state
  $(L',i)$. When $(L',i)$ was entered, a residual right subinterval
  $Z$ was split into $ZL$ and $ZR$, and $ZL$ became the new allocation
  interval. Recall that $x_Z\geqslant 2$, since $(L',i)$ can only be
  entered after a collision followed by one or more idles. A success
  in $(L',i)$ implies $x_{ZL}=1$.  Hence, $(R',i)$ is entered if and
  only if both these events occurs, i.e., $x_Z\geqslant 2$ and
  $x_{ZL}=1$. Therefore, $(R',i)$ can be entered if and only if
  $x_{ZR}\geqslant 1$. Note that there can never be an idle from
  $(R',i)$.

  The probability of success in $(R',i)$ is the probability that
  $x_{ZR}=1$ conditional on $x_{ZR}\geqslant 1$, i.e.,
  \begin{eqnarray}
    P_{R_i',R_0}
    &=& \Pr(x_{ZR}=1|x_{ZR}\geqslant 1), \nonumber \\
    &=& \frac{\Pr(x_{ZR}\geqslant 1|x_{ZR}=1)\Pr(x_{ZR}=1)}{\Pr(x_{ZR}\geqslant 1)}, \nonumber \\
    \therefore P_{R_i',R_0}
    &=& \frac{G_ie^{-G_i}}{1-e^{-G_i}}. \label{eq:prob_Riprime_R0}
  \end{eqnarray}

  Let $x_{ZRL}$ and $x_{ZRR}$ denote the number of packets in $ZRL$
  and $ZRR$ respectively. Note that $x_{ZR}=x_{ZRL}+x_{ZRR}$.
  $x_{ZRL}$ and $x_{ZRR}$ are independent Poisson r.v.s of mean
  $G_{i+1}$ each. The probability of capture in state $(R',i)$ is the
  probability that $x_{ZRL}=1$ and $x_{ZRR}=1$ conditional on
  $x_{ZR}\geqslant 1$, i.e.,
  \begin{eqnarray}
    P_{R_i',C_{i+1}}
    &=& \Pr(x_{ZRL}=1,x_{ZRR}=1|x_{ZR}\geqslant 1), \nonumber \\
    &=& \frac{\Pr(x_{ZR}\geqslant 1|x_{ZRL}=1,x_{ZRR}=1)\Pr(x_{ZRL}=1,x_{ZRR}=1)}{\Pr(x_{ZR}\geqslant 1)}, \nonumber \\
    &=& \frac{1.\Pr(x_{ZRL}=1)\Pr(x_{ZRR}=1)}{\Pr(x_{ZR}\geqslant 1)}, \nonumber \\
    &=& \frac{G_{i+1}^2e^{-2G_{i+1}}}{1-e^{-G_i}}, \nonumber \\
    \therefore P_{R_i',C_{i+1}}
    &=& \frac{\frac{G_i^2}{4}e^{-G_i}}{1-e^{-G_i}}.
          \label{eq:prob_Riprime_Ciplus1}
  \end{eqnarray}

  From (\ref{eq:prob_Riprime_R0}) and (\ref{eq:prob_Riprime_Ciplus1}),
  we obtain
  \begin{eqnarray}
    P_{R_i',L_{i+1}}
    &=& 1 - P_{R_i',R_0} - P_{R_i',C_{i+1}}, \nonumber \\
    \therefore P_{R_i',L_{i+1}}
    &=& \frac{1-\big(1+G_i+\frac{G_i^2}{4}\big)e^{-G_i}}{1-e^{-G_i}}.
  \end{eqnarray}
\end{proof}

In summary, Figure \ref{fig:dtmc_pcfcfs} is a DTMC and the transition
probabilities are given by (\ref{eq:probability_R0_R0}),
(\ref{eq:probability_R0_C1}), (\ref{eq:probability_R0_L1}) and
(\ref{eq:probability_Ci_R0}), and Lemmas
\ref{lem:outgoing_probabilities_Li},
\ref{lem:outgoing_probabilities_Ri},
\ref{lem:outgoing_probabilities_Liprime} and
\ref{lem:outgoing_probabilities_Riprime}.

We now analyze the DTMC in Figure \ref{fig:dtmc_pcfcfs}. Observe that
no state can be entered more than once before the return to $(R,0)$.
Let $Q_{X_i}$ denote the probability that state $(X,i)$ is entered
before returning to $(R,0)$, where $X \in \{L,L',R,R',C\}$ and $i \in
{\mathbb Z}^+$.  In other words, $Q_{X_i}$ denotes the probability of
hitting $(X,i)$ in a CRP given that we start from $(R,0)$. Note that
$Q_{C_1}=P_{R_0,C_1}$ and $Q_{L_1}=P_{R_0,L_1}$. The probabilities
$Q_{X_i}$ can be calculated iteratively from the initial state $(R,0)$
as follows:
\begin{eqnarray}
Q_{C_1} &=& \frac{G_0^2}{4}e^{-G_0}, \\
Q_{L_1} &=& 1-\Big(1+G_0+\frac{G_0^2}{4}\Big)e^{-G_0}, \\
Q_{C_2} &=& Q_{L_1}P_{L_1,C_2} + Q_{R_1}P_{R_1,C_2}, \\
Q_{L_2'} &=& Q_{L_1}P_{L_1,L_2'}, \\
Q_{L_2} &=& Q_{L_1}P_{L_1,L_2} + Q_{R_1}P_{R_1,L_2}, \\
Q_{L_i} &=& Q_{L_{i-1}'}P_{L_{i-1}',L_i} + Q_{L_{i-1}}P_{L_{i-1},L_i}
  + Q_{R_{i-1}}P_{R_{i-1},L_i} \nonumber \\
      && {} + Q_{R_{i-1}'}P_{R_{i-1}',R_i} \;\; \forall \;\; i \geqslant 3, \\
Q_{L_i'} &=& Q_{L_{i-1}'}P_{L_{i-1}',L_i'} + Q_{L_{i-1}}P_{L_{i-1},L_i'}
  \;\; \forall \;\; i \geqslant 3, \\
Q_{R_i} &=& Q_{L_i}P_{L_i,R_i} \;\; \forall \;\; i \geqslant 1, \\
Q_{R_i'} &=& Q_{L_i'}P_{L_i',R_i'} \;\; \forall \;\; i \geqslant 2, \\
Q_{C_i} &=& Q_{L_{i-1}'}P_{L_{i-1}',C_i} + Q_{L_i}P_{L_i,C_i} 
  + Q_{R_{i-1}}P_{R_{i-1},C_i} + Q_{R_{i-1}'}P_{R_{i-1}',C_i}
    \;\; \forall \;\; i \geqslant 3.
\end{eqnarray}

Let random variable $K$ denote the number of slots in a CRP. Thus, $K$
equals the number of states visited in the Markov chain, including the
initial state $(R,0)$, before the return to $(R,0)$. Thus
\begin{eqnarray}
E[K] &=& 1 + \sum_{i=1}^\infty (Q_{L_i}+Q_{L_i'}+Q_{R_i}+Q_{R_i'}+Q_{C_i}),
\label{eq:expected_slots}
\end{eqnarray}
where we assume $Q_{L_1'}=Q_{R_1'}=0$.

We evaluate the change in $T(k)$ from one CRP to the next, i.e., we
evaluate the difference in left endpoints of initial allocation
intervals of successive CRPs. For the assumed initial interval of size
$\phi_0$, this change is at most $\phi_0$. However, if left allocation
intervals have collisions or captures (e.g., $RL$ in Figure
\ref{fig:pcfcfs_example1}), then the corresponding right allocation
intervals (e.g., $RR$ in Figure \ref{fig:pcfcfs_example1}) are
returned to the waiting interval, and the change is less than
$\phi_0$. Let random variable $F$ denote the fraction of $\phi_0$
returned in this manner over a CRP, so that $\phi_0(1-F)$ is the
change in $T(k)$. We distinguish between two cases:
\begin{enumerate}
\item If a left allocation interval of type $(L,i)$ has a collision or
  a capture, then the corresponding right allocation interval $(R,i)$
  is returned to the waiting interval. Let $U_{L_i}$ denote the
  probability that $(L,i)$ has a collision or a capture. Hence,
  $U_{L_i}$ denotes the probability that $(L,i)$ has two or more
  packets. Thus, $U_{L_i}=P_{L_i,L_{i+1}}+P_{L_i,C_{i+1}}$. Using
  (\ref{eq:probability_Li_Ciplus1}) and
  (\ref{eq:probability_Li_Liplus1}), we obtain
  \begin{eqnarray}
    U_{L_i}
    &=& \frac{1-(1+G_i)e^{-G_i}}{1-\Big(1+G_{i-1}+\frac{G_{i-1}^2}{4}\Big)e^{-G_{i-1}}}  \;\; \forall \;\; i \geqslant 1.
  \end{eqnarray}

\item If a left allocation interval of type $(L',i)$ has a collision
  or a capture, then the corresponding right allocation interval
  $(R',i)$ is returned to the waiting interval. Let $U_{L_i'}$ denote
  the probability that $(L',i)$ has a collision or a capture. Hence,
  $U_{L_i'}$ denotes the probability that $(L',i)$ has two or more
  packets. Thus, $U_{L_i'}=P_{L_i',L_{i+1}}+P_{L_i',C_{i+1}}$. Using
  (\ref{eq:probability_Liprime_Ciplus1}) and
  (\ref{eq:probability_Liprime_Liplus1}), we obtain
  \begin{eqnarray}
    U_{L_i'}
    &=&\frac{1-\big(1+G_i\big)e^{-G_i}}{1-(1+G_{i-1})e^{-G_{i-1}}}
    \;\; \forall \;\; i \geqslant 2.
  \end{eqnarray}
\end{enumerate}

In either case, the fraction of the original allocation interval
returned on such a collision or a capture is $2^{-i}$. Therefore, the
expected value of $F$ is
\begin{eqnarray}
E[F] &=& \sum_{i=1}^{\infty} (Q_{L_i}U_{L_i}+Q_{L_i'}U_{L_i'})2^{-i},
\label{eq:expected_fraction}
\end{eqnarray}
where we assume $U_{L_1'}=0$.

From (\ref{eq:expected_packets}), (\ref{eq:expected_slots}) and
(\ref{eq:expected_fraction}), we observe that $E[K]$ and $E[F]$ are
functions only of the product $\lambda \phi_0$.  Note that as $i
\rightarrow \infty$, $G_i = 2^{-i}\lambda\phi_0 \rightarrow 0$.  Using
the Taylor series expansion for $e^x$ or L'H\^{o}pital's Rule, we can
easily prove that:
\begin{enumerate}

\item
  \begin{eqnarray}
    \lim_{i \rightarrow \infty } P_{L_i',R_i'} &=& \frac{1}{2}, 
       \label{eq:limit_Liprime_Riprime} \\
    \lim_{i \rightarrow \infty } P_{L_i',L_{i+1}'} &=& \frac{1}{4}, 
      \label{eq:limit_Liprime_Liplus1prime} \\
    \lim_{i \rightarrow \infty } P_{L_i',C_{i+1}} &=& \frac{1}{8}, 
      \label{eq:limit_Liprime_Ciplus1} \\
    \lim_{i \rightarrow \infty } P_{L_i',L_{i+1}} &=& \frac{1}{8},
      \label{eq:limit_Liprime_Liplus1}
  \end{eqnarray}

\item
  \begin{eqnarray}
    \lim_{i \rightarrow \infty } P_{R_i',R_0} &=& 1, 
      \label{eq:limit_Riprime_R0} \\
    \lim_{i \rightarrow \infty } P_{R_i',C_{i+1}} &=& 0, \\
    \lim_{i \rightarrow \infty } P_{R_i',L_{i+1}} &=& 0,
  \end{eqnarray}

\item
  \begin{eqnarray}
    \lim_{i \rightarrow \infty } P_{L_i,R_i} &=& 0, \\
    \lim_{i \rightarrow \infty } P_{L_i,L_{i+1}'} &=& \frac{1}{2}, \\
    \lim_{i \rightarrow \infty } P_{L_i,C_{i+1}} &=& \frac{1}{4}, \\
    \lim_{i \rightarrow \infty } P_{L_i,L_{i+1}} &=& \frac{1}{4},
  \end{eqnarray}

\item
  \begin{eqnarray}
    \lim_{i \rightarrow \infty } P_{R_i,C_{i+1}} &=& \frac{1}{2}, \\
    \lim_{i \rightarrow \infty } P_{R_i,L_{i+1}} &=& \frac{1}{2}.
      \label{eq:limit_Ri_Liplus1}
  \end{eqnarray}

\end{enumerate}
The proofs of these results are given in Appendix
\ref{ap:limiting_probabilities}.  Hence, $Q_{L_i}$, $Q_{L_i'}$,
$Q_{R_i'}$ and $Q_{C_i}$ tend to zero with increasing $i$ as $2^{-i}$,
while $Q_{R_i}$ tends to zero with increasing $i$ as $4^{-i}$. Thus,
$E[K]$ and $E[F]$ can be easily evaluated numerically as functions of
$\lambda\phi_0$.

Define the time backlog to be the difference between the current time
and the left endpoint of the allocation interval, i.e., $k-T(k)$.
Note that all packets that arrived in the interval $T(k),k$ have not
yet been successfully transmitted, i.e., they are backlogged.
Moreover, we define the drift $D$ to be the expected change in time
backlog, $k-T(k)$, over a CRP, assuming an initial allocation interval
of $\phi_0$. Thus, $D$ is the expected number of slots in a CRP less
the expected change in $T(k)$, and is given by
\begin{eqnarray}
D &=& E[K] - \phi_0(1-E[F]).
\end{eqnarray}
The drift is negative if $E[K]<\phi_0(1-E[F])$. Equivalently, the drift
is negative if
\begin{eqnarray}
\lambda &<& \frac{\lambda\phi_0(1-E[F])}{E[K]} =: \zeta.
\label{ineq:negative_drift}
\end{eqnarray}

\begin{figure}[thbp]
\centering
\includegraphics[width=5.5in]{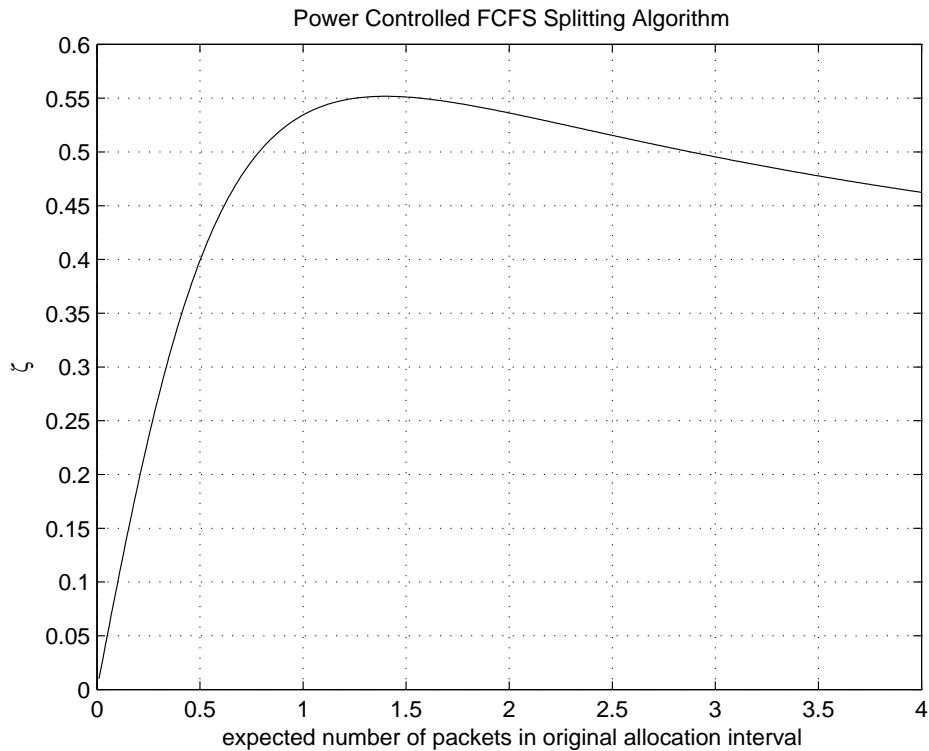}
\caption{Plot of $\zeta$ versus $\lambda\phi_0$.}
\label{fig:zeta_plot}
\end{figure}

The right hand side of (\ref{ineq:negative_drift}), $\zeta$, is a
function of $\lambda\phi_0$ and is plotted in Figure
\ref{fig:zeta_plot}.  We observe that $\zeta$ takes its maximum value
at $\lambda\phi_0=1.4$.  More precisely, $\zeta$ has a numerically
evaluated maximum of $0.5518$ at $\lambda\phi_0=1.4$.  If $\phi_0$ is
chosen to be $\frac{1.4}{0.5518}=2.54$, then
(\ref{ineq:negative_drift}) is satisfied for all $\lambda < 0.5518$.
Thus, the expected time backlog decreases whenever it is initially
larger than $\phi_0$, and we infer that the algorithm is stable for
$\lambda < 0.5518$. We have therefore proved the following result.

\begin{proposition}
The maximum stable throughput of the PCFCFS algorithm is 0.5518.
\end{proposition}

\section{Numerical Results}
\label{sec:results_pcfcfs}

In our numerical experiments, we use values of system parameters that
are commonly encountered in wireless networks
\cite{kim_lim_hou__improving_spatial}.  We compare the performance of
the following algorithms:
\begin{enumerate}
\item 
  FCFS with uniform power $P_1$,
\item
  PCFCFS.
\end{enumerate}

For each algorithm, the value of the initial allocation interval is
chosen so as to achieve maximum stable throughput.  For FCFS, maximum
stable throughput occurs when its initial allocation interval,
$\alpha_0=2.6$ \cite{bertsekas_gallager__data_networks}.  From Section
\ref{sec:throughput_pcfcfs}, the maximum throughput of PCFCFS occurs
at $\phi_0=2.54$.  Let $n_{suc}$ denote the number of successful
packets in $[0,\tau)$ and $d_i$ denote the departure time of $i^{th}$
packet.

For a given set of system parameters, we compute the following
performance metrics:
\begin{eqnarray}
\mbox{Throughput} &=& \frac{n_{suc}}{\tau}, \\
\mbox{Average Delay} &=& \frac{\sum_{i=1}^{n_{suc}}(d_i-a_i)}{n_{suc}}, \\
\mbox{Average Power} &=& 
 \frac{\sum_{i=1}^{n_{suc}}\sum_{k=\lceil a_i\rceil}^{d_i}P_i(k)}{n_{suc}}.
\end{eqnarray}
Keeping all other parameters fixed, we observe the effect of increasing the
arrival rate on the throughput, average delay and average power.

\begin{table}[tbhp]
\centering
\begin{tabular}{|l|l|l|} \hline
  Parameter & Symbol & Value \\ \hline
  communication threshold & $\gamma_c$ & 7 dB \\ \hline
  noise power spectral density & $N_0$ & -90 dBm \\ \hline
  path loss exponent & $\beta$ & 4 \\ \hline
  transmitter-receiver distance & $D$ & 100 m \\ \hline
  initial allocation interval of FCFS & $\alpha_0$ & 2.6 s \\ \hline
  initial allocation interval of PCFCFS & $\phi_0$ & 2.54 s \\ \hline
  algorithm operation time & $\tau$ & $3 \times 10^5$ s \\ \hline
\end{tabular}
\caption{System parameters for performance evaluation of PCFCFS and FCFS
algorithms.}
\label{tab:system_parameters_pcfcfs}
\end{table}

\begin{figure}[thbp]
\centering
\includegraphics[width=5in]{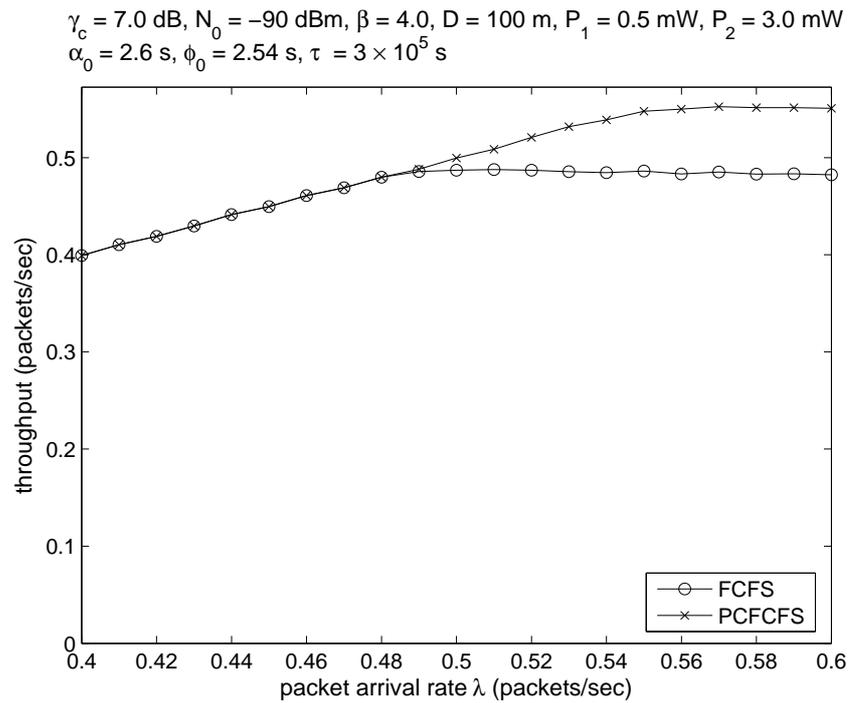}
\caption{Throughput versus arrival rate for PCFCFS and FCFS algorithms.}
\label{fig:throughput_pcfcfs}
\end{figure}

\begin{figure}[thbp]
\centering
\includegraphics[width=5in]{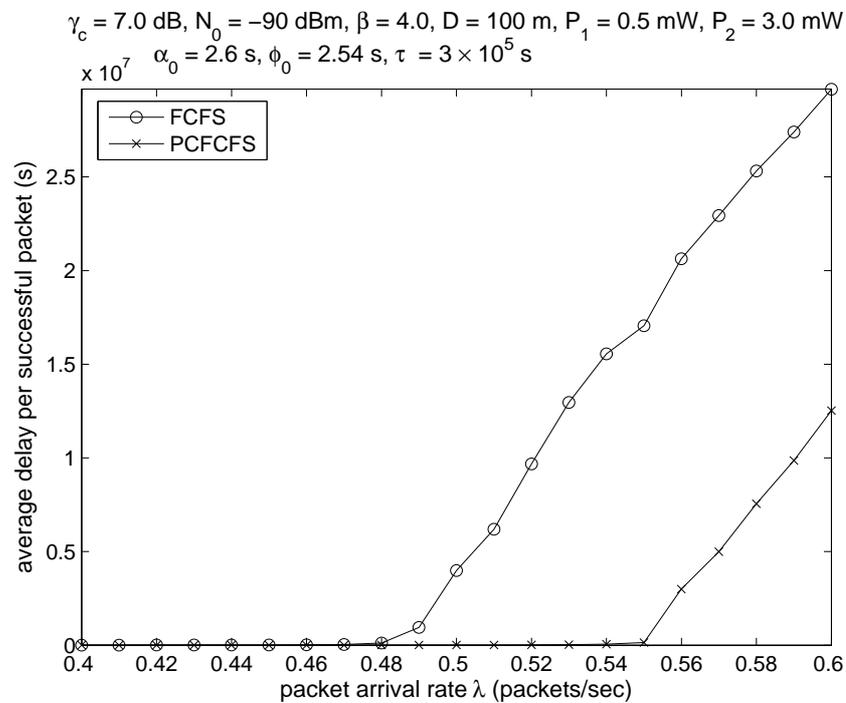}
\caption{Average delay versus arrival rate for PCFCFS and FCFS algorithms.}
\label{fig:delay_pcfcfs}
\end{figure}

\begin{figure}[thbp]
\centering
\includegraphics[width=5in]{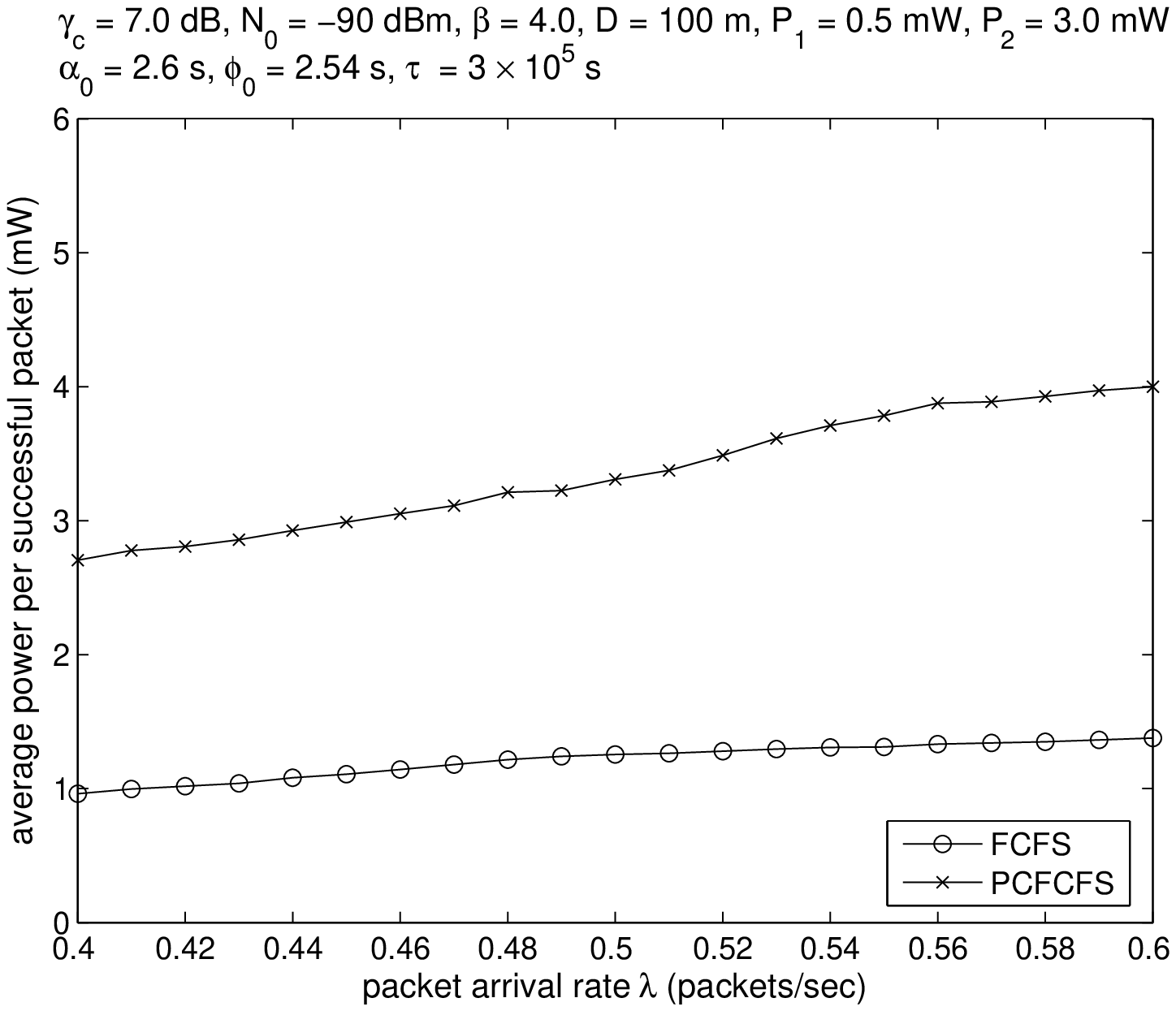}
\caption{Average power versus arrival rate for PCFCFS and FCFS algorithms.}
\label{fig:power_pcfcfs}
\end{figure}

The system parameters for our numerical experiments are shown in Table
\ref{tab:system_parameters_pcfcfs}.  From (\ref{eq:nominal_power}) and
(\ref{eq:higher_power}), we obtain $P_1=0.50$ mW and $P_2=3.01$ mW. We
vary the arrival rate $\lambda$ from $0.40$ to $0.60$ packets/s in
steps of $0.01$. Figure \ref{fig:throughput_pcfcfs} plots the
throughput versus arrival rate for the PCFCFS and FCFS algorithms.
Figure \ref{fig:delay_pcfcfs} plots the average delay per successful
packet versus arrival rate for both the algorithms.  Finally, Figure
\ref{fig:power_pcfcfs} plots the average power per successful packet
versus arrival rate for both the algorithms.

For arrival rates exceeding $0.56$, the throughput of PCFCFS is less
than the arrival rate (Figure \ref{fig:throughput_pcfcfs}) and the
average delay of PCFCFS increases rapidly (Figure
\ref{fig:delay_pcfcfs}), which leads to a substantial increase in the
number of backlogged packets and system instability.  Hence, the
maximum stable throughput of PCFCFS is between $0.55$ and $0.56$.
Thus, Figures \ref{fig:throughput_pcfcfs} and \ref{fig:delay_pcfcfs}
corroborate our result that the maximum stable throughput of PCFCFS is
$0.5518$ (see Section \ref{sec:throughput_pcfcfs}).

For both PCFCFS and FCFS, the departure rate (throughput) equals the
arrival rate for all arrival rates up to 0.487 (Figure
\ref{fig:throughput_pcfcfs}). Hence, both these algorithms are stable
for arrival rates below 0.487. For arrival rates exceeding 0.487, the
departure rate of FCFS is strictly lower than its arrival rate,
leading to packet backlog and system instability. On the other hand,
for PCFCFS, the departure rate still equals its arrival rate for
arrival rates between 0.487 and 0.5518.  In other words, the PCFCFS
algorithm is stable for a higher range of arrival rates compared to
FCFS algorithm. However, the PCFCFS algorithm becomes unstable for
arrival rates exceeding 0.5518.

The PCFCFS algorithm achieves higher throughput and lower average
delay than the FCFS algorithm, albeit at the cost of expending higher
average power.  For example, at $\lambda=0.55$, PCFCFS achieves
$13.3\%$ higher throughput and $96.7\%$ lower average delay than FCFS,
at the cost of $170\%$ higher power.

\section{Conclusions}
\label{sec:conclusions_pcfcfs}

In this chapter, we have considered random access in wireless networks
under the physical interference model.  By recognizing that the
receiver can successfully decode the strongest packet in presence of
multiple transmissions, we have proposed PCFCFS, a splitting algorithm
that modulates transmission powers of users based on observed channel
feedback.  PCFCFS achieves higher throughput and substantially lower
delay than those of the well known FCFS algorithm with uniform
transmission power. We show that the maximum stable throughput of
PCFCFS is 0.5518.  PCFCFS can be implemented in those scenarios where
users are willing to trade some power for a substantial gain in
throughput.  Moreover, if users can estimate the arrival rate of
packets, then they can employ FCFS algorithm for arrival rates up to
0.4871 and PCFCFS algorithm for higher arrival rates, thus leading to
further reduction in average transmission power.

\clearpage{\pagestyle{empty}\cleardoublepage}

\chapter{Flow Control: An Information Theory Viewpoint}
\label{ch:flow_control}

This thesis has so far explored various aspects of link scheduling in
wireless networks. An equally interesting problem is to analyze flow
control.  We formulate the problem of controlling the rate of packets
at the ingress of a packet network (possibly a wireless link) so as to
maximize the mutual information between a source and a destination.
We discuss various nuances of the problem and describe related work.
We then derive the maximum entropy of a packet level flow that
conforms to linearly bounded traffic constraints, by taking into
account the covert information present in the randomness of packet
lengths. Our results provide insights towards the design of flow
control mechanisms employed by an Internet Service Provider (ISP).

The rest of the chapter is organized as follows.  In Section
\ref{sec:general_system_model}, we define the problem of information
theoretic analysis of flow control in a packet network.  In Section
\ref{sec:gtbr}, we introduce a Generalized Token Bucket Regulator
(GTBR) as our flow control mechanism.  The concepts of flow entropy
and information utility are defined in Section
\ref{sec:information_utility}. We formulate the problem of determining
the GTBR with maximum information utility in Section
\ref{sec:problem_formulation_gtbr}. In Section \ref{sec:results_gtbr},
we derive a necessary condition for the optimal GTBR and compute its
parameters. We explain the results from an information theoretic
viewpoint in Section \ref{sec:info_theoretic_interpret} and discuss
the implications of our work in Section
\ref{sec:discussion_info_theoretic}.

\section{System Model}
\label{sec:general_system_model}

\begin{figure}[thbp]
  \centering
  \includegraphics[width=6in]{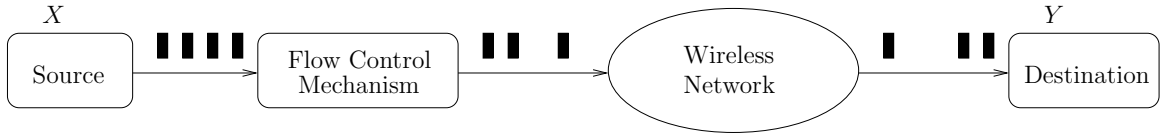}
  \caption{Flow control of a source's packets over a packet network.}
  \label{fig:flow_control}
\end{figure}

Our system model is shown in Figure \ref{fig:flow_control}, wherein a
source sends packets to a destination over a packet-switched network
(possibly a wireless network).  The packets transmitted by the source
are regulated (policed) by a flow control mechanism at the ingress of
the network. We are interested in the packet probability distribution
that maximizes the mutual information between the source and the
destination. In other words, given the description of the flow control
mechanism and the stochastic characterization of the packet network,
we seek the maximum amount of information (in the Shannon sense) that
can be transmitted from the source to the destination.

The problem can be stated as:
\begin{eqnarray}
  \max_{\mathbf p_X} I(X;Y),
  \label{eq:mutual_info}
\end{eqnarray}
where
\begin{eqnarray*}
X &=& 
\mbox{random variable representing randomness in packet contents, lengths}\\ 
&& \mbox{and timings at the source}, \\
y &=& 
\mbox{random variable representing randomness in packet contents, lengths}\\ 
&& \mbox{and timings at the destination}, \\
\mathbf p_X 
 &=& \mbox{probability distribution of $X$}.
\end{eqnarray*}
(\ref{eq:mutual_info}) can be simplified to:
\begin{eqnarray}
  \max_{\mathbf p_X} \big(H(X)-H(X|Y)\big).
  \label{eq:entropy_diff}
\end{eqnarray}
Thus, to maximize the information transfer from the source to the
destination, we not only have to characterize the entropy of the
source's packets $H(X)$, but also the conditional entropy of the
source's packets given the packets received at the destination
$H(X|Y)$.

We state the following remarks about our problem formulation:
\begin{enumerate}

\item It is well-known that, in a packet-switched network, information
  can be transmitted not only by the contents, but also by the lengths
  and timings of packets.  \cite{gallager__basic_limits} is perhaps
  the first work to recognize this fact. Information transmitted by
  the lengths and timings of packets is referred to as {\em covert
    information} or {\em side information}. The channel that is used
  to convey covert information is called {\em covert channel}. Covert
  channels have been investigated in \cite{gallager__basic_limits},
  \cite{anantharam_verdu__bits_through},
  \cite{shah_karandikar__information_utility}.

\item By flow control, we mean a rate control mechanism that regulates
  the packets transmitted by a source (subscriber) at the ingress of a
  network. Note that we do not consider end-to-end flow control
  mechanisms such as Transmission Control Protocol (TCP).  For
  simplicity, we consider a flow control mechanism that is described
  by a linearly bounded service curve\footnote{Consider a flow through
    a system $\mathcal S$ with input and output functions $A(t)$ and
    $B(t)$ respectively. $\mathcal S$ offers to the flow a service
    curve $\vartheta(t)$ if and only if $\vartheta(t)$ is a wide sense
    increasing function, with $\vartheta(0)=0$, and $B(t) \geqslant
    \inf_{s \leqslant t}\{A(s)+\vartheta(t-s)\}$ for all $t \geqslant
    0$.}  \cite{boudec_thiran__network_calculus}.

\item In the packet network, packets can be received incorrectly at
  the destination due to fluctuations in the channel, like that in a
  wireless channel.  We assume the existence of link layer mechanisms
  such as Forward Error Correction (FEC) which ensure that all packets
  are correctly received at the destination.

\end{enumerate}

The packet network shown in Figure \ref{fig:flow_control} only
guarantees that the contents and lengths of the packets transmitted by
the source are the same as those at the destination.  However, the
network can arbitrarily vary the timings between packets.
Equivalently, the network can highly distort the covert timing
information carried by the packets.

Taking a cue from this, we only take into account information that is
carried by the contents and lengths of the packets. Consequently, the
probability distribution of packet contents and lengths at the
destination is the same as that at the source. Hence, $H(X|Y)=0$
and (\ref{eq:entropy_diff}) simplifies to
\begin{eqnarray}
  \max_{\mathbf p_X} H(X).
  \label{eq:source_H}
\end{eqnarray}
In other words, we seek the probability distribution of packet
contents and lengths that maximize the source entropy $H(X)$.

Typically, the entity that owns a network, say an Internet Service
Provider (ISP), implements certain mechanisms to ensure that packets
transmitted by a subscriber are not lost in the network.  However, to
allocate network resources efficiently and guarantee zero loss of
packets, the entity also mandates that the aggregate traffic of a
subscriber be upper bounded by an envelope or a service curve. For
example, the entity can mandate that the aggregate traffic of the
subscriber be {\em linearly bounded.}  A linearly bounded service
curve can be implemented by a class of regulators known as token
bucket regulators.

\begin{figure}[thbp]
  \centering
  \includegraphics[width=6in]{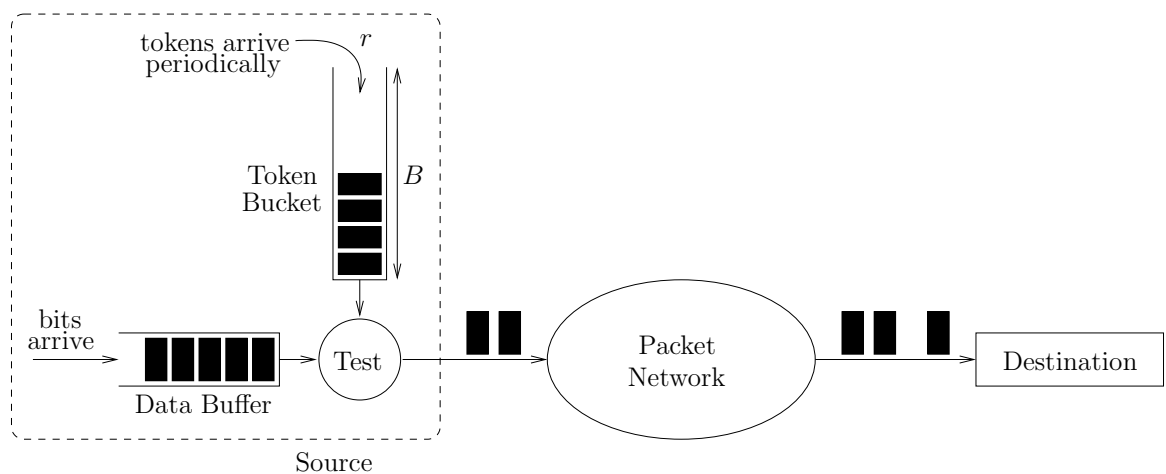}
  \caption{Token bucket regulation of a source's packets over a packet
    network.}
  \label{fig:sys_mod_tbr}
\end{figure}

The system model that we analyze incorporates a Token Bucket Regulator
(TBR) and is shown in Figure \ref{fig:sys_mod_tbr}.  A source
transmits packets to a destination over a network, where every packet
consists of an integer number of bits.  The packets transmitted by the
source are regulated by a TBR or leaky bucket regulator
\cite{keshav__engineering_approach}. Intuitively, the regulator
collects tokens in a bucket of depth $B$, which fills up at a certain
rate $r$.  Each token corresponds to the permission to transmit one
bit into the network. The packets to be transmitted by the source
accumulate in its data buffer over time. If there is a packet of
length $n$ bits in the data buffer at a given time, then it can be
sent into the network only if $n \leqslant B+r$. If the packet is
transmitted, then $n$ tokens are depleted from the token bucket.

A TBR can be used to smoothen the bursty nature of a subscriber's
traffic. We assume that the network is owned by an ISP.  From a
Quality of Service (QoS) perspective, a TBR can be considered to be a
part of the Service Level Agreement (SLA) between a subscriber and an
ISP.  The SLA mandates that the ISP should provide end-to-end loss and
delay guarantees to a subscriber's packets, provided the traffic
profile of the subscriber adheres to certain TBR constraints.
Specifically, the onus of the ISP is to ensure that every packet of a
conforming source successfully reaches its destination within a
certain permissible delay.

The Standard Token Bucket Regulator (STBR), as defined by the Internet
Engineering Task Force (IETF) and shown in Figure
\ref{fig:sys_mod_tbr}, enforces linear-boundedness on the flow.  An
STBR is characterized by its token increment rate $r$ and bucket depth
$B$.  We will be more general and consider a TBR in which the token
increment rate and bucket depth (maximum burst size) can vary from
slot to slot. Such a TBR, which we define as a Generalized Token
Bucket Regulator (GTBR), can be used to regulate Variable Bit Rate
(VBR) traffic\footnote{For example, a pre-recorded video stream.} from
a source \cite{shah_karandikar__optimal_packet}.  The continuous-time
analogue of a GTBR is the time-varying leaky bucket shaper
\cite{giordano_boudec__class_time} in which the token rate and bucket
depth parameters can change at specified time instants.

The idea is to develop the notion of information utility of a GTBR.
Specifically, we derive the maximum information that a GTBR-conforming
traffic flow can convey in a finite time interval, by taking into
account the additional information present in the randomness of packet
lengths.  These aspects are further elucidated in subsequent sections.

\section{Generalized Token Bucket Regulator}
\label{sec:gtbr}

In this section, we mathematically describe our system model and
define a GTBR.  We also explain the differences between our system
model and those considered in existing literature.

\begin{figure}[thbp]
\centering
\includegraphics[width=3in]{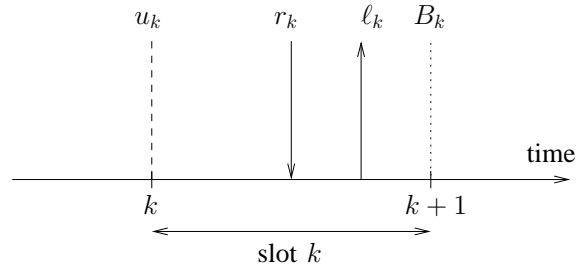}
\caption{Relative time instants of  parameters defined in (\ref{eq:gtbr_notation}).}
\label{fig:relative_instants}
\end{figure}

Consider a system in which time is divided into slots and a source
which has to complete its data transmission within $S$ slots. In our
discrete-time model, we will evaluate the system at time instants
$0,1,\ldots,S-1,S$. Slot $k$ is defined to be the time interval
$[k,k+1)$, i.e., data transmission commences with slot 0 and
terminates with slot $(S-1)$. 

The traffic from the source is regulated by a GTBR.  We define:
\begin{eqnarray}
r_k    &=& \mbox{token increment for slot $k$}, \nonumber \\
B_k    &=& \mbox{bucket depth for slot $(k+1)$}, \nonumber \\
\ell_k &=& \mbox{length of packet (in bits) transmitted in slot $k$}, 
            \nonumber \\
u_k    &=& \mbox{residual  tokens at start of slot $k$}.
\label{eq:gtbr_notation}
\end{eqnarray}
$r_k$, $B_k$, $\ell_k$ and $u_k$, whose relative time instants are
shown in Figure \ref{fig:relative_instants}, are all non-negative
integers.  Let ${\mathbf r} := (r_0,r_1,\ldots,r_{S-1})$ denote the
token increment sequence and ${\mathbf B} := (B_0,B_1,\ldots,B_{S-2})$
denote the bucket depth sequence.  The system starts with zero tokens.
So, $u_0=0$.  A GTBR with the above parameters is denoted as $\mathcal
R_g(S,\mathbf r,\mathbf B)$.

The constraints imposed by $\mathcal R_g(S,\mathbf r,\mathbf B)$ on
the packet lengths is
\begin{eqnarray}
\ell_i \leqslant u_i + r_i \; \; \forall \;\; i= 0,1,\ldots,S-1.
\label{eq:conforming_packet_lengths}
\end{eqnarray}
If (\ref{eq:conforming_packet_lengths}) is satisfied, then $\mathbf
\ell=(\ell_0,\ell_1,\ldots,\ell_{S-1})$ is a conforming packet length
vector and the number of residual tokens will evolve according to
\begin{eqnarray}
u_0 &=& 0, \nonumber \\
u_{i+1} &=& \min(u_i + r_i - \ell_i, B_i) \; \; \forall \; \;
    i= 0,1,\ldots,S-2, \nonumber \\
u_S &=& u_{S-1}+r_{S-1}-\ell_{S-1}.
\label{eq:token_evolution}
\end{eqnarray}
(\ref{eq:token_evolution}) is referred to as the token evolution
equation.

Note that if $r_i=r$ $\forall$ $i=0,1,\ldots,S-1$ and $B_i=B$
$\forall$ $i=0,1,\ldots,S-2$, then the GTBR $\mathcal R_g(S,\mathbf
r,\mathbf B)$ degenerates to the STBR $\mathcal R_s(S,r,B)$.

We should point out that our system model is similar to that of
\cite{chang_chao_thomas__fundamental_limits}. However, unlike
\cite{chang_chao_thomas__fundamental_limits}, our traffic regulator
is a deterministic mapping of an input sequence to an output sequence.
Also, the rate of our regulator is defined by the average token
increment rate and not by the peak rate.

The system model encompasses that of
\cite{shah_karandikar__information_utility}, wherein the authors have
derived the information utility of an STBR and suggested a pricing
viewpoint for its application. Our interest, however, is more
theoretical.  Specifically, we consider an STBR as a special case of a
GTBR and describe a framework for their information-theoretic
comparison. The main objective is to investigate whether a GTBR can
achieve higher flow entropy than an STBR and explain the properties of
entropy-maximizing GTBRs. These aspects are addressed in the following
sections.

\section{Notion of Information Utility}
\label{sec:information_utility}

In this section, we introduce the concept of information utility of a
GTBR. We derive the entropy of a flow that is regulated by a GTBR by
considering the information present in the contents and lengths of the
packets. We formulate the problem of computing the maximum flow
entropy and subsequently describe a technique to compute the
information utility of the GTBR.

Consider a source which has a large amount of data to send and whose
traffic is regulated by a GTBR. We seek to maximize the information
that the source can convey to the destination in the given time
interval or the entropy present in the source traffic flow in an
information-theoretic sense.  For a given transmission interval $S$,
token increment sequence $\mathbf r$ and bucket depth sequence
$\mathbf B$, the maximum entropy achievable by any flow which is
constrained by the GTBR $\mathcal R_g(S,\mathbf r,\mathbf B)$ is
termed as the {\em information utility} of the GTBR $\mathcal
R_g(\cdot)$.

The source can send information to its destination via two channels:
\begin{enumerate}
\renewcommand{\theenumi}{\roman{enumi}}
\item Overt channel: The contents of each packet.  Let $\ell_i$ denote
  the length of a packet in bits.  The value of each bit is $0$ or $1$
  with equal probability and is independent of the values taken by the
  preceding and succeeding bits.  Thus, this packet contributes
  $\ell_i$ bits of information.
  \label{it:overt}

\item Covert channel: We consider the length of a packet as an event
  and associate a probability with it.  Thus, side information is
  transmitted by the randomness in the packet lengths.
  \label{it:covert}
\end{enumerate}
The joint entropy of \ref{it:overt}. and \ref{it:covert}. is the sum of
their entropies.

During any slot $k$, the only method by which past transmissions can
constrain the rest of the flow is by the residual number of tokens
$u_k$.  So, $u_k$ captures the state of the system.  The key
observation is that the future entropy depends only on the token
bucket level $u_k$ in slot $k$.  Hence, entropy is a function of
system state $u_k$ and is denoted by $H_k(u_k)$.

During slot $S$, the source signals the termination of the current
flow by transmitting a special string of bits (flag).  The information
transmitted by this fixed sequence of bits is zero. Thus
\begin{eqnarray}
H_S(u_S) &=& 0 \; \; \forall \; u_S.
\label{eq:entropy_slot_N}
\end{eqnarray}

For a given state $u_k$ of the system, if a packet of length $\ell_k$
bits is transmitted with probability $p_{\ell_k}(u_k)$, then
\begin{enumerate}
\item The overt information transmitted is $\ell_k$ bits,

\item As the event occurs with probability $p_{\ell_k}(u_k)$, the
  covert information transmitted is $(-\log_2 p_{\ell_k}(u_k))$ bits,

\item Since $\ell_k$ is random, $u_{k+1}$ is also random (from
  (\ref{eq:token_evolution})).  Thus, $H_{k+1}(u_{k+1})$ is also a
  random variable.
\end{enumerate}
Adding all of the above and averaging it over all conforming packet
lengths, we obtain the entropy in the current slot (stage)
\begin{eqnarray}
H_k(u_k) &=& \sum_{\ell_k=0}^{u_k+r_k} p_{\ell_k}(u_k) \Big( \ell_k 
- \log_2\big(p_{\ell_k}(u_k)\big) 
+  H_{k+1}\big(\min(u_k+r_k-\ell_k,B_k)\big) \Big) 
  \nonumber \\
&& \; \; \forall \; k=0,\ldots,S-1.
\label{eq:flow_entropy}
\end{eqnarray}
The equation above, which will be referred to as the {\em flow entropy
  equation}, intuitively states that the flow entropy of the current
state is given by the sum of the entropy of the packet contents, the
entropy of the packet lengths and the flow entropy of possible future
states in the next slot. Note that (\ref{eq:flow_entropy}) is similar
to the backward recursion equation from dynamic programming
\cite{bertsekas__dynamic_programming}.  Finally, the packet length
probabilities must satisfy
\begin{eqnarray}
\sum_{\ell_k=0}^{u_k+r_k} p_{\ell_k}(u_k) 
&=& 1 \; \; \forall \; k=0,\ldots,S-1.
\label{eq:probability_sum_unity}
\end{eqnarray}
Let ${\mathbf p}_k(u_k) =
\left(p_0(u_k),p_1(u_k),\cdots,p_{u_k+r_k}(u_k)\right)$ denote the
vector of packet length probabilities for slot $k$ with $u_k$ residual
tokens.  The dependence of $p_{\ell_k}$ and ${\mathbf p}_k$ on $u_k$
is assumed to be understood and is not always stated explicitly. So,
${\mathbf p}_k = \left(p_0,p_1,\cdots,p_{u_k+r_k}\right)$.

Our objective is to determine the sequence of probability mass
functions $({\mathbf p}_{S-1}^*$, ${\mathbf p}_{S-2}^*$, $\cdots$,
${\mathbf p}_0^*)$ which maximizes the flow entropy $H_0(0)$ for a
given GTBR $\mathcal R_g(S,\mathbf r,\mathbf B)$.  From
(\ref{eq:entropy_slot_N})
\begin{eqnarray}
H_S^*(u_S) &=& 0.
\label{eq:optimal_entropy_N}
\end{eqnarray}
From (\ref{eq:flow_entropy})
\begin{eqnarray}
H_k(u_k) 
&=& \sum_{\ell_k=0}^{u_k+r_k} p_{\ell_k} \Big( \ell_k - \log_2(p_{\ell_k}) 
  + H_{k+1}^*\big(\min(u_k+r_k-\ell_k,B_k)\big) \Big) \nonumber \\
&& \; \; \forall \; k=0,1, \ldots S-1.
\label{eq:intermediate_Hk} 
\end{eqnarray}
Given $H_{k+1}^*(u_{k+1})$ $\forall$ $u_{k+1}$, there exists an
optimum probability vector ${\mathbf p}_k^* =
(p_0^*,p_1^*,\ldots,p_{u_k+r_k}^*)$ which maximizes the flow entropy
$H_k(u_k)$, i.e.,
\begin{eqnarray}
H_k^*(u_k) 
&=& \sum_{\ell_k=0}^{u_k+r_k} p_{\ell_k}^* \Big( \ell_k - \log_2(p_{\ell_k}^*)
  + H_{k+1}^*\big(\min(u_k+r_k-\ell_k,B_k)\big) \Big) \nonumber \\
&& \; \; \forall \; k=0,1,\ldots, S-1.
\label{eq:optimal_entropy_k}
\end{eqnarray}

Thus, the problem of computing the entire sequence of probability
vectors $(\mathbf p_{S-1}^*$, $\mathbf p_{S-2}^*$, $\cdots$, $\mathbf
p_0^*)$ decouples into a sequence of subproblems.  The subproblem for
slot $k$ is: Given the function $H_{k+1}^*(u_{k+1}) \; \forall \;
u_{k+1}$, determine the probability vector ${\mathbf p}_k =
(p_0,p_1,\ldots,p_{u_k+r_k})$ so as to
\begin{eqnarray}
\mbox{maximize} & & \sum_{\ell_k=0}^{u_k+r_k} p_{\ell_k} \Big( \ell_k - 
  \log_2(p_{\ell_k}) 
  + H_{k+1}^*\big(\min(u_k+r_k-\ell_k,B_k)\big) \Big), \nonumber \\
\mbox{subject to} & & \sum_{\ell_k=0}^{u_k+r_k} p_{\ell_k}  = 1.
\label{eq:optimize_stagewise}
\end{eqnarray}

(\ref{eq:optimize_stagewise}) is an equality-constrained optimization
problem and can be solved using the technique of Lagrange multipliers
\cite{bertsekas__nonlinear_programming}.  Define the Lagrangian
\begin{eqnarray}
{\mathcal L}({\mathbf p}_k,\lambda_k) 
 &=& \sum_{\ell_k=0}^{u_k+r_k} p_{\ell_k} 
\Big( \ell_k - \log_2(p_{\ell_k}) + H_{k+1}^*\big(\min(u_k+r_k-\ell_k,B_k)\big)
 \Big) \nonumber \\
&& {} + \lambda_k \bigg( \sum_{\ell_k=0}^{u_k+r_k} p_{\ell_k} - 1 \bigg).
\end{eqnarray}
At the optimal point $({\mathbf p}_k^*,\lambda_k^*)$, we must have
\begin{eqnarray}
\left. \frac{\partial{\mathcal L}}{\partial p_{\ell_k}} \right| 
\begin{array}{c} \\ ({\mathbf p}_k^*,\lambda_k^*) \end{array}
&=& 0 \; \; \; \;
\forall \; \; 0 \leqslant \ell_k \leqslant u_k+r_k,
\label{eq:derivative_probability} \\
\left. \frac{\partial{\mathcal L}}{\partial \lambda_k} \right| 
\begin{array}{c} \\ ({\mathbf p}_k^*,\lambda_k^*) \end{array} 
&=& 0.
\label{eq:derivative_lagrange}
\end{eqnarray}
Solving (\ref{eq:derivative_lagrange}) yields
\begin{eqnarray}
\sum_{\ell_k=0}^{u_k+r_k} p_{\ell_k}^*(u_k) &=& 1,
\label{eq:optimum_probabilities_sum}
\end{eqnarray}
which is (\ref{eq:probability_sum_unity}) for the case of optimal
probabilities.  Solving (\ref{eq:derivative_probability}), we obtain
\begin{eqnarray}
p_{\ell_k}^*(u_k) 
 &=& 2^{\ell_k - \log_2 e + H_{k+1}^*(\min(u_k+r_k-\ell_k,B_k))
+\lambda_k^*(u_k)}
\; \; \forall \; \; 0 \leqslant \ell_k \leqslant u_k + r_k.
\label{eq:optimum_probability}
\end{eqnarray}
From (\ref{eq:optimum_probabilities_sum}) and
(\ref{eq:optimum_probability}), the optimal Lagrange multiplier is
given by
\begin{eqnarray}
\lambda_k^*(u_k) = \log_2 e - \log_2 \bigg( \sum_{\ell_k=0}^{u_k+r_k} 
2^{\ell_k + H_{k+1}^*(\min(u_k+r_k-\ell_k,B_k))} \bigg).
\label{eq:optimum_lagrange}
\end{eqnarray}
From (\ref{eq:optimum_probability}) and (\ref{eq:optimum_lagrange}),
the optimum packet length probability is given by
\begin{eqnarray}
p_{\ell_k}^*(u_k) &=& 
\frac{2^{\ell_k + H_{k+1}^*(\min(u_k+r_k-\ell_k,B_k))}}
{\sum_{\alpha_k=0}^{u_k+r_k} 
2^{\alpha_k + H_{k+1}^*(\min(u_k+r_k-\alpha_k,B_k))}}.
\label{eq:optimum_probability_ratio}
\end{eqnarray}
From (\ref{eq:optimal_entropy_k}) and
(\ref{eq:optimum_probability_ratio}), we finally obtain
\begin{eqnarray}
H_k^*(u_k) 
 &=& \log_2 \bigg( \sum_{\ell_k=0}^{u_k+r_k} 2^{\ell_k + 
      H_{k+1}^*(\min(u_k+r_k-\ell_k,B_k))} \bigg). 
\label{eq:optimal_flow_entropy}
\end{eqnarray}
(\ref{eq:optimal_flow_entropy}) will be referred to as the {\em
  optimal flow entropy equation}.

The {\em information utility} of the GTBR $\mathcal R_g(S,\mathbf
r,\mathbf B)$ is defined to be $H_0^*(0)$, the maximum flow entropy.
$H_0^*(0)$ is computed by starting with $H_S^*(u_S)=0$, and using
(\ref{eq:optimal_flow_entropy}) to compute the optimal flow entropy
$H_k^*(u_k)$ for all $u_k$ and then proceeding backward recursively
for $k=S-1,S-2,\ldots,0$.

\section{Problem Formulation}
\label{sec:problem_formulation_gtbr}

Having developed a method to compute the information utility of a GTBR
in Section \ref{sec:information_utility}, we seek answers to the
following questions:
\begin{enumerate}
\renewcommand{\theenumi}{\alph{enumi}}
\item Can a GTBR achieve higher information utility than that of an
  STBR?

\item If yes, what is the increase in information utility?
\end{enumerate}

For the information-theoretic comparison of a
GTBR ${\mathcal R}_g(S,{\mathbf r},{\mathbf B})$ and an STBR
${\mathcal R}_s(S',r,B)$, we impose the following conditions:
\begin{enumerate}

\item $\mathcal R_g(\cdot)$ and $\mathcal R_s(\cdot)$ must operate
  over the same number of slots, i.e.,
  \begin{eqnarray}
    S &=& S'.
    \label{eq:equal_slots}
  \end{eqnarray}
  \label{cond:equal_slots}

\item The aggregate tokens of $\mathcal R_g(\cdot)$ and $\mathcal
  R_s(\cdot)$ must be equal, i.e.,
  \begin{eqnarray}
    \sum_{i=0}^{S-1} r_i &=& Sr.
    \label{eq:equal_tokens}
  \end{eqnarray}
  \label{cond:aggregate_tokens}

\item The aggregate bucket depth of $\mathcal R_g(\cdot)$ must not
  exceed that of $\mathcal R_s(\cdot)$\footnote{Equality is present in
    (\ref{eq:equal_tokens}) because every additional token directly
    translates to the permission to transmit one more bit, leading to
    increase in information utility.  As this is not necessarily true
    for bucket depth, we permit inequality in
    (\ref{ineq:bucket_depth}).}, i.e.,
  \begin{eqnarray}
    \sum_{i=0}^{S-2} {B_i} &\leqslant& (S-1)B.
    \label{ineq:bucket_depth}
  \end{eqnarray}
  \label{cond:aggregate_bursts}

\item The bucket depth of $\mathcal R_s(\cdot)$ cannot be very high
  compared to its token increment rate. To quantify this, we
  mandate\footnote{This assumption is practically justifiable. For
    example, in \cite{giordano_boudec__class_time}, the authors use
    $r=6$ Mbps and $B=12$ Mbps for their simulations.}
  \begin{eqnarray}
    2r \leqslant B \leqslant 5r.
    \label{restrictions_bucket_depth}
  \end{eqnarray}
  \label{cond:moderate_bursts}

\item The token increment rate of $\mathcal R_g(\cdot)$ in every slot
  must not exceed the bucket depth of $\mathcal R_s(\cdot)$, i.e.,
  \begin{eqnarray}
    r_i &\leqslant& B.
    \label{eq:restriction_tokens}
  \end{eqnarray}
  \label{cond:moderate_tokens}

\end{enumerate}
If Conditions \ref{cond:equal_slots}, \ref{cond:aggregate_tokens},
\ref{cond:aggregate_bursts}, \ref{cond:moderate_bursts} and
\ref{cond:moderate_tokens} are satisfied, then GTBR $\mathcal
R_g(\cdot)$ and STBR $\mathcal R_s(\cdot)$ are said to be {\em
  comparable} to each other.

The optimal GTBR problem is formally stated as:\\
Given an STBR $\mathcal R_s(S,r,B)$, determine the token increment
sequence $\mathbf r$ and bucket depth sequence $\mathbf B$ of a
comparable GTBR $\mathcal R_g(S,\mathbf r,\mathbf B)$ so as to
\begin{eqnarray}
& \mbox{maximize} & H_0^*(0), \nonumber \\
\mbox{subject to} & \sum_{i=0}^{S-1} r_i &= Sr, \label{eq:token_sum} \\
& \sum_{i=0}^{S-2} B_i &\leqslant (S-1)B. \label{ineq:burst_sizes}
\end{eqnarray}
Note that we are maximizing a real-valued function
over two finite sequences of non-negative integers. \\

\section{Results}
\label{sec:results_gtbr}

In this section, we derive a necessary condition for the optimal GTBR
in terms of aggregate bucket depth. We also compute the parameters of
the optimal GTBR for some representative cases.

\subsection{Analytical Result}

\begin{proposition}
  For an optimal GTBR, equality must hold in (\ref{ineq:burst_sizes}),
  except when $S$ is small.  In other words, if $\mathbf B^*$ is the
  bucket depth sequence of an optimal GTBR, it must satisfy
  \begin{eqnarray}
    \sum_{i=0}^{S-2} B_i^* &=& (S-1)B.
    \label{eq:burst_sizes}
  \end{eqnarray}
  \label{prop:equality_burst_sizes}
\end{proposition}

\begin{proof}
  We prove by contradiction.  Define $g_k(u)=2^{H_k^*(u)}$.  Since
  $H_k^*(u) \geqslant 0$, $g_k(u) \geqslant 1$.  From
  (\ref{eq:optimal_flow_entropy})
  \begin{eqnarray}
    g_k(u) &=& 
      \sum_{\ell=0}^{u+r_k} 2^\ell g_{k+1}\big(\min(u+r_k-\ell,B_k)\big).
    \label{eq:alphabet_size}
  \end{eqnarray}
  $g_{S-1}(u)=2^{u+r_{S-1}+1}-1$ is an increasing sequence in $u$.
  Using (\ref{eq:alphabet_size}), we can show that $g_k(u)$ is an
  increasing sequence in $u$ $\forall$ $k=0,\ldots,S-1$.  Let
  $\mu_i=$ maximum number of tokens possible in slot $i$. Thus
  \begin{eqnarray}
    \mu_0 &=& 0,\\
    \mu_i &=& \min(\mu_{i-1}+r_{i-1},B_{i-1}) \; \; \forall \; i=1,\ldots,S-1.
  \end{eqnarray}
  If $u_i \leqslant \mu_i$, then we say that state $u_i$ is {\em
    reachable} in slot $i$, otherwise it is {\em unreachable.}

  Let $\mathcal R_g(S,\mathbf r,\mathbf B)$ be an optimal GTBR, for
  which equality does not hold in (\ref{ineq:burst_sizes}). Then
  $\sum_{i=0}^{S-2} B_i \leqslant (S-1)B-1$. Consider another GTBR
  $\mathcal R_g'(S,\mathbf r',\mathbf B')$ with $\mathbf r'=\mathbf r$
  and $\mathbf B'=(B_0,\ldots,B_{k-1},B_k+1,B_{k+1},\ldots,B_{S-2})$
  for some $k$.  Let ${H_k'}^*(u)$ denote the optimal flow entropy of
  $\mathcal R_g'(\cdot)$ in slot $k$ with $u$ residual tokens. Define
  $g_k'(u)=2^{{H_k'}^*(u)}$.  From (\ref{eq:optimal_flow_entropy})
  \begin{eqnarray}
    g_k'(u) &=& 
      \sum_{\ell=0}^{u+r_k} 2^\ell g_{k+1}'\big(\min(u+r_k'-\ell,B_k')\big).
    \label{eq:alphabet_size_prime}
  \end{eqnarray}
  $\mathbf B'$ satisfies (\ref{ineq:burst_sizes}).  $g_i'(u)=g_i(u)$
  $\forall$ $i=k+1,\ldots,S$ and $\forall$ $u$.  Since
  $\min(u+r_k-\ell,B_k+1) \geqslant \min(u+r_k-\ell,B_k)$, it follows
  that $g_k(\min(u+r_k-\ell,B_k+1)) \geqslant
  g_k(\min(u+r_k-\ell,B_k)) \geqslant 1$.  If we determine a reachable
  state $u$ such that $g_k'(u) > g_k(u)$, then $g_0'(0) > g_0(0)$,
  since the flow entropy in slot $0$ is computed slot-by-slot as a
  linear sum of future possible flow entropies with positive weights.
  Thus, the problem now reduces to determining a slot $k$ and a
  reachable state $u$ such that $g_k'(u) > g_k(u)$.  One of the
  following must hold:
  \begin{enumerate}
  \item There exists an $i$ $\in$ $\{1,\ldots,S-1\}$ such that
    $\mu_i=B_{i-1}<\mu_{i-1}+r_{i-1}$, or \label{exists_stage}

  \item There is no $i$ such that $\mu_i=B_{i-1}<\mu_{i-1}+r_{i-1}$.
    \label{non_existent_stage}
  \end{enumerate}

  Case \ref{exists_stage}: Consider the smallest $i$ such that
  $\mu_i=B_{i-1}<\mu_{i-1}+r_{i-1}$. Substituting $k=i-1$ in
  (\ref{eq:alphabet_size}), we obtain
  \begin{eqnarray}
    g_{i-1}(u) 
    &=& \sum_{\ell=0}^{u+r_{i-1}} 2^\ell 
         g_i\big(\min(B_{i-1},u+r_{i-1}-\ell)\big),
      \nonumber \\
    \therefore g_{i-1}(u)
    &=& \sum_{\ell=0}^{u+r_{i-1}-B_{i-1}-1} 2^\ell g_i(B_{i-1}) 
    + \sum_{\ell=u+r_{i-1}-B_{i-1}}^{u+r_{i-1}} 2^\ell g_i(u+r_{i-1}-\ell). 
    \label{g_iminus1}
  \end{eqnarray}
  Substituting $k=i-1$ in (\ref{eq:alphabet_size_prime}), we obtain
  \begin{eqnarray}
    g_{i-1}'(u) 
    &=& \sum_{\ell=0}^{u+r_{i-1}} 2^\ell 
    g_i\big(\min(B_{i-1}+1,u+r_{i-1}-\ell)\big),
    \nonumber \\
    \therefore g_{i-1}'(u)
    &=& \sum_{\ell=0}^{u+r_{i-1}-B_{i-1}-1} 2^\ell g_i(B_{i-1}+1) 
    + \sum_{\ell=u+r_{i-1}-B_{i-1}}^{u+r_{i-1}} 2^\ell g_i(u+r_{i-1}-\ell).
    \label{gprime_iminus1}
  \end{eqnarray}
  (\ref{g_iminus1}) and (\ref{gprime_iminus1}) hold only if
  \begin{eqnarray}
    u+r_{i-1}-B_{i-1}-1 \geqslant 0.
    \label{ineq:upper_index_nonnegative}
  \end{eqnarray}
  $u=\mu_{i-1}$ is a state which is reachable in the original system
  as well as in the primed system and satisfies
  (\ref{ineq:upper_index_nonnegative}).  Since $g_i(u)$ is an
  increasing sequence in $u$, (\ref{g_iminus1}) and
  (\ref{gprime_iminus1}) imply $g_{i-1}'(\mu_{i-1}) >
  g_{i-1}(\mu_{i-1})$. Consequently, $g_0'(0)>g_0(0)$.

  Case \ref{non_existent_stage}: If no such $i$ exists, then we must
  have $B_i \geqslant r_0 + \cdots + r_i$ $\forall$ $i=0,\ldots,S-2$.
  Adding these $(S-1)$ inequalities and using $r_i \leqslant B$ (from
  (\ref{eq:restriction_tokens})),
  \begin{eqnarray}
    \sum_{i=0}^{S-2}B_i 
    &\geqslant& (Sr-r_{S-1}) + (Sr-r_{S-1}-r_{S-2}) 
    + (Sr-r_{S-1}-r_{S-2}-r_{S-3}) + \cdots, \nonumber \\
    &\geqslant& (Sr-B)+(Sr-2B)+(Sr-3B) + \cdots, \label{ineq:second_last} \\
    &=& S(S-1)r - \alpha B. \label{eq:bound_aggregate_bucket}
  \end{eqnarray}
  We cannot have $r_i = B$ $\forall$ $i$ (from
  (\ref{restrictions_bucket_depth}), (\ref{eq:restriction_tokens}) and
  (\ref{eq:token_sum})).  Thus, $\alpha$ cannot be of the order of
  $S^2$.  Thus, the lower bound on $\sum_{i=0}^{S-2} B_i$ given by
  (\ref{ineq:second_last}) and (\ref{eq:bound_aggregate_bucket}) is a
  loose lower bound.  From (\ref{restrictions_bucket_depth}),
  (\ref{ineq:burst_sizes}) and (\ref{eq:bound_aggregate_bucket}),
  $\sum_{i=0}^{S-2} B_i$ grows as $S^2$ and is upper-bounded by
  $5(S-1)r$, which is impossible (except when $S$ is small).  So, we
  discard Case \ref{non_existent_stage}.

  From the result of Case \ref{exists_stage}, ${H_0'}^*(0) >
  H_0^*(0)$.  So, our assumption that $\mathcal R_g(\cdot)$ is an
  optimal GTBR is incorrect. Therefore, equality must hold in
  (\ref{ineq:burst_sizes}) for every optimal GTBR.
\end{proof}

\subsection{Numerical Results}

\begin{table}[hbtp]
\small
\begin{tabular}{|c|c|c|c|c|c|} \hline
STBR & optimal token increment & optimal bucket depth & 
  \multicolumn{3}{c|}{information utility} \\ \cline{4-6}
parameters &  sequence of GTBR & sequence of GTBR & $H_s$ & $H_g^*$ & 
  percentage \\
($S$,$r$,$B$) & $\mathbf r^*$ & $\mathbf B^*$ & (bits) & (bits) & 
  increase \\ \hline
(4,3,6) & (6 3 3 0) & (6 6 6) & 20.04  & 20.92  & 4.4\% \\ \hline
(4,3,7) & (6 4 2 0) & (6 8 7) & 20.08  & 21.16  & 5.4\% \\ \hline
(4,3,8) & (7 3 2 0) & (7 9 8) &   &   &  \\
 & (8 3 1 0) & (8 9 7) & 20.10  & 21.32  & 6.1\% \\ \hline
(4,3,9) & (8 3 1 0) & (8 10 9) &   &   &  \\
 & (9 2 1 0) & (9 10 8) & 20.10  & 21.44  & 6.7\% \\ \hline
(4,3,10) & (9 3 0 0) & (9 12 9) & 20.10  & 21.51  & 7.0\% \\ \hline
(4,3,11) & (10 2 0 0) & (10 12 11) &   &   &  \\
 & (11 1 0 0) & (11 12 10) & 20.10  & 21.54  & 7.2\% \\ \hline
(4,3,12) & (12 0 0 0) & (12 12 12) & 20.10  & 21.56   & 7.2\% \\ \hline
(4,3,13) & (12 0 0 0) & (13 13 13) & 20.10  & 21.56  & 7.2\% \\ \hline
(4,4,8) & (8 4 4 0) & (8 8 8) & 25.08  & 26.04  & 3.8\% \\ \hline
(4,4,9) & (8 5 3 0) & (8 10 9) &   &   &  \\
 & (9 4 3 0) & (9 10 8) & 25.11  & 26.24  & 4.5\% \\ \hline
(4,4,10) & (9 5 2 0) & (9 12 9) & 25.13  & 26.39  & 5.0\% \\ \hline
(4,4,12) & (11 4 1 0) & (11 14 11) & 25.14  & 26.59  & 5.8\% \\ \hline
(4,4,16) & (16 0 0 0) & (16 16 16) & 25.14  & 26.70  & 6.2\% \\ \hline
(4,5,10) & (10 5 5 0) & (10 10 10) & 29.91  & 30.92  & 3.4\% \\ \hline
(4,5,12) & (11 6 3 0) & (11 14 11) & 29.96  & 31.24  & 4.3\% \\ \hline
(4,6,12) & (11 7 6 0) & (11 13 12) &  &  &  \\
 & (12 7 5 0) & (12 13 11) & 34.60  & 35.66  & 3.1\% \\ \hline
(5,3,6) & (6 3 3 3 0) & (6 6 6 6) & 25.68  & 26.57  & 3.5\% \\ \hline
(5,3,9) & (8 3 3 1 0) & (8 10 10 8) & 25.88  & 27.33  & 5.6\% \\ \hline
(5,3,12) & (11 2 2 0 0) & (11 13 13 11) & 25.90  & 27.59  & 6.5\% \\ \hline
(5,3,15) & (15 0 0 0 0) & (15 15 15 15) & 25.90  & 27.64  & 6.7\% \\ \hline
(6,2,4) & (4 2 2 2 2 0) & (4 4 4 4 4) & 23.00  & 23.77  & 3.4\% \\ \hline
(6,3,6) & (6 3 3 3 3 0) & (6 6 6 6 6) & 31.33  & 32.23  & 2.9\% \\ \hline
\end{tabular}
\caption{Entropy-maximizing GTBR for given data transmission time,
token rate and bucket depth of a comparable STBR.}
\label{tab:optimal_gtbr}
\end{table}

\begin{figure}[thbp]
\centering
\includegraphics[width=5in]{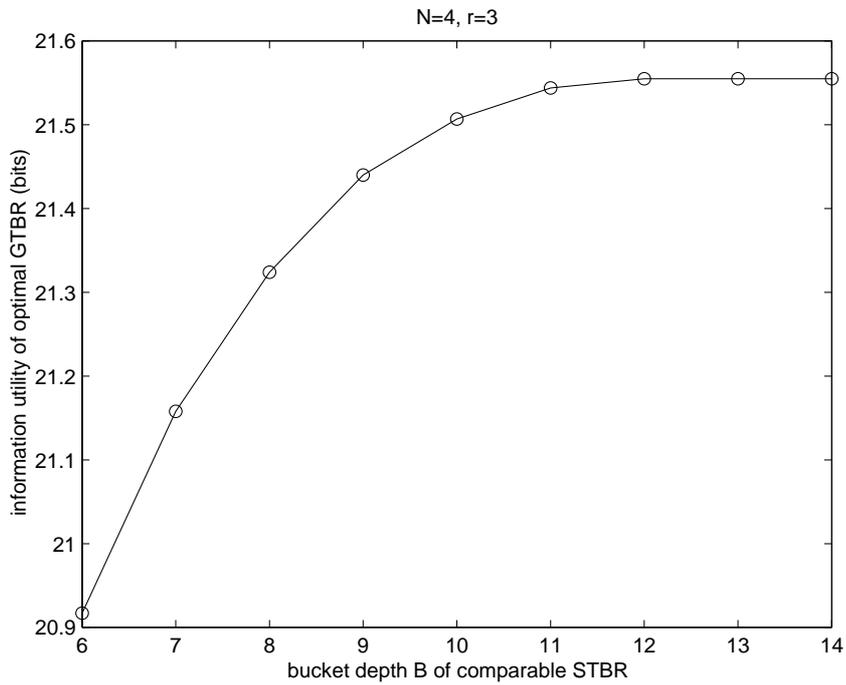}
\caption{Information utility of GTBR vs. bucket depth of comparable STBR.}
\label{fig:flow_entropy_vs_B}
\end{figure}

\begin{figure}[thbp]
\centering
\includegraphics[width=5in]{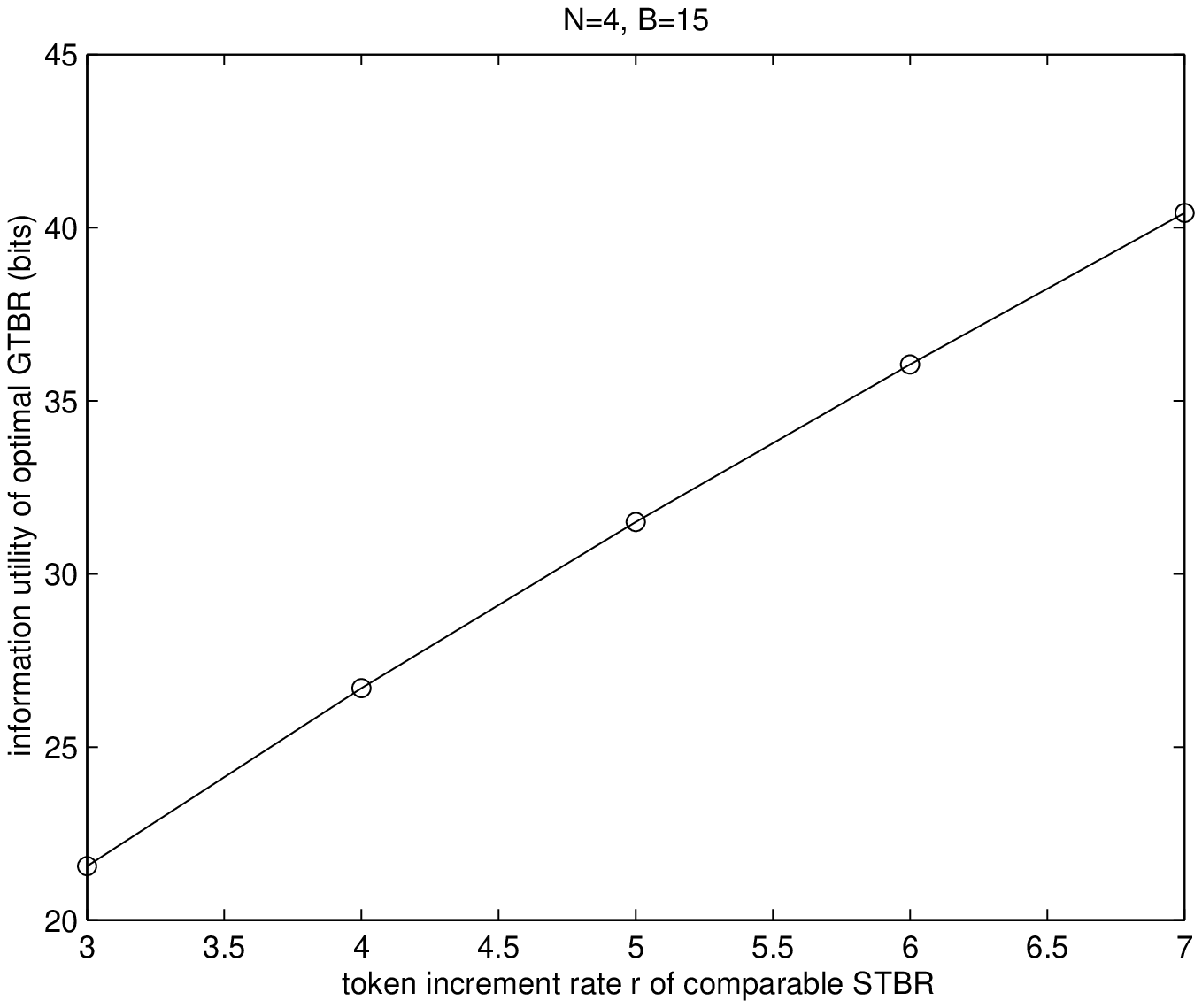}
\caption{Information utility of GTBR vs. token increment rate of comparable STBR.}
\label{fig:flow_entropy_vs_r}
\end{figure}

For a given data transmission time $S$, token increment sequence
$\mathbf r$ and bucket depth sequence $\mathbf B$, we determine the
optimal GTBR by exhaustive search over the reduced search space
obtained from Proposition \ref{prop:equality_burst_sizes}.  Our
computation results are shown in Table \ref{tab:optimal_gtbr}. $H_s$
and $H_g^*$ denote the information utility of the STBR $\mathcal
R_s(S,r,B)$ and the optimal GTBR $\mathcal R_g(S,\mathbf r^*,\mathbf
B^*)$ respectively.  We also observe the variation in information
utility of the optimal GTBR with important parameters of the
comparable STBR, namely its bucket depth $B$ and token increment rate
$r$.  For a data transmission time of 4 slots and token increment rate
of 3 bits, Figure \ref{fig:flow_entropy_vs_B} shows the variation of
information utility of the GTBR versus the bucket depth of the
comparable STBR.  For a data transmission time of 4 slots and bucket
depth of 15 bits, Figure \ref{fig:flow_entropy_vs_r} shows the
variation of information utility of the GTBR versus the token
increment rate of the comparable STBR.

Based on our computations, we draw the following inferences:
\begin{enumerate}

\item A generalized token bucket regulator can achieve {\it higher}
  information utility than that of a standard token bucket regulator.
  The increase in information utility is significant (up to 7.2\%),
  especially for higher values of $B$.

\item The optimal bucket depth sequence $\mathbf B^*$ is
  uniform\footnote{$B_0^* = B_1^* = \cdots = B_{S-2}^*$.} or
  near-uniform (the standard deviation is very small compared to the
  mean).

\item The optimal token increment sequence $\mathbf r^*$ is a
  decreasing sequence and is not uniform.

\item For a fixed data transmission time $S$ and token increment rate
  $r$:
  \begin{enumerate}
  \item If $B=2r$, ${\mathbf B^*}$ is always uniform and ${\mathbf
      r^*}$ is uniform except for the terminal values.

  \item As $B$ increases from $2r$ to $\min(5,S)r$, the variance of
    ${\mathbf r^*}$ increases rapidly with a concentration of tokens
    in first few stages, the variance of ${\mathbf B^*}$ increases
    slowly, while $H_g^*$ initially increases and then saturates at
    some final value.  $H_g^*$ is an increasing and concave
    sequence\footnote{The sequence of first-order differences
      $(B_1^*-B_0^*,B_2^*-B_1^*,\cdots,B_{S-2}^*-B_{S-3}^*)$ is a
      decreasing and non-negative sequence.} in $B$ (see Figure
    \ref{fig:flow_entropy_vs_B}).
    \label{obs:concave_B}
  \end{enumerate}

\item For a fixed data transmission time $S$ and bucket depth $B$,
  $H_g^*$ an increasing, highly linear and slightly concave sequence
  in $r$ (see Figure \ref{fig:flow_entropy_vs_r}).
  \label{obs:linear_r}

\end{enumerate}

\section{Information-Theoretic Interpretation}
\label{sec:info_theoretic_interpret}

In this section, we provide explanations for empirical results in
Section \ref{sec:results_gtbr}. The explanations are intuitive and
rely on basic results from information theory.

Consider a system with $n$ states, where $p_i$ denotes the probability
of state $i$ and $\sum_{i=1}^n p_i =1$.  From classical information
theory, system entropy $H$ increases with decreasing Kullback-Leibler
distance between the given probability mass function (pmf) and the
uniform pmf \cite{cover_thomas__elements_information}.  $H$ is
maximized only if $p_1=\cdots=p_n=\frac{1}{n}$. Also, $H^*$ increases
with $n$.  Analogously, a GTBR can achieve higher information utility
than that of an STBR because the pmfs of the packet lengths at each
stage are more uniform and have a larger support.  Recall that, for
given $\mathbf r$ and $\mathbf B$, information utility is computed
recursively using (\ref{eq:token_evolution}) and
(\ref{eq:optimal_flow_entropy}).

We argue that the optimal bucket depth sequence $\mathbf B^*$ must be
uniform or near-uniform for maximum information utility. If $\mathbf
B^*$ is neither uniform nor near uniform, then $B_j = \min_i B_i$ is
much smaller than $B$. This restricts the range of values taken by
$u_{j+1}$ and $\ell_{j+1}$ (from (\ref{eq:conforming_packet_lengths})
and (\ref{eq:token_evolution})). The support of packet length pmfs at
stage $j+1$ is reduced, leading to lower flow entropy at stage $j+1$
and consequently lower information utility.  Thus, ${\mathbf B^*}$
must be uniform or near-uniform to maximize the minimum support of
packet length pmfs {\it at each stage}.  In Table
\ref{tab:optimal_gtbr}, the observation that $\min_{i}B_i^* = B-1$ or
$\min_{i}B_i^* = B$ throughout corroborates our claim that $\mathbf
B^*$ is near-uniform.

We argue that for maximum information utility, the optimal token
increment sequence $\mathbf r^*$ must be a decreasing sequence,
subject to $r_i \leqslant B_i$ for every $i$.  If $r_i > B_i$ for any
$i$, then a packet of length zero cannot be transmitted in slot $i$
(from (\ref{eq:token_evolution})) and will have zero probability. This
decreases the support of the packet length pmfs in slot $i$ and leads
to lower information utility. More importantly, from
(\ref{eq:flow_entropy}),
\begin{eqnarray}
  H_0^*(0) 
  &=& \sum_{\ell_0=0}^{r_0} p_{\ell_0}^*(0) \Big( \ell_0 - 
  \log_2\big(p_{\ell_0}^*(0)\big) + H_1^*\big(\min(r_0-\ell_0,B_0)\big) \Big).
\end{eqnarray}
The major contribution to information utility $H_0^*(0)$ is from the
support of the packet lengths $[0,r_0]$ and the pmf of the packet
lengths $(\mathbf p_0^*(0))$, while the contribution from
$H_1^*(\cdot)$ is relatively smaller.  So, to maximize $H_0^*(0)$,
$r_0$ should be allowed to take its maximum possible value, subject to
$r_0 \leqslant B_0$, and the pmf of the packet lengths should be close
to the uniform pmf.  The observation that $r_0=B_0$ consistently in
Table \ref{tab:optimal_gtbr} corroborates this.  Also, a high value of
$r_0$ leads to larger supports of packet length pmfs at intermediate
and later stages. Similarly, the first few elements of $\mathbf r^*$
tend to take large values till the aggregate tokens are exhausted.
However, their contribution to $H_0^*(0)$ is not as pronounced and
equality may not hold in $r_i \leqslant B_i$. Thus, $\mathbf r^*$ must
be a decreasing sequence and the first few elements of $\mathbf r^*$
tend to take their maximum possible values, subject to $r_i \leqslant
B_i$, to achieve uniformity and larger supports of packet length pmfs
{\it at intermediate and later stages}.

This ``greedy'' nature of $\mathbf r^*$ is evident when $S$ and $r$
are kept constant and $B$ increases (Result \ref{obs:concave_B}). A
similar argument is applicable when $S$ and $B$ are kept constant and
$r$ increases (Result \ref{obs:linear_r}).  The only difference is
that a unit increase in $r$ will necessarily increase $H_g^*$ by at
least $S$ bits ($S$ bits are contributed by the packet contents alone,
which also explains the dominant linear variation in Figure
\ref{fig:flow_entropy_vs_r}), while a unit increase in $B$ will
increase $H_g^*$ only by an amount equal to the difference in covert
information.  The increase in covert information is positive only if
the optimal token increment and bucket depth sequences (${\mathbf
  r^*}$ and ${\mathbf B^*}$) result in larger support and more
uniformity for the packet length pmfs.  Indeed, when $B$ increases
beyond the maximum number of tokens possible at any stage
($\max_i\{\mu_i\}$), clamping the residual number of tokens at every
stage becomes ineffective and the system behaves as if bucket depth
constraints were not imposed at all (Figure
\ref{fig:flow_entropy_vs_B}).

\section{Discussion}
\label{sec:discussion_info_theoretic}

In this chapter, we have studied linearly bounded flows over a packet
network.  We considered a source whose traffic is regulated by a
generalized token bucket regulator and which seeks to maximize the
entropy of the resulting flow.  Recognizing that the randomness in
packet lengths acts as a covert channel in the network, the source can
achieve maximum entropy by sizing its packets appropriately.  We have
formulated the problem of computing the GTBR with maximum information
utility in terms of constrained token increment and bucket depth
sequences. A GTBR can achieve higher information utility than that of
a standard IETF token bucket regulator.  Finally, we have
information-theoretically interpreted the observation that an
entropy-maximizing GTBR always has a near-uniform bucket depth
sequence and a decreasing token increment sequence.

\clearpage{\pagestyle{empty}\cleardoublepage}

\chapter{Conclusions}
\label{ch:conclusions}

The recent revolution in wireless communications has motivated
researchers and engineers alike to design ever better wireless
networks that deliver high data rates to users. The {\em joint design}
of physical and MAC layers is the key to breaking the ``bandwidth
bottleneck'' of wireless networks, which has been the primary
inspiration for this thesis.

This thesis has focused on link scheduling in wireless mesh networks
by taking into account physical layer characteristics.  The assumption
made throughout this thesis is that a packet is received successfully
only if the SINR at the receiver exceeds a certain threshold, termed
as communication threshold.  The thesis has also discussed the
complementary problem of flow control.

The first part of this thesis has considered link scheduling in STDMA
wireless networks. The network is modeled by a finite set of
store-and-forward nodes that communicate over a wireless channel
characterized by propagation path loss. We have considered two nuances
of the scheduling problem: point to point link scheduling wherein a
transmitted packet is intended for a single neighbor only, and point
to multipoint link scheduling wherein a transmitted packet is intended
for all neighbors in the vicinity.

Specifically, in Chapter \ref{ch:framework_link}, we have introduced
the system model of an STDMA wireless network.  We have discussed two
prevalent models for specifying the criteria for successful packet
reception: the protocol interference model which mandates a ``silence
zone'' around a receiver and is better suited to represent WLANs, and
the physical interference model which mandates that the SINR at a
receiver be no less than the communication threshold and is more
appropriate to represent mesh networks.  We have described the
equivalence between a link schedule and the coloring of edges of a
certain graph representation of the network, termed as communication
graph.  We have argued that STDMA link scheduling algorithms can be
broadly categorized into three classes: those based entirely on a
communication graph representation of the network, those based on
communication graph and SINR threshold conditions and those based
entirely on an SINR graph representation of the network.  We have
reviewed representative research papers from each of these classes.
We have described limitations of algorithms that are based only on the
communication graph. Subsequently, we have introduced spatial reuse as
a performance metric that corresponds to aggregate network throughput.

Next, in Chapter \ref{ch:comm_graph}, we have critically examined
ArboricalLinkSchedule, a point to point link scheduling algorithm
proposed in \cite{ramanathan_lloyd__scheduling_algorithms}.  While
this is one of the earlier works on link scheduling with nice
theoretical properties, it could yield a higher schedule length in
practice. Specifically, the methodology employed by
ArboricalLinkSchedule is to represent the network by a communication
graph, partition the graph into minimum number of subgraphs and color
each subgraph in a greedy manner.  We have modified the algorithm to
reuse colors while coloring successive subgraphs of the communication
graph.  We have shown that the modified algorithm yields lower
schedule length in practice, albeit at a cost of slightly higher
running time complexity. Subsequently, we have proposed the
ConflictFreeLinkSchedule algorithm that not only utilizes the
communication graph, but also verifies SINR threshold conditions at
receivers. We have demonstrated that the proposed algorithm achieves
higher spatial reuse than existing algorithms, even under fading and
shadowing channel conditions. We have argued that the running time
complexity of the proposed algorithm is only marginally higher than
those of existing algorithms.

Taking a step ahead, in Chapter \ref{ch:sinr_graph}, we have provided
a somewhat different perspective on point to point link scheduling.
For an STDMA network, we recognize that interferences between pairs of
links can be embedded into edge weights and normalized noise powers at
receivers of links can be embedded into vertex weights of a certain
graph representation of the network, termed as SINR graph.  We have
then proposed SINRGraphLinkSchedule, a novel link scheduling algorithm
that is based on the SINR graph. We have proved the correctness of the
algorithm and shown that it has polynomial running time complexity. We
have demonstrated that the proposed algorithm achieves high spatial
reuse compared to algorithms which utilize a communication graph model
of the network, including ConflictFreeLinkSchedule algorithm.

In Chapter \ref{ch:broadcastschedule}, we have considered point to
multipoint link scheduling and generalized the definition of spatial
reuse for this scenario. We have proposed a scheduling algorithm based
on a communication graph representation of the network and
``neighbor-average'' SINR threshold conditions. Moreover, we have
demonstrated that the proposed algorithm achieves higher spatial reuse
than existing algorithms, without any increase in running time
complexity.

Overall, we have observed the tradeoff between accuracy of the network
representation, spatial reuse and algorithm running time complexity in
our successive results. For a more accurate network representation,
higher spatial reuse is achieved, but at a cost of higher running time
complexity. For example, since the SINR graph representation of an
STDMA network is more accurate than the communication graph
representation, SINRGraphLinkSchedule achieves higher spatial reuse
than that of ConflictFreeLinkSchedule, but at a cost of increased
running time complexity.

A summary of existing and proposed link scheduling algorithms
investigated in the first part of the thesis is provided in Table
\ref{tab:summary_algos}.

\begin{table}[hbtp]
\footnotesize
\begin{tabular}{|c|c|c|c|} \hline
Type of link  & Wireless network  & Existing  & Proposed  \\ 
scheduling & model & algorithms & algorithm \\ \hline
 & communication graph 
 & ArboricalLinkSchedule \cite{ramanathan_lloyd__scheduling_algorithms} 
 & ALSReuseColors \\ 
 & & & (Chapter \ref{ch:comm_graph}) \\ \cline{2-4}
Point & communication graph 
 & GreedyPhysical \cite{brar_blough_santi__computationally_efficient} &  \\ 
 to  & and  & TGSA \cite{behzad_rubin__performance_graph} 
 & ConflictFreeLinkSchedule \\
point & SINR conditions & & (Chapter \ref{ch:comm_graph}) \\ \cline{2-4}
 & SINR graph & & SINRGraphLinkSchedule \\ 
 & & & (Chapter \ref{ch:sinr_graph}) \\ \hline
Point  & communication graph 
 & BroadcastSchedule \cite{ramanathan_lloyd__scheduling_algorithms}
 & \\ \cline{2-4}
to & communication graph & & \\
multipoint & and & & MaxAverageSINRSchedule \\
 & SINR conditions & & (Chapter \ref{ch:broadcastschedule}) \\ \hline
\end{tabular}
\caption{Link scheduling algorithms investigated in Chapters
\ref{ch:comm_graph}, \ref{ch:sinr_graph} and \ref{ch:broadcastschedule}.}
\label{tab:summary_algos}
\end{table}

The second part of this thesis has considered link scheduling in
random access wireless networks.  Specifically, it has focused on
random access algorithms for wireless networks that take into account
channel effects and SINR conditions at the receiver.

In Chapter \ref{ch:review_random}, we have reviewed representative
research papers on such random access techniques. We have also
motivated the use of variable transmission power in random access
wireless networks.

Subsequently, in Chapter \ref{ch:powercontrolled}, we have
investigated a random access scenario wherein multiple transmitters
(users) attempt to communicate with a single receiver over a wireless
channel characterized by propagation path loss. We have assumed that
the receiver is capable of power based capture and proposed an
interval splitting algorithm that varies transmission powers of users
based on their arrival times and quaternary channel feedback. We have
modeled the algorithm dynamics by a Discrete Time Markov Chain and
consequently shown that its maximum stable throughput is 0.5518.  We
have demonstrated that the proposed algorithm has higher throughput
and lower delay than the FCFS interval splitting algorithm with
uniform transmission power.

The third and final part of this thesis has considered
information-theoretic analysis of flow control in packet networks. We
have defined the problem of maximizing the information carried by
packets from a source to a destination, subject to a flow control
mechanism at the ingress of the network. We have considered a linearly
bounded flow and focused on the information carried by the randomness
in packet contents and lengths.  Consequently, we have formulated the
problem of maximizing the entropy of a packet level flow that is
shaped by a generalized token bucket regulator.  We have demonstrated
that the optimal regulator has a decreasing token increment sequence
and a near-uniform bucket depth sequence.  Finally, we have provided
information theoretic interpretations for these observations.

To sum it up, in this thesis, we have investigated both fixed and
random access flavors of link scheduling problems in wireless networks
from a physical layer viewpoint.  Finally, we have discussed a flow
control problem in packet networks.

Various avenues for further research have emerged from our
investigations. We outline some possible directions for future work.

\begin{enumerate}

\item It would be interesting to derive approximation bounds of
  ConflictFreeLinkSchedule and SINRGraphLinkSchedule algorithms under
  reasonable assumptions on node deployment and interference regions.
  The assumptions and approximation techniques employed in
  \cite{brar_blough_santi__computationally_efficient} may provide some
  pointers in this direction.

\item Though distributed link scheduling algorithms for STDMA wireless
  networks under the protocol interference model have been proposed in
  \cite{salonidis_tassiulas__distributed_dynamic},
  \cite{djukic_valaee__distributed_link}, the design of distributed
  link scheduling algorithms under the physical interference model
  remains a challenging problem.

\item Various generalizations of the PCFCFS algorithm are worth
  investigating.  For example:
  \begin{enumerate}
  \item Design a variable power splitting algorithm under the
    assumption that users are at unequal distances from the receiver
    and can adjust their minimum transmission powers accordingly.
  \item Design a splitting algorithm for the case when the receiver is
    capable of decoding more than one packet correctly (as in wideband
    systems) and the users can employ $n$ transmission power levels,
    where $n>2$.
  \item Analyze the throughput improvement in CSMA/CA based WLANs when
    power control is employed in conjunction with binary exponential
    backoff. The work done in \cite{colbourn_cui__carrier_sense} can
    be a useful starting point.
  \end{enumerate}

\item A challenging task would be to analyze the expected delay of the
  PCFCFS algorithm.  A useful starting would be
  \cite{huang_berger__delay_analysis_1},
  \cite{huang_berger__delay_analysis_2}, which have employed
  techniques to obtain upper and lowers bounds on the expected delay
  of the FCFS algorithm.

\item Our results in Chapter \ref{ch:flow_control} show the existence
  of upper bounds on the entropy of regulated flows. It would be
  interesting to construct source codes which come close to this
  bound. Furthermore, it would be insightful to develop a
  rate-distortion framework for a generalized token bucket regulator,
  perhaps using the techniques employed in
  \cite{chang_chao_thomas__fundamental_limits}.

\end{enumerate}

\clearpage{\pagestyle{empty}\cleardoublepage}

\renewcommand{\chaptermark}[1] {\markboth{{\appendixname\ \thechapter.\ #1}}{}}
\renewcommand{\sectionmark}[1] {\markright{{\thesection.\ #1}}}
\lhead[\small \thepage]{\small \rightmark}
\rhead[\small \leftmark]{\small \thepage}
\lhead[\small \thepage]{\small \rightmark}
\rhead[\small \leftmark]{\small \thepage}

\begin{appendix}
\chapter{Proofs  of Limiting Transition Probabilities}
\label{ap:limiting_probabilities}

According to L'H\^{o}pital's Rule, if $\lim_{x \rightarrow c} f(x)$
and $\lim_{x \rightarrow c} g(x)$ are both zero or are both $\pm
\infty$ and, if $\lim_{x \rightarrow c} \frac{f(x)}{g(x)}$ has a
finite value or if the limit is $\pm \infty$, then
\begin{eqnarray}
\lim_{x \rightarrow c} \frac{f(x)}{g(x)} 
  &=& \lim_{x \rightarrow c} \frac{f'(x)}{g'(x)}.
\end{eqnarray}
We will employ L'H\^{o}pital's Rule to prove
(\ref{eq:limit_Liprime_Riprime}) - (\ref{eq:limit_Ri_Liplus1})

In this appendix, we will only provide the proofs of
(\ref{eq:limit_Liprime_Riprime}),
(\ref{eq:limit_Liprime_Liplus1prime}),
(\ref{eq:limit_Liprime_Ciplus1}) and (\ref{eq:limit_Liprime_Liplus1}).
The proofs of (\ref{eq:limit_Riprime_R0}) -
(\ref{eq:limit_Ri_Liplus1}) are similar to those of
(\ref{eq:limit_Liprime_Riprime}) - (\ref{eq:limit_Liprime_Liplus1})
and are omitted.

\section{Proof of (\ref{eq:limit_Liprime_Riprime})}

\begin{proof}
In (\ref{eq:probability_Liprime_Riprime}), substitute $G_i=x$.
From (\ref{eq:expected_half}), $G_{i-1}=2G_i=2x$.
As $i \rightarrow \infty$, $G_i=2^{-i}\lambda\phi_0 \rightarrow 0$.
Thus, using L'H\^{o}pital's Rule successively, we obtain

\begin{eqnarray*}
\lim_{i \rightarrow \infty } P_{L_i',R_i'}
  &=& \lim_{x \rightarrow 0} \frac{(1-e^{-x})xe^{-x}}{1-(1+2x)e^{-2x}}, \\
  &=& \lim_{x \rightarrow 0} \frac{\frac{d}{dx}(xe^{-x}-xe^{-2x})}{\frac{d}{dx}(1-e^{-2x}-2xe^{-2x})}, \\
  &=& \lim_{x \rightarrow 0} \frac{e^{-x}+xe^{-x}-e^{-2x}}{4xe^{-2x}}, \\
  &=& \lim_{x \rightarrow 0} \frac{\frac{d}{dx}(e^{-x}+xe^{-x}-e^{-2x})}{\frac{d}{dx}(4xe^{-2x})}, \\
  &=& \lim_{x \rightarrow 0} \frac{2e^{-2x}-xe^{-x}}{4e^{-2x}-8xe^{-2x}}, \\
\therefore \lim_{i \rightarrow \infty } P_{L_i',R_i'} &=& \frac{1}{2}.
\end{eqnarray*}
\end{proof}

\section{Proof of (\ref{eq:limit_Liprime_Liplus1prime})}

\begin{proof}
In (\ref{eq:probability_Liprime_Liplus1prime}), substitute $G_i=x$.
Thus, using $G_{i-1}=2x$ and applying L'H\^{o}pital's Rule
successively, we obtain

\begin{eqnarray*}
\lim_{i \rightarrow \infty } P_{L_i',L_{i+1}'} 
  &=& \lim_{x \rightarrow 0} \frac{(1-e^{-x}-xe^{-x})e^{-x}}{1-(1+2x)e^{-2x}},\\
  &=& \lim_{x \rightarrow 0} \frac{\frac{d}{dx} (e^{-x}-e^{-2x}-xe^{-2x})}{\frac{d}{dx}(1-e^{-2x}-2xe^{-2x})}, \\
  &=& \lim_{x \rightarrow 0} \frac{e^{-2x}+2xe^{-2x}-e^{-x}}{4xe^{-2x}}, \\
  &=& \lim_{x \rightarrow 0} \frac{\frac{d}{dx}(e^{-2x}+2xe^{-2x}-e^{-x})}{\frac{d}{dx}(4xe^{-2x})}, \\
  &=& \lim_{x \rightarrow 0} \frac{e^{-x}-4xe^{-2x}}{4e^{-2x}-8xe^{-2x}}, \\
\therefore \lim_{i \rightarrow \infty } P_{L_i',L_{i+1}'} &=& \frac{1}{4}.
\end{eqnarray*}
\end{proof}

\section{Proof of (\ref{eq:limit_Liprime_Ciplus1})}

\begin{proof}
In (\ref{eq:probability_Liprime_Ciplus1}), substitute $G_i=x$.
Thus, using $G_{i-1}=2x$ and applying L'H\^{o}pital's Rule
successively, we obtain

\begin{eqnarray*}
\lim_{i \rightarrow \infty } P_{L_i',C_{i+1}} 
  &=& \lim_{x \rightarrow 0} \frac{\frac{x^2}{4}e^{-x}}{1-(1+2x)e^{-2x}}, \\
  &=& \lim_{x \rightarrow 0} \frac{\frac{d}{dx}(\frac{x^2}{4}e^{-x})}{\frac{d}{dx}(1-e^{-2x}-2xe^{-2x})}, \\
  &=& \lim_{x \rightarrow 0} \frac{\frac{1}{2}e^{-x}-\frac{x}{4}e^{-x}}{4e^{-2x}}, \\
\therefore \lim_{i \rightarrow \infty } P_{L_i',C_{i+1}} &=& \frac{1}{8}.
\end{eqnarray*}
\end{proof}

\section{Proof of (\ref{eq:limit_Liprime_Liplus1})}

\begin{proof}
In (\ref{eq:probability_Liprime_Liplus1}), substitute $G_i=x$.
Thus, using $G_{i-1}=2x$ and applying L'H\^{o}pital's Rule
successively, we obtain

\begin{eqnarray*}
\lim_{i \rightarrow \infty } P_{L_i',L_{i+1}}
  &=& \lim_{x \rightarrow 0} \frac{1-(1+x+\frac{x^2}{4})e^{-x}}{1-(1+2x)e^{-2x}}, \\
  &=& \lim_{x \rightarrow 0} \frac{\frac{d}{dx}(1-e^{-x}-xe^{-x}-\frac{x^2}{4}e^{-x})}{\frac{d}{dx}(1-e^{-2x}-2xe^{-2x})}, \\
  &=& \lim_{x \rightarrow 0} \frac{\frac{1}{2}-\frac{x}{2}}{4e^{-x}}, \\
\therefore \lim_{i \rightarrow \infty } P_{L_i',L_{i+1}} &=& \frac{1}{8}.
\end{eqnarray*}
\end{proof}

\clearpage{\pagestyle{empty}\cleardoublepage}
\end{appendix}

\bibliographystyle{ieeetr}
\bibliography{IEEEfullform,thesis}

\chapter*{Publications from the Thesis}

\begin{enumerate}

\item A.D. Gore and A. Karandikar, ``Entropy-Optimal Generalized Token
  Bucket Regulator,'' in Proc. $12^{th}$ National Conference on
  Communications - NCC 2006, Indian Institute of Technology - Delhi,
  Joint Telematics Group of IITs/IISc, pp. 344-348, Jan. 2006.

\item A.D. Gore and A. Karandikar, ``On High Spatial Reuse Broadcast
  Scheduling in STDMA Wireless Ad Hoc Networks,'' in Proc. $13^{th}$
  National Conference on Communications - NCC 2007, Indian Institute
  of Technology - Kanpur, Joint Telematics Group of IITs/IISc, pp.
  74-78, Jan. 2007.

\item A.D. Gore and A. Karandikar, ``Power Controlled FCFS Splitting
  for Wireless Networks,'' in Proc. IEEE MILCOM, Orlando, Florida,
  USA, Oct.  2007.

\item A.D. Gore, S. Jagabathula and A. Karandikar, ``On High Spatial
  Reuse Link Scheduling in STDMA Wireless Ad Hoc Networks,'' in Proc.
  $50^{th}$ IEEE GLOBECOM, Washington, DC, USA, Nov. 2007.

\item N.P. Kumar, A.D. Gore and A. Karandikar, ``Link Scheduling in
  STDMA Wireless Networks: A Line Graph Approach,'' in Proc.
  $14^{th}$ National Conference on Communications - NCC 2008, Indian
  Institute of Technology - Bombay, Joint Telematics Group of
  IITs/IISc, pp. 108-111, Feb. 2008.

\end{enumerate}

\end{document}